\documentclass[11pt, openany]{book}
\pdfoutput=1	% force arXiv to use pdflatex

\usepackage[letterpaper, top=1in, bottom=1in, left=1in, right=1in]{geometry}
\usepackage{indentfirst}
\usepackage[T1]{fontenc}
\usepackage{lmodern}	% scalable fonts (needed for microtype expansion)

\usepackage{titlesec}
\usepackage{titletoc}
\titleformat{\chapter}[display]
    {\fontsize{16pt}{16pt}\bfseries}
    {\chaptertitlename\ \thechapter}{16pt}{}
\titleformat{\section}
  {\fontsize{15pt}{15pt}\bfseries}
  {\thesection}{1em}{}
\titlecontents{chapter}% <section-type>
  [0pt]% <left>
  {\addvspace{10pt}%
  \bfseries}% <above-code>
  {\chaptername\ \thecontentslabel : \,}% <numbered-entry-format>
  {}% <numberless-entry-format>
  {\hspace{\fill} \hspace{-10mm} \contentspage}% <filler-page-format>
  
\usepackage{amsfonts, amssymb, amsmath, amsthm}
\usepackage{latexsym}
\usepackage[tracking=smallcaps]{microtype}	% for an arXiv submission, set \pdfoutput=1 near the top so it uses pdflatex
\usepackage{url}
\usepackage{color}
\definecolor{DarkGray}{rgb}{0.1,0.1,0.5}
\definecolor{curgray}{rgb}{0.5,0.5,0.5}
\definecolor{curblue}{rgb}{0.04,0.11,0.64}
\definecolor{curpurple}{rgb}{0.65,0.16,0.58}
\definecolor{curorange}{rgb}{1,0.32,0}
\usepackage[colorlinks=true,breaklinks, linkcolor=DarkGray,citecolor=DarkGray,urlcolor=DarkGray, hyperfootnotes=false]{hyperref}	% linkcolor=blue,citecolor=blue,urlcolor=blue %DarkGray
\usepackage{bookmark}

\usepackage[nottoc]{tocbibind}

\usepackage{graphicx}
\usepackage{subcaption}	%% incompatible with revtex
\usepackage{multirow}

\newcommand{\ket}[1]{{|#1\rangle}}

\newcommand{\abs}[1]{{\lvert #1\rvert}}	% since the delimiters do not scale, it might be a good idea to add a dummy {} at the end, so \abs{big expression}^2 has the superscript at a low height

\newcommand{\NP}{\ensuremath{\mathsf{NP}}}%{{\mathcal{NP}}}
\newcommand{\identity}{\ensuremath{\boldsymbol{1}}} %\mathbb{I}
\newcommand{\Id}{\identity} 
\DeclareMathOperator{\CNOT}{\operatorname{CNOT}}
\DeclareMathOperator{\SWAP}{\operatorname{SWAP}}
\DeclareMathOperator{\CZ}{\operatorname{CZ}}

\DeclareMathOperator{\pa}{\operatorname{parent}}
\DeclareMathOperator{\ro}{\operatorname{root}}
\DeclareMathOperator{\spam}{\operatorname{SPAM}}

\def\CZ {C\!Z}	% control-Z gate

\DeclareMathOperator{\depth}{\operatorname{depth}}

\newtheorem{theorem}{Theorem}%[section]
\newtheorem{lemma}{Lemma}

\newtheorem{claim}{Claim}

\newtheorem{condition}{Condition}
\newtheorem{procedure}{Procedure}

\newfont{\subsubsecfnt}{ptmri8t at 11pt}
\renewcommand{\subparagraph}[1]{\smallskip{\subsubsecfnt #1.}}

\newcommand{\eqnref}[1]{\hyperref[#1]{{(\ref*{#1})}}}
\newcommand{\thmref}[1]{\hyperref[#1]{{Theorem~\ref*{#1}}}}
\newcommand{\lemref}[1]{\hyperref[#1]{{Lemma~\ref*{#1}}}}
\newcommand{\corref}[1]{\hyperref[#1]{{Corollary~\ref*{#1}}}}
\newcommand{\defref}[1]{\hyperref[#1]{{Definition~\ref*{#1}}}}
\newcommand{\secref}[1]{\hyperref[#1]{{Sec.~\ref*{#1}}}}
\newcommand{\chapref}[1]{\hyperref[#1]{{Chap.~\ref*{#1}}}}
\newcommand{\figref}[1]{\hyperref[#1]{{Fig.~\ref*{#1}}}}  % for revtex use Fig. except at the start of a sentence
\newcommand{\figureref}[1]{\hyperref[#1]{{Figure~\ref*{#1}}}}  % for revtex use Fig. except at the start of a sentence
\newcommand{\tabref}[1]{\hyperref[#1]{{Table~\ref*{#1}}}}
\newcommand{\remref}[1]{\hyperref[#1]{{Remark~\ref*{#1}}}}
\newcommand{\appref}[1]{\hyperref[#1]{{Appendix~\ref*{#1}}}}
\newcommand{\claimref}[1]{\hyperref[#1]{{Claim~\ref*{#1}}}}
\newcommand{\factref}[1]{\hyperref[#1]{{Fact~\ref*{#1}}}}
\newcommand{\propref}[1]{\hyperref[#1]{{Proposition~\ref*{#1}}}}
\newcommand{\exampleref}[1]{\hyperref[#1]{{Example~\ref*{#1}}}}
\newcommand{\conjref}[1]{\hyperref[#1]{{Conjecture~\ref*{#1}}}}
\newcommand{\condref}[1]{\hyperref[#1]{{Condition~\ref*{#1}}}}
\newcommand{\procref}[1]{\hyperref[#1]{{Procedure~\ref*{#1}}}}
\newcommand{\expref}[1]{\hyperref[#1]{{Experiment~\ref*{#1}}}}

\allowdisplaybreaks[1]
\def\COLOR{}
\ifdefined\COLOR

\else

\fi

\definecolor{Cayenne}{rgb}{0.5,0,0}
\definecolor{Midnight}{rgb}{0,0,0.5}
\definecolor{Plum}{rgb}{0.5,0,0.5}
\definecolor{Teal}{rgb}{0,0.5,0.5}
\definecolor{Clover}{rgb}{0,0.5,0}
\definecolor{Maroon}{rgb}{0.5,0,0.25}
\definecolor{Ocean}{rgb}{0,0.25,0.5}
\definecolor{Tangerine}{rgb}{1,0.5,0}
\definecolor{Strawberry}{rgb}{1,0,0.5}
\definecolor{Fern}{rgb}{0.25,0.5,0}
\definecolor{Aqua}{rgb}{0,0.5,1}
\definecolor{Moss}{rgb}{0,0.5,0.25}
\definecolor{Mocha}{rgb}{0.5,0.25,0}
\definecolor{Lemon}{rgb}{1,1,0}
\definecolor{Asparagus}{rgb}{0.5,0.5,0}
\definecolor{Grape}{rgb}{0.5,0,1}
\definecolor{Iron}{rgb}{.3,.3,.3}
\definecolor{Steel}{rgb}{.4,.4,.4}
\definecolor{Purple}{rgb}{.5,0,.5}

\def\llbracket{{[\![}}
\def\rrbracket{{]\!]}}

\usepackage[utf8]{inputenc}

\usepackage{enumitem}
\usepackage{color}
\usepackage{dsfont}
\usepackage{amsmath}
\usepackage{amssymb,bm}
\usepackage{amsthm}
\usepackage{mathtools}
\usepackage{url}

\usepackage{booktabs}

\usepackage{array, makecell}
\usepackage{boldline}
\usepackage{enumitem}

\usepackage{color,dsfont}
\usepackage{graphicx}
\usepackage{ragged2e}
\usepackage[normalem]{ulem}
\usepackage[capitalise]{cleveref}
\usepackage{mathrsfs}
\usepackage{keyval}
\usepackage{mathtools}
\usepackage{float}
\usepackage{verbatim}
\usepackage{latexsym}
\usepackage{amsmath}
\usepackage{amssymb}
\usepackage{setspace}
\usepackage{amsfonts}
\usepackage{enumitem}
\usepackage{lipsum}
\setlist[enumerate]{label*=\arabic*.}

\definecolor{cardinal}{rgb}{0.827, 0, 0}

\begin{document}

\frontmatter
\begin{titlepage}
    \begin{center}
        \vspace*{2cm}
         \large
        LOWER OVERHEAD FAULT-TOLERANT BUILDING BLOCKS FOR NOISY QUANTUM COMPUTERS

        \normalsize
        \vspace{0.8cm}
        by\\
        \vspace{0.8cm}
        Prithviraj Prabhu

        \vfill
        
        \begin{singlespace}
        A Dissertation Presented to the\\
        FACULTY OF THE USC GRADUATE SCHOOL\\
        UNIVERSITY OF SOUTHERN CALIFORNIA\\
        In Partial Fulfillment of the\\
        Requirements for the Degree\\
        DOCTOR OF PHILOSOPHY\\
        (ELECTRICAL ENGINEERING)
        \end{singlespace}
        
        \vspace{2cm}

        May $2024$
            
        \vspace{2cm}
        
    \end{center}

    Copyright 2024
    \hfill
    Prithviraj Prabhu
        \vspace{1cm}
    
\end{titlepage}

\chapter{Acknowledgements}
\label{pref:Ack}

First and foremost, I would like to express my deepest gratitude to my supervisor, Dr. Benjamin Reichardt, for their invaluable encouragement, guidance and support throughout my PhD journey. Their expertise in the field of quantum error correction and fault tolerance has been crucial in shaping my thesis. Additionally, their influence over the past six years has been instrumental in shaping me into the researcher I am today. 

I would also like to thank the members of my qualifying exam and thesis defense committees, Dr. Todd Brun, Dr. Daniel Lidar, Dr. Eli Levenson-Falk and Dr. Keith Chugg. Their constructive criticism at crucial junctures of my PhD have been greatly appreciated. My deepest thanks to Dr. Brun, who served as my secondary advisor. Their constant support and supply of ideas kept me motivated and exploring new possibilities.

I am grateful to my colleagues in the Electrical Engineering department at USC for their stimulating discussions and collaborative spirit. Special thanks to Dr. Rui Chao, Dr. Sourav Kundu, Yuanjia Wang and Anirudh Lanka for their insightful conversations and constructive criticisms.

I would also like to thank the quantum computing team at Intel and Dr. Christopher Chamberland previously at Amazon for providing me the opportunity to pursue internships. My forays in the industry have certainly increased the value of this thesis.

Finally, I want to express my deepest appreciation to my family and friends for their unwavering love and support during this challenging but rewarding experience. I am grateful for the support of my parents, Govindaraj Prabhu and Sharmila Devi Prabhu, and my sister Prasidha.
\begin{singlespace}
\tableofcontents
\listoftables
\listoffigures
\end{singlespace}
% \chapter*{Abbreviations}
% \label{pref:abb}

\chapter{Abstract}
\label{sec:abs}

Quantum computation holds the promise of solving certain complex problems exponentially faster than classical computers. However, the high prevalent noise in current quantum devices impedes the accurate execution of even basic quantum algorithms. This can be remedied by protecting quantum information with a quantum error-correcting code, in which the logical information of an algorithmic qubit is spread across multiple physical qubits. Individual quantum errors are then located and corrected by the fault-tolerant measurement of multi-qubit stabilizer operators (parity checks). Unfortunately, error correction and fault tolerance both impose large demands on the qubit overhead: hundreds to thousands of physical qubits per logical qubit.

In this thesis, we reduce the qubit and time cost of fault tolerance by redesigning key building blocks of an error-corrected quantum computer. First, we develop a combinatorial proof with flag fault tolerance that exponentially reduces the number of qubits needed to measure a stabilizer of any size, while tolerating one fault. We then leverage the combinatorial proofs to develop fault-tolerant circuits to prepare cat states deterministically with only one ancillary qubit. These results then enable the construction of few-qubit fault-tolerant circuits for the preparation of complex encoded states with $100\%$ yield. Next, we optimize the overhead of error correction on a planar $25$-qubit layout. We show with extensive simulations that a distance-four code encoding six logical qubits protects information as well as the distance-five surface code, using one-tenth as many physical qubits. Finally, we optimize the time overhead of logical gates in surface code quantum computers. For computations executed via lattice surgery measurements of multi-qubit Pauli operators, we show that protecting measurement results with a classical code cuts computation time by a factor of two to six. Our hardware-agnostic optimizations of the space and time costs of fault tolerance thus suggest new routes to advance the timeline of error-free quantum computing.

\mainmatter
\chapter{Introduction}
\label{chap:prelim}

% In this chapter, we motivate the construction of a fault-tolerant quantum computer, introduce its building blocks, and describe the context in which they are redesigned. We start with a quick discussion on the necessity of quantum computing, before progressing to the challenges of building a quantum computer. We then show how models of fault tolerance adopted from classical computing can be applied to build robust quantum computers. After an initial wave of results proving that fault-tolerant quantum computers can compute with arbitrary accuracy, the goal has been to reduce the overhead. In \secref{sec:redesign}, we introduce the building blocks of fault-tolerant quantum computation, before delving into the context and results of the redesign of each building block. We end this chapter with a discussion of current thrusts in the field of quantum error correction, and highlight what potential future directions may be.

\section{What is a fault-tolerant quantum computer?} 

\begin{figure}
    \centering
    \includegraphics[width=\textwidth]{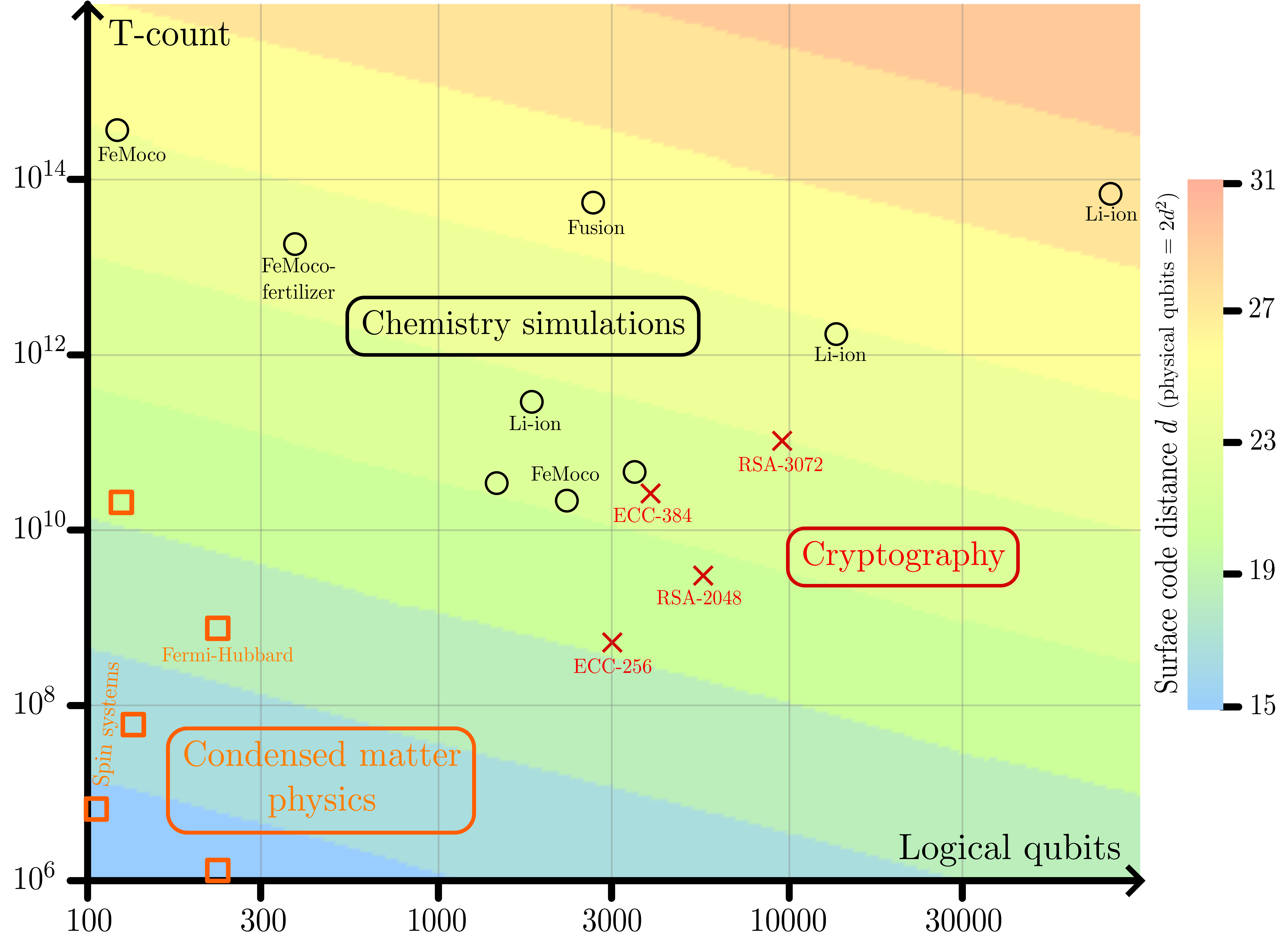}
    \vspace{-.6cm}
    \caption{Qubit and gate counts ($T$) needed for academic and commercial quantum use~\cite{scholten24}. Quantum algorithms decomposed into a sequence of Clifford and $T$ gates assume only the latter is computationally expensive. 
    To ensure reliable results for large and long computations, qubits are protected by error-correcting codes. When using the surface code, we show the minimum code distance needed to achieve an algorithmic accuracy of $1\%$, with a physical gate error rate of $10^{-3}$. We make modest assumptions about the routing space and ignore magic state distillation factories.
    }
    \label{fig:overheadgraph}
    \vspace{-.6cm}
\end{figure}

% Fundamental issues:
% 1) Maintaining scale below threshold (hardware, level 1)
% 2) Resource overheads (EC, Logical level 3, 4)
% 3) Control electronics (hardware level 1)

% Sources of noise that make prevent accurate quantum computation include hardware manufacturing imperfections, decoherence on idle qubits, coherent and crosstalk errors during one- and two-qubit gates, and other environmental sources such as cosmic rays.

% Some of the applications of a fully fault-tolerant quantum computer include molecular simulation for drug discovery and material design, solving combinatorial optimization problems.

The study of the complexity of classical algorithms informs us that some problems cannot be solved by computers in reasonable timeframes. Quantum computers can solve these problems quickly, and the construction of these devices will bring profound technological changes in the upcoming decades. Some practical applications for which quantum algorithms currently exist include molecular simulation for drug discovery and material design~\cite{Berry19,Lee21chemistry,Yuan21, rubin2023quantum}, algorithms that break public-key cryptography~\cite{Gidney2021rsa} and simulations of quantum physics~\cite{Babbush18, Childs18}. In \figref{fig:overheadgraph}, we show the number of qubits and computation depth needed for some of these algorithms. Quite surprisingly, public-key cryptography can be rendered obsolete with a $3000$-qubit quantum computer~\cite{Litinski23}.

Experiments with current quantum devices show that there are hints of quantum advantage around the corner~\cite{Google2019, IBM23}. However, the execution of the algorithms in \figref{fig:overheadgraph} requires the ability to run very deep computations. This is not possible with current quantum devices. Minor control inaccuracies and the inevitable interaction between a quantum system and its environment corrupt quantum states, rendering computations inaccurate. For example, an algorithm requiring $N$ gates executed with an output accuracy of $1\%$ requires gate error rates below $\frac{1}{100 N}$. For $N=10^6$, we must engineer gate accuracy to be five orders of magnitude better than what is available today~\cite{Acharya23}.

The current state of development of quantum computers is reminiscent of the age when classical computers were built with vacuum tubes and mechanical relays. These devices suffered from high failure rates and needed technicians to constantly replace components. With the onset of the transistor, failures decreased. 
% Soon, integrated circuits were developed with many thousands of transistors working in harmony. The argument is not quite the case with quantum devices. 
The quantum equivalent of the transistor has not yet been discovered. This is the case despite decades of research developing quantum devices with different platforms, such as superconductors, photons, ion traps, neutral atoms, semiconductor quantum dots, nuclear magnetic resonance qubits and nitrogen-vacancy centers in diamonds. Useful applications providing quantum advantage may only be possible with more robust qubits. Quantum error correction paves a route to robustness~\cite{Shor95decoherence, Gottesman97thesis, terhal15review}. Here, the quantum information taking part in a logical algorithm is encoded into a quantum error-correcting code on a large number of physical qubits. Any accumulated noise is removed using carefully designed error diagnosis and correction protocol. Using this technique, we can achieve arbitrary levels of accuracy.

Protocols for error correction can locate and remove errors. However these protocols are executed using faulty components and control sequences. Errors occurring \textit{during} the protocol thus corrupt ideal execution, potentially introducing more errors into the system. This can be problematic for very long computations. Hence to run deep quantum algorithms on hundreds of thousands of qubits, we will require fault-tolerant implementations of the fundamental processes of a quantum computer.

\subsection*{Theoretical models for fault-tolerant quantum computing}

% The first developments in quantum fault tolerance where built using the earliest model of quantum computation, the circuit model. 
Shor and Steane suggested the first techniques for fault-tolerant error correction. Shor's scheme \cite{Shor96}, which is applicable to a large class of codes called stabilizer codes, relies on a slow error diagnosis process where subsets of qubits are checked sequentially. A highly parallelized version of this process applied on topological codes yields very high fidelities. We provide more information on these codes later. Steane's scheme~\cite{Steane97}, applicable to a subset of stabilizer codes called CSS codes~\cite{Steane96codes}, determines error locations by copying errors onto an analogous state in one step. The catch is that the fault-tolerant preparation of this resource state is quite complex~\cite{Paetznickgolay2013}. 

Clever choices of early quantum error-correcting codes permitted some elementary logical gates to be transversal\footnote{An encoded gate performed by executing the same physical operation on all qubits in parallel}~\cite{Steane96css}. Further, magic state protocols enable quantum universality~\cite{Bravyi05}. For arbitrarily accurate computation, quantum codes can be concatenated. This is a process wherein encoded qubits are further encoded using another quantum error-correcting code to polynomially improve the protection. Early proposals for planar devices using concatenation for arbitrary accuracy required unattainably-high-fidelity physical operations~\cite{Gottesman00local, SvoreTerhalDiVincenzo04, Hollenberg06, Svore07localSteane}.

% Different models of quantum ocmptuing are interchangeable with each other.
% another fault-tolerant quantum computing model was devised - raussendorff fundamentally diff physical operations required. resource state acts as a substrate on which quantum computation is performed by performing a chain of measurements across the substrate.

The above results represent the first era in fault-tolerant quantum error correction. The second era began with the connection of ideas in topological quantum field theories to quantum error correction~\cite{Kitaev97faulttolerance}. The surface code~\cite{Fowler12surface}, first explored in the early 2000s, has been comprehensively investigated over the last two decades due to its potential for large-scale quantum computation. This attention has unveiled numerous advantages and opportunities, along with some pitfalls and challenges. A compelling advantage is its high fault tolerance threshold, permitting arbitrarily accurate computation at higher error rates than concatenated protocols. Within these families of codes, choosing larger codes guarantees arbitrarily high accuracy below a fault tolerance threshold.

The surface code exhibits some weaknesses from the perspective of quantum coding theory. The associated qubit overhead for information storage and performing quantum logic is prohibitively large~\cite{Gidney2021rsa, Litinski19}. In \tabref{tab:overheads}, we show some estimates for the physical resources required to implement different algorithms in \figref{fig:overheadgraph}. Additionally, the classical computing power needed to keep up with the error decoding problem can be substantial, even in moderate-sized implementations. Recent developments simplify the error decoding problem by partitioning the decoder and decoding multiple chunks in parallel~\cite{ChaoPrallel22, EarlParallel}. However with large devices, bottlenecks will be encountered in the global control and heat dissipation of multiple fast parallel decoders. Possibly different codes and fault tolerance techniques must be studied.

% quantum simulations, $10^7$, $100$,~\cite{}
% fertilizer manufacture, $10^{10}$, $2000$~\cite{Berry19, vonBerg21, Lee21chemistry}
% battery materials, $10^{13}$, $2000$~\cite{Delgado22, Shokrian23}

\begin{table}
\caption{\label{tab:overheads} Overhead estimates for different fault-tolerant algorithm implementations.}
\centering
\setlength{\tabcolsep}{5pt}
\begin{tabular}{c | c}
\hline \hline
Reference & Overheads \\
\hline
\cite{Gidney2021rsa} & $20$ million qubits in 8 hours to break RSA-2048 \\
\cite{Gouzien21} & 13436 qubits in 177 days to break RSA-2048 \\
\cite{Litinski23} &  6000 resource state generators at 580MHz with \\
 & logarithmic non-local connections in 58 seconds to break ECC-256\\ 
\cite{Gouzien23} & 126133 cat qubits in 9 hours to break 256-bit ECC \\
\cite{Matteo20} & 8 GB QRAM query in 2 ms with $10^{15}$ qubits or 2.4 years with 300000 qubits \\
% \cite{} &  \\
% \cite{} &  \\
% \cite{} &  \\
\hline \hline 
\end{tabular}
\end{table}

\subsection*{Experiments on quantum error correction}

\begin{figure}
    \centering
    \includegraphics[width =\textwidth]{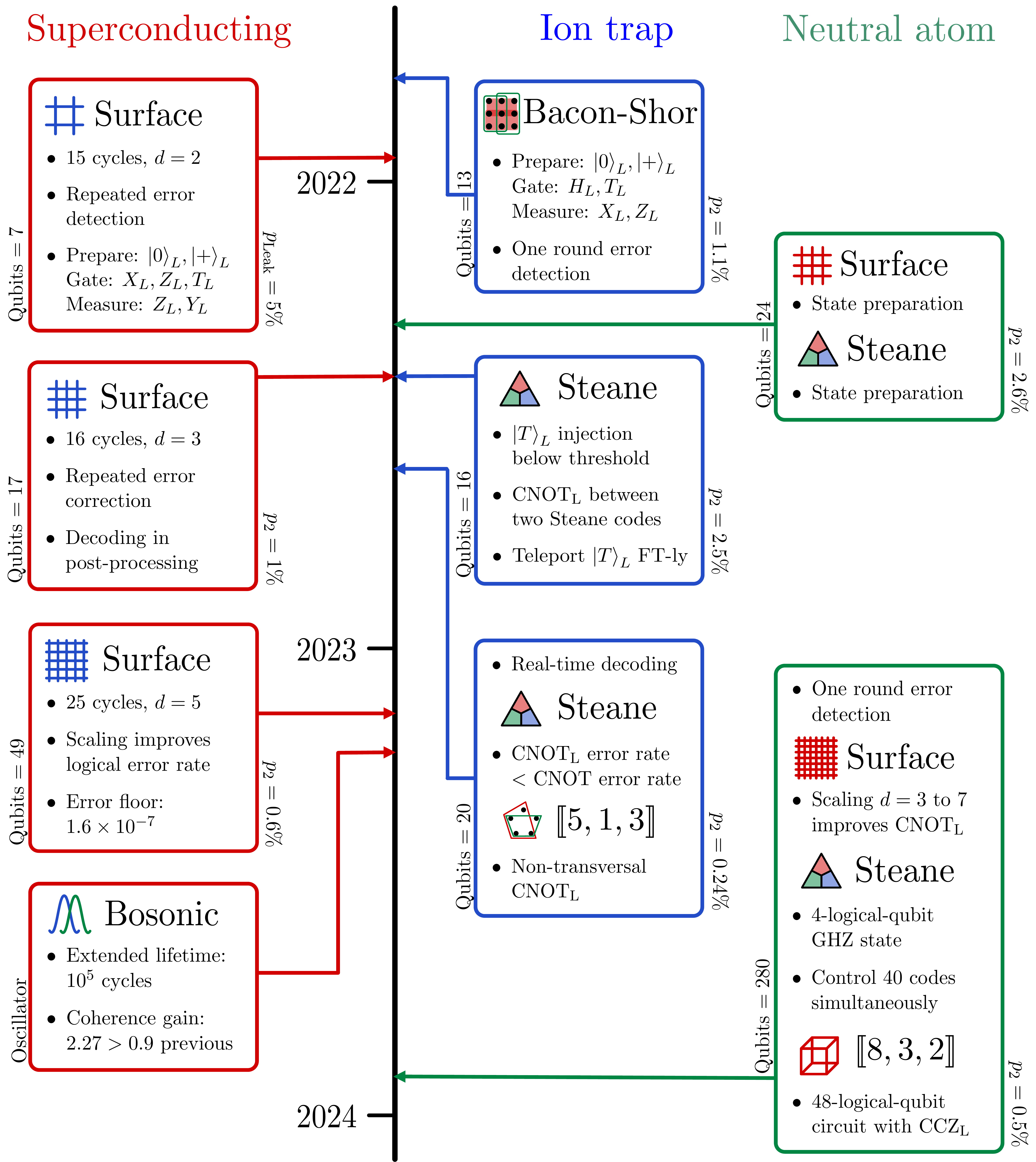}
    \caption{Progress of experimental demonstrations of quantum error correction on superconducting, ion trap, and neutral atom systems. Currently the objective is  to benchmark error correction primitives with small codes. In the next few years, devices will become larger while maintaining low noise levels. Neutral atom systems can control many qubits and maintain coherence, but lack repeated real-time error correction. Superconducting and ion trap qubits have demonstrated all the primitives for fault tolerance, and now face the engineering challenge of scaling up system sizes.
    }
    \label{fig:timelineexps}
\end{figure}

Despite the theoretical advancements outlined in the previous subsection, experimental quantum error correction is still nascent. Drastic improvements to device manufacturing and control optimization have permitted device error rates that are sufficiently low enough to see improvements from fault-tolerant error correction protocols. In \figref{fig:timelineexps} and \tabref{tab:timelineEC}, we highlight some of the most recent experiments on quantum error correction conducted on superconducting, ion trap and neutral atom quantum computers. Superconducting and ion trap technologies have matured from industrial funding, while neutral atom computers are still gaining traction within academia. 

\begin{table}
\caption{\label{tab:timelineEC} References for the error correction experiments in \figref{fig:timelineexps}.}
\centering
\setlength{\tabcolsep}{5pt}
\begin{tabular}{c c}
\hline \hline
Published date & Reference  \\
\hline
\multicolumn{2}{l}{\textbf{Superconducting qubits}} \\
Dec 16 2021 & \cite{Marques22} \\
May 25 2022& \cite{Krinner22}  \\
Feb 22 2023 & \cite{Acharya23}  \\
March 22 2023 & \cite{Sivak23}  \\[.2cm]
\multicolumn{2}{l}{\textbf{Ion trap}} \\
Oct 04 2021 & \cite{duke21BScode}  \\
May 25 2022 & \cite{Postler22}  \\
Aug 3 2022 & \cite{RyanAnderson22}  \\[.2cm]
\multicolumn{2}{l}{\textbf{Neutral atoms}} \\
April 20 2022 & \cite{Bluvstein22} \\
% Oct 11 2023 & \cite{Evered23} & Parallel multi-qubit entangling gates \\
Dec 06 2023 & \cite{Bluvstein24} \\
\hline \hline 
\end{tabular}
\end{table}

Currently, small-scale fault-tolerant implementations of the fundamental building blocks of a computation have been executed on all three platforms. What remains to be seen is whether these systems are able to retain low gate error rates as they are scaled up in size. As larger systems become more prevalent, more complex error correction protocols will be executed. Potent classical computers will be needed to keep up with decoding for error correction, and developments will need to be made to operate these circuits in cryogenic environments.

\section{Redesigns of fault-tolerant building blocks}
\label{sec:redesign}

Ideally, a quantum computation proceeds by preparing a quantum state, operating on it, and finally measuring it to observe an outcome. When this model is generalized to an error-corrected fault-tolerant quantum computation, the building blocks of the computation are slightly altered. We show a layered architecture to construct large-scale fault-tolerant quantum machines in Fig.~\ref{fig:arch}~\cite{Jones12}. This hierarchy abstracts the instructions for a quantum algorithm into five levels, ranging from high-level algorithmic operations to low-level physical gates on the device. Crucially, the gap between the algorithmic and physical layers is bridged by error correction, acting as a vital buffer against the ubiquitous noise. 

With an appropriate choice of quantum error-correcting code, the building blocks of a quantum computation must be executed with fault-tolerant physical circuits. These building blocks include state preparation, error correction, and logical gates, with stabilizer measurement serving as a foundational subroutine. 
Note that we do not include measurement as a building block.
% as we include it in the bracket of logical gates. Lattice-surgery and measurement-based computation are two example models that show that Fault-tolerant measurements may be performed using transversal single-qubit measurements for some CSS codes. 
We include fault-tolerant measurements of logical qubits under the bracket of logical gates, as in many models of quantum computation, they are used to perform logical operations on quantum states~\cite{Chamberland22, Briegel09}. In \figref{fig:summary}, we show a summary of the important results in this thesis. For stabilizer measurement, state preparation, and error correction, we show improved qubit overheads. Alternatively, we show an improvement in the execution time for fault-tolerant error correction protocols and logical gate sequences.

\begin{figure}
    \centering
    \includegraphics[width = 0.9\textwidth]{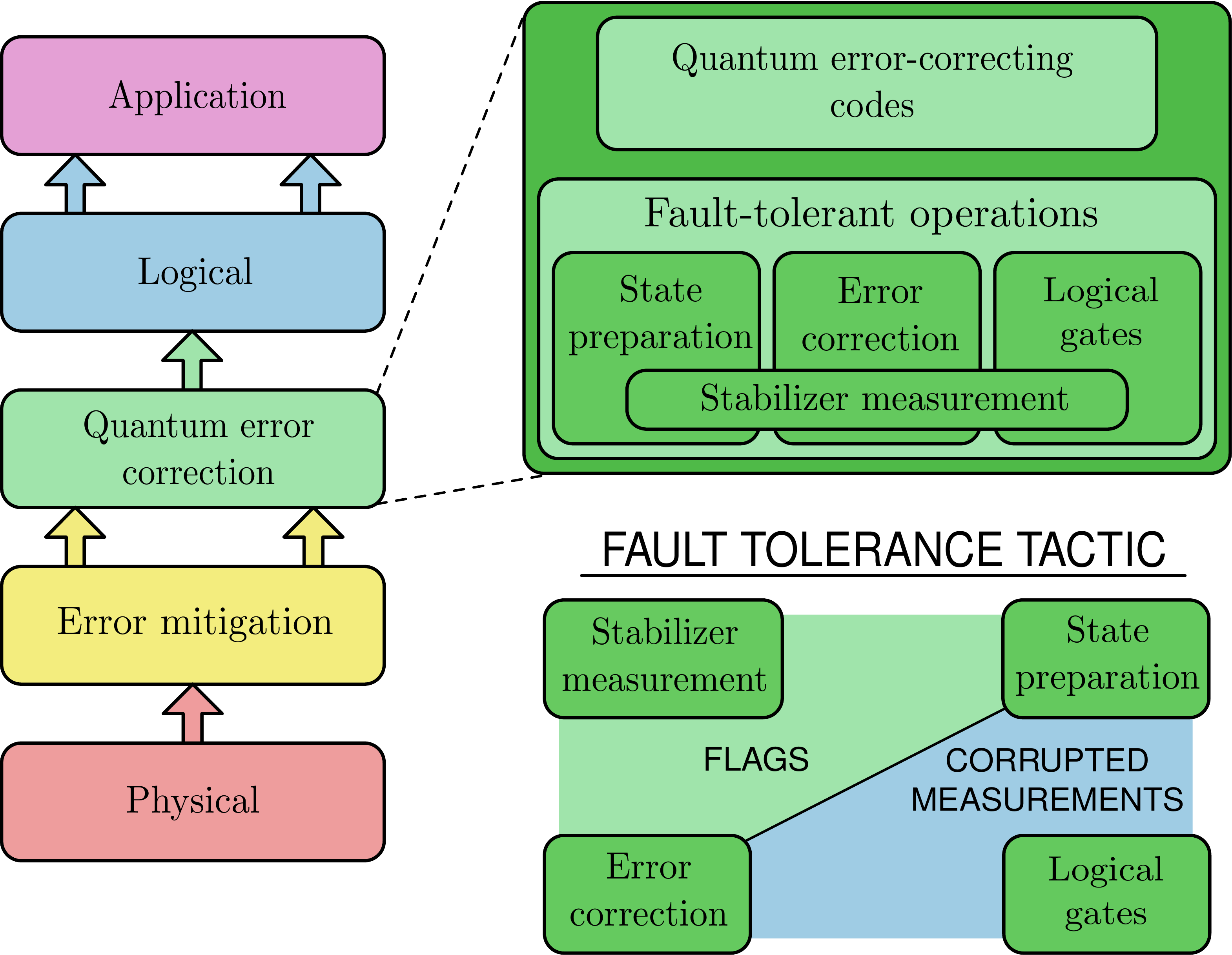}
    \caption{Layered architecture of a fault-tolerant quantum computer~\cite{Jones12}. The functioning of a quantum computer can be decomposed into a five-layer vertical stack, each handling a different level of abstracted instructions. From hardware to software, they are the physical pulses, error mitigation strategies, quantum error correction procedures, substrates for quantum logic and the programs of quantum algorithms. Given a quantum code, the fundamental fault-tolerant operations that underpin the error correction layer are state preparation, error correction and logical gates, with stabilizer measurement as a foundational subroutine. Faults affect each building block differently. Flag-based fault tolerance curbs error spread due to physical two-qubit gates. Higher up the stack, faults affect measurement results, requiring fault tolerance tactics against corrupted measurements.}
    \label{fig:arch}
\end{figure}

\subsection*{Tactics for fault tolerance}

\begin{figure}
    \centering
    \includegraphics[width = 0.99\textwidth]{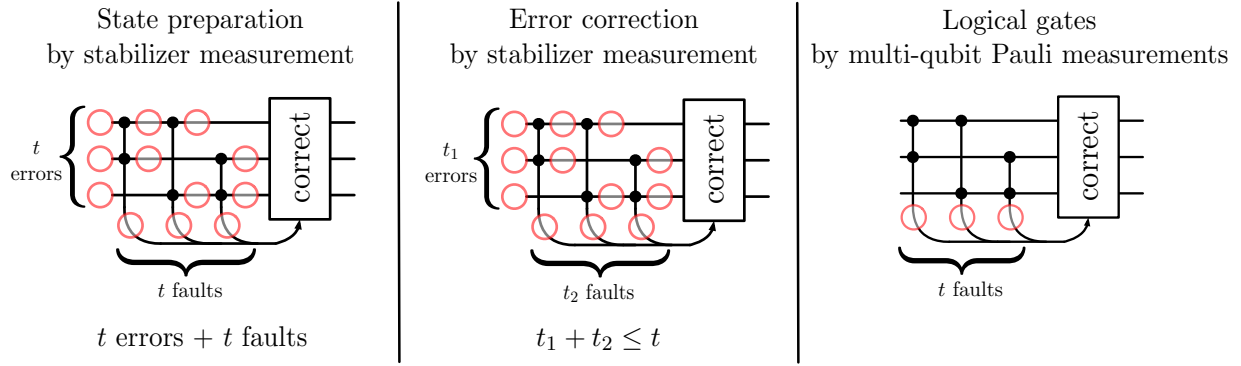}
    \caption{Three different models of fault tolerance against faults that corrupt measurement results. The results of stabilizer measurement in state preparation and error correction are used to apply Pauli corrections. Similarly, results of multi-qubit Pauli measurements are used to apply Pauli or Clifford corrections. We show fault-tolerant circuits that are deterministic (no postselection). These models are distance-$d$ fault-tolerant, for $d=2 t + 1$.  A model is fault-tolerant if for $k \leq t$ faults, the weight of the output error after corrections is at most $k$. 
    State preparation is the hardest fault tolerance problem as we must ensure tolerance against $t$ input errors and also up to $t$ faults in the execution of the measurements. Error correction is simpler as the total number of input errors and internal faults is constrained to be at most $t$. 
    Finally, measurement faults in multi-qubit Pauli measurements can be corrected by protecting results with a classical error-correcting code.
    }
    \label{fig:modelscorrupt}
\end{figure}

% In classical computing, errors are generally modeled as bit flips. With quantum computers, a widely-accepted model is the depolarizing error model, comprised of bit flips ($X$), phase flips ($Z$) and combinations of the two ($Y$).

Errors affecting quantum computers are typically modeled as continuous processes~\cite{AhnDohertyLandahl01continuous}. However, continuous noise is challenging to analyze from a fault tolerance perspective. Discrete noise models are simpler. Here, faults are modeled as the probabilistic injection of discrete errors at specific circuit locations. This is a valid approach since stabilizer measurements during error correction discretize continuous errors on individual qubits. With this model of faults, we can analyze the effects of introduced errors. A fault occurring midway through an error correction protocol corrupts subsequent stabilizer measurements, misdiagnosing existing errors, which eventually leads to the application of corrections on the wrong qubits. The resulting effect is an increase in the number of errors. Additionally, errors can spread when qubits interact. Hence, the main objective of fault tolerance is to curb the spread of errors.

Faults affect the operation of the building blocks in different ways, hence we must employ different models to solve each fault tolerance problem. For stabilizer measurement and state preparation by operator encoding, the concern is that faults occuring partway through a circuit may eventually spread to errors on many qubits. In these cases we employ flag fault tolerance to curb error spread. Here, circuits are engineered such that errors on code qubits spread onto a set of flag qubits, which are then measured to determine if a fault occurred. Based on the flag measurement outcomes, corrections can be applied to reduce error spread. Early theoretical and current experimental results on flags rely on postselection-style fault tolerance~\cite{RyanAnderson22, ChaoReichardt18fewqubitcomputation}. Recently, there has been a push to make these protocols deterministic to permit repeated circuit execution for extended periods of time~\cite{DiVincenzoAliferis06slow, stephens14colorcodeft, YoderKim16trianglecodes, ChaoReichardt17errorcorrection}. In this thesis, we consider deterministic flag protocols in \chapref{chap:SMcat}, \chapref{chap:catstateprep} and \chapref{chap:encodedstateprep}. Additionally, we consider some circuits that combine postselection with correction in \chapref{chap:D4codes}. We highlight the building blocks using flags for fault tolerance in \figref{fig:arch}.

Flags prevent the spread of errors due to qubit interactions, but these errors can still corrupt measurement results. A noisy error correction procedure may erroneously apply corrections at wrong locations. This can be prevented by ensuring the measured data is protected from errors. Naively, it may seem like encoding the measurement results in a classical error-correcting code may offer some protection. In fact, this is the solution for a model where faults only affect measurement outcomes. In \figref{fig:modelscorrupt}, we outline three models of fault tolerance against corrupted measurements. The model for fault-tolerant error correction with stabilizer measurements is historically important as threshold results are derived using it~\cite{AliferisGottesmanPreskill05}. The model for state preparation is more difficult to find solutions for as more errors must be tolerated. For these models, faults on qubits leave residual errors while also corrupting the measurement result. In the model for logical gates, the qubits are assumed to be robust, hence we only consider faults affecting measurements.

Consider the model for error correction. The protocol is deemed fault-tolerant to $t$ errors if for $t_1$ input errors and $t_2$ internal faults (faults \textit{during} the protocol), the output has at most $t_2$ errors when $t_1+t_2<t$~\cite{delfosse2020short}. In the absence of faults, the protocol must correct errors up to the capacity of the code. If there are no input errors and only faults occur, this model ensures that ``errors don't spread". We use this model in \chapref{chap:D4codes} to make fault-tolerant error correction protocols for distance-four quantum codes. 

\begin{figure}
    \centering
    \includegraphics[width=.95\textwidth]{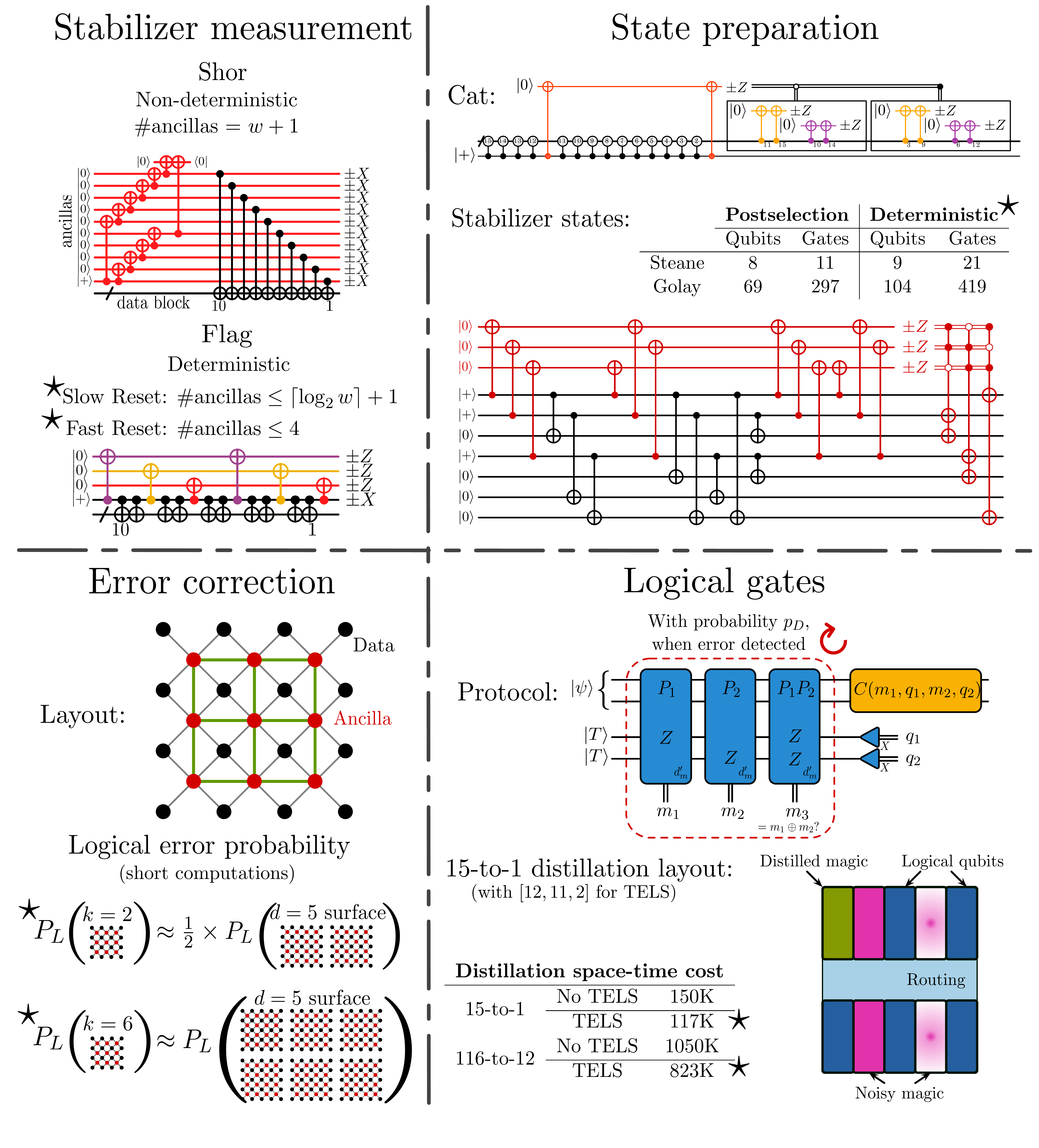}\vspace{-.7cm}
    \caption{Overview of the thesis. New results are indicated by a $\star$. \textbf{(Top left)} The flag scheme needs $\sim \log_2 w$ ancilla qubits to deterministically measure a stabilizer while tolerating one fault, exponentially fewer than Shor's probabilistic scheme.  \textbf{(Top right)} First overhead estimates on \textit{deterministic} state preparation. One-shot fault-tolerant stabilizer state preparation is optimized with either flags or tolerance to corrupted measurements. \textbf{(Bottom left)} Qubit-efficient distance-four codes protect information as well as the distance-five surface code, assisted by postselection. The qubit layout is an augmented surface code layout. Flags prevent error spread in stabilizer measurement, and a clever selection of stabilizers can tolerate corrupted measurements.
     \textbf{(Bottom right)} 
    Parallelizable Pauli measurements run 
    % two to six times 
    faster by pre-encoding measurement results into a classical code (Temporal Encoding of Lattice Surgery). This can reduce magic state distillation~costs.}
    \label{fig:summary}
\end{figure}

The model for state preparation requires tolerance to more errors. When preparing states by stabilizer measurement, projection into the desired state is not guaranteed. Instead the initial set of measurements are random (due to anticommuting operators), leading to a projection that leaves residual errors on the qubits. This is why we require tolerance to at least $t$ input errors. In the absence of faults, the stabilizer measurement results are decoded to diagnose the exact error space, prompting a correction into the desired state. When faults occur, erroneous corrections can project the state into a different error space. As long as the number of residual errors after the completion of the state preparation protocol scales with the number of faults, and fault spread is curbed in stabilizer measurement, this protocol is fault-tolerant. This model is used to prepare cat states and stabilizer states in \chapref{chap:catstateprep} and \chapref{chap:encodedstateprep}.

Finally, we consider a model where the qubits are assumed to be noise free, and faults only affect measurement results. This is the case when considering qubits that are arbitrarily well-protected by error correction, but still suffer from noisy quantum logic. In this model, the measurements are affected by bit flip errors on the results. Since there are no residual quantum errors on the data, we use a classical error-correcting code to protect the measurement results. This approach is only applicable when the set of measurements to be performed commute. We use this model to speed up quantum logic operations in \chapref{chap:MSD} and \chapref{sec:MSDMSD}.

\subsection*{Distance-three stabilizer measurement}

Stabilizer measurement is the most fundamental and repeated operation in a fault-tolerant quantum computer. Not only is it the backbone of Shor-type error correction schemes~\cite{Shor96}, it is also used exclusively for operations with the surface code~\cite{Fowler12surface}, and has a similar circuit design to operator encoding circuits, allowing fault tolerance ideas to be shared between the two. Steane-style error correction~\cite{Steane97} doe not involve the measurement of individual stabilizers, but the circuits used to prepare the complex ancilla states for Steane error correction require either operator encoding circuits or stabilizer measurement itself, as we show in \chapref{chap:encodedstateprep}. Due to its prevalence, optimizing the overhead of fault-tolerant stabilizer measurement is paramount. Previous methods exhibit overheads that scale linearly in the size of the measured stabilizer~\cite{DiVincenzoAliferis06slow, stephens14colorcodeft, YoderKim16trianglecodes, ChaoReichardt17errorcorrection}. We show logarithmic overhead to measure a stabilizer while tolerating one fault. 

In \chapref{chap:SMcat}, we demonstrate that distance-three fault-tolerant measurement of a weight-$w$ stabilizer needs at most $\lceil \log_2 w \rceil +1$ ancilla qubits, with non-local connectivity. If qubits reset quickly, four ancillas suffice. These improvements arise from mapping flag qubit measurement outcomes to the vertices of a hypercube. With $a$ flag qubits, previous methods use $O(a)$ flag patterns to identify faults. We use the same flag qubits more efficiently by constructing maximal-length paths through the $a$-dimensional hypercube to use nearly all $2^a$ possible flag patterns. Finally, we extend this technique to distance-five and -seven stabilizer measurement, where we demonstrate overhead that scales as $\sim w/2$ and $\sim w$.

\subsection*{Deterministic state preparation}

The combinatorial approach of constructing paths in a hypercube can also be extended to fault-tolerant cat state preparation. An even more efficient method is discussed in \chapref{chap:catstateprep}, where $w$-qubit cat states are prepared with \textbf{one} ancilla qubit measured $O(\log_2 w)$ times. Asymptotic upper bounds on the qubit overhead in different scenarios are tabulated in \tabref{f:resultsCSP}. While these methods primarily focus on checking a cat state after it has been prepared noisily, we also explored an approach that prepares cat states by measuring stabilizers. 
% Noisy measurement results are tolerated by selecting sequences of stabilizers while analyzing the effects of fault propagation. By carefully selecting minimal-length sequences of stabilizers, measurement faults can be tolerated while increasing the fidelity of prepared cat states. 
This method also requires one ancilla, but performs $w + O(\log_2 w)$ measurements. This technique was then extended to tolerate an arbitrary number of faults during cat state preparation when placing qubits on a $1$-D chain. 

Both techniques developed in \chapref{chap:catstateprep} can be applied in a straightforward manner to the preparation of more complex ancilla states, as detailed in \chapref{chap:encodedstateprep}. Here, we compare the fidelity and overhead of using different deterministic circuits for the preparation of the encoded $\ket{0}$ states of the Steane and Golay quantum codes. Previous results on the fault-tolerant preparation of these states leveraged postselection-based fault-tolerance, where a protocol is rejected when a non-trivial outcome is observed. Here, we decode measurement results and assign compatible corrections to ensure states are never thrown away. The high-fidelity preparation of stabilizer states is motivated by the use of these states for Steane-style error correction. In this protocol, qubits of a quantum code are engineered to copy errors over to a stabilizer state. Measuring the state and decoding the results reveals the location of errors in the quantum code. For some codes, the Steane-style protocol exhibits higher thresholds than Shor-type schemes~\cite{Steane97, Escobar24}. This has also been verified experimentally~\cite{Postler23, Huang23b, Bluvstein24}.

\subsection*{Distance-four error correction}

Finding the perfect quantum error-correcting code is an elusive goal. For near-term implementations of quantum memory, we desire a code with high-fidelity error correction, high encoding rate, and circuits that can be implemented on near-term hardware. Topological codes like the surface code are well-suited to near-term hardware and exhibit very high fidelity. Unfortunately, these codes have low encoding rate and require many physical qubits per logical qubit. Quantum LDPC codes, on the other hand, have high rate but need circuits that are difficult to implement on current hardware~\cite{Panteleev22, Bravyi2023}. In contrast, \chapref{chap:D4codes} of this thesis considers distance-four, efficient encodings of multiple qubits into a modified planar patch of the 16-qubit surface code. These codes satisfy all the desired properties for near-term error correction. We use postselection techniques to achieve high-fidelity error correction. These codes can encode up to six logical qubits with distance-four protection using the same number of physical qubits as the distance-four surface code. Finally, the circuits for error correction are designed for planar hardware. 

Choosing distance four codes allows us to correct single faults and detect double faults. Thus logical errors occur due to at least three faults, which is rare at current noise rates. We perform extensive simulations of error correction and compare the logical error rate of our codes against the distance-four and -five surface codes~\cite{tomita14surface}. These codes serve as the perfect benchmarks as they satisfy most of the desired properties for a quantum memory while only suffering from low code rate. A spectrum of trade-offs were revealed: at one end, a code encoding six logical qubits can achieve comparable error protection as the distance-five surface codes using one-tenth as many qubits. At the other end of the spectrum, a two-logical-qubit code can achieve logical error rate half that of the surface code, using one third of the physical qubits; an improvement on both fronts! Hence distance-four codes, using postselection and in a planar geometry, are qubit-efficient candidates for fault-tolerant, moderate-depth computations.

The circuits in this chapter are amenable to architectures that are constrained to be planar with local connectivity, such as superconducting systems and semiconductor spin systems. Currently IBM develops superconducting systems where qubits interact with at most three neighbors~\cite{Sundaresan2023}, and with Google's devices, qubits can interact with up to four neighbors~\cite{Acharya23}. Constructing a device with a connectivity of eight neighbors naively permits the use of our protocols for error correction. But low-cost modifications can decrease the connectivity requirements. 

\subsection*{Logical gates by temporally encoded lattice surgery}

From the perspective of experimental quantum computing, the leading candidate quantum error-correcting code is the surface code. Stabilizer measurements for error correction are simple, and can be parallelized when placing the qubits on a square planar grid. The only caveat of these codes is that they encode very little logical information, and hence require many physical qubits. This caveat may indeed be a boon, as performing logical operations on surface code qubits is simpler than with codes encoding multiple qubits. Further, we consider biased-noise implementations of surface codes, as they are backed by experiment and theory~\cite{Puri20, Claes2023}. A biased noise model has reduced noise for one type of error, thereby reducing the overhead of qubits needed to protect for it. In the context of biased noise, different variants of the surface code possess different desirable properties. Our focus is to reduce the number of physical resources per logical qubit for a CPU-style architecture of logical qubits~\cite{Chamberland22}. Recent research on biased noise surface codes has shown that the total fault tolerance overhead is lowest when rectangular surface codes are used with biased noise models~\cite{HiggottBP}. When equipped with magic state distillation factories, the lowest-overhead logical gates that permit universal computation are lattice surgery measurements of multi-logical-qubit Pauli operators~\cite{Litinski19}. 

Surface code logical qubits are protected from errors by the spacelike distance of the code. However when performing lattice surgery measurements, the accuracy of the result depends on the timelike distance, which is the number of rounds of syndrome measurements needed to protect against strings of timelike failures. As such, a larger timelike distance requirement will result in the slowdown of an algorithm’s runtime. Temporal encoding of lattice surgery (TELS) is a technique which can be used to reduce the number of syndrome measurement rounds that are required during a lattice surgery protocol~\cite{Chamberland22b}. In a regular lattice surgery protocol, multi-qubit Pauli operators are sequentially measured. With TELS, only commuting subsequences of Pauli operators can be measured. In this technique, the classical measurement results of the commuting subsequence are encoded into a classical error-correcting code. This results in the measurement of a larger set of mutually commuting multi-qubit Pauli operators that are derived from the initial sequence. Up to the distance of the classical code, the multi-qubit Pauli measurement results are protected against timelike lattice surgery failures. We evaluate the speedups that could be achieved by TELS at different desired accuracy ($10^{-10}$, $10^{-15}$, $10^{-20}$, $10^{-25}$) and observe an algorithmic speedup by a factor of two to six respectively. In these calculations, we considered codes of size up to $128$, and distance up to $11$. The application of TELS is not restricted to the surface code. In fact, it can be applied for any topological or non-topological code that executes logical computations by lattice surgery.

We also apply TELS to magic state distillation protocols with biased noise in \chapref{sec:MSDMSD}. Magic state distillation is a short algorithm executed on encoded qubits~\cite{Bravyi05}, and thus served as the perfect benchmark to test the improvements of TELS and perform a full cost estimation. In order to make TELS compatible with magic state distillation protocols, we designed a new distillation technique, where multi-qubit measurements are performed in the Clifford frame\footnote{Computations performed in a certain frame contain errors of that frame's type, but the locations of the errors are known and they are tracked in software until they can be corrected.}. Previous techniques distilled in the Pauli frame~\cite{Litinski19magic}. Using TELS and optimally-designed layouts, we demonstrate a reduction in the space-time cost of magic state factories by as much as $22 \%$.

\section{Current thrusts and future outlook}

\subsection*{Current thrusts}

Research at the forefront of improvements to quantum error correction can be mapped to different levels of abstraction in the layered architecture of \figref{fig:arch}. Here we discuss improvements to hardware architecture designs, developments on new quantum error-correcting codes, and advancements in fault-tolerant quantum algorithms.

For the first three years of the 2020s, quantum computers built using superconducting and ion trap technologies were touted as game-changing. Superconducting quantum circuits are fast, and gate fidelities are improving, however scaling and cryogenics are formidable challenges~\cite{Huang2020}. On the other hand, ion trap quantum computers can demonstrate impressive gate fidelities and all-to-all connectivity. But they suffer from scaling issues, and slow circuit operations~\cite{Bruzewicz19}. 
Neutral atom systems have recently emerged as another potential candidate to build moderately sized fault-tolerant quantum computers. Non-local gates are performed by transporting qubits across a $2$-D plane. 
Although qubit transport is slow, it can still be performed much faster than the timescale of decoherence~\cite{Bluvstein22}. Additionally, many qubits can be simultaneously loaded, providing a simple route to scale to larger devices. However most importantly, this architecture permits the execution of highly optimal circuits for encoded quantum logic. Steane error correction, shown to exhibit higher thresholds than some Shor schemes~\cite{Escobar24}, can form the backbone of error correction with neutral atoms. 
In \chapref{chap:encodedstateprep}, we show low-overhead fault-tolerant circuits to prepare code states for Steane error correction.

As highlighted in \chapref{chap:SMcat}, the measurement of stabilizers of high weight can require many extra qubits. The surface code measures stabilizers of size at most four. But recent work has highlighted new codes with measurements of at most two or three qubits. These operators are trivial to measure fault-tolerantly. In most cases, these operators are not stabiliszers of the code, but are representations of ignored logical qubits, colloquially known as gauge qubits~\cite{Hastings21, gidney2021faulttolerant, Higgott21}. 
Similar to the surface code, these codes are topological in nature and encode a constant number of logical qubits. The low encoding rates of these codes may prohibit wide-scale use for logical computation but these codes may serve as a primary level of protection in a concatenated setting.

The field of quantum error-correcting codes has recently been inundated with the study of high-rate codes constructed by quantum low-density parity check techniques (QLDPC). This ranges from the discovery of new optimal codes~\cite{Panteleev22, Pannteleev22b, Lin2023}, codes on planar layouts~\cite{Bravyi2023}, and analyses of concatenations of high-rate QLDPC codes with other high-fidelity codes~\cite{Ruiz24, Gidney23yoked}. Additionally, different encoding schemes are being constructed based on the physics of quantum devices~\cite{Grimm2020, GottesmanKitaevPreskill00}. These advancements offer flexibility in code choice throughout the architecture stack.

Converting all noise to erasure is an exciting new paradigm. Erasures are errors at known locations, hence they are easier to correct than errors whose position is unknown. There have been recent experimental proposals to convert two-qubit gate errors in neutral atom systems and amplitude damping noise in superconducting systems to erasure noise~\cite{Wu22, Ma23, Kubica23}. This has been complemented with an experimental implementation of a dual-rail cavity qubit in both systems~\cite{Chou23, Scholl23, Levine23} and in ion trap systems~\cite{Kang23}. New codes are also being developed for systems with biased erasure errors~\cite{Sahay23}. Other types of noise have also been suppressed in recent experiments. For example, leakage errors have been reduced~\cite{Battistel21, Varbanov2020}, and high-energy low frequency errors from cosmic rays have also been studied~\cite{McEwen2022}.

Quantum devices entered the era of noisy intermediate-scale devices (NISQ) with the emergence of devices boasting around fifty qubits capable of running circuits with up to a hundred gates~\cite{bharti2021noisy}. This prompted the discovery of a new class of algorithms for these devices, termed NISQ algorithms~\cite{cerezo2020variational}. The hope was that one such NISQ algorithm could be found that shows a definite quantum advantage. However in the half decade or so that has passed, no such results have been shown. 
Research on quantum algorithms has now shifted focus. With larger devices and more robust error correction primitives being shown experimentally, research on quantum algorithms has shifted to fault tolerance. General fault-tolerant algorithms for an arbitrarily accurate quantum computer are converted to permit execution on near-term devices, while incorporating limited error correction benefits. These results indicate the dawn of early fault-tolerant quantum computing (EFTQC)~\cite{Campbell22, Zhang22, Lin22, Wang22, Kshirsagar22, Wan22, Wang23, Zhiyan23, Akahoshi23}. 

\subsection*{Future outlook}

% \textbf{Fault-tolerant quantum computer architecture.} 
Current early architectures for fault-tolerant quantum computing involve using just one quantum error-correcting code. In the future, a quantum computer will require different encodings throughout the architecture stack. The earliest models concatenated small codes with themselves~\cite{Knill05nat}, however we believe that codes with different properties should be concatenated together. 

Consider the following thought experiment to construct a fault-tolerant computer. A GKP qubit is a continuous-variable logical qubit with a strong ability to suppress errors. Concatenating GKP qubits with the surface code can then produce low overhead logical qubits with very good protection to noise~\cite{Noh20, Noh22}. Considering qubits are fairly well protected, a QLDPC code may now be added on top to further increase protection while keeping the encoding rate fairly high~\cite{Gidney23yoked, Ruiz24}. In the process of error correction or logical gates, classical codes may be used to protect measurement results, as explored in \chapref{chap:MSD}. This simple thought experiment uses four levels of codes, however future quantum computers may require more. 

Additionally, one must consider the decoding problem when using multiple codes. Physical systems that require cryogenics have stringent constraints on the heat output of the quantum computer, hence low-energy techniques must be devised for complex decoding algorithms with classical computers. Even with a simple decoder such as a lookup table, there can be complications. Consider as an example the decoding of a $28$-bit string of classical results for the Golay code in \chapref{chap:encodedstateprep}. At least ten gigabytes are required to store a lookup table of all the corrections. On the other hand, slow complex fault-tolerant decoders can potentially stall computations for long periods of time~\cite{ChambsLocalNN22}. To combat this, efforts have been made to improve the speed and scalability of decoders for surface code quantum computers~\cite{Barber23, Liyanage23, Tan23}. A recent summary on real-time decoding is shown in Ref.~\cite{Battistel23}.

Conventionally, universal quantum computing is achieved using magic state distillation factories. However, code-switching techniques may be a useful alternative~\cite{BevsUniversal}. A quantum computer then can be partitioned into multiple units such as a memory unit with high rate codes and a computational unit with a plethora of codes permitting different types of logic. Accumulating a library of Clifford and non-Clifford operations can speed up the execution of complex algorithms while also permitting the execution of a wide range of algorithms.

\chapter[Syndrome measurement]{Fault-tolerant syndrome measurement}
\label{chap:SMcat}

A critical component of quantum error correction is syndrome measurement: a set of circuits that are used to pinpoint which qubits have errors.  This process of error identification is itself susceptible to noise and may fail.  To make this robust, extra (ancilla) qubits can be used to identify damaging mid-circuit faults and mitigate the spread of errors.  The objective of this paper is to reduce the overhead of ancilla qubits used in imparting this fault tolerance.  In particular, we focus on optimizing the flag technique for distance-three fault-tolerant stabilizer measurement.  

\begin{figure}
    \centering
    \subfloat[\label{f:genove}]{
        \includegraphics[width = 0.45\textwidth]{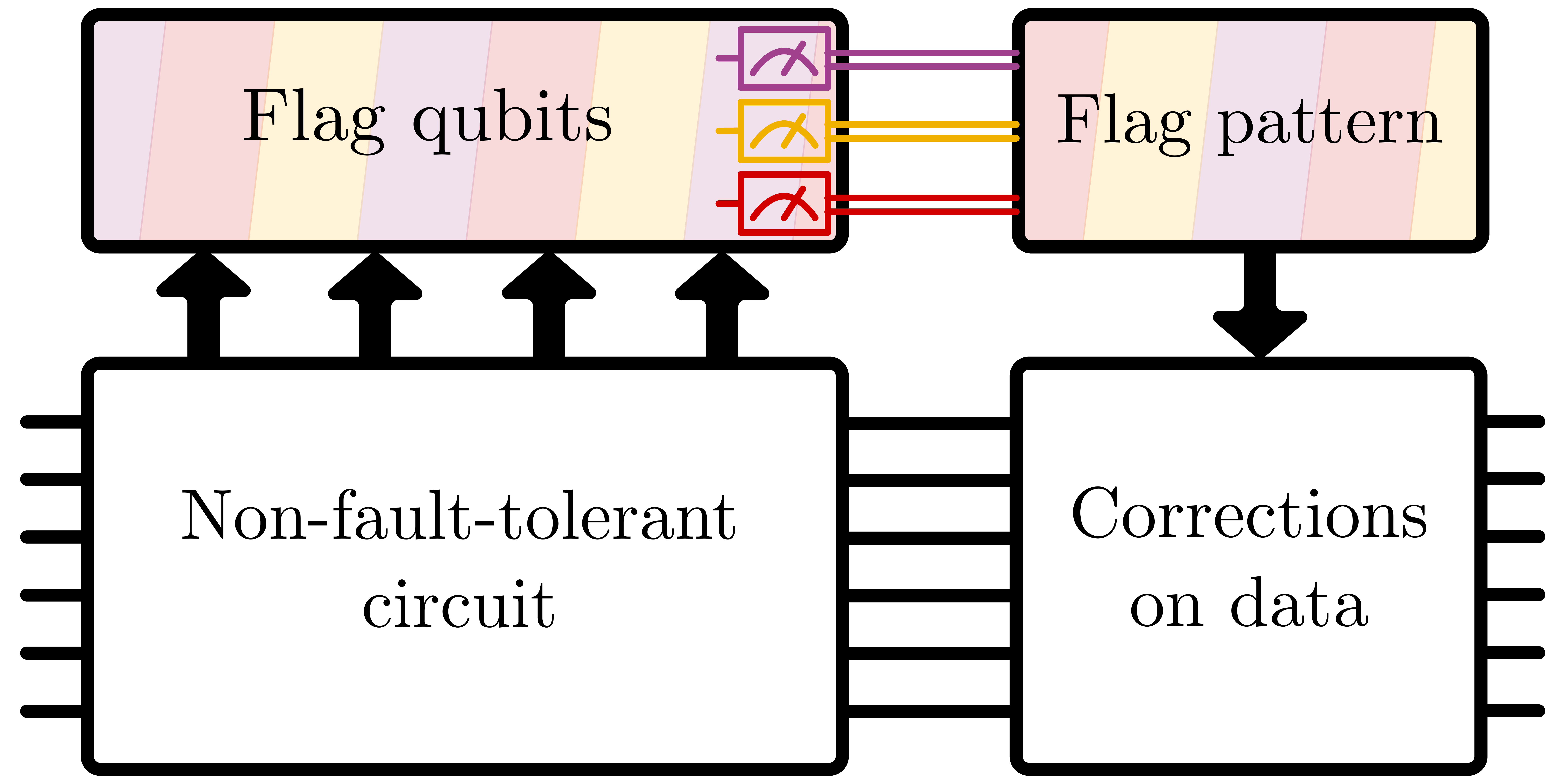}}
    \hspace{0.25cm}
    \subfloat[\label{f:syndromemeasurementd3slowresetw10}]{
        \includegraphics[width = 0.5\textwidth]{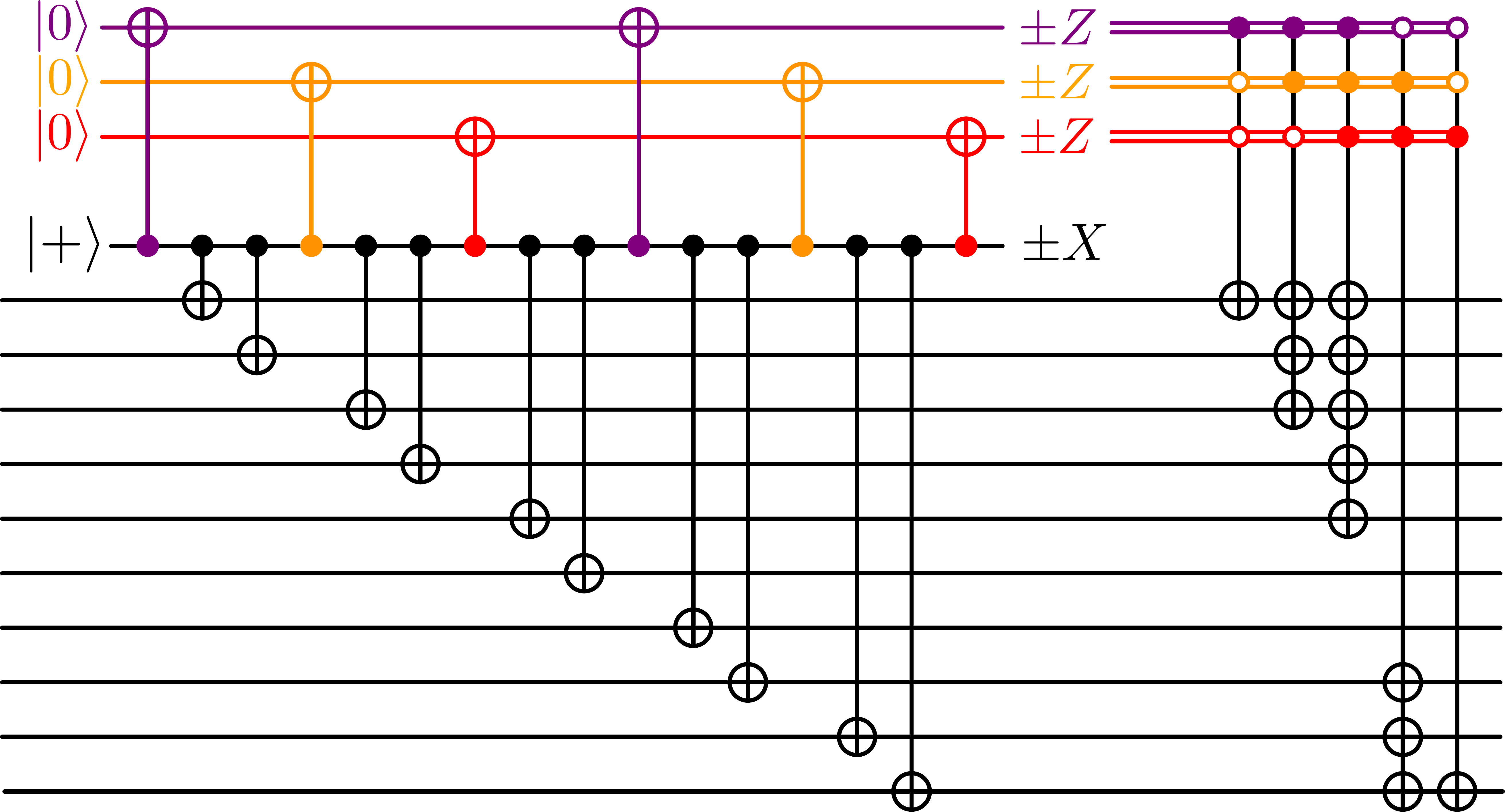}}
\caption{(a)~Function of a flag scheme. Errors in a non-fault-tolerant circuit can be made to spread into flag qubits.  On measurement, the flag qubits yield a pattern of $1$s and~$0$s, based on which the data is corrected.  (b)~Circuit to measure the stabilizer $X^{\otimes 10}$, using three flag qubits, in color, to protect against one $X$ fault (distance three).}
\end{figure}

We strive for low qubit overhead since quantum computers with limited qubits count resources preciously, and even minor improvements can free up extra qubits for other tasks. In topological codes where stabilizers are localized in space and are of  low weight, only a few flag qubits close to each stabilizer suffice to impart fault tolerance~\cite{YoderKim16trianglecodes, Chamberlandtriangularcodesflag2020, Chamberland20heavycodes}.  It has also been shown that with adaptive control and quickly resetting qubits, only four ancillas are required for the universal fault-tolerant operation of some distance-three codes~\cite{ChaoReichardt17errorcorrection, ChaoReichardt18fewqubitcomputation}.  In this paper, we present a general fault-tolerant protocol that works for a stabilizer of any size.  If qubits are connected well enough, we show that only logarithmic overhead is required for fault-tolerant stabilizer measurement, an exponential space improvement over the previous linear overhead.
 
\begin{figure}
\centering
    \subfloat[\label{f:Shorcard} \cite{Shor96}]{
        \includegraphics[width = 0.45\textwidth]{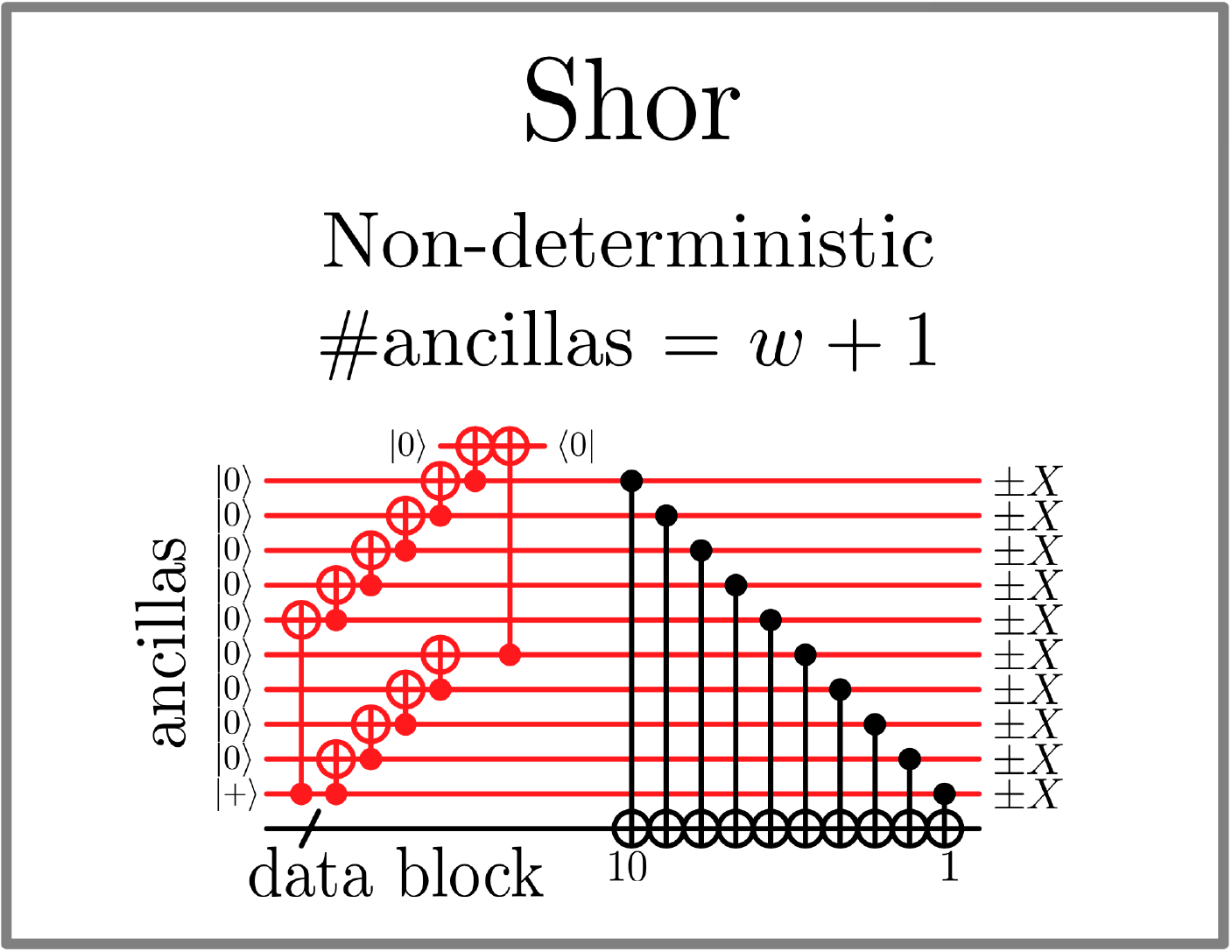}}
    \hspace{0.2cm}
    \subfloat[\label{f:DAcard} \cite{DiVincenzoAliferis06slow}]{
        \includegraphics[width = 0.45\textwidth]{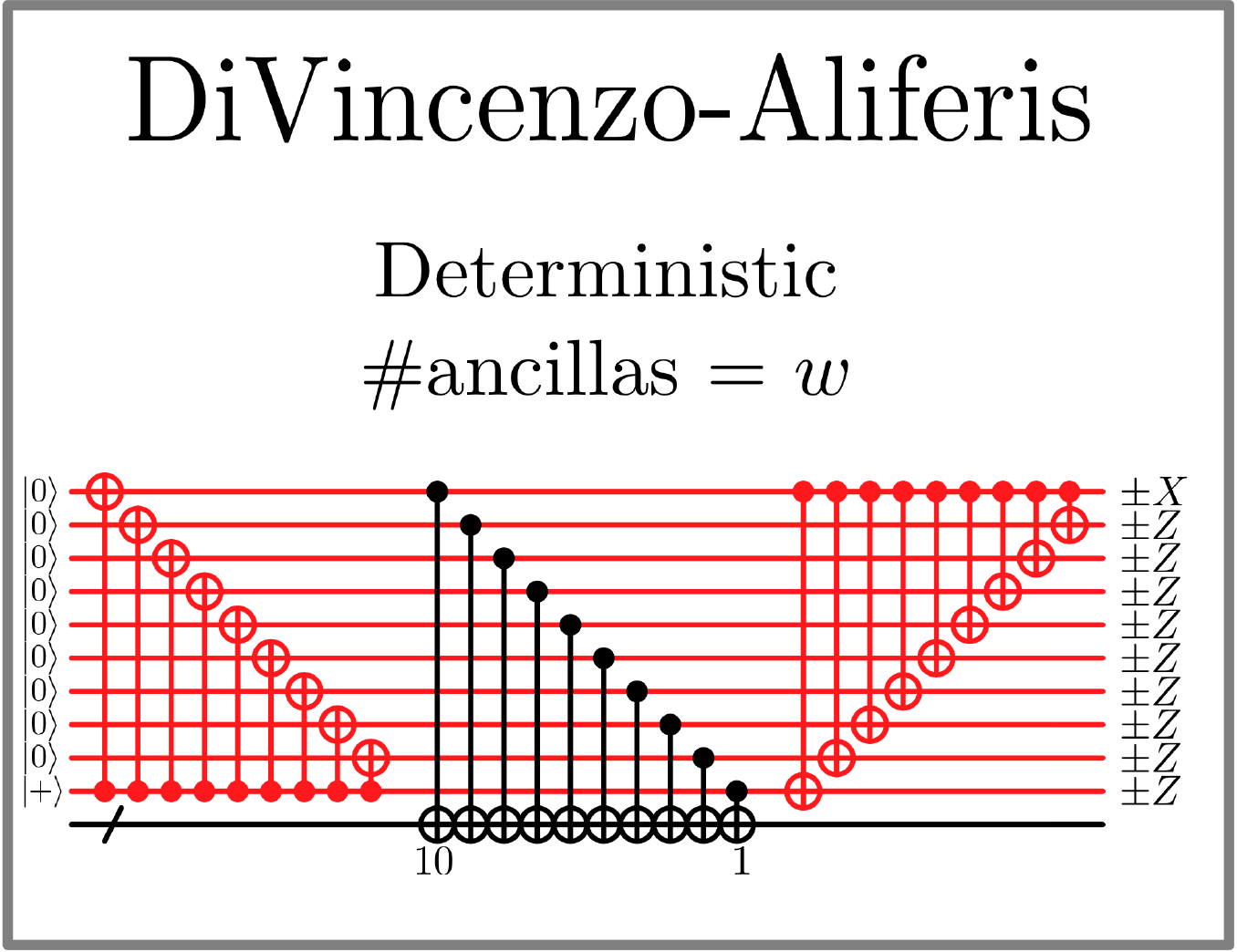}}
    \hspace{0.05cm}
    \subfloat[\label{f:CompDAcard} \cite{stephens14colorcodeft, YoderKim16trianglecodes, ChaoReichardt18fewqubitcomputation}]{
        \includegraphics[width = 0.45\textwidth]{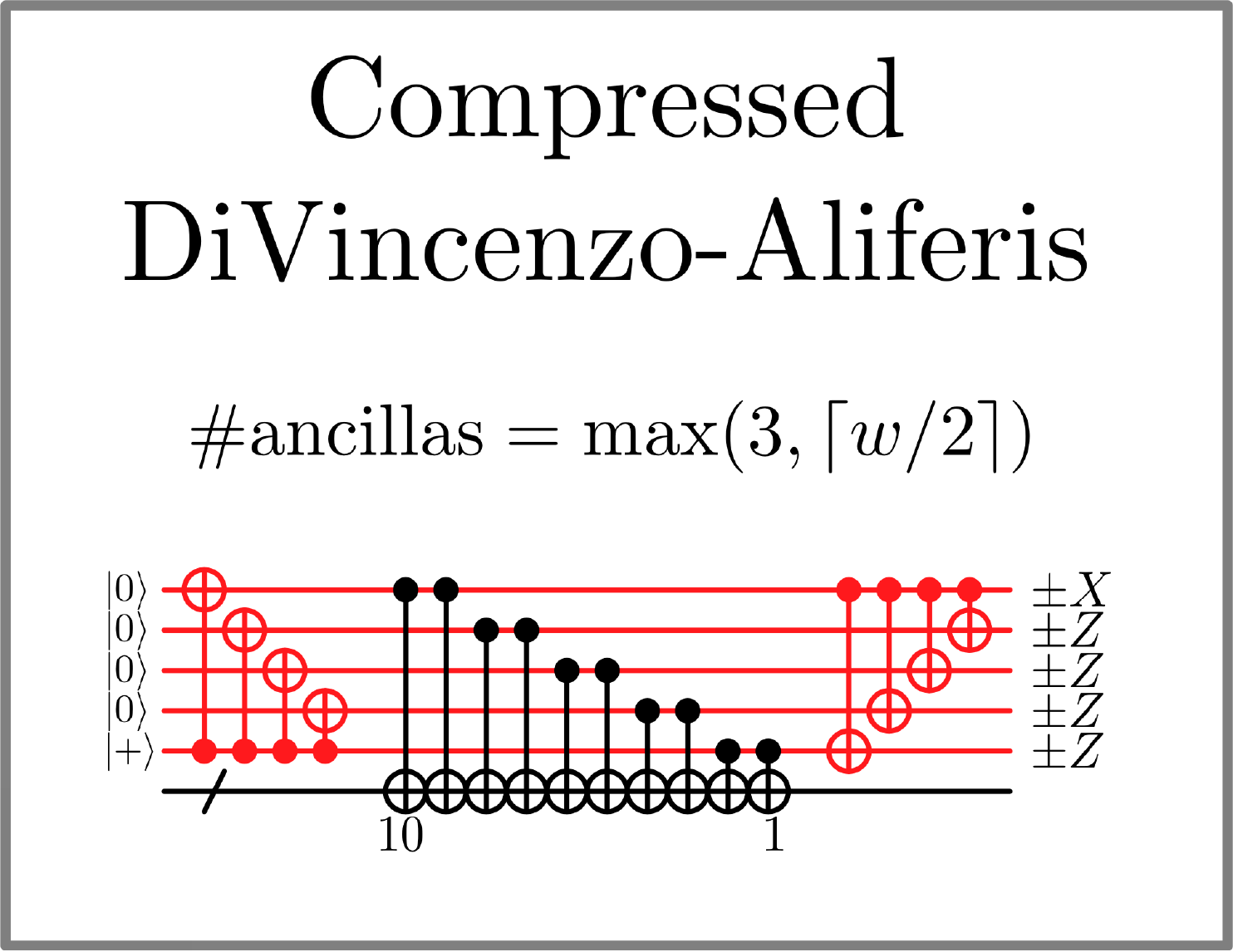}}
    \hspace{0.2cm}
    \subfloat[\label{f:Flagcard}]{
        \includegraphics[width = 0.45\textwidth]{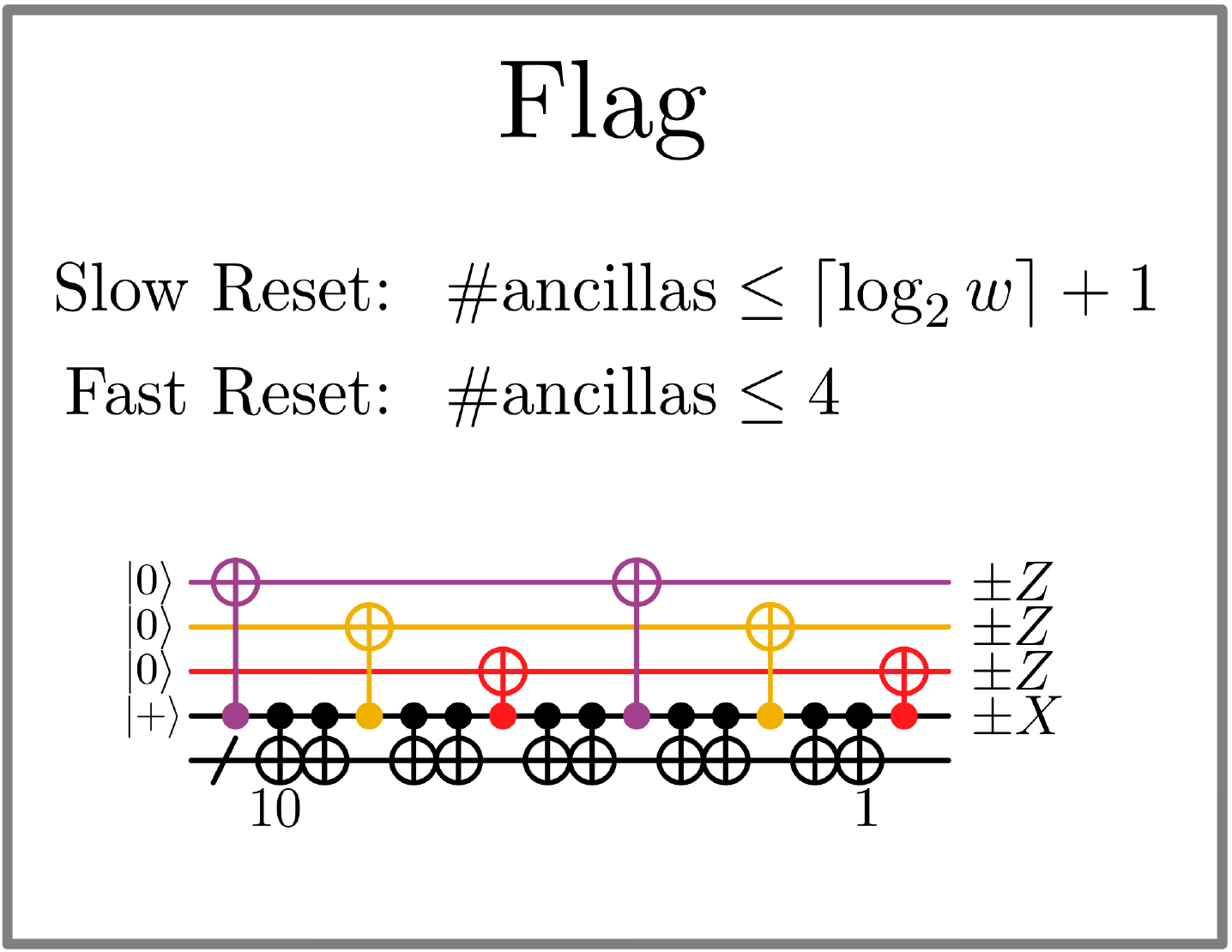}}
    \caption{Historical progression of stabilizer measurement circuits, illustrated by a weight-$10$ $X$ stabilizer measurement. 
    The black CNOTs have targets on the 10 data qubits, collectively represented by a black wire.  In~(b-d), fault-tolerance is only guaranteed to distance three and Pauli corrections, or frame updates, are applied to the data based on the $Z$ basis measurements.
    (a)~Shor's method uses $w+1$ ancillas and requires a fault-tolerantly prepared cat state. 
    (b,c)~These methods use unverified cat states with subsequent error decoding, 
    %.  Hence non-deterministic cat state verification is replaced by 
    giving a deterministic circuit.
    (d)~Our flag method prepares and unprepares an ancilla cat state while collecting the stabilizer.  Exponentially more flag patterns can thus be accessed for fault diagnosis.
    }
    \label{f:stabmeascomparison}
\end{figure}

The general model of flag-based fault tolerance is displayed in \figref{f:genove}. Here a set of flag ancilla qubits monitor operations in a non-fault-tolerant circuit and when measured at the end, produce flag patterns that identify mid-circuit faults. Based on the observed flag pattern, a correction is applied to the data to minimize the spread of errors. As an example, \figref{f:syndromemeasurementd3slowresetw10} measures a stabilizer on 10 data qubits while tolerating one fault. The three colored qubits are the flags and the measured flag patterns each imply different corrections.  Also note that the sequence of flag patterns $100, 110, 111, 011, 001$ is a path on the hypercube and corresponds to the order of the flag CNOTs, e.g., between $100$ and $110$ a CNOT targets flag qubit~$2$.

In this chapter, we restrict discussion to the measurement of individual stabilizers of a quantum code, as in Shor-style fault-tolerant stabilizer measurement~\cite{Shor96}. We do not consider measuring multiple stabilizers in parallel, as in~\cite{Steane97, Knill05, Huang21ShorSteane}.  \cref{f:stabmeascomparison} displays improvements made to Shor's method.  Note that Shor's method can tolerate any number of faults by increasing the fault tolerance of the cat state preparation.  The subsequent schemes forgo this property and are only fault-tolerant to distance three.  DiVincenzo and Aliferis first make the circuit deterministic by removing the need for cat state verification~\cite{DiVincenzoAliferis06slow}. This ensures that a circuit designer need not wait for a fault-tolerantly prepared cat state before measuring the stabilizer.  Subsequent improvements were made in~\cite{stephens14colorcodeft, YoderKim16trianglecodes, ChaoReichardt18fewqubitcomputation} to reduce ancilla count by coupling each ancilla qubit to two data qubits instead of one.  

With our flag method, the ancilla cat state is prepared and unprepared while collecting the stabilizer.  As in \figref{f:syndromemeasurementd3slowresetw10}, an $X$ fault occurring anywhere on the $\ket +$ qubit may spread into the data, but will also leaves its imprint on the flags. This is then measured out as a flag pattern. Due to the particular arrangement of the flag CNOTs, any fault that can spread to a data error of weight more than one triggers one of the five shown flag patterns. Each flag pattern then applies a unique correction that ensures that there is at most one data qubit in error.  This satisfies the condition for fault tolerance, which states that $k$ faults in a circuit should cause no more than $k$ qubits to have~errors. 

For the distance-three fault-tolerant measurement of a weight-$w$ stabilizer, we propose two methods based on the speed of qubit reset. With fast qubit reset, \thmref{t:fastresetd3syndromemeasurement}, only three flag ancillas are required in total, but each flag needs to be measured once per four data qubits.  If more flags are used in parallel, the number of accessible flag patterns grows exponentially and the number of measurements per ancilla converges to one. This is the regime of slow qubit reset, \thmref{t:syndromemeasurementd3slowreset}, which uses at most $\lceil \log_2 w \rceil$ flag ancillas measured only at the end.  

The rest of the chapter is divided into two sections. \secref{s:flagseq} details the construction of the two paths on the hypercube that we use as flag sequences. \secref{s:SMcatSDSM} describes how to use these sequences for distance-three fault-tolerant stabilizer measurement. In \secref{s:SMcatSDSM5}, we detail the distance-five and distance-seven cases.

\section{Flag sequences}
\label{s:flagseq}

A flag pattern is a string of $1$s and $0$s that arises from measuring flag qubits.  A flag pattern with $a$ flags is a vertex of the $a$-dimensional hypercube~$\{0,1\}^a$.  We show how to construct two maximal-length paths through the hypercube.  
Between sequential flag patterns only one bit changes, which in the fault-tolerant circuit constructions below will correspond to a CNOT from the syndrome qubit to that flag qubit.

The first type of flag sequence just requires a maximal-length traversal of the $a$-dimensional hypercube. A simple choice for this is the Gray code~\cite{Gray1953, Gardner1986Doughnuts}.  

\begin{lemma} \label{t:graycode}
For $a \geq 1$, the Gray code creates a length-$2^a$ Hamming path in the $a$-dimensional hypercube~$\{0,1\}^a$.
\end{lemma}

\begin{proof}
We construct the sequence inductively.
For $a = 1$, use $0, 1$.  For $a > 1$, first run the sequence for $a-1$ with $0$s appended, then run it backwards with $1$s appended.  
\end{proof}

%\vspace{-.1cm}%DEBUG
\noindent
For $a = 2$, e.g., the sequence is $00,\allowbreak 10,\allowbreak 11,\allowbreak 01$.
For $a = 3$, the sequence is $000,\allowbreak 100,\allowbreak 110,\allowbreak 010,\allowbreak 011,\allowbreak 111,\allowbreak 101,\allowbreak 001$.  

\smallskip%DEBUG

The second type of sequence is related to the degree of fault tolerance of the circuit.  By definition, fault tolerance to distance $d$ implies that for all $k \leq t =  \lfloor \frac{d-1}{2} \rfloor$, correlated errors of weight $k$ occur with $k$-th order probability.  For distance-three Calderbank-Shor-Steane (CSS) fault-tolerant syndrome measurement, any single fault should result in a data error with $X$ and $Z$ components having weight zero or one.  

In order to ensure that the circuit is distance-three fault-tolerant, we need to ensure that a measurement fault on any one ancilla qubit does not trigger corrections of weight greater than one.  Hence the second maximal-length sequence requires that there are no weight-one strings except at the start and end.  As shown in \figref{f:syndromemeasurementd3slowresetw10}, we may assign weight-one corrections to these two patterns, but for all the others, multi-qubit corrections~are~required.

%\pagebreak %DEBUG  \comment{I don't think this is needed. -BR}

\begin{figure}
\centering
\includegraphics[width = 0.472\textwidth]{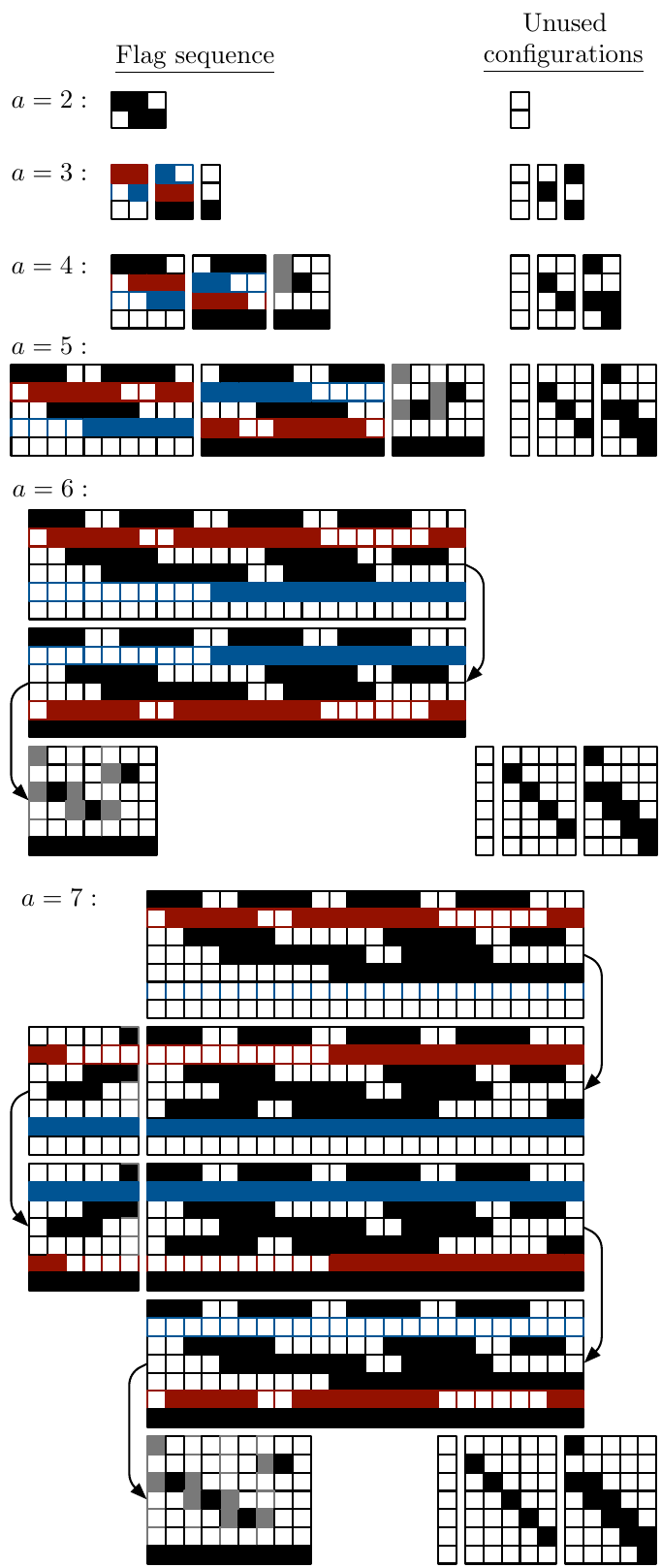}
\caption{
Flag sequences for distance-three fault-tolerant syndrome measurement, using $a$ flag qubits, each measured once (the slow reset model).  These sequences are walks through the $a$-dimensional hypercube, from $10^{a-1}$ to $0^{a-1}1$; passing through each vertex at most once and no other weight-one vertices.  Flag patterns are stacked vertically and ordered initially left to right, with solid and empty squares representing $1$ and $0$, respectively, e.g., \protect\raisebox{-.05cm}{\protect\includegraphics[scale=.4]{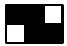}} represents $10, 11, 01$.} 
\label{f:slowresetdistance3flagsequences}
\end{figure}

\begin{lemma} \label{t:slowresetdistance3flagsequences}
For $a \geq 2$, in the $a$-dimensional hypercube $\{0,1\}^a$ there exists a path $v_1 = 10^{a-1}, \ldots, v_n = 0^{a-1}1$, with length $n = 2^a - 2 a + 3$, such that all $v_2, \ldots, v_{n-1}$ have weight at least two, and none repeat.  
\end{lemma}

\begin{proof}
\def\indicator{{1}}
Let $\indicator_S \in \{0,1\}^a$ denote the vertex that is $1$ exactly for indices in $S$.  
\cref{f:slowresetdistance3flagsequences} illustrates the inductive construction.  For $a=2$, the sequence is the same as that in \lemref{t:graycode}. The base case of our inductive proof is with $a=3$, where the sequence is $100, 110, 111, 011, 001$. For $a > 3$, first run the previous sequence for $b = a - 1$ with $0$s added to the bottom, up to the second-to-last element 
(which for $b \geq 3$ is $\indicator_{\{2, b\}}$).  Then run the sequence backward, except with $1$s added to the bottom, and swapping coordinates $2$ and $a-1$ (the red and blue rows in the figure).  Finally, finish the sequence from $\indicator_{\{1, a\}}$ by walking through $\indicator_{\{3, a\}}, \indicator_{\{4, a\}}, \ldots, \indicator_{\{a - 2, a\}}, \indicator_{\{2, a\}}$, with the appropriate weight-three sequences $\indicator_{\{1, 3, a\}}, \indicator_{\{3, 4, a\}}, \ldots,$ $\indicator_{\{a - 1, a - 2, a\}}, \indicator_{\{a - 2, 2, a\}}$ (shown in gray) interposed.

To ensure that no vertex is visited more than once, one need only check that the last $2a - 5$ sequences are distinct from those that came before.  For this, one can track by induction the $2a - 3$ hypercube vertices that are not visited by each walk: $0^a$, the $a - 2$ weight-one strings $\indicator_2, \ldots, \indicator_{a - 1}$, and the $a - 2$ weight-two strings $\indicator_{\{1,3\}}, \indicator_{\{3,4\}}, \indicator_{\{4,5\}}, \ldots, \indicator_{\{a - 1, a\}}$.  
Thus, the sequence has total length $2^a - (2 a - 3)$.  
\end{proof}

The length $2^a - 2a + 3$ is maximal.  
This follows since there are $2^{a-1} - a$ vertices with odd weight more than one,
and vertices must alternate odd and even weights.

\section[Distance-three stabilizer measurement]{Distance-$3$ fault-tolerant stabilizer measurement}
\label{s:SMcatSDSM}

In this section, we outline two protocols for distance-three CSS fault-tolerant stabilizer measurement.  They differ based on the speed of qubit measurement and reset.

For $w \in \{4, 5, 6\}$, flag-fault-tolerant circuits are constructed the same way regardless of qubit reset speed.  We show in \figref{f:slowresetd3w6stabilizermeasurement} that for $w = 6$, only two flag qubits are required.  Lower-weight stabilizers can be measured by removing data CNOTs and making appropriate changes to the Pauli corrections.  For $7 \leq w \leq 10$, the different constructions yield the same circuits.  It is only for $w > 10$ that the effects of qubit reset speed are~pronounced.

\subsection{Fast reset}

\begin{theorem} \label{t:fastresetd3syndromemeasurement}
If qubits can be measured and reset quickly, then for any $w$, four ancilla qubits are sufficient to measure the syndrome of $X^{\otimes w}$, CSS fault-tolerantly to distance three.  Moreover, the number of measurements needed is $\lceil \tfrac{w + 2}{4} \rceil + 1$.  
\end{theorem}

\begin{proof}
For $w \in \{ 4, 5, 6\}$, the circuit using two flag ancillas is shown in \figref{f:slowresetd3w6stabilizermeasurement}.  It runs through a sequence of three flag patterns and a multi-qubit correction is only applied for the flag pattern $11$.  For $w > 6$, the general construction is shown in \figref{f:syndromemeasurementd3fastreset}.  Each repetition of the highlighted region adds the $X$ parity of four more data qubits, while measuring and quickly reinitializing one flag qubit.  In terms of the number of measurements~$m$, the construction achieves up to $w = 4 \, (m - 1) - 2$.  
It is fault-tolerant because $X$ faults on the control wire cause flag patterns of alternating weights two or three, that localize the fault to three possible consecutive locations along the control wire: before, between or after two CNOT gates.  The appropriate correction, ensuring distance-three fault tolerance, is for a fault between the CNOT gates.  
\end{proof}

% \subfigure[$w = 6$, $a = 7$]{\includegraphics[width=.325\textwidth]{imagesChap2/dist5w6}}    

% \label{f:slowresetd3w6stabilizermeasurement}
% \label{f:slowresetd3w6CSP}

\begin{figure}
    \centering
    \subfloat[\label{f:slowresetd3w6stabilizermeasurement}]{
        \includegraphics[width = 0.47\textwidth]{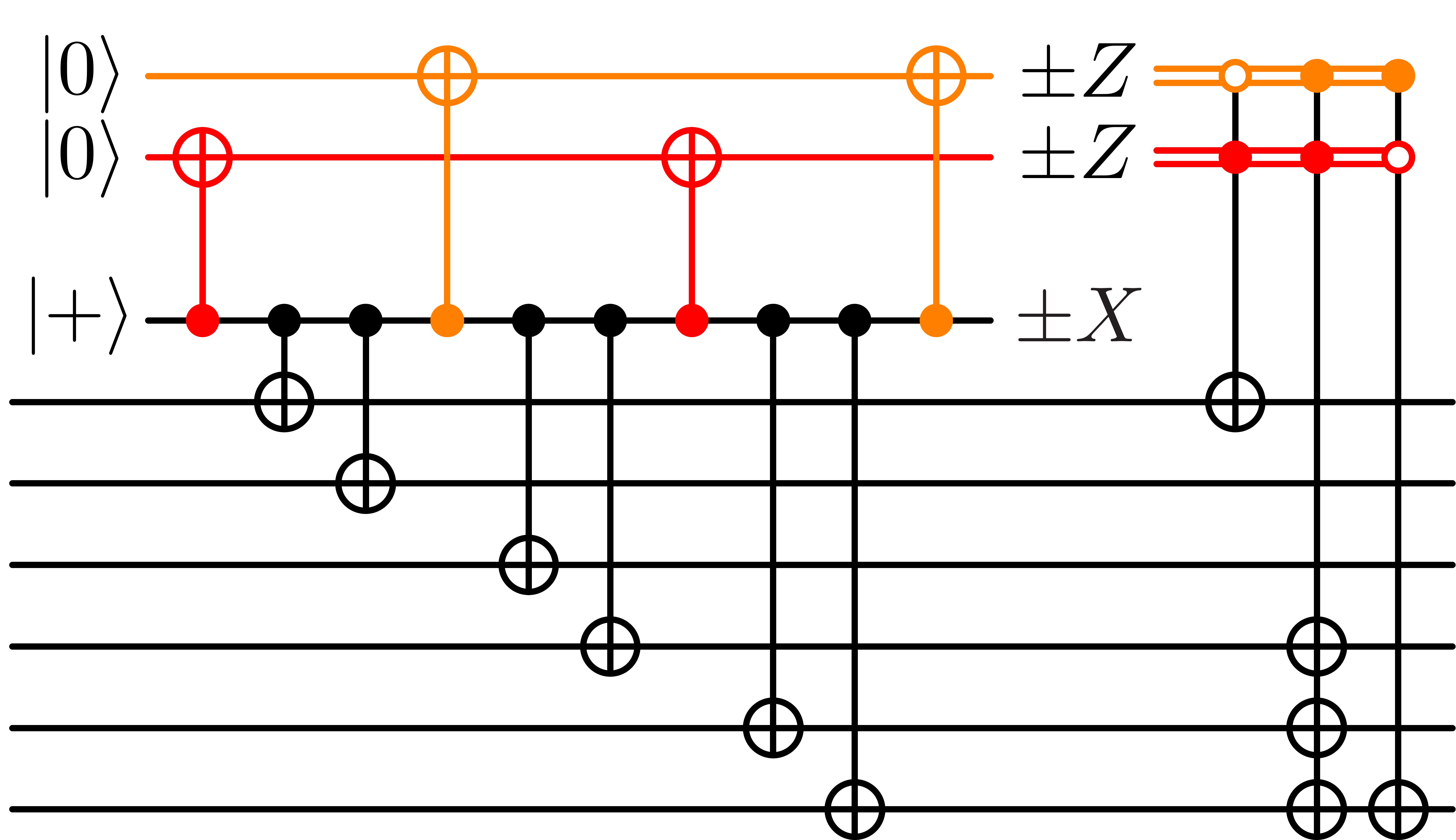}}
    \hspace{0.2cm}
    \subfloat[\label{f:slowresetd3w6CSP}]{
        \includegraphics[width = 0.44\textwidth]{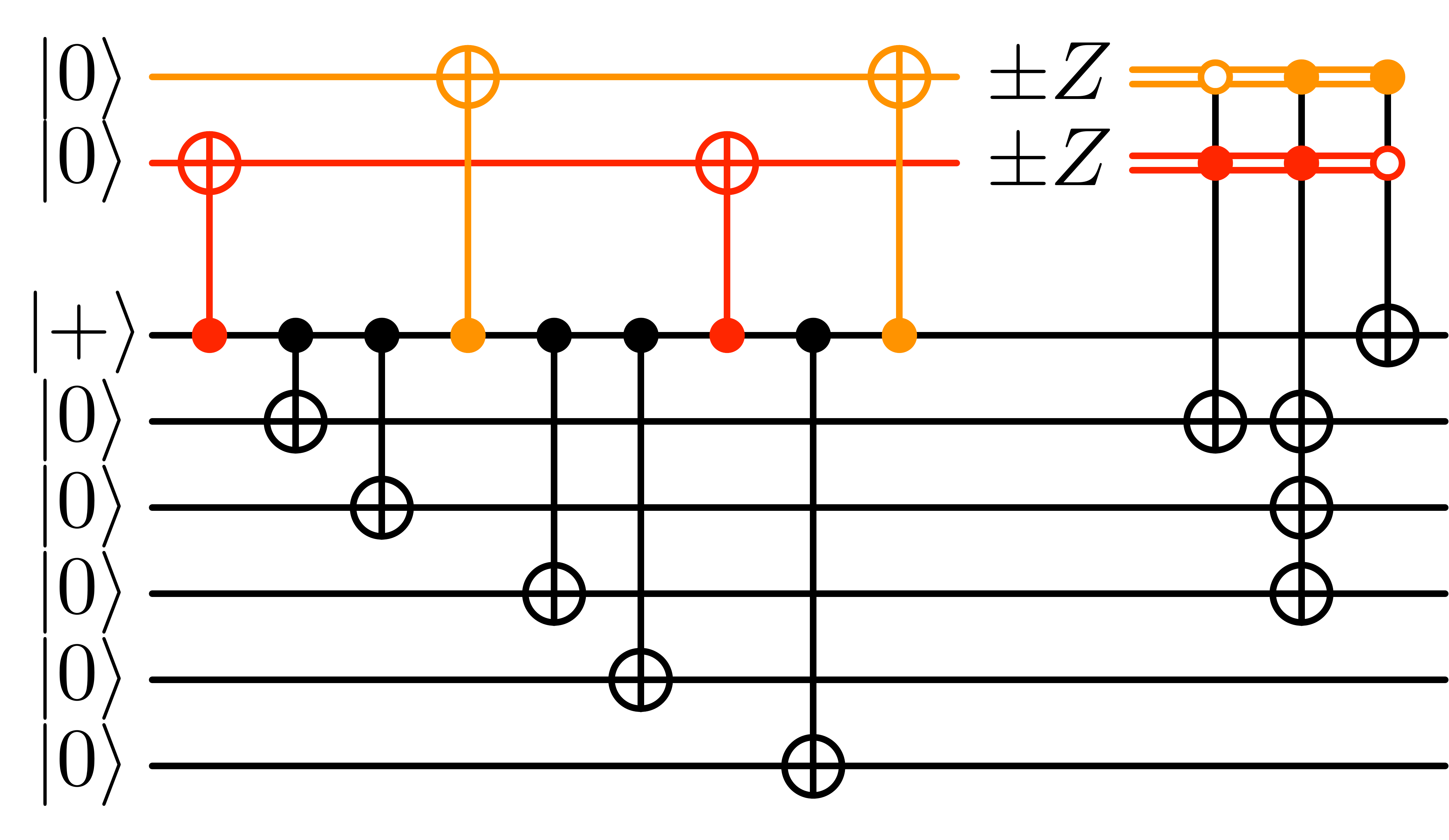}}
    \caption{(a) Circuit to measure an $X^{\otimes 6}$ stabilizer, CSS fault-tolerant to distance three.  (b) Circuit to prepare a six-qubit cat state, fault-tolerant to distance three.}
    \label{f:distance3w6SMandcat}
\end{figure}

\begin{figure}
    \centering
    \includegraphics[width=0.999\textwidth]{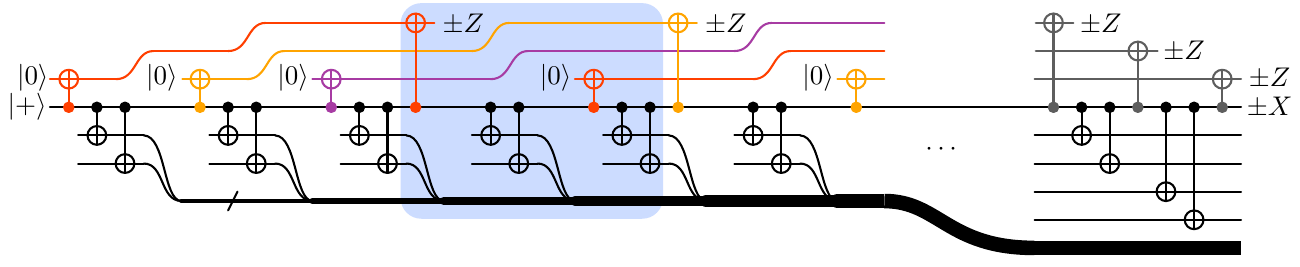}
    \caption{Distance-three fault-tolerant syndrome bit measurement only needs three flag qubits. The highlighted region can be repeated to fit the weight of the stabilizer being measured.}
    \label{f:syndromemeasurementd3fastreset}
\end{figure}

\begin{figure}
    \centering
    \includegraphics[width=.41\textwidth]{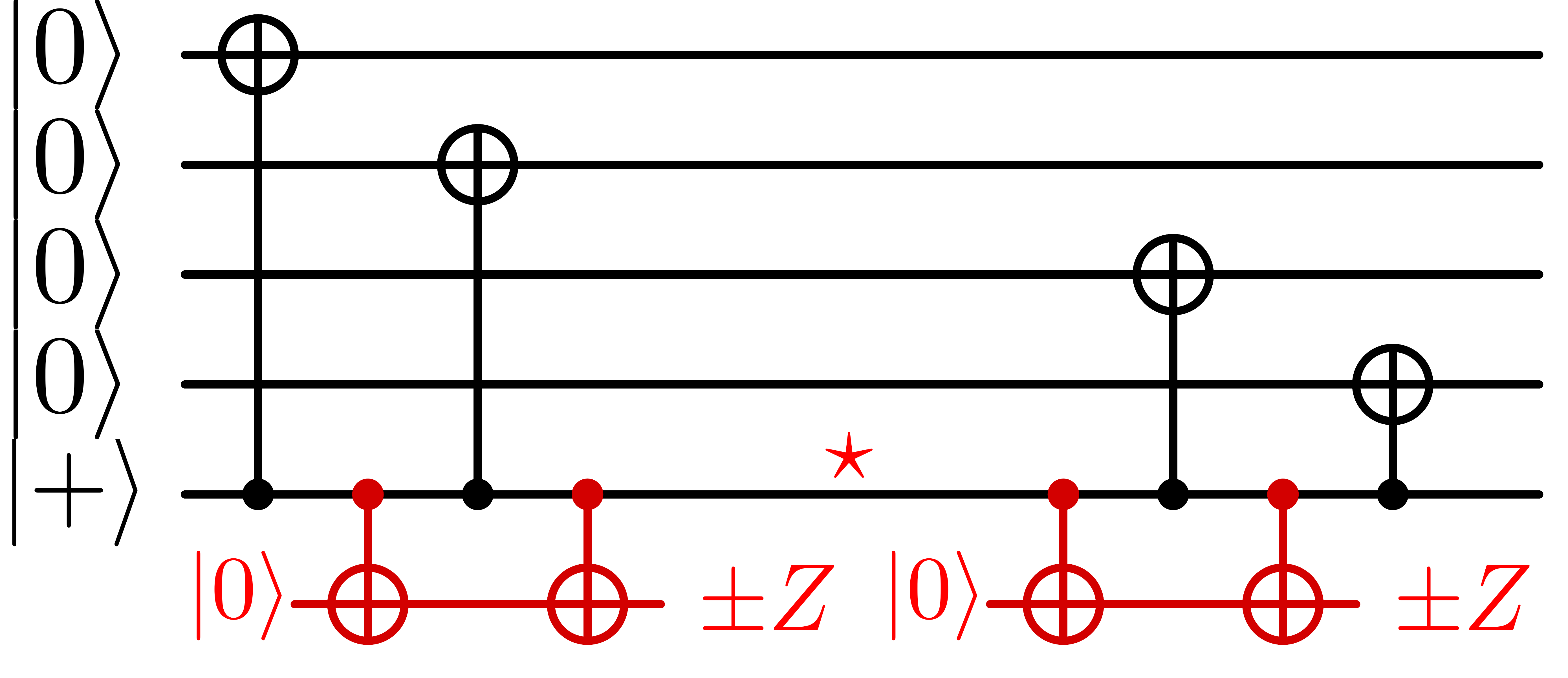}
    \includegraphics[width=.38\textwidth]{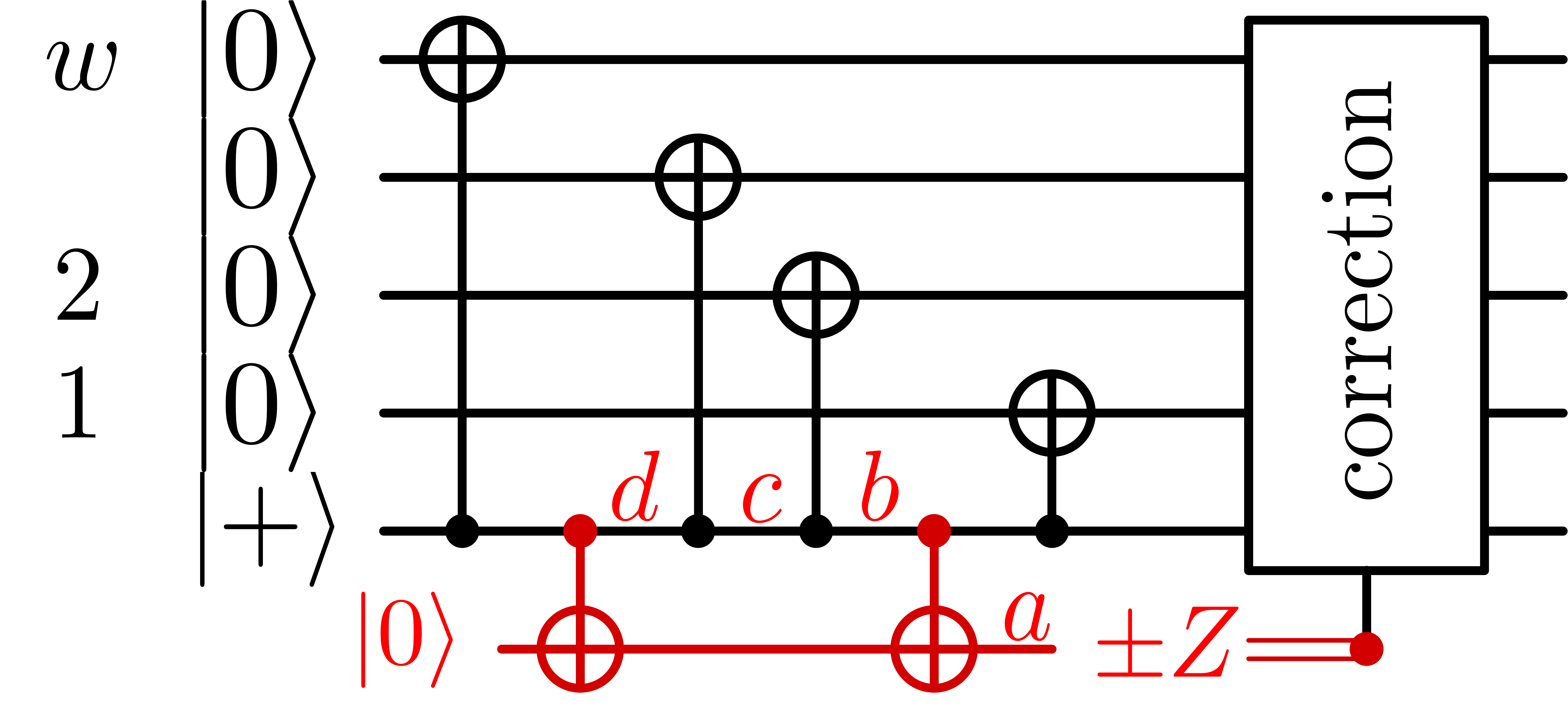} %oneancillad3
    \caption{Distance-three error correction is not possible with one flag qubit.  Either (top) the control wire is unprotected at some point $\star$, from which an $X$ fault can propagate to an error of weight at least two; or (bottom) faults at $a$, $b$, $ c$, $ d$, causing respective errors $I$, $X_1$, $X_1 X_2$, $X_w$ have no consistent~correction.} 
    \label{f:oneflagprotection}
\end{figure}

\thmref{t:fastresetd3syndromemeasurement} may be optimal; it does not appear to be possible to use fewer than three flag qubits.  With just one flag qubit, one can detect that an error has occurred, but not where.  As illustrated in \figref{f:oneflagprotection}, either the control wire is unprotected at some point or for $w \geq 4$ there is no consistent correction rule.  

By a similar argument, two flag qubits are not enough.  Any correction based on a single flag can have weight at most one, since the flag measurement itself could be faulty.  However, if at some point in the middle the control wire is protected by just a single flag, a weight-one correction will not suffice.  On the other hand, if both flags are used to protect the control wire across the entire sequence of CNOT gates, we are unable to locate faults well enough to correct them.  

We remark that this construction can also be used to prepare a $w$-qubit cat state fault-tolerantly to distance three. The conversion follows three steps: 1.\ Remove one data qubit.  2.\ Initialize the data qubits as $\ket 0$.  3.\ Remove the syndrome ancilla measurement, so as to retain it in the support of the stabilizer.  An example of this conversion is shown for $w = 6$ in \figref{f:slowresetd3w6CSP}.  In \chapref{chap:catstateprep}, we suggest different protocols that use just one ancilla qubit.

\subsection{Slow reset}
\label{subsec:stabmeasslow}

\begin{theorem} \label{t:syndromemeasurementd3slowreset}
The syndrome of $X^{\otimes w}$ can be measured CSS fault-tolerantly to distance three using $m \geq 3$ measurements, provided that 
\begin{equation*}
w \leq 2 \, (2^{m - 1} - 2 (m - 1) + 3) \, .
\end{equation*}
\end{theorem}

\begin{proof}
Two examples are shown in \figref{f:slowresetd3w6stabilizermeasurement}, for $w = 6$, and \figref{f:syndromemeasurementd3slowresetw10}, for $w = 10$.  
As in these figures, in general we collect the syndrome two qubits at a time into a syndrome qubit that is initialized as $\ket +$.  Between each of these pairs of CNOT gates, a CNOT is applied from the syndrome qubit into one of $m - 1$ flag qubits.  This leads to a sequence of flag patterns, e.g., $100, 110, 111, 011, 001$ for the $w = 10$ example.  Based on the observed flag pattern, a correction is applied as if an $X$ fault had occurred between the corresponding pair of flag CNOT gates.  

Observe that the flag sequence changes one bit at a time; it can be thought of as a path on the hypercube.  It begins and ends with weight-one patterns, but otherwise the patterns all have weight at least two.  This is important for distance-three fault tolerance because a fault could affect the flags, and only the first and last data corrections have weight one.  Also, the flag patterns along the sequence are distinct, so each is associated with only one correction.  The theorem then just follows using the flag sequence construction in \lemref{t:slowresetdistance3flagsequences}.  
\end{proof}

Note that the approach of \thmref{t:syndromemeasurementd3slowreset}, with slow reset, is different from the fast reset case of \thmref{t:fastresetd3syndromemeasurement}, in that a flag qubit is active and able to detect faults in more than one region of the circuit.  

\subsection{Space-time cost}
\label{subsec:SMcatSDSMspacetime}

Here, we count the circuit depth and number of ancillas used in our distance-three fault-tolerant stabilizer measurement circuits. Parallelization can substantially reduce circuit depth. \tabref{f:resourcereqs} compares our flag method for measuring a weight-$w$ stabilizer to the earlier methods in \figref{f:stabmeascomparison}. Also considered is a parallelized Shor method, in which the initial cat state is prepared in logarithmic depth, with $w/4$ extra ancilla qubits for postselection checks. The Shor methods must pass the postselection checks, and so they are non-deterministic protocols. \tabref{f:resourcereqs} shows the best case, where all the checks pass. Note that the flag and parallelized Shor methods both have $\text{space} \times \text{depth}$ cost scaling as $O(w \log w)$, with the leading coefficient in favor of the flag method.

\begin{table}
    \caption{Space and time costs for measuring a weight-$w$ stabilizer using different distance-three fault-tolerant stabilizer measurement circuits. In the following, all the logarithms are base $2$. The flag method requires the fewest ancillas and has low depth, allowing for the smallest cost when computing $\# \text{ancillas} \times \text{depth}$. }
    \centering
    \begin{tabular}{c c c c}
        \hline
        \hline
        Protocol & Ancillas & Depth & Ancillas$\times$Depth \\
        \hline
        Shor & $w+1$ & $w/2 + 3$ & $O(0.5 w^2)$ \\
        Shor-Par & $5w/4$ & $3\log w - 1 $ & $O(3.75 w \log w)$ \\
        DA & $w$ & $2w - 1$ & $O(2 w^2)$ \\
        Compressed DA & $w/2$ & $3w/2 - 2$ & $O(0.75 w^2)$ \\
        Flag & $\log w +1$  & $3w/2 + O(1)$ & $O(1.5w \log w)$ \\[0.25cm]
        Not fault-tolerant & $1$ & $w$ & $O(w)$ \\
        \hline
        \hline
    \end{tabular}
    \label{f:resourcereqs}
\end{table}

Using a standard depolarizing noise model, we simulate noisy versions of the different circuits to determine statistics of the weight-one and weight-two errors. Specifically:
\begin{itemize}
    \item With probability $p$, the preparation of $\ket 0$ is replaced by $\ket 1$ and vice versa---similarly for $\ket +$ and $\ket -$.
    \item With probability $p$, an $X$ or $Z$ measurement has its outcome flipped.
    \item With probability $p$, a one-qubit gate is followed by a Pauli error drawn uniformly at random from $\{ X, Y, Z\}$.
    \item With probability $p$, the two-qubit CNOT gate is followed by a two-qubit Pauli error drawn uniformly at random from $\{ I, X, Y, Z\}^{\otimes 2} \setminus \{I \otimes I \}$.
\end{itemize}
There are no errors on idle resting qubits. 

\begin{figure}
    \centering
    \includegraphics[width=.999\textwidth]{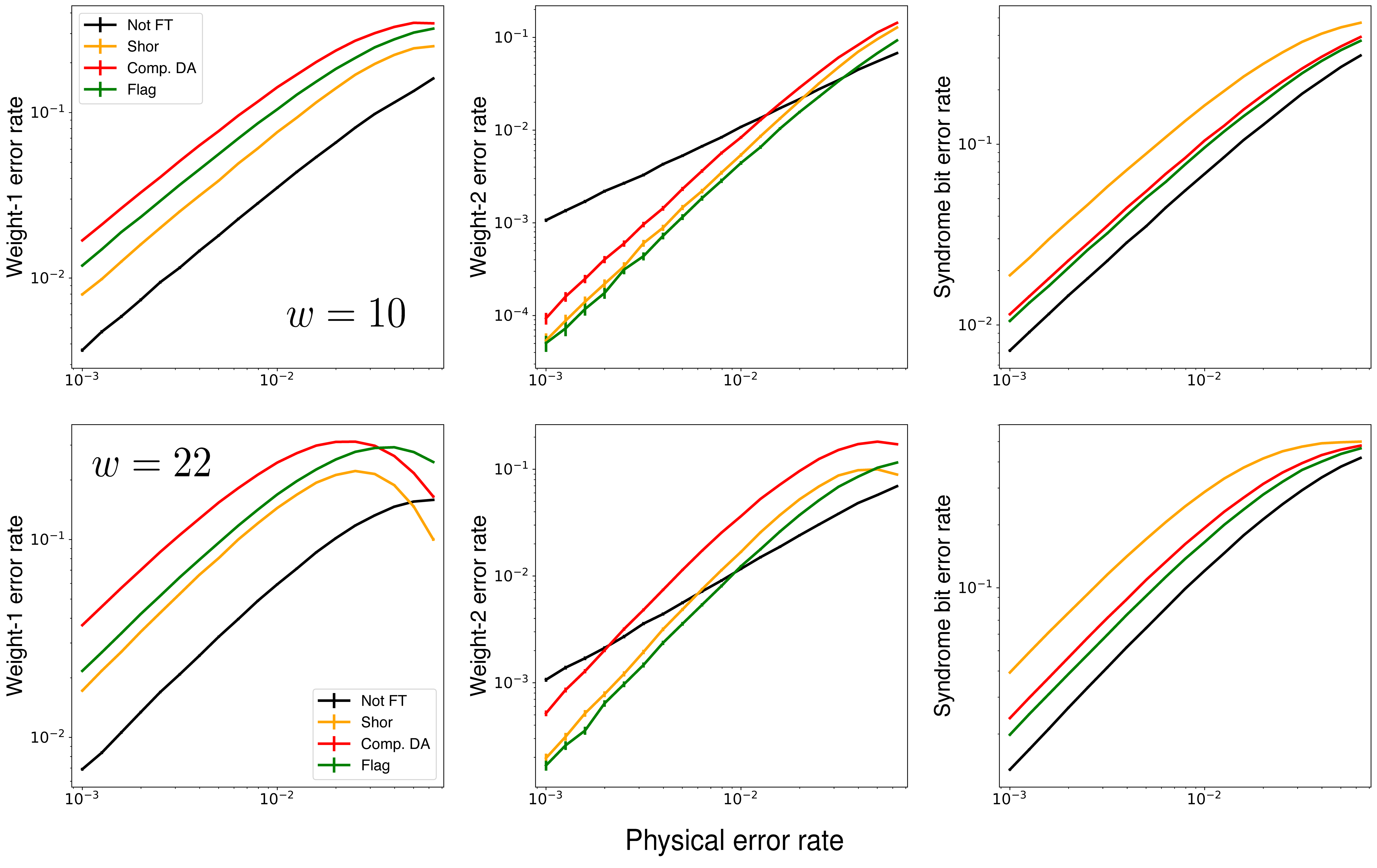}
    \caption{Simulation of the noisy measurement of an $X ^{\otimes 10}$ and $X ^{\otimes 22}$ stabilizer at physical error rate $p \in \{ 10^{-3}, 10^{-2} \}$ using different distance-three fault-tolerant circuits: Shor-style, compressed Divincenzo-Aliferis, and the flag method of \secref{subsec:stabmeasslow}. In the first and second column of graphs, we show the rate of weight-one and weight-two data errors due to these circuits, with $99 \%$ error bars. In the third column, we show the rate at which the measured syndrome bit is wrong.} 
    \label{f:sims}
\end{figure}

\cref{f:sims} shows the rates of weight-one errors and weight-two errors for different input physical error rates $p$. The rate of weight-one errors is lowest in the non-fault-tolerant circuit, since it contains the fewest locations for faults. The Shor method has a lower weight-one error rate than the flag method, but among the deterministic fault-tolerant methods, the flag method performs the best. For larger stabilizers ($w=22$), the curves for the Shor method and the flag method are closer, implying that the difference in the rate of weight-one errors between the Shor method and the flag method is reduced. 

As expected, the rate of weight-two errors of the three fault-tolerant protocols scales quadratically with $p$, allowing for a lower probability of weight-two errors below a pseudothreshold (physical error rate below which a fault-tolerant method achieves lower weight-two error rate than the non-fault-tolerant method). Notice that the flag method has the highest pseudothreshold. Moreover, as the stabilizer weight is increased, the pseudothreshold of the flag method decreases slower than those of the other fault-tolerant methods. Asymptotically, the flag method admits the highest pseudothreshold for weight-two errors, but incurs more weight-one errors than the probabilistic Shor method. Additionally, we compute the rate of errors on the syndrome bit, as this determines how much fault tolerance will be needed to correct faulty syndrome information~\cite{Delf22}. The rate of faulty syndrome bits is lowest when using the flag method for fault tolerance.

\section[Distance-five/seven stabilizer measurement]{Stabilizer measurement tolerating more than 1 fault}
\label{s:SMcatSDSM5}

Distance-five fault tolerance is interesting for stabilizers of weight $w \geq 6$.  For $w \in \{ 6, 7, 8\}$, the circuits in \figref{f:syndromemeasurementd5} with seven ancilla qubits are distance-five fault-tolerant.  We present a general method to construct stabilizer measurement circuits for arbitrary $w$ in \figref{f:d5arbwsyndmeas}.  By computer simulation, we verify the fault tolerance of this construction for $w$ up to $90$ qubits.  The general construction proceeds as follows. First, five flag qubits are activated. For each additional flag that is needed, the gates in the shaded blue region are applied.  These gates deactivate an existing flag and activate a new flag.  Finally, when no additional flags are needed, the flags are deactivated in the order $\{ 2,4,1,5,3\}$. $1$ denotes the flag that has been active for the longest time and $5$, the flag that has been active for the shortest time.  To ensure faults are correctly flagged, it is necessary to ensure there is asymmetry between the order in which flags are activated and the order in which they are deactivated. This is in contrast to the distance-three DiVincenzo-Aliferis method in \figref{f:CompDAcard}, where both the orders are symmetric.

In \figref{f:d5arbwsyndmeas}, the thick black line indicates the $w$-qubit register of data qubits that are in the support of the stabilizer.  Data CNOT gates (in black) are applied to qubits $\{ w, w-1, \mathellipsis , 1\}$ after every flag CNOT (in red).  The last data CNOT must be placed either before the third-last or second-last flag CNOT.  The addition of another data CNOT gate before the last flag CNOT results in uncorrectable errors. 

If there are $a$ ancilla qubits, one can measure a weight-$(2 a - 5)$ or weight-$(2 a - 4)$ stabilizer.  Hence a weight-$w$ CSS stabilizer may be fault-tolerantly measured to distance-five, for $w \leq 2a - 4$.  Note that at most five flag qubits are active at any instant.  Hence with fast qubit reset, one only requires five flag ancillas and one syndrome ancilla to measure an arbitrary weight stabilizer fault-tolerantly to distance-five.

For distance-seven fault-tolerance, there are differences in the spacing between data CNOT gates and the order in which flag ancillas are activated and deactivated.  We show how to construct circuits for stabilizer of arbitrary weight $w$ by first discussing a circuit for a weight-$17$ stabilizer, shown in \figref{f:d7w17}. We chose $w=17$ since the circuit is non-trivial and its construction encompasses all the tricks needed to construct circuits for arbitrary weight.  In general, compared to \figref{f:d5arbwsyndmeas}, the number of ancilla CNOT gates between data CNOT gates is doubled, except in the center of the circuit, where it is tripled for the length of four data CNOT gates.  For odd $w$, the number of ancilla CNOTs between the $w-1$ subsequent pairs of data CNOT gates is the sequence $\{ (\lceil \frac{w-6}{2} \rceil  \text{ 2's} ), 3, 3, 3, 3 , (\lfloor \frac{w-6}{2} \rfloor \text{ 2's}), 1\}$, as shown in \figref{f:d7w17}. For even $w$, the sequence is $\{ ( \frac{w-6}{2} \text{ 2's} ), 3, 3, 3, 3 , (\frac{w-6}{2} \text{ 2's}), 1\}$. Note that, as shown in \figref{f:d7w17}, one additional ancilla CNOT gate is required at the start.

Next we comment on the order in which ancilla qubits are deactivated as flags. Similar to the distance-five case, after initially activating seven flags, a flag is deactivated to activate a new flag qubit. An active group of seven flags is closed in the order $\{ 2,4,6,1,3,5,7 \}$. As these seven flags are closed, seven new flags are simultaneously opened. The process repeats unitl there are exactly seven remaining flags to close. These last seven flags are also closed in the same order $\{ 2,4,6,1,3,5,7 \}$. In \figref{f:d7w17}, flag ancillas are shown in alternating colors to highlight the order that flags are activated and deactivated. Distance-seven fault-tolerance was verified by computer simulation for stabilizer weight up to $32$. The number of flag ancillas needed to measure a weight-$w$ stabilizer is $w+1$.  

The techniques described in this section may also be used to develop resource-efficient circuits that are fault-tolerant to higher distance.

\begin{figure}
    \centering
    \subfloat[$w = 6$, $a = 7$]{
        \includegraphics[width = 0.45\textwidth]{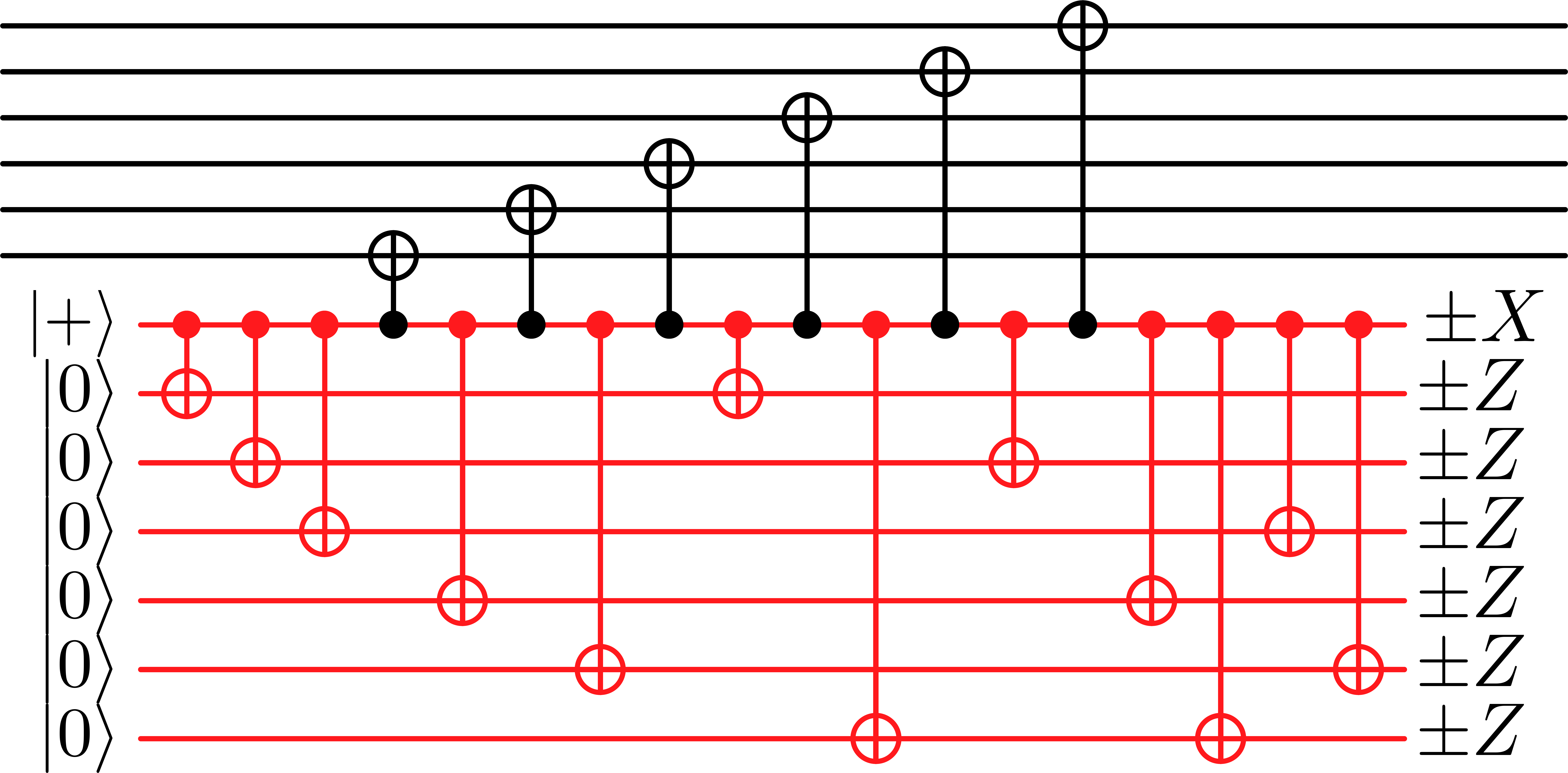}}
    \hspace{0.4cm}
    \subfloat[$w = 7$, $a = 7$]{
        \includegraphics[width = 0.45\textwidth]{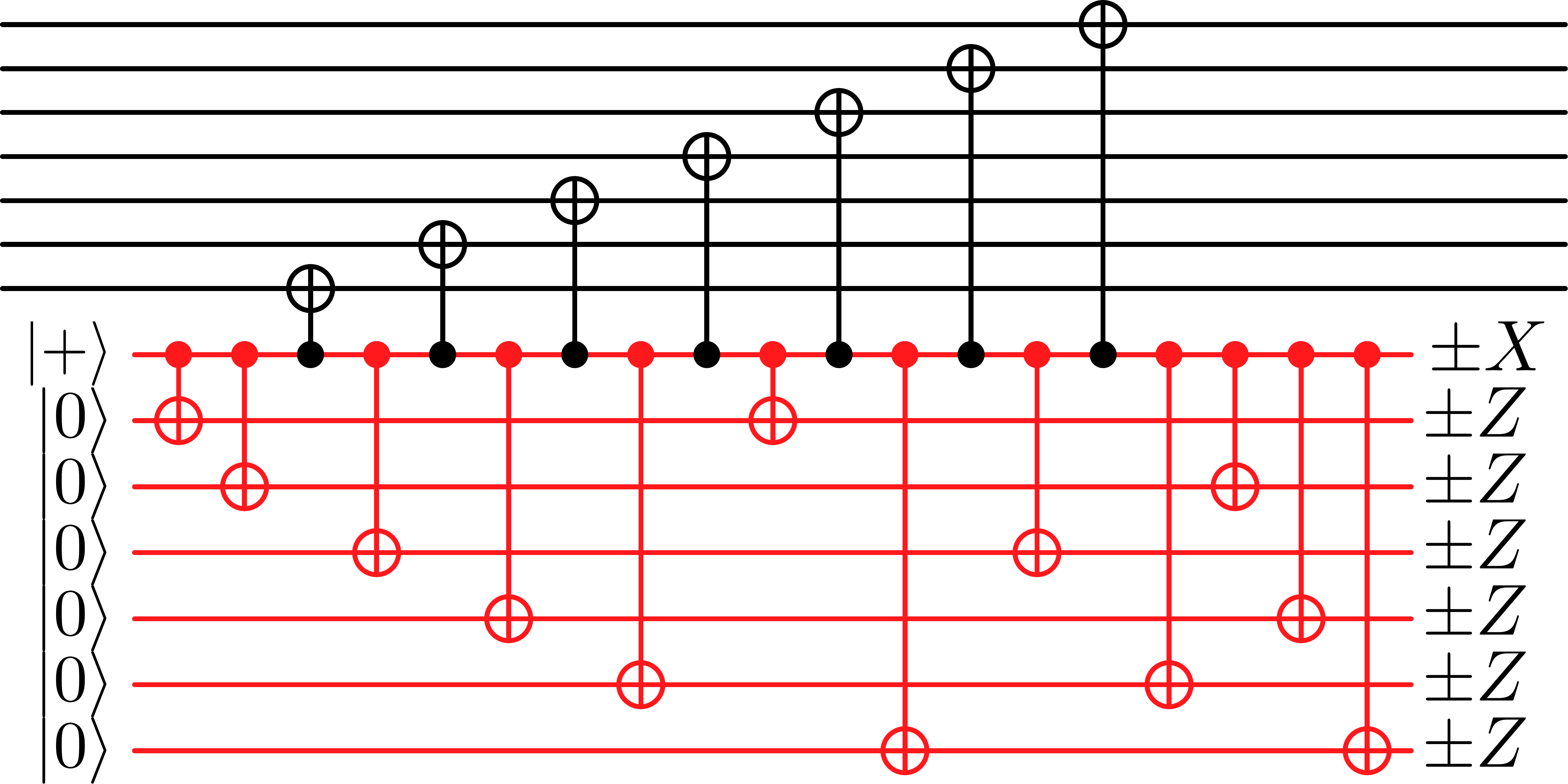}}
    \hspace{0.4cm}
    \subfloat[$w = 8$, $a = 7$]{
        \includegraphics[width = 0.45\textwidth]{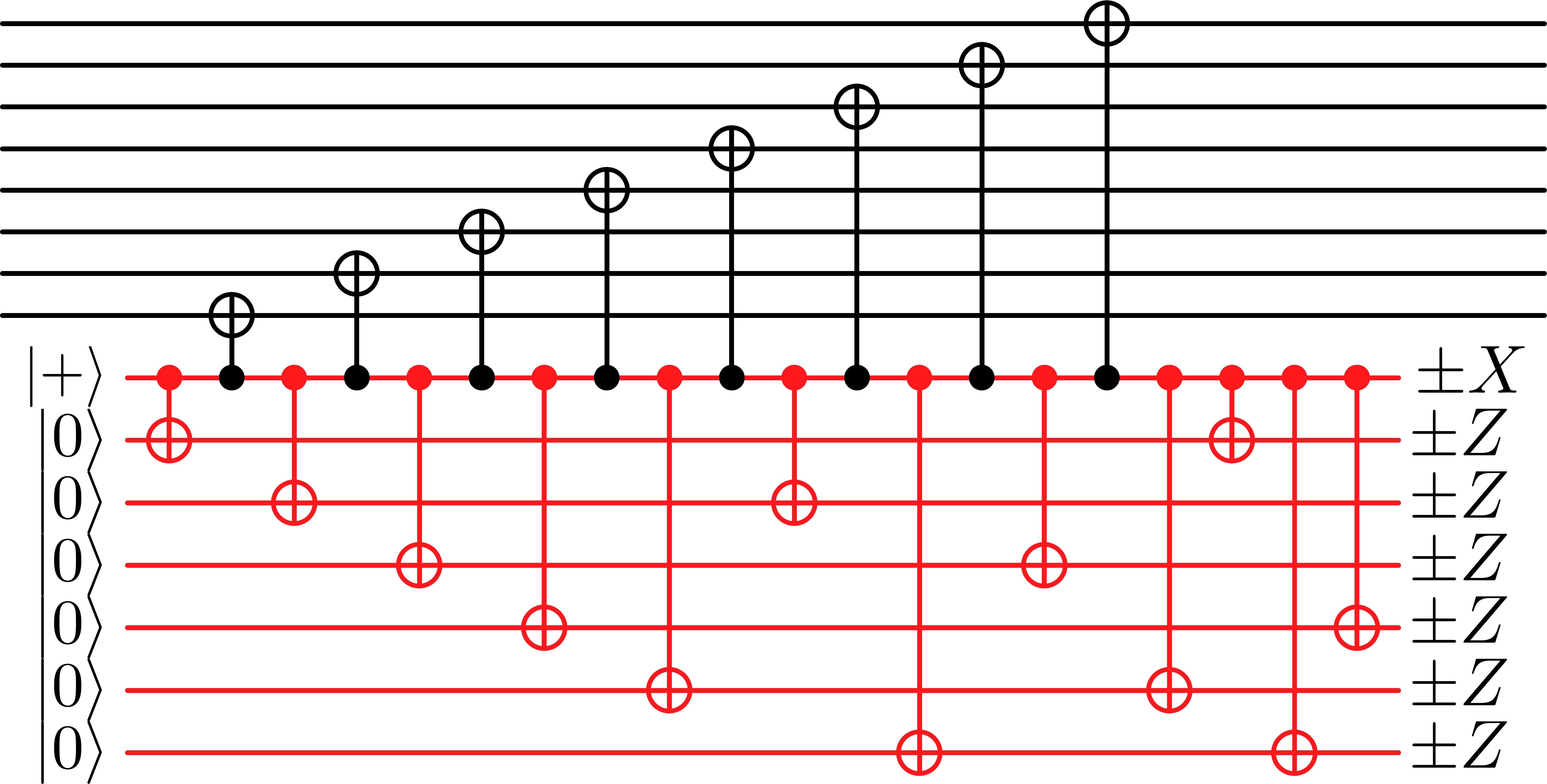}}
    \caption{Distance-five CSS stabilizer measurement with slow qubit reset for $w \in \{ 6, 7, 8\}$.  Red wires indicate syndrome and flag qubits.} 
    \label{f:syndromemeasurementd5}
\end{figure}

\begin{figure}
    \centering
    \includegraphics[width=.8\textwidth]{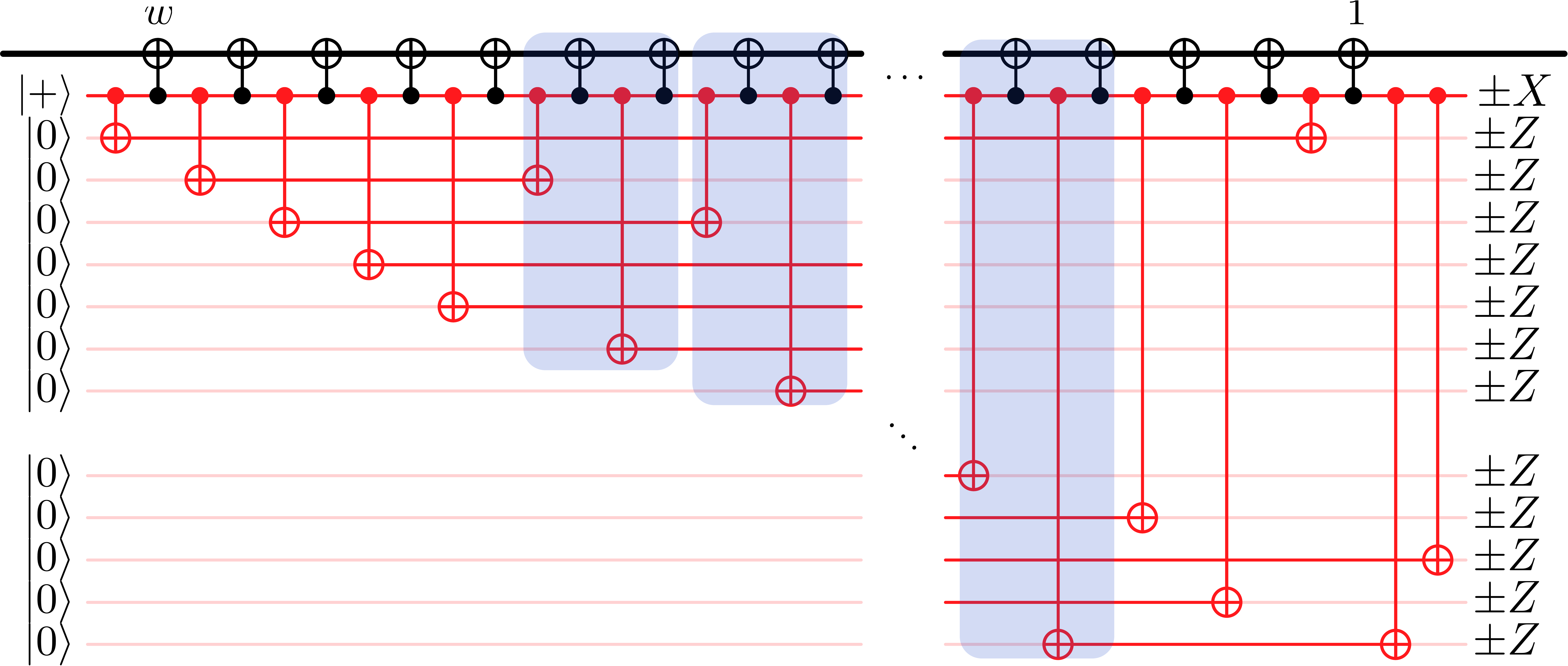} 
    \caption{Distance-five syndrome measurement with slow qubit reset for a weight-$w$ $X$ stabilizer.  The thick black wire indicates a register of $w$ qubits.  An opaque red wire implies the flag is currently inactive and not catching faults. The gates in the blue section can be repeated to construct stabilizer measurement circuits for arbitrary stabilizer weight $w$. At any instant, only five flags are active. Hence this circuit can be performed with fast qubit reset using only five flag qubits.} 
    \label{f:d5arbwsyndmeas}
\end{figure}

\begin{figure}
    \centering
    \includegraphics[width=.99\textwidth]{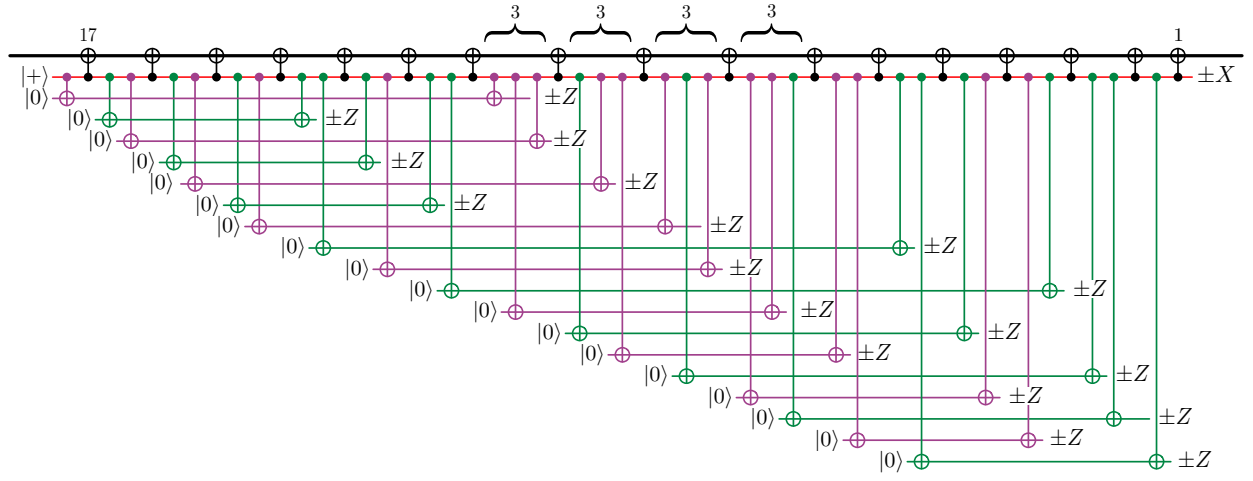} 
    \caption{Distance-seven syndrome measurement with slow qubit reset for a weight-$17$ $X$ stabilizer. At any instant, only seven flags are active. Hence this circuit can be performed with fast qubit reset using only seven flag qubits.} 
    \label{f:d7w17}
\end{figure}

\chapter[Cat state preparation]{Fault-tolerant cat state preparation}
\label{chap:catstateprep}

In this section we show protocols for distance-three fault-tolerant cat state preparation with overhead that is logarithmic in the size of the cat state. Alternatively with fast reset, only one ancilla is needed, and it is used a logarithmic number of times. Cat states~\cite{Greenberger1989} have applications in many areas of quantum computing, including communication~\cite{Gisin07}, information processing~\cite{PanChen2012}, and error correction~\cite{Shor96}.  Besides practical applications, our results on cat state preparation are theoretically interesting since: $i$) we introduce the study of asymptotic estimates of qubit overhead for the fault-tolerant preparation of cat states of \textit{arbitrary} size, and, $ii$) ideas developed for cat state preparation may provide clues for the fault-tolerant preparation of logical states of more complex codes.  We explore these ideas further in \chapref{chap:encodedstateprep}.

We briefly summarize the results of this chapter. We first consider making cat state preparation deterministic by using flag techniques from \chapref{chap:SMcat}. Flag sequences developed in \secref{s:flagseq} are utilized to tolerate a single fault in a non-fault-tolerant preparation circuit. This leads to protocols requiring just one ancilla qubit, measured $m$ times, where $m$ is logarithmic in the size of the cat state. Flag circuits of higher distance, such as in \secref{s:SMcatSDSM5} or Ref.~\cite{chao2019flag}, can be used for arbitrary distance fault tolerance. We also considered preparing cat states by solely performing parity measurements on single-qubit states. Note that a stabilizer state may be ideally prepared by measuring a minimal set of stabilizer generators of the state. However fault tolerance usually requires the repetition of measurements to facilitate a majority voting. Here we consider cat states prepared by performing $ZZ$ measurements, and show that in addition to the $n-1$ generators that must be measured ideally, only $\lceil \log_2 (n/3) \rceil +1$ extra measurements are needed to tolerate a single fault.

\begin{table}
    \centering
    \caption{\label{f:resultsCSP} Cat state size for distance-$3$ preparation methods that use $m$ ancilla qubit measurements.}
    \begin{tabular}{ c @{\hspace{.65cm}} c }
        \hline \hline
        \textbf{Method} & \textbf{Cat state size~$w$} \\
        \hline 
        \textit{Deterministic} Error Correction & $w \leq 3 \, (2^m - 2m + 2)$ \\
        (\thmref{t:catstated3}) & $\depth = (w - 1) + 2^{m - 2}$ \\
        \textit{Adaptive} Error Correction &  $w \leq 3 \, (2^m - 2m + 3)$ \\
        (\thmref{t:adaptiveslowresetd3}) & \\[.15cm]
        Error Detection &  $w \leq 3 \cdot 2^{m -1}$ \\
        (\thmref{t:errordetcatd3}) & \\[.15cm]
        \textit{Parallelized} Error Correction & $w = 2m = 2 \cdot 2^j , j\in \mathbb{N}$ \\
        (\thmref{t:parallelcatd3s}) & $\depth = 2 + \log_2 w$\\
        \hline \hline
    \end{tabular}
\end{table}

\tabref{f:resultsCSP} contains bounds on the ancilla overhead for preparing weight-$w$ cat states fault-tolerantly with flags.  If the flag qubits can reset quickly, \thmref{t:catstated3} states that only one flag qubit is required and it needs to be reset and measured $m$ times.  Since the flag qubits operate independently, it is also possible to use $m$ flag qubits, with each one being measured once.  We further show how to use an adaptive circuit in \thmref{t:adaptiveslowresetd3} to marginally increase the number of flag patterns in use.

When preparing the cat states by stabilizer measurement, we consider two scenarios of qubit connectivity. First, we determine overhead bounds for state preparation with non-local gates. This facilitates parity measurements between data qubits that are far apart. Next, we consider a layout where qubits are connected as a 1-D chain. Along the chain, data and ancilla qubits alternate, and two-qubit gates are local. Under this model, the logarithmic overhead from the nonlocal scenario was lost, however a scalable construction for arbitrary distance fault tolerance was uncovered. We find an upper bound scaling as $O(nd)$ for $n$-qubit cat state preparation with distance-$d$ fault tolerance. We explicitly show the parity measurements needed for small cat state sizes, as this may be interesting from an experimental standpoint.

In the appendix, \secref{sec:SMcatPCSP} and \secref{sec:catparallelized} contain two additional protocols for distance-three weight-$w$ cat state preparation. In \thmref{t:errordetcatd3}, we show how to use postselection to prepare cat states while tolerating \textit{two} faults.  Finally \thmref{t:parallelcatd3s} details how to create low-depth circuits for distance-three fault-tolerant cat state preparation, which may be useful in technologies with many qubits or long two-qubit gate~times.

\section{Using flags to tolerate faults in a non-fault-tolerant circuit}
\label{sec:SMcatFDCSP}

We start by outlining the general procedure used to construct the fault-tolerant circuits in this section. An $n$-qubit cat state is a simple ancilla state defined with the stabilizer generators $\{X^{\otimes n}, Z_iZ_{i+1}$ $(\forall \, i<n)\}$.

\begin{procedure}
    \label{proc:gateSP}
    \textbf{Flags to tolerate faults for cat state preparation: }
    \begin{enumerate}
        \item Construct a non-fault-tolerant cat state preparation circuit by applying CNOT gates from a control $\ket +$ qubit on to a set of $n-1$ $\ket 0$ target qubits.
        \item Using flag sequences from \secref{s:flagseq}, perform parity measurements on targeted qubits to detect if a fault has caused an $X$ error of high weight. 
        \item Apply corrections based on the syndromes.
    \end{enumerate}
\end{procedure}

For preparing a two- or three-qubit cat state, any preparation circuit is automatically fault-tolerant, because every error has weight zero or one.  For example, on three qubits $XXI \sim IIX$, since $XXX$ is a stabilizer.  Fault tolerance becomes interesting for preparing cat states on $w \geq 4$ qubits.  

\begin{theorem} \label{t:catstated3}
For $m \geq 2$, one ancilla qubit, measured $m$ times, is sufficient to prepare a cat state on $w$ qubits fault-tolerantly to distance three, for
\begin{equation*}
w \leq 3 \, \big( 2^m - 2 m + 2 \big)\,.
\end{equation*}
\end{theorem}

\noindent
Let $[m] = \{1, 2, \ldots, m\}$ and $X_S = \prod_{j \in S} X_j$.

\begin{proof}[Proof of \thmref{t:catstated3}]
\cref{f:catstated3} illustrates our construction for the cases $m = 3$ and $m = 4$.  In general, we prepare a $w$-qubit cat state using CNOT gates from the first qubit, so that the possible $X$ errors from a single fault are $\identity, X_1, X_{[2]}, X_{[3]}, \ldots$.  
We then compute parities of subsets of the qubits into the ancillas, following the flag sequence from \lemref{t:slowresetdistance3flagsequences} and \figref{f:slowresetdistance3flagsequences}.  Although for clarity \figref{f:catstated3} shows the $m$ parity checks being made in parallel, they can also be made sequentially with just one ancilla qubit.  

\begin{figure}
    \centering
    \subfloat[\label{f:catstatew12}]{
        \includegraphics[width = 0.7\textwidth]{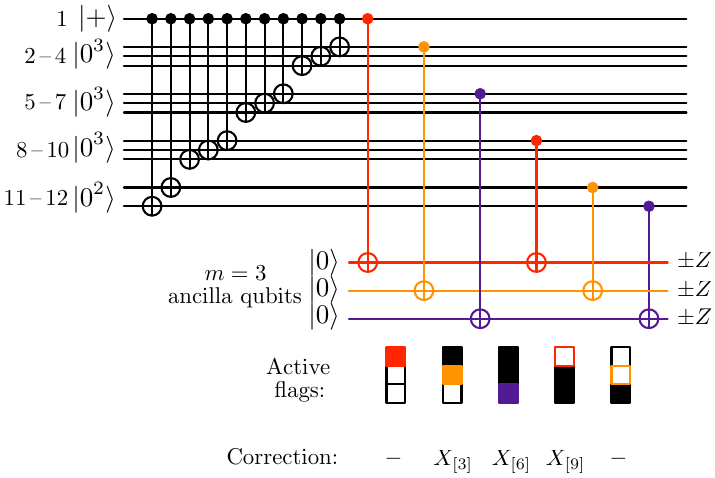}}
    \hspace{0.25cm}
    \subfloat[\label{f:catstatew30}]{
        \includegraphics[width = 0.95\textwidth]{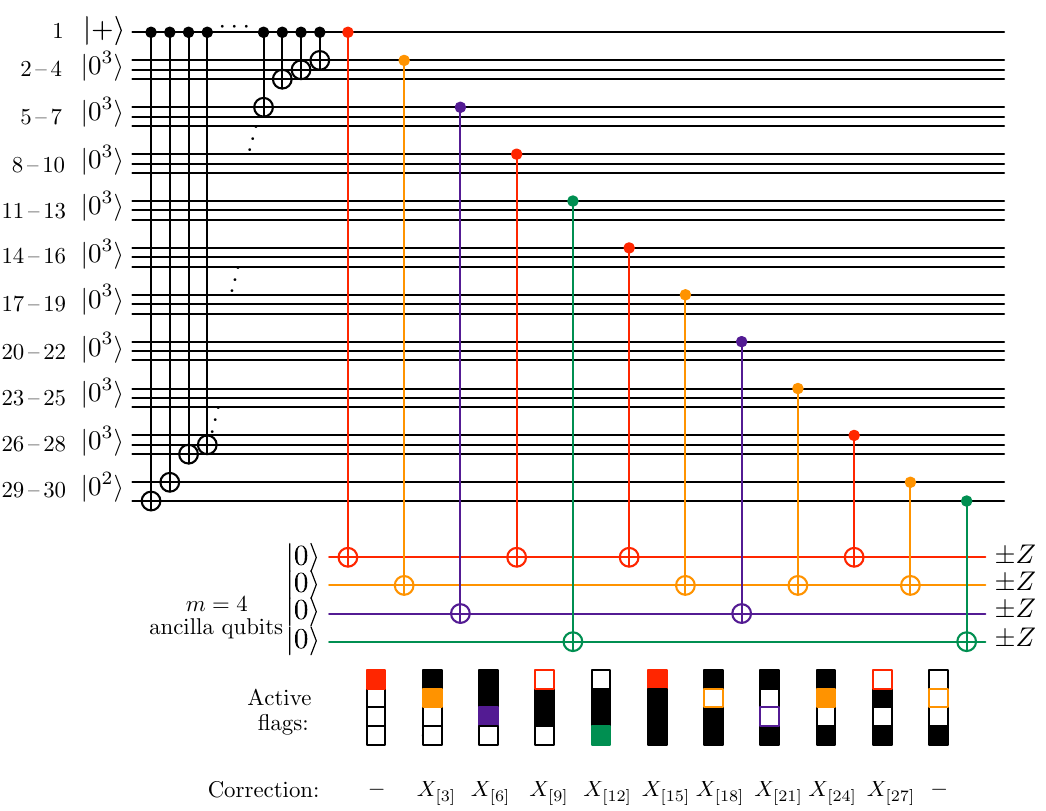}}
    \caption{Distance-three fault-tolerant cat state preparation circuits.  Note that, with fast reset, only one ancilla qubit is required.} 
    \label{f:catstated3}
\end{figure}

With the given correction rules, errors due to single faults are corrected up to possibly a weight-one remainder.  (For example, in \figref{f:catstatew12}, errors $X_{[5]}$, $X_{[6]}$ and $X_{[7]}$ all result in the parity checks $111$, for which the correction $X_{[6]}$ is applied.)  The circuit also tolerates faults within the parity-check sub-circuit, because a single fault here can flip at most one parity, and no correction is applied for the weight-one patterns.  
\end{proof}

By this method, the cat state is prepared in depth $w - 1$.  The depth of the parity check circuit increases exponentially as $2^{m - 2}$ for $m \geq 3$ if we consider slow reset ($a = m$).  This is evident from the flag sequences in \figref{f:slowresetdistance3flagsequences} as the maximum number of times any flag bit is switched.  The total depth of the circuit is then $(w - 1) + 2^{m - 2}$.  

Note that the construction from \thmref{t:catstated3} does not help for syndrome measurement, because the parity checks would in general become entangled with the data. However the ideas of Theorems~\ref{t:fastresetd3syndromemeasurement} and~\ref{t:syndromemeasurementd3slowreset} can also be applied to cat state preparation.  For example, just as in \figref{f:slowresetd3w6CSP} a circuit for measuring $X^{\otimes 6}$ with three ancilla qubits corresponds to a circuit to prepare a six-qubit cat state with two ancillas, similarly adapting the construction of \thmref{t:syndromemeasurementd3slowreset} allows preparing a $2 (2^a - 2 a + 3)$-qubit cat state using $a$ ancilla qubits each measured once.  \thmref{t:catstated3} shows a protocol needing just one ancilla qubit hence that is better for state preparation. This technique can however be very useful for the preparation of more complex ancilla states, as in general we cannot perform parity measurements after preparation as we do with cat states.

We can do slightly better than \thmref{t:catstated3} if we allow an \emph{adaptive} circuit, in which the parity checks are chosen based on the outcome of a flag qubit measurement.  For example, \figref{f:adaptiveslowresetd3w15} gives a circuit to prepare a $15$-qubit cat state using $m = 3$ measurements.  Here, the result of measuring the red ancilla determines how the other two ancillas are used.  

\begin{figure}
    \centering
    \includegraphics[width=.999\linewidth]{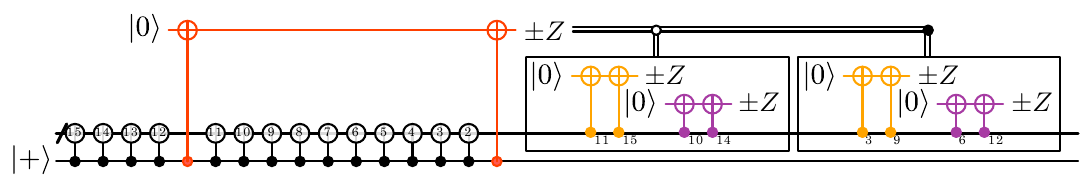}
    \caption{Circuit to prepare a $15$-qubit cat state by adaptive error correction, fault-tolerant to distance three.  Labels on the thick black wire indicate which data qubit in the block is being addressed as the control or target of the CNOT.  If a fault occurs while preparing the cat state on the $\ket +$ qubit, it is partially localized by the red flag ancilla.  The measurement result of this flag then determines a set of parity checks to completely localize a possible fault.  After all the ancilla qubits have been measured, corrections are applied based on \tabref{f:adaptiveslowresetd3w15corrections}.}
    \label{f:adaptiveslowresetd3w15}
\end{figure}

%\pagebreak%DEBUG

\begin{theorem} \label{t:adaptiveslowresetd3}
Using an adaptive circuit, for $m \geq 2$, one ancilla qubit, measured $m$ times, can be used to prepare a cat state on $w$ qubits fault-tolerantly to distance three, for
\begin{equation*}
w \leq 3 \, \big( 2^m - 2 m + 3 \big)\, .
\end{equation*}
\end{theorem}

\begin{figure}
\centering
    \subfloat[$w = 15, a = 3$]{
        \includegraphics[width=0.37\textwidth]{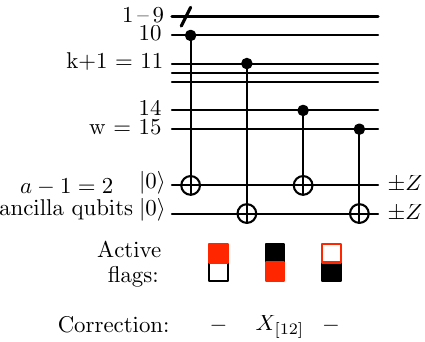}}
    \hspace{0.25cm}
    \subfloat[$w = 33, a = 4$]{
        \includegraphics[width=0.45\textwidth]{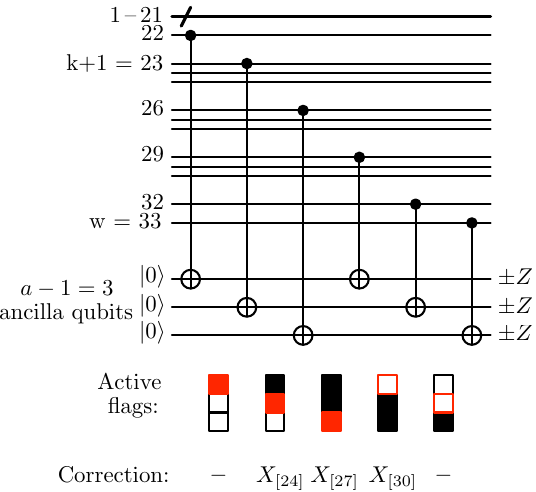}}
    \hspace{0.25cm}
    \subfloat[$w = 75, a = 5$]{
        \includegraphics[width=0.76\textwidth]{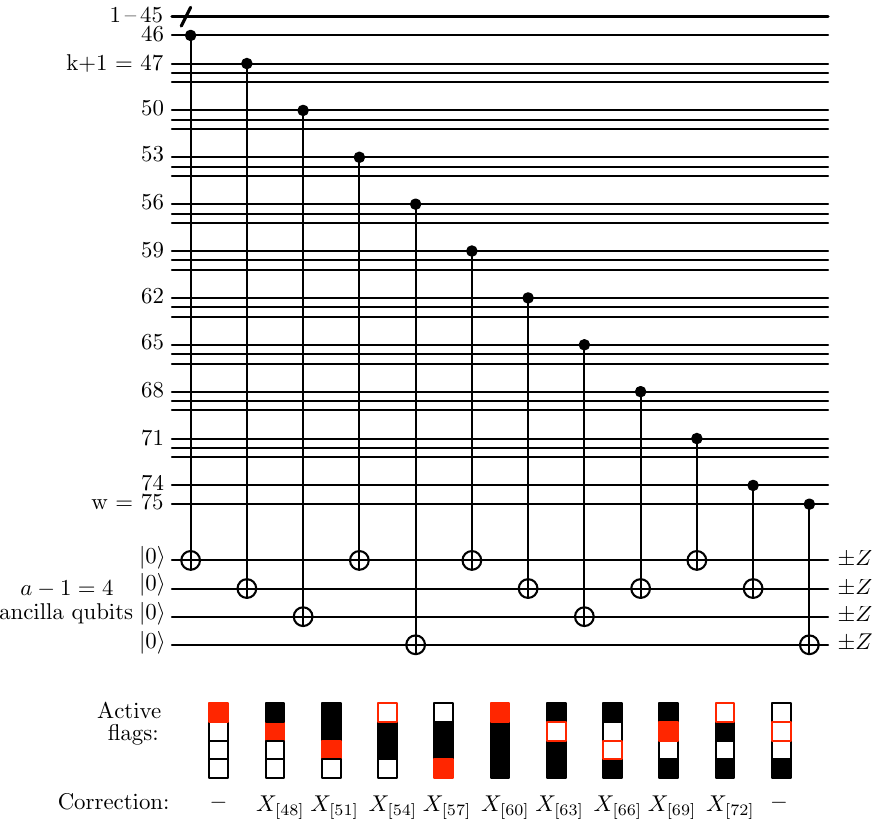}}
    \caption{If the red ancilla flag in \figref{f:adaptiveslowresetd3w15} is not triggered, these circuits are used to find and correct a possible error.  The flag sequences (from \figref{f:slowresetdistance3flagsequences}) and corresponding corrections are listed at the bottom.   Note that these sequences are nonadaptive, and can be used either with $a$ ancilla qubits in a slow reset model, or with just one ancilla qubit in a fast reset model, since all the CNOT gates commute.  
    } 
    \label{f:adaptived3noflagcorrections}
\end{figure}

\begin{proof}
Our construction will follow the same basic structure as the circuit in \figref{f:adaptiveslowresetd3w15}.  
Prepare the $w$ data qubits as $\ket{{+} 0^{w-1}}$, then apply $\CNOT_{1,w},$ $ \CNOT_{1,w-1}, \ldots, \CNOT_{1,2}$ to get a cat state.  Let $k = 3 ( 2^{ m - 1 }) - 2 $.  Just before $\CNOT_{1,k+1}$ and just after $\CNOT_{1,2}$, apply CNOTs into the first ancilla qubit, the red qubit in \figref{f:adaptiveslowresetd3w15}, and measure it.  

The remainder of the circuit depends on the measurement result.  If it is~$1$, then a fault has been detected.  The error on the cat state can be one of 
\begin{equation*}
\identity, X_1, X_{[2]}, \;\;\; X_{[3]}, X_{[4]}, X_{[5]}, \;\;\; \ldots, \;\;\; X_{[k-1]}, X_{[k]}, X_{[k+1]} \, .
\end{equation*}
The correction procedure needs to determine in which of the above $1 + \tfrac{k-1}{3}$ groups-of-three the error lies; then for any error in $\{X_{[3 j]}, X_{[3 j + 1]}, X_{[3 j + 2]}\}$ the correction $X_{[3 j + 1]}$ works.  Perhaps the easiest way to locate the error is by binary search using the Gray code in \lemref{t:graycode}, e.g., by computing parities between qubits $3 j$ for~$j \in \{ 1, 2, \ldots, 1 + \tfrac{k-1}{3} \}$.  Since the measurement of the red ancilla could have been incorrect, it is important that the all-$0$s outcome of the binary search correspond to the $\identity, X_1, X_{[2]}$ error triple, as in~\tabref{f:adaptiveslowresetd3w15corrections}.  Using $m - 1$ measurements, we can search $2^{m-1}$ possibilities, which indeed is $1 + \tfrac{k-1}{3}$.  (The search circuit can also be made nonadaptive, as in \figref{f:adaptiveslowresetd3w15}.)  

%\begin{table}
%\caption{Parity checks and correction rules when the red flag ancilla in \figref{f:adaptiveslowresetd3w15} is measured as~$1$.} \label{f:adaptiveslowresetd3w15correctionsflag}
%\begin{tabular}{ @{\hspace{.2cm}}c @{\hspace{.4cm}} c @{\hspace{1.5cm}} c@{\hspace{1.1cm}} c}
%\hline
%\hline
% & & & \\[-0.35cm]
%$3 \oplus 9$ & $6 \oplus 12$ & Possible errors & Correction \\
%\hline
% & & & \\[-0.3cm]
%$0$ & $0$ & $\identity, X_1, X_{[2]}$ & $X_1$ \\ 
%$1$ & $0$ & $X_{[3]}, X_{[4]}, X_{[5]}$ & $X_{[4]}$ \\ 
%$1$ & $1$ & $X_{[6]}, X_{[7]}, X_{[8]}$ & $X_{[7]}$ \\ 
%$0$ & $1$ & $X_{[9]}, X_{[10]}, X_{[11]}$ & $X_{[10]}$ \\ 
% & & & \\[-0.35cm]
%\hline 
%\hline
%\end{tabular}
%\end{table}

Next consider the case that the first measurement result is~$0$, so no fault has been detected.  The error on the cat state can be one of $X_{[k+1]}, X_{[k+2]}, \ldots, X_{[w]} \sim \identity$.  We again use the remaining $m - 1$ ancilla qubits to measure parities of subsets of cat state qubits.  Since there is no guarantee of a fault having occurred yet, we use flag sequences from \lemref{t:slowresetdistance3flagsequences}, where the length of the weight-at-least-two flag sequence is $J = 2^{m - 1} - 2 ( m - 1 ) + 1$.  The parity checks are now done between qubits $\{ k, k + 1 + 3 j, k + 2 + 3 J \} $ for $ j \in \{ 0, 1, \ldots , J \}$, as shown in \figref{f:adaptived3noflagcorrections} and \tabref{f:adaptiveslowresetd3w15corrections}.  We do not allow weight-one flag patterns to be able to correct any errors since they may also be triggered by a measurement fault on any one of the data qubits involved in the parity check.

\begin{table}
    \caption{Possible data errors and associated corrections for the different observed flag patterns in \figref{f:adaptiveslowresetd3w15}. } 
    \label{f:adaptiveslowresetd3w15corrections}
    \centering
    \begin{tabular}{c @{\hspace{0.35cm}} c c c @{\hspace{0.35cm}} c @{\hspace{0.35cm}} c}
        \hline \hline
         % & & & & & \\[-0.33cm]
        Red flag & \multicolumn{2}{c}{Parity checks} & & Possible errors & Correction \\
        \hline
         % & & & & \\[-0.3cm]
         $1$ & $3 \oplus 9$ & $6 \oplus 12$ & & &  \\
         \cline{2-3}
         % & & & & & \\[-0.3cm]
         & $0$ & $0$ & &  $\identity , X_1, X_{[2]}$ & $X_1$ \\ 
         & $1$ & $0$ & &  $X_{[3]}, X_{[4]}, X_{[5]}$ & $X_{[4]}$ \\ 
         & $1$ & $1$ & &  $X_{[6]}, X_{[7]}, X_{[8]}$ & $X_{[7]}$ \\ 
         & $0$ & $1$ & &  $X_{[9]}, X_{[10]}, X_{[11]}$ & $X_{[10]}$ \\[0.25cm]
         $0$ & $11 \oplus 15$ & $10 \oplus 14$ & & &  \\
         \cline{2-3}
         % & & & & & \\[-0.3cm]
         & $0$ & $0$ & & $\identity$ & None \\ 
         & $1$ & $0$ & &  $X_{11}, X_{15}, X_{[14]}$ & None \\
         & $1$ & $1$ & & $X_{[11]}, X_{[12]}, X_{[13]}$ & $X_{[12]}$ \\
         & $0$ & $1$ & & $ X_{10}, X_{14}$ & None \\
         % & & & & & \\[-0.3cm]
        \hline \hline
          % & & & & & \\[-0.7cm] %DEBUG
    \end{tabular}
\end{table}

Consolidating, we are allowed up to $3 J + 1$ CNOTs before the red ancilla is initialized, and up to $k$ CNOTs in the monitored region of the red ancilla. In total we can create a cat state on up to 
\begin{equation*}
w \leq 3 J + k + 2 = 3 \, \big( 2^m - 2 m + 3 \big)
\end{equation*}
qubits, with $m$ total measurements.
\end{proof}

We also tested protocols where multiple flags are used for the initial partial localization of a fault (in place of the red flag qubit).  We found no improvement to our bounds on ancilla overhead.  It appears that ancillas are better used in the parity checks than for partial fault localization. 

In this section we show multiple cat state preparation circuits that have been made tolerant to a single fault by using flags. Moreover, given the non-local connectivity, the number of ancillas needed is shown to be logarithmic in the cat state size. We derive these upper bounds via combinatorial proof techniques, and hope our results inform the theoretical analysis of the asymptotic resource requirements of general state preparation. In the next section, we consider the preparation of cat states without two-qubit gates between data qubits. Instead, cat states are prepared solely by measuring stabilizers.

\section{State preparation by measurement} 
\label{sec:catbymeas}

As before, we outline the general procedure we use in this section to prepare $n$-qubit cat states by stabilizer measurements.
\begin{procedure}
    \label{proc:measureSP}
    \textbf{Overcomplete sequences to tolerate measurement faults: }
    \begin{enumerate}
        \item Prepare all the qubits in the $X$-basis, ensuring that we can then project into the $+1$-eigenspace of the $X^{\otimes n}$ operator.
        \item Measure $Z$-type stabilizers to project into one of the unique $X$ error spaces. For fault tolerance, an overcomplete sequence of stabilizers is chosen according to the techniques developed in Refs.~\cite{delfosse2020short, Delf22}.
        \item Based on the observed syndrome, apply an $X$ correction to return to the desired codespace.
    \end{enumerate}
\end{procedure}
Careful gate scheduling and the use of extra ancillas allows measurements to be parallelized, reducing circuit depth. In \figref{fig:catsavings}, we show the reduced number of $ZZ$ parity measurements needed for our determinstic circuits, compared to a Shor-type postselective protocol. The first graph shows the number of parity measurements needed with all-to-all connectivity, and the subsequent graph considers cat state preparation on a 1-D chain. 

An alternative to \procref{proc:measureSP} is to initialize all the qubits in the $Z$ basis and measure the $X^{\otimes n}$ operator fault-tolerantly. Only one operator needs to be measured, however performing the measurement fault-tolerantly can be quite expensive. Flag techniques that are fault-tolerant to arbitrary distance can require many extra qubits~\cite{chao2019flag,Anker22}. To cut costs, we consider measuring the $Z$-type stabilizers after initializing all the qubits in the $X$-basis. This requires only one quickly resetting ancilla qubit if the measurements are performed sequentially, which also reduces the circuit complexity.

Note that the fault tolerance setting considered here is not entirely the same as that in Refs.~\cite{delfosse2020short, Delf22}. In those references, an overcomplete sequence of measurements was determined for error correction, where it is clear that logical errors must be suppressed. With cat state preparation, the objective is simply to reduce the residual error weight at the end of the state preparation circuit. Moreover from a technical standpoint, a distance-$d$ fault-tolerant protocol from Refs.~\cite{delfosse2020short, Delf22} tolerates up to $t=\lfloor \frac{d-1}{2} \rfloor$ combined input errors and internal faults. In our state preparation circuits, we must tolerate any number of input errors, while also tolerating up to $t$ internal faults.

\subsection{Non-local measurements}
\label{subsub:nonlocal}

Qubit layouts with non-local (especially all-to-all) connectivity are great tools for analytic resource estimates, however are difficult to construct experimentally. In this subsection, we first attempt to understand the limitations of non-local connectivity by deriving tight upper bounds for the number of parity measurements needed to tolerate one fault. We then apply these techniques to optimize the overhead when considering local connectivity. In doing so, we observe a scalable method of choosing parity measurements that may be fault-tolerant to arbitrary distance.

\begin{figure}
    % \hspace{-0.4cm}
    \centering
    \includegraphics[width=.75\textwidth]{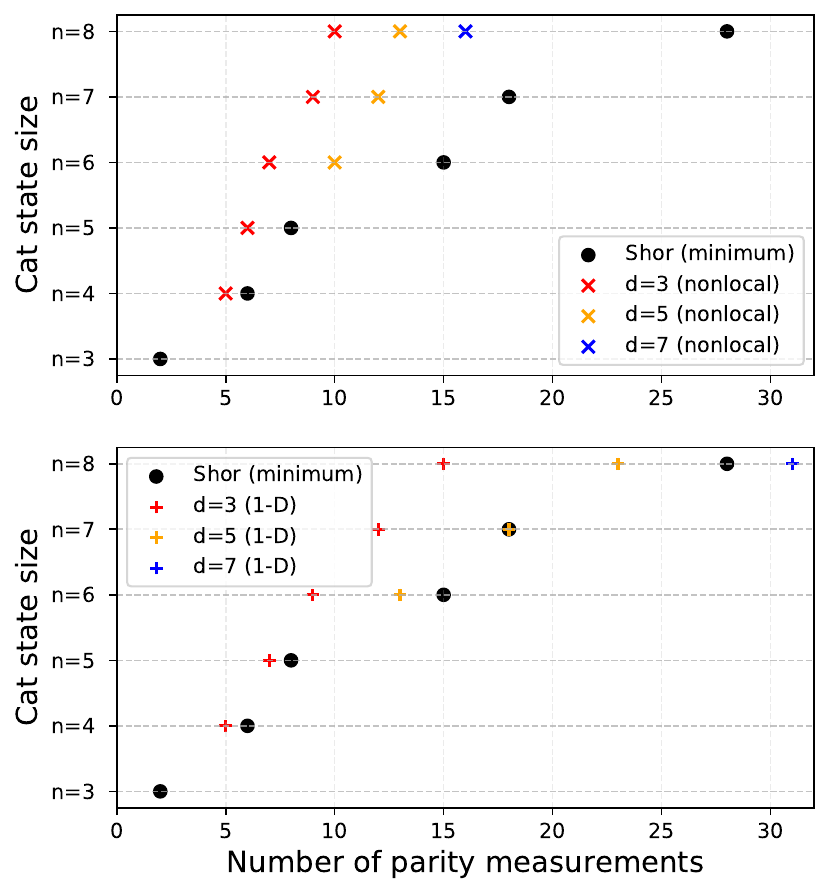}
    \caption{Reduced number of $ZZ$ parity measurements required to prepare a cat state fault-tolerantly. We first consider non-local connectivity, proving that $n + \lceil log_2(n/3)\rceil$ measurements are sufficient to tolerate one fault. By random search, we found sequences of parity checks that can tolerate two or three faults too. In the second graph, qubits are laid on a 1-D chain, and CNOT gates are local. The number of parity measurements is generally larger than with non-local connectivity. The $n=8$, $d=7$ solution for local measurements is conjectured but not proved.}
    \label{fig:catsavings}
\end{figure}

\begin{theorem}
Using~\procref{proc:measureSP}, length-$n$ cat state preparation tolerating a single fault needs at most $n + \lceil log_2(n/3)\rceil$ parity measurements. 
\end{theorem}

\begin{proof}
We partition the proof into first determining the corrections for all the even-weight syndromes, followed by a rigorous analysis of the correction for odd-weight syndromes. 

Prepare all the qubits as $\ket+$. The first $n$ measurements are of the $n-1$ stabilizer generators of the code ($Z_i Z_{i+1}$ $(\forall i<n)$) followed by $Z_1 Z_n$. This ensures every qubit is checked twice. Every $X$ error space (in the absence of any faults) is now characterized by a unique even-weight syndrome. The correction for these syndromes are the $X$ errors that cause them.

We now prove that a single fault during these measurements causes an odd-weight syndrome. Note that a measurement fault on any of the $n$ syndrome qubits causes a weight-one syndrome. At this stage, there are also exactly $n$ possible locations for $X$ faults (on each qubit, after its first check), which result in a weight-one error and a weight-one syndrome. Since we already consider the syndrome of every $X$ error space, the effect of an $X$ fault is the same as that of a measurement fault in a different error space. This leaves us with only $n$ unique faults, affecting all the measurement error spaces. A fault at any of these $n$ locations converts an even-weight syndrome to odd weight. Each odd-weight syndrome thus occurs as a result of a projection onto an $X$ error space and one of the $n$ different measurement faults. 

Looking closer at the error spaces signaled by an odd-weight syndrome, they consist of a set of errors $\{e_i + cyc(f,i) | f \in F\}$, for an $X$ error space $e_i$ and 
\begin{align*}
F= \{ & 0 0 0\mathellipsis 0 0, \\
   & 1 0 0\mathellipsis 0 0, \\
   & 1 1 0\mathellipsis 0 0, \\
   & \mathellipsis, \\
   & 1 1 1\mathellipsis 0 0, \\
   & 1 1 1\mathellipsis 1 0 \}
\end{align*}
where $cyc(f,i)$ is a uniform right cyclic shift of the bitstring $f$ by $i$ indices, where the value of $i$ depends on the error space $e_i$. Note that $i$ does not index the error spaces. The resulting $n$ distinct error spaces can then be partitioned into $2^{\lceil log(n/3)\rceil}$ groups, using $\lceil log(n/3)\rceil$ parity measurements. Using techniques from \ref{t:catstated3}, and with an appropriate Gray code, we can choose the qubits taking part in the parity checks. Note that each syndrome due to these extra parity checks must signals one of up to three overlapping errors, allowing us to choose the middle of the three errors as a correction. \cref{tab:nonlocalparities} shows a choice for these parity checks up to $n=24$. For every odd-weight syndrome in the first $n$ parity measurements, we have found an assignment of corrections. Ignore all syndromes that do not trigger any of the first $n$ measurements. This completes the proof as all possible syndromes have assigned corrections.
\end{proof}

By performing measurements in parallel, all of the parity measurements can be implemented in three rounds, with $n/2$ extra qubits instead of one. Through exhaustive numerical search, we verified that for $n \in \{4,5,6 \}$, the minimum required number of non-local two-qubit parity measurements matches our upper bound. We conjecture that the lower bound to tolerate a fault is at least $n$ measurements.

\setlength{\tabcolsep}{2pt}
\begin{table}
    \centering
    \begin{tabular}{r| l}
   \multicolumn{2}{c}{$\mathbf{d=3}$} \\
    $n=4$ & $\{ 1 \mathellipsis n \},  (2,3)$ \\
    $n=5$ &  $\{ 1 \mathellipsis n \},  (2,4)$  \\
    $n=6$ &  $\{ 1 \mathellipsis n \},  (2,5)$  \\
    $n=7$ &  $\{ 1 \mathellipsis n \},  (2,4),(3,6)$  \\
    $n=8$ & $\{ 1 \mathellipsis n \},  (2,5), (3,7)$  \\[.2cm]    
    $n=12$ & $\{ 1 \mathellipsis n \},  (2,8), (5,11)$  \\
    $n=13$ & $\{ 1 \mathellipsis n \},  (2,8), (4,10), (6,12)$  \\[.2cm]
    $n=18$ & $\{ 1 \mathellipsis n \},  (2,11), (5,14), (8,17)$  \\
    $n=19$ & $\{ 1 \mathellipsis n \},  (2,6,11,16), (4,14), (8,18)$  \\[.2cm]
    $n=24$ & $\{ 1 \mathellipsis n \},  (2,8,14,20), (5,17), (11,23)$  \\[.4cm]
    \multicolumn{2}{c}{$\mathbf{d=5}$} \\
    $n=6$ &  $\{ 1 \mathellipsis n \},  (2,5), (1,5),(1,6),(3,5)$  \\
    $n=7$ &  $\{ 1 \mathellipsis n \},  (2,6),(3,5),(1,3),(2,4),(3,6)$  \\
    $n=8$ &  $\{ 1 \mathellipsis n \},  (2,5),(3,8),(2,8),(4,8),(2,7)$  \\
    $n=9$ &  $\{ 1 \mathellipsis n \},  (2,6),(2,8),(5,9),(4,8),(7,9)$  \\
    $n=10$ & $\{ 1 \mathellipsis n \},  (3,8), (5,10), (3,10), (1,6), (1,3)$  \\[.4cm]
    \multicolumn{2}{c}{$\mathbf{d=7}$} \\   
    $n=8$ & $\{ 1 \mathellipsis n \},  (2,5),(3,8),(2,8),(4,8),(2,7),(3,6),(1,8),(3,5)$  \\
    \end{tabular}
    \caption{Sequences of parity measurements needed to prepare cat states of size $n$ fault-tolerantly to distance $d$. We assume non-local gates are possible, permitting parity measurement of distant data qubits. The sequences for the distance-five and -seven cases were generated at random, while the corrections were calculated and fault tolerance was verified using Mathematica programs. $\{ 1 \mathellipsis n \}:= (1, 2), (2, 3), \mathellipsis (n-1, n), (n, 1)$.}
    \label{tab:nonlocalparities}
\end{table}

\subsubsection{Simulations}
\label{subsubsec:nonlocal}

To verify the correctness of our state preparation circuits and compare against previous work, we perform Monte Carlo simulations of our circuits using the Gottesman-Knill framework under an independent circuit-level noise model as described below.

\begin{itemize}[leftmargin=*]
\item With probability $p$, the preparation of $\ket 0$ is replaced by $\ket 1$ and vice versa---similarly $\ket +$ and $\ket -$.
\item With probability $p$, $\pm X$ or $\pm Z$ measurement on any qubit has its outcome flipped.
\item With probability $p$, the two-qubit CNOT gate is followed by a random two-qubit Pauli error drawn uniformly from $\{ I, X, Y, Z\}^{\otimes 2} \setminus \{I \otimes I \}$.
\end{itemize}

We perform simulations of a weight-$8$ cat state preparation circuit, and target fault tolerance to distance seven. We are interested in observing what the probability of a residual weight-$k$ error is, for $k<n$, using the different fault-tolerance techniques we have at hand. In \figref{fig:catwt8det}, we compare the residual error probabilities using five different circuits, three of which are fault-tolerant to distance-seven. These include the $n=8$, $d=7$ sequence from \tabref{tab:nonlocalparities}, the measurement of the $X^{\otimes n}$ stabilizer fault-tolerantly to distance-seven and finally the encoding of an $X^{\otimes n}$ operator. The last two techniques are adopted from techniques in \ref{chap:SMcat}. Additionally, we consider a $d=3$ and a $d=5$ sequence of parity measurements for completeness. The distance-seven techniques do suppress errors fault-tolerantly, with the stabilizer measurement method showing the lowest residual error rates for high-weight errors. Due to the reduced number of operations in the $d=3$ and $d=5$ sequences, low-weight errors are suppressed better than with $d=7$.

\begin{figure*}
    % \hspace{-0.5cm}
    % % \centering
    \includegraphics[width=\textwidth]{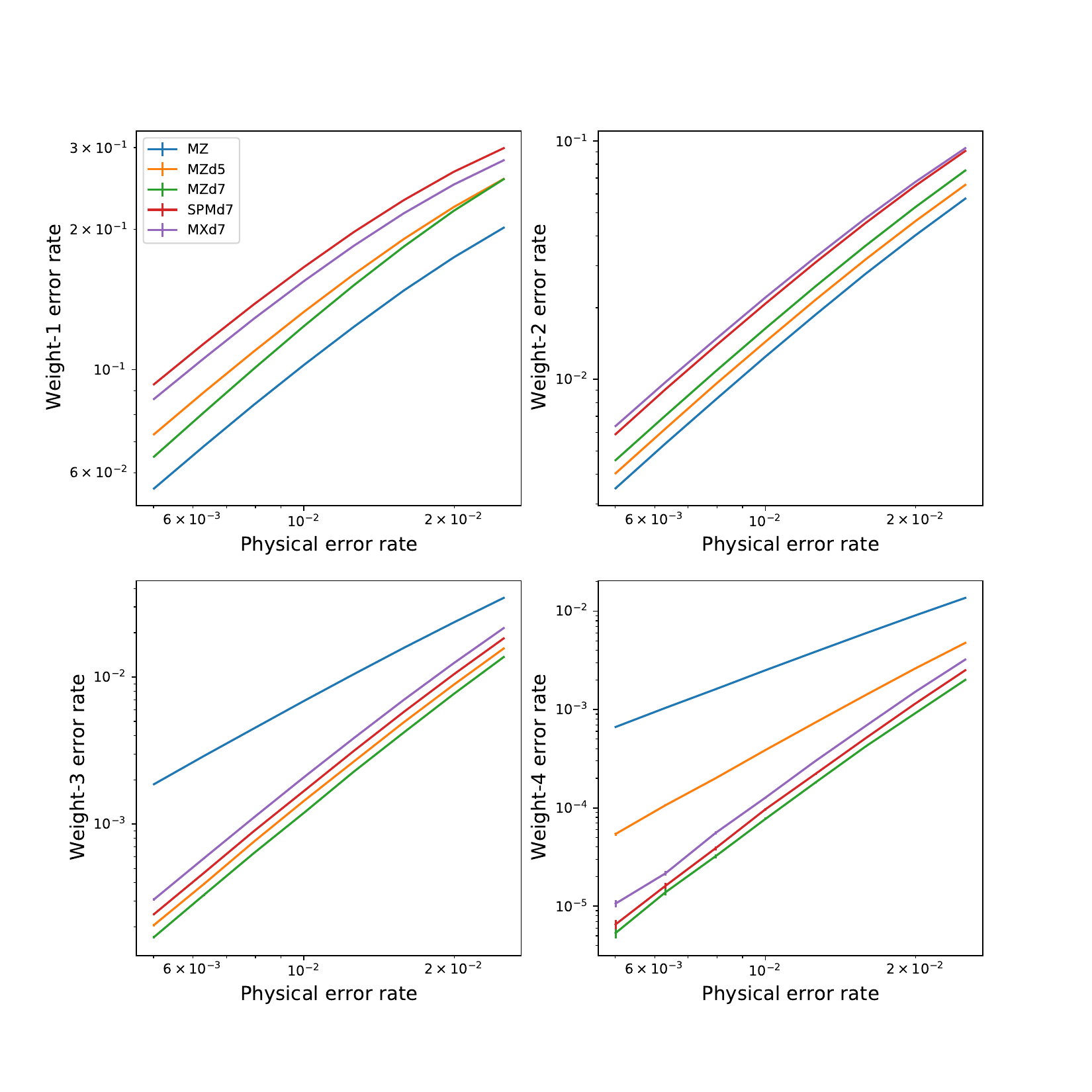}
    \caption{Rates of residual errors of weight $w \in \{1,2,3,4\}$ after size-eight cat state preparation, for physical error rates $p \in [5 \times 10^{-3}, 2.5 \times 10^{-2}]$.The methods used consist of sequences from \tabref{tab:nonlocalparities}, the measurement of the $X^{\otimes n}$ stabilizer fault-tolerantly to distance-seven (MXd7) and finally the encoding of an $X^{\otimes n}$ operator fault-tolerantly to distance-seven (SPMd7). The distance-seven cases are fully fault-tolerant, as opposed to the distance-five case where three faults can result in a weight-four error. The lowest residual error rates are observed with the distance-seven fault-tolerant stabilizer measurement sequence in \tabref{tab:nonlocalparities}.}
    \label{fig:catwt8det}
\end{figure*}

\subsection{Local measurements on a 1-D chain}
\label{subsub:local}

We finally consider the problem of preparing cat states on a line, using only local measurements. We show solutions that need fewer measurements than Shor-type schemes to achieve fault tolerance. This translates to fewer operations on the qubits and hence a higher fidelity output state.

\begin{claim}
    With $1$-D connectivity and for $n \geq 6$, deterministic length-$n$ cat state preparation fault-tolerant to distance-$d$ needs at most $d (n-4) +3$ measurements. Length-$4$ and $5$ cat states need $5$ and $7$ measurements respectively, and are fault-tolerant to distance-three.
    \label{claim:1Dcat}
\end{claim}

First, we choose the number of times each parity is measured, and then define the parity checks in each layer. Measure the parity $i$ (on qubits ($i,i+1$)) the number of times specified by the entry of index $i$ in this length-$(n-1)$ string: $\{ 1, \frac{d+1}{2}, d, d, \mathellipsis, d, d, \frac{d+1}{2}, 1\}$. The parity measurements are executed in at most $2 d$ layers by executing them in parallel. This is done by splitting the $n-1$ unique parities into two groups, where a qubit appears no more than once in each group. Measurements in each layer are drawn from the two groups in an alternating fashion. Over the first layer (and second layer for even $n$), the measurements at the ends of the line will be completed and not need to be repeated. Proceed until all the $d (n-4) +3$ measurements are performed.

With this construction, we have computationally verified the preparation of cat states on a line up to size $12$ for distance-three fault tolerance, and size $8$ for distance-five fault tolerance. Due to CPU limitations we could not numerically verify sequences for larger cat states or higher distance. We leave the proof of \claimref{claim:1Dcat} to future work.

\chapter{Encoded state preparation}
\label{chap:encodedstateprep}

Building on improvements from the previous chapter, we now turn our attention to the preparation of more complex states. With cat states, fault tolerance is only needed to curb $X$ error spread. $ZZ$ stabilizers provides tolerance to $Z$ errors for free. In this chapter, we consider general CSS stabilizer states for which both $X$ and $Z$ fault tolerance is required. We attempt to prepare the circuits using two models. In one, flags are used to encode stabilizer operators fault-tolerantly. In the other, states are prepared solely by performing fault-tolerant stabilizer measurements. With both methods, we may choose to postselect. However we make these protocols deterministic. This is a first in stabilizer state preparation, and can greatly simplify the construction and scheduling of quantum state distillation factories. 

We first describe what it means to prepare an ancilla state fault-tolerantly. 
For fault tolerance to distance $d'$, the final condition is relaxed to $k \leq \frac{d'+1}{2}$.
\begin{condition}
    \label{cond:faulttolerance}
     \textbf{Fault-tolerant state preparation}: To fault-tolerantly prepare an ancilla state for an $\llbracket n,k,d\rrbracket$ quantum code, any $k$ faults in the circuit should propagate to a residual error of weight at most $k$, for $k \leq \frac{d+1}{2}$.
\end{condition}

Quantum error-correcting codes~\cite{Shor95decoherence, Steane96css, CalderbankShor96, Gaitan13, LidarBrun13} used fault-tolerantly can perform arbitrarily accurate computations with noisy components, provided the noise stays below a threshold as the system size is scaled up~\cite{Knill05nat, Aharonov97, AliferisGottesmanPreskill05, Reichardt06}. It is natural, therefore, to push the threshold up as high as possible. In recent years, topological codes such as the surface code~\cite{Fowler12surface} and floquet codes~\cite{Hastings21, gidney2021faulttolerant} gained prominence due to their high thresholds and simple hardware requirements. However Shor-style schemes were used to derive these thresholds~\cite{Shor96, Fowler12surface}. 

Steane-style error correction with low-error CSS stabilizer states may permit higher thresholds of operation~\cite{Steane97, Escobar24}. This has also been explored experimentally on ion trap and neutral atom systems~\cite{Postler23, Huang23b, Bluvstein24}. While the Shor scheme measures stabilizers individually, the Steane method extracts all the information for fault-tolerant error correction in a single round of transversal CNOT gates. This implies there are fewer potential locations for faults, leading to a very high-fidelity error correction routine. Additionally, it was shown that allied codestates could facilitate single-shot error correction with Knill's method~\cite{Knill05}, and Clifford computation in $O(1)$ time~\cite{Zheng_2020}. Universality can then be achieved by transversal non-Clifford gates and code switching~\cite{BevsUniversal, Paetznick13}.

The study of the preparation of CSS stabilizer states is fairly extensive. A lot of focus has been conferred to the preparation of logical states of the quantum Steane and Golay codes~\cite{Steane03, Goto16, Paetznickgolay2013}. However recent work has demonstrated a process to construct states of arbitrary CSS stabilizer codes fault-tolerantly~\cite{Brun15, Lai17stabstates, ZhengBrunFTancprep18}. For the codes we consider in this chapter, we benefit from a more granular analysis of fault propagation. Refs.~\cite{Huang21ShorSteane, Huang23a} perform a hybrid Shor-Steane scheme of error correction. Our results may aid in the development of stabilizer states for these protocols too.

Previous protocols for stabilizer state preparation relied on postselection for fault tolerance. If a non-trivial syndrome is observed, the erroneous state is discarded. Its major advantage is the simplicity of decoding, however, postselection raises many technical issues when actually constructing a fault-tolerant quantum computer. We consider two scenarios based on the availability of real-time control of quantum operations. Considering state preparation is one of the earlier protocols in a quantum computation, without real-time control, a lot of time is wasted in running computations without knowing if they should be rejected. On the other hand, acting on the measurements early and stopping the computation can save time. The drawback is that real-time scheduling of quantum operations in a large quantum computer can become computationally expensive in itself.

Due to these shortcomings, an interesting use case for deterministic protocols is in the first round of state preparation in a state distillation factory. Distillation is generally expensive and non-deterministic, which is remedied with our low-overhead deterministic protocol. In addition, as opposed to simple state injection techniques that encode information with error rate $O(p)$, the protocols we describe in this chapter inject into the code with a failure rate $O(p^\frac{d+1}{2})$. 

\section{Deterministic fault-tolerant preparation of CSS ancilla states}
\label{sec:SPdetCSS}

In this chapter we show how to prepare code states of quantum codes using two techniques for fault tolerance, similar to \chapref{chap:catstateprep}. The first is to use flags to tolerate faults in ideal state preparation circuits and the second involves preparing the state by repeatedly measuring stabilizers. This method of projecting into a codestate is the same as \procref{proc:measureSP}, with the exception that the stabilizers are larger and thus will need flag techniques to make their measurement fault-tolerant. In either case, we may choose to let the protocol be postselective (if a non-trivial syndrome arises, reject), however we attempt deterministic implementations. The side effect of $100 \%$ yield is that the logical error rate of the prepared states are higher than with postselection techniques. To make the  protocol deterministic, we determine what correction permits projection into the code state while satisfying \condref{cond:faulttolerance}, for every possible observed syndrome.

\begin{procedure}
    \label{proc:gateSPstab}
    \textbf{Flags to tolerate faults for stabilizer state preparation: }
    \begin{enumerate}
        \item Construct a non-fault-tolerant ancilla state preparation circuit using $\ket0,\ket1,\ket+,\ket-$ states and CNOT gates.
        \item Every qubit that acts as a control spreads $X$ errors. For a qubit that controls $(w-1)$ CNOT gates, use the flag circuit for weight-$w$ fault-tolerant cat state preparation from \figref{f:slowresetd3w6CSP} to prevent malignant $X$ error spread.
        \item Similarly, target qubits spread $Z$ errors. This is mitigated with the dual (``Hadamarding") of the flag circuits for controlling $X$ error spread.
    \end{enumerate}
\end{procedure}

Circuits that are fault-tolerant to larger distance can be constructed using the techniques developed in \secref{s:SMcatSDSM5} and Ref.~\cite{chao2019flag}. The lowest overheads were observed when the total number of qubits acting as control or target was lowest. Thus, non-fault-tolerant circuits constructed using the Latin rectangle method~\cite{Steane03} may be preferred over the overlap method of Ref.~\cite{Paetznickgolay2013}.

\section{Steane code}
\label{sec:steane}

% Next we look at the preparation of more complex states than the cat state. For example, we will consider the preparation of codestates for the $\llbracket 7,1,3 \rrbracket$ Steane code and the $\llbracket 23,1,7 \rrbracket$ Golay code.

\begin{figure}
    \centering
    \includegraphics[width=.35\textwidth]{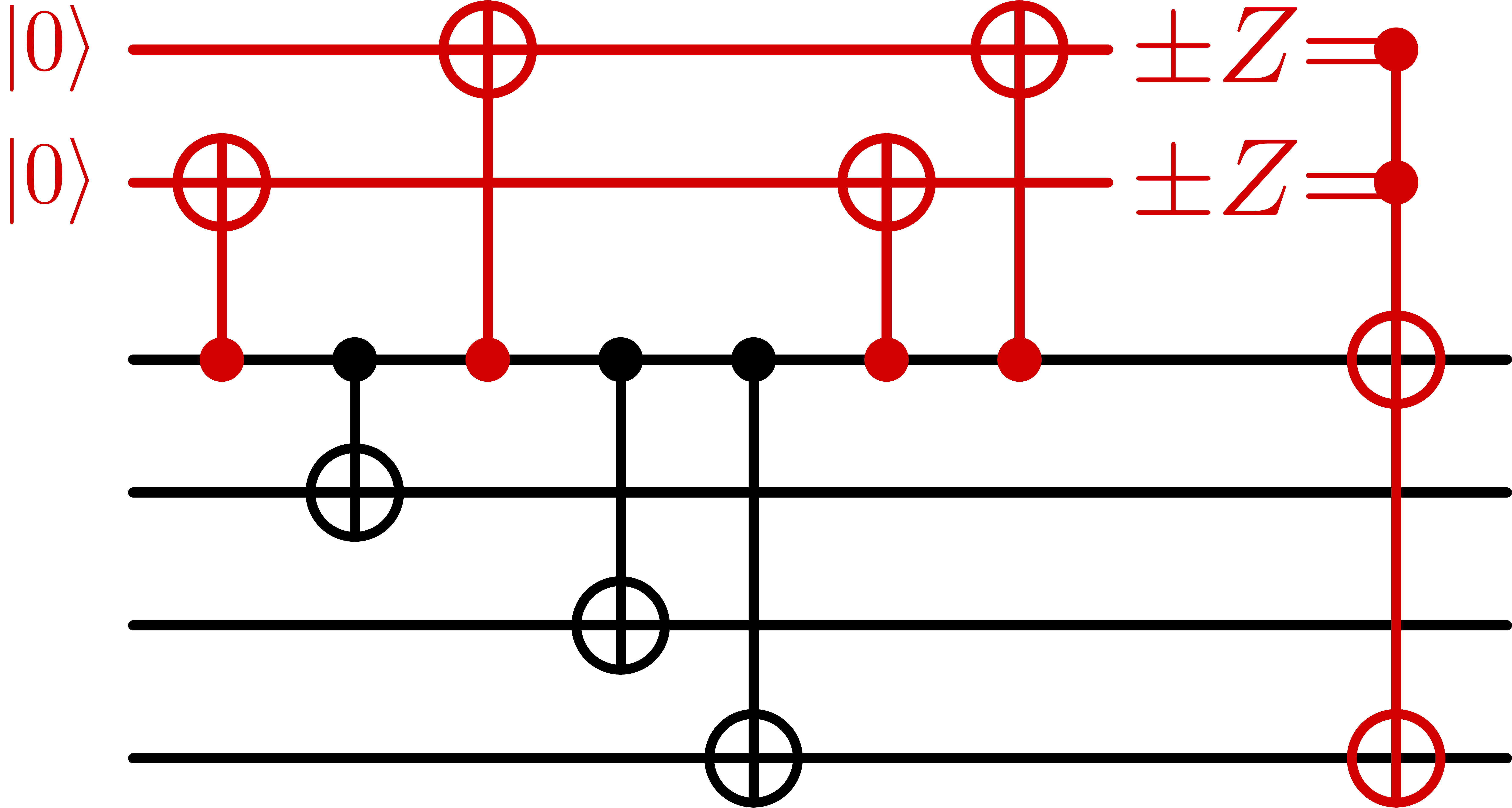}
    \caption{Fault-tolerant circuit for encoding the operator $XXXX$, given the first data qubit does not start in a $Z$ eigenstate. Data qubits are shown in black and ancilla qubits in red. Ancilla qubits flag correlated errors and apply appropriate corrections on to the data qubits. Note that this circuit is derived from the flag-fault-tolerant stabilizer measurement circuits in \chapref{chap:SMcat}.}
    \label{fig:wt4d3}
\end{figure}

\begin{figure*}
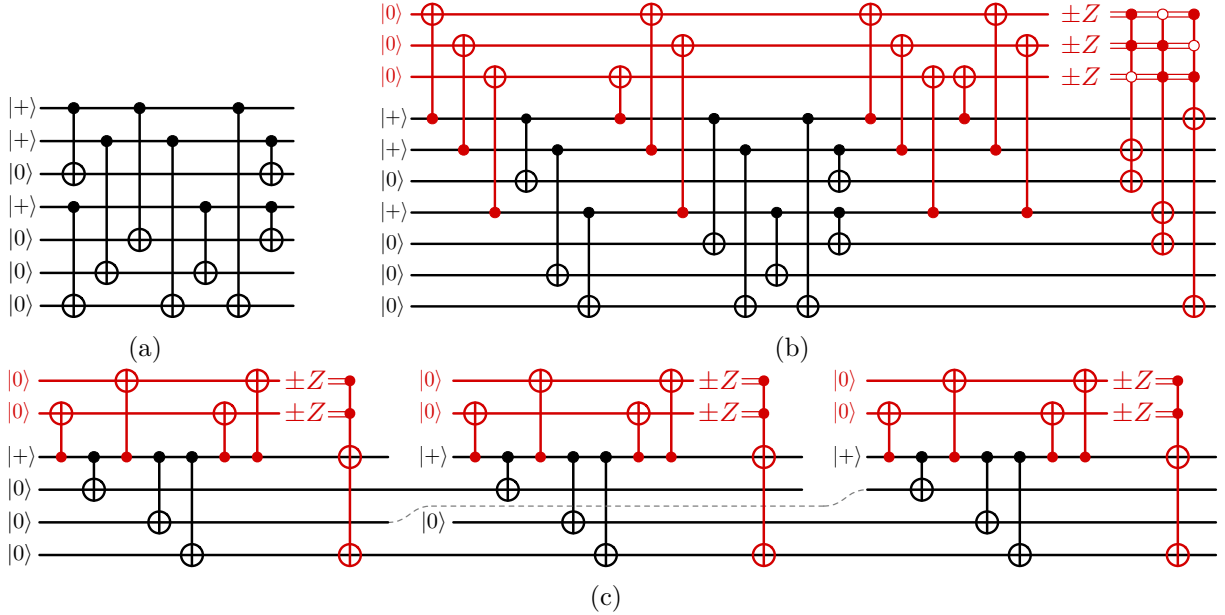

    \centering
     \subfloat[\label{fig:steaneprepnonft}]{\includegraphics[width=.23\textwidth]{imagesChap4/steanenonft2.pdf}}
     \hspace{1.cm}
     \subfloat[\label{fig:steaneprepshort}]{\includegraphics[width=.67\textwidth]{imagesChap4/steaneftshort.pdf}}
     \hspace{.1cm}
     \subfloat[\label{fig:steanepreplong}]{\includegraphics[width=.97\textwidth]{imagesChap4/steaneftlong.pdf}}
    \caption{Fault-tolerant preparation of a $\llbracket 7,1,3\rrbracket$ $\ket{0}_L$ state. (a)~Circuit to ideally prepare the $\ket{0_L}$ state of the Steane code, needing nine CNOT gates over three rounds. (b)~Condensed circuit to fault-tolerantly  and deterministically prepare the $\ket{0_L}$ state of the Steane code. This circuit uses the same number of CNOTs as in (c), but has circuit depth seven, as opposed to $21$ in (c). (c)~Circuit to fault-tolerantly prepare the $\ket{0_L}$ state of the Steane code, using the weight-$4$ operator encoding circuit of \cref{fig:wt4d3}. 
    Two ancillas are needed if qubits can be measured and reset quickly.
    }
    \label{fig:steanecircs}
\end{figure*}

Circuits for the preparation of encoded states of the Steane code go through periodic updates every 10 years~\cite{Steane97, Reichardt06, Goto16}. The update of this decade is a circuit to prepare $\ket{0}_L$ fault-tolerantly and deterministically, as shown in \figref{fig:steanecircs}. The circuit uses $21$ CNOT gates and either nine or ten qubits total, depending on the qubit measurement speed. It is constructed by using flags to make the circuit in \figref{fig:steaneprepnonft} fault-tolerant. Note that qubits $1$, $2$ and $4$ in \figref{fig:steaneprepnonft} contain the controls for $3$ CNOT gates each, similar to the circuit for preparing a weight-$4$ cat state. This permits the use of flag techniques from \chapref{chap:SMcat} and \chapref{chap:catstateprep}. Using the fault-tolerant $X^{\otimes 4}$ encoding circuit in \figref{fig:wt4d3}, we can construct the fault-tolerant state preparation circuit in \figref{fig:steanepreplong}.

Notice that for each pair of flags, only the weight-two flag configuration is used. No corrections are applied for any of the weight-one flag configurations. This provides a clue for further optimization. Consider that with three total flag qubits, there is access to three distinct weight-two flag configurations, and even one of weight three. We show in \figref{fig:steaneprepshort} that the flag circuits for the three stabilizer encodings can be merged to use a common pool of ancilla qubits. This has the effect of minimizing the space overhead and circuit depth. This technique of sharing flags has not been explored in great depth in this thesis, but this method shows very high potential for reducing overhead. For example, similar to encoding multiple stabilizers in parallel, stabilizers may also be measured in parallel with a common pool of flag qubits. Note that the number of flag configurations increases exponentially, but the number of faults to be tolerated only increases linearly in the number of stabilizers measured.

To prepare the $\ket{0}_L$ state of the Steane code fault-tolerantly, we must only ensure that weight$>2$ $X$ errors do not occur due to a single fault. Considering the $Z_L$ operator is included as a stabilizer, any $Z$ error is equivalent to a weight-$1$ $Z$ error. This is because the Steane code is perfect CSS, greatly simplifying the fault-tolerant circuit when preparing by stabilizer measurements. The $\ket{0}_L$ state can be prepared by initializing the seven physical qubits in the $\ket 0$ state, and subsequently measuring three of the $X$-type stabilizer generators. Note that it is fault-tolerant to just measure the stabilizer generators once since regardless of whether a fault occurs, any $Z$ error will be equivalent to a $Z$ error of weight at most one. If there is no fault, the projected $Z$ error space has been identified and a correction can return the state to the code space. If there was a fault, then it is fault-tolerant to project into any of the weight-$1$ error spaces. 

Alternatively, the state may be prepared by initializing the qubits in the $\ket +$ state and measuring the $Z$-type operators. This approach can have even lower overhead. Instead of measuring weight-four operators with flag qubits for fault tolerance, weight-three $Z$ operators can be measured fault-tolerantly using just one qubit. The circuit will require $8$ qubits overall, as in Ref.~\cite{Goto16}, if the ancilla can reset quickly. The caveat is that potentially more than four operators will need to be measured for distance-three $X$-type fault tolerance.

We can also prepare logical magic states of the Steane code fault-tolerantly using a transversal injection technique similar to that in Ref.~\cite{Gavriel23}. First, initialize all qubits in the $\ket H$ state to ensure the logical state is projected into the $+1$-eigenspace of the $H_L$ operator. To complete the logical state preparation, the state is projected into the code space of the code using the Shor-type measurements of individual stabilizers. In this scenario, $Z$ errors of weight two are possible, hence more than three $X$-type stabilizers will need to be measured for fault tolerance. We find that four weight-four stabilizers are sufficient for the Steane code for fault tolerance by postselection. It is unclear if this protocol can be made deterministic. 

This method of encoded magic state preparation is interesting because it possesses properties of both magic state injection and magic state distillation. In the former, individual magic states are injected from a physical qubit into a quantum error-correcting code with encoded error probability scaling as the error rate of physical magic state preparation $O(p)$. In the latter, many magic states with error rate $p$ are processed through a Clifford circuit to produce one magic state with error rate $O(p^k)$, physical gate fidelities permitting. Our preparation technique captures the better properties of both these methods. Magic states are injected from physical qubits into a quantum error-correcting code, while simultaneously improving the error probability to $O(p^2)$.

\subsection{Simulations}

We simulate the circuits for state preparation using a stabilizer circuit simulator to collect statistics on the residual weight-one and weight-two errors. These simulations are performed according to the model in \ref{subsubsec:nonlocal}, for $p \in [10^{-2.5}, 10^{-1.6}]$. In doing so, we also compare the post-selection protocols with the new deterministic ones.
We plot the probability of weight-one and weight-two $X$ errors, and the yield of the protocol. In \figref{fig:SteaneHeraldedsims}, we first compare results for different heralded methods as described in Fig.~1 of Ref.~\cite{Goto16}. The fourth method, labeled 'Goto (d)', replaces the measurement of the $X_0 X_5 X_6$ in (c) with $X_2 X_4 X_5$. With this measurement, the logical error rate is decreased with a small increase in the weight-one error rate. The yields of these postselection based protocols are very high, considering very few gates need to be performed.

We compare the Goto (c) method with deterministic methods in \figref{fig:SteaneDeterministicsims}. We consider four circuits. The first, from Ref.~\cite{Reichardt06}, prepares two Steane code states with different preparation circuits followed by a Steane-style error correction procedure to identify the locations of residual $X$ errors. These residual errors can eventually be corrected to ensure \ref{cond:faulttolerance} is satisfied. The two flag methods used are the circuits shown in \ref{fig:steanecircs}. Finally, the measurement-based (MB) method measures the three weight-four $X$ stabilizers to prepare the state. The best performing deterministic method is the Steane-style error correction based method of Ref.~\cite{Reichardt06}.

\begin{figure}
    \centering
    \includegraphics[width=.7\textwidth]{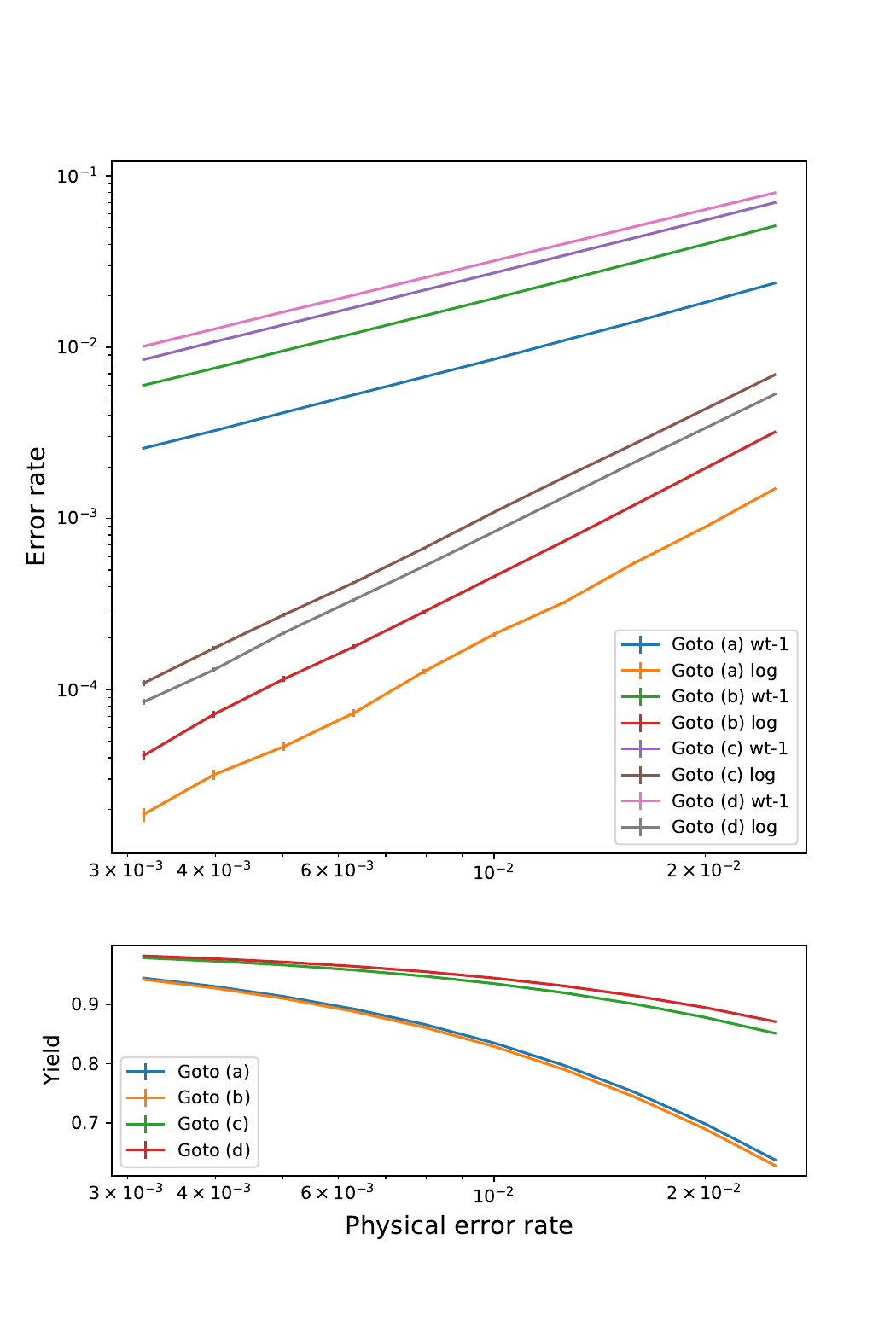}
    \caption{Rates of residual weight-1 and logical errors after preparing $\ket{0}_L$ of the Steane code using postselection-based fault-tolerant circuits. We also plot the postselection yield. The circuits considered here are found in Fig.~1 of Ref.~\cite{Goto16}. The method labeled `(d)' replaces the $X_0 X_5 X_6$ from (c) with $X_2 X_4 X_5$. This improves the logical error rate at the cost of a higher weight-one error rate.}
    \label{fig:SteaneHeraldedsims}
\end{figure}

\begin{figure}
    \centering
    \includegraphics[width=.7\textwidth]{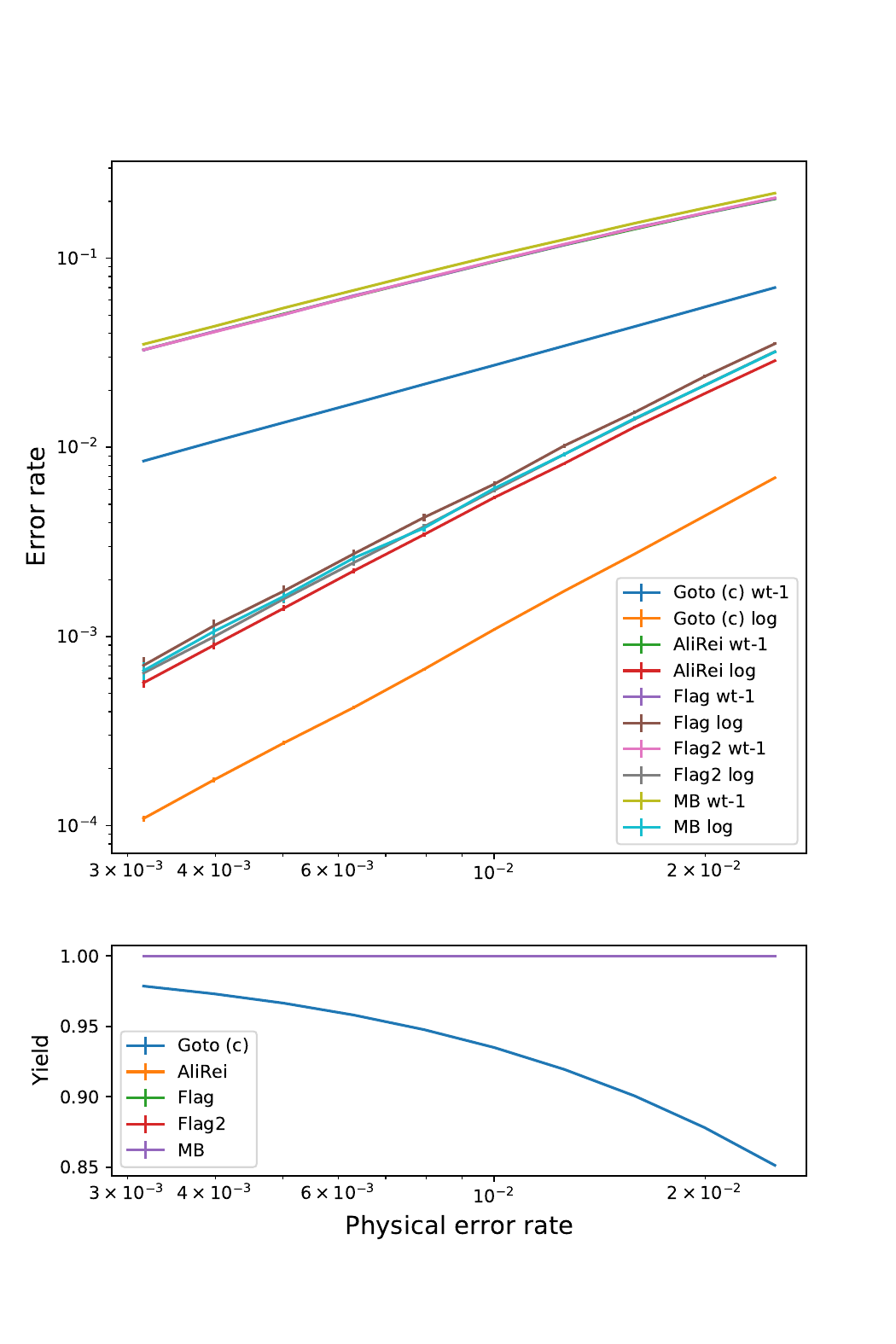}
    \caption{Rates of residual weight-1 and logical errors after preparing $\ket{0}_L$ of the Steane code using deterministic fault-tolerant circuits. We consider a method suggested by Aliferis and Reichardt in Ref.~\cite{Reichardt06}, the flag-based circuits in \figref{fig:steanecircs}, and a circuit to prepare the state by measuring $X$-type stabilizers. For physical error rates below $2 \times 10^{-2}$, the logical error rate of the best deterministic protocol is about five times worse than the postselection-based protocol.}
    \label{fig:SteaneDeterministicsims}
\end{figure}

\section{Golay code}
\label{sec:golay}

The Golay code has also enjoyed the same amount of attention as the Steane code, with work on state preparation circuits established from the early days of quantum fault tolerance~\cite{Steane97, Reichardt06, Paetznickgolay2013, ZhengBrunFTancprep18}. The protocol we demonstrate is the first, to our knowledge, to show that the $\ket{0}_L$ state of the Golay code can be prepared fault-tolerantly with $100\%$ yield. 

First, we consider state preparation by using flags to tolerate faults in operator encoding circuits.
Similar to the distance-three Steane code, the distance-seven Golay code is also a perfect CSS code. Since we prepare $\ket{0}_L$, the logical $Z$ operator is also included in the stabilizer group. This ensures that any $Z$ error of weight greater than three is equivalent to an error of weight at most three, implying only distance-five fault tolerance is required for $Z$-type errors. This fact was also observed in Refs.~\cite{Reichardt06, Paetznickgolay2013}. 

We choose the Latin rectangle circuit from Ref.~\cite{Paetznickgolay2013} with $77$ CNOT gates as the base non-fault-tolerant circuit for the state preparation. The encoding of the stabilizer operators in this circuit is then made fault-tolerant using flags. We require that the eleven weight-eight $X$ stabilizers are encoded fault-tolerantly to distance-seven and the twelve $Z$ stabilizers are encoded fault-tolerantly to distance-five. This can be done using the flag circuits shown in \figref{fig:highdcircs}. We perform all the $X$-type stabilizer encodings sequentially, allowing reuse of the nine flag qubits for each $X$-stabilizer. As a consequence, the flag qubits for the $Z$-type operator encodings will need to be kept active for the entire duration of the state preparation. In total, the number of qubits needed is $23 + 9 + 12*6 = 104$, with a total of $419$ CNOT gates. In contrast, the best performing postselection circuit, from Ref.~\cite{Paetznickgolay2013}, uses $69$ qubits simultaneously, with $297$ CNOT gates.

\begin{figure*}
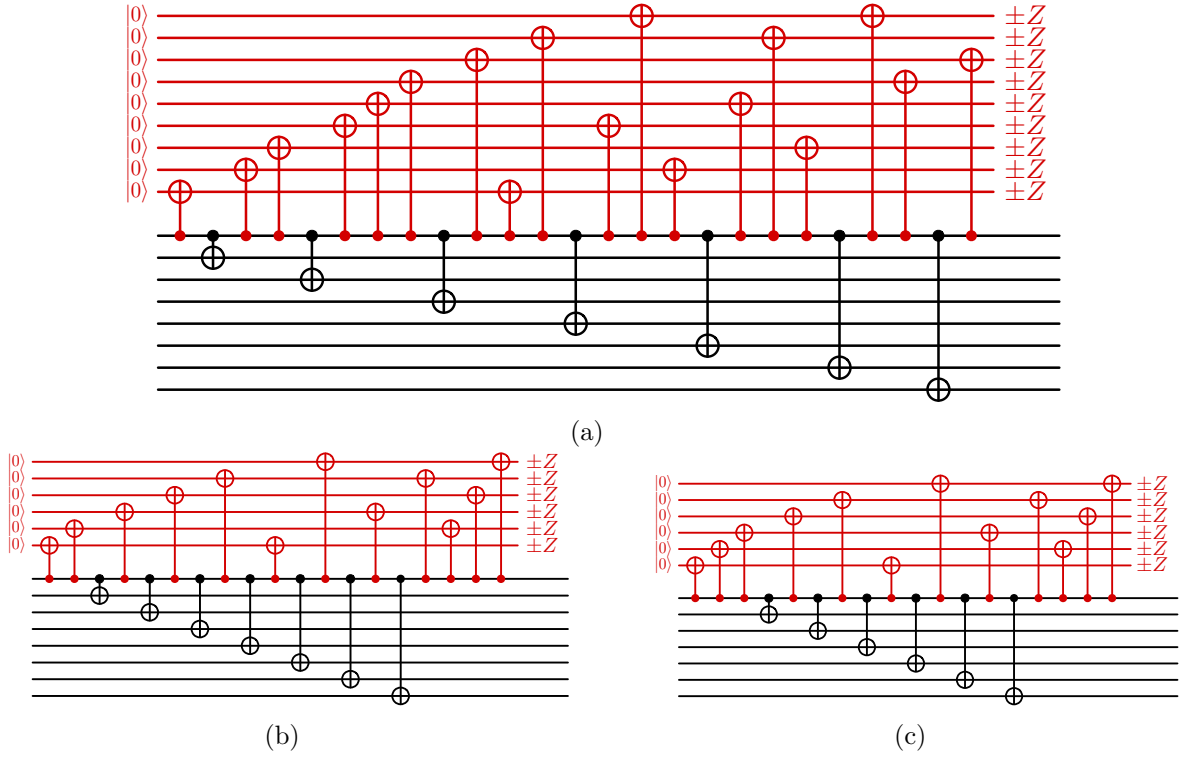

    \centering
     \subfloat[\label{fig:wt8d7}]{\includegraphics[width=.75\textwidth]{imagesChap4/wt8Xoperatord7.pdf}}
     \hspace{5cm}
     \subfloat[\label{fig:wt8d5}]{\includegraphics[width=.45\textwidth]{imagesChap4/wt8Xoperatord5.pdf}}
     \hspace{1cm}
     \subfloat[\label{fig:wt7d5}]{\includegraphics[width=.42\textwidth]{imagesChap4/wt7Xoperatord5.pdf}}
    \caption{Flag-based circuits to encode high-weight $X$-type operators fault-tolerantly to distance five and seven. Corrections are computed based on the flag measurement outcomes based on the techniques in \ref{s:SMcatSDSM5}. (a)~A circuit to encode a weight-eight operator fault-tolerantly to distance seven. (b)~A circuit to encode a weight-eight operator fault-tolerantly to distance five. (c)~A circuit to encode a weight-seven operator fault-tolerantly to distance five. } 
    \label{fig:highdcircs}
\end{figure*}

We finally consider the preparation of the $\ket{0}_L$ state of the Golay code by measuring stabilizer operators. We first consider preparation by initializing all the physical qubits in the $\ket 0$ basis, followed by $X$-type stabilizer measurements. This is favored since we are only required to find a sequence of $X$-type stabilizers to measure that is distance-five fault-tolerant to faults corrupting measurements (any $Z$ error is at most weight-three). Even with this simplification, a fast Mathematica script struggled to verify the fault tolerance of sequences longer than $28$ stabilizers. Unfortunately, an appropriate distance-five sequence was not found. However, distance-three fault-tolerant sequences were determined at length $19$. 

Checking random sequences can be very time-consuming, as there are many options of stabilizers to consider measuring. We tried an alternative method to find a distance-five sequence found in Sec. $IX$ C of Ref.~\cite{delfosse2020short}. Stabilizers in the sequence are chosen as cyclic shifts of the following weight-eight stabilizer generator $$10000000000111110010010.$$ We verified that a sequence of $20$ stabilizers is sufficient for distance-three fault tolerance. Our computer could not completely verify $29$ stabilizers for distance-five fault tolerance. We conjecture that a sequence of $29$ to $31$ stabilizers will be required.

\subsection{Simulations}
\label{subsec:simsSteaneGolay}

Using a similar approach as in \secref{sec:steane}, we simulate the state preparation circuits under a depolarizing noise model to collect error rate statistics. The rates of residual $X$ errors of weight-$k$ are determined, for $k \leq 4$. We compare the performance of the postselection-based protocol of Ref.~\cite{Paetznickgolay2013} and our circuit based on the flag-fault-tolerant encoding of operators. The residual error rates with Paetznick's postselection protocol match those of the deterministic protocol, albeit at a physical error rate that is a factor of ten higher. At similar physical error rates, the deterministic protocol performs on par with some of the higher overhead postselection-based protocols of Ref.~\cite{ZhengBrunFTancprep18}. These protocols are much more useful for the distillation of codestates of much larger codes. Quantum codes as small as the $\llbracket 23,1,7 \rrbracket$ Golay benefit from a more granular analysis of fault propagation.

\begin{figure}
    \centering
     \label{fig:golay}
     \includegraphics[width=.7\textwidth]{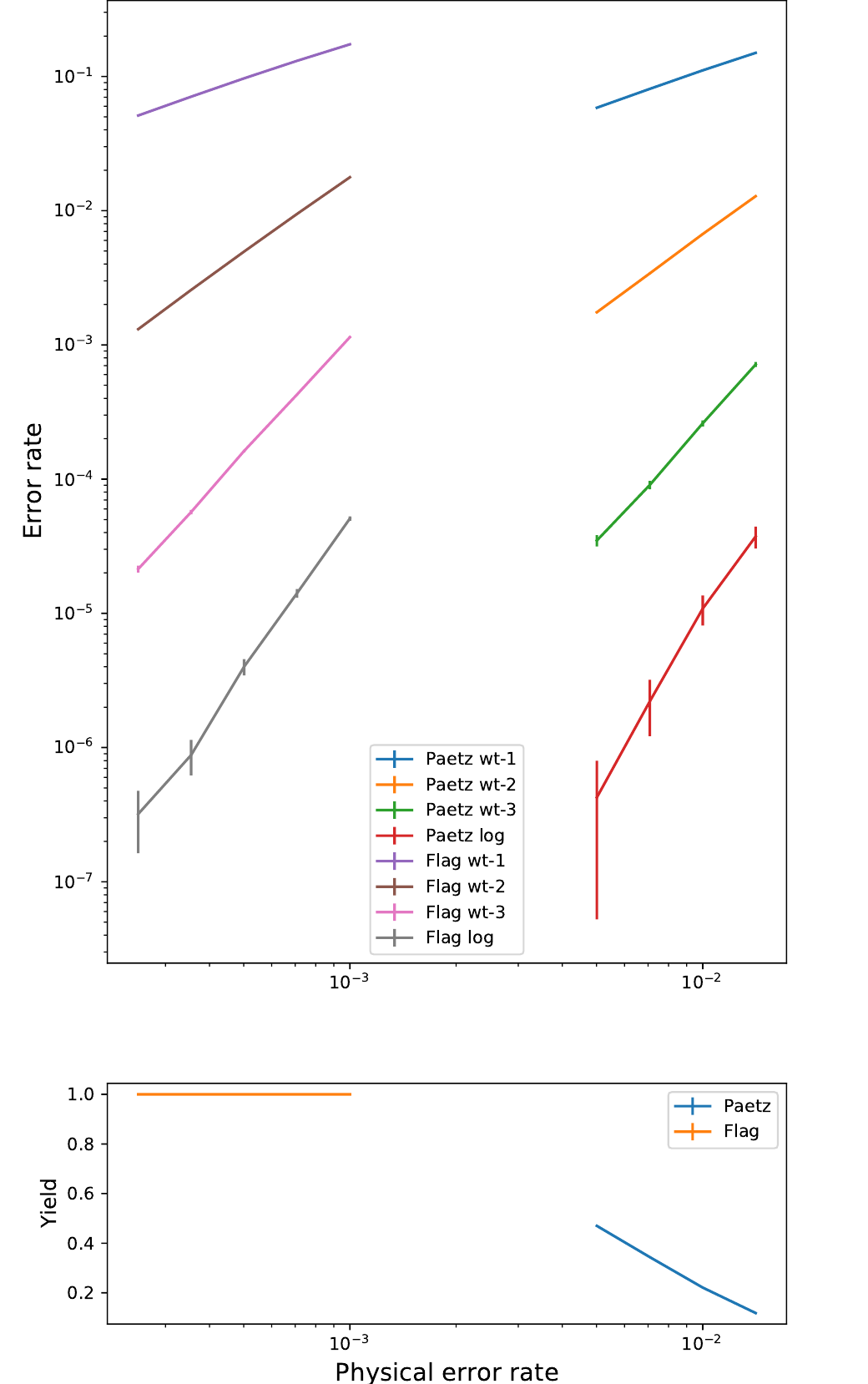}
    \caption{Comparison of the rates of residual errors between the $69$-qubit postselection-based protocol of Ref.~\cite{Paetznickgolay2013} and a circuit that uses flags to fault-tolerantly encode the stabilizer operators of the state. The flag-based corrections are deterministic, allowing us to demonstrate the first state preparation circuit for the Golay code that is deterministic. The postselection method is very robust and works well even at high physical error rates. The deterministic method is only really useful at a physical error rate below $10^{-3}$.}
\end{figure}

\chapter[Distance-four quantum codes]{Fault-tolerant error correction for distance-four quantum codes with postselection}
\label{chap:D4codes}

\textbf{Error correction and postselection.}
Noisy Intermediate-Scale Quantum (NISQ) algorithms for eigensolvers~\cite{Peruzzo2014, Wecker15eigen} and machine learning~\cite{Biamonte2017} are growing as popular applications for state-of-the-art few-qubit quantum systems.  Unfortunately, these devices are still prone to large amounts of noise~\cite{mooney2021wholedevice, Google2019, Pogorelov21ionexp}.  Although error correction can decrease error rates~\cite{Shor95decoherence, Gottesman97thesis, terhal15review}, current experiments encode only one logical qubit that is still fairly noisy~\cite{honeywell21steane, Google2021, duke21BScode, Wootton20ignis}.
In this paper we simulate storing multiple logical qubits in a lattice, as a first step toward modeling few-qubit computations. We repeatedly correct and remove single-qubit errors.  On detecting a more dangerous---and less common---two-qubit error, we reject and restart.  This ``postselection" technique allows distance-four codes to achieve similar logical error rates to distance-five codes.  For example, as shown in \tabref{f:tab6lsumm}, a distance-four code can correct errors on six logical qubits with a similar failure rate as the distance-five surface code, using only $10 \%$ as many physical~qubits.  The table also shows that acceptance rates are fairly high, so occasional restarts should not be a major issue for low-depth NISQ~\mbox{algorithms}~\cite{cerezo2020variational, bharti2021noisy}.

Postselection is a versatile tool in the quantum toolkit. 
In experiments, it has been used to decode the $\llbracket 4,2,2 \rrbracket$ error-detecting code~\cite{Linke17edexp, Takita17edexp} and the $\llbracket 4,1,2 \rrbracket$ surface code~\cite{Andersen2020, Google2021}.  In theoretical research, it has been used to reduce the logical error probability of state preparation~\cite{Paetznickgolay2013, ZhengBrunFTancprep18} and magic state distillation~\cite{Bravyi05, Meier13MSD}. Recently, postselection has been used to improve quantum key distribution~\cite{Yumang20QKDpost, Sekga21QKD} and learning quantum states~\cite{Scott18shadow, McClean2020}.  

\begin{table}
    \caption{Postselected error correction for $6$ logical qubits using $\llbracket 16,k,4\rrbracket$ codes on the $25$-qubit planar layout of~\figref{f:layout}. The probability of logical error, acceptance, and expected time to complete are shown for 300 time steps, with noise rate $p = 5 \times 10^{-4}$.  The $k = 6$ code achieves logical error rate close to the distance-$5$ surface code using only $10\%$ of the qubits.  In comparison, for $6$ physical qubits at memory error rate $p/10$, the probability of error is about~$6\times 300 \times p/10 = 0.09$.
    \label{f:tab6lsumm}}
    %\setlength{\tabcolsep}{1.9pt}
    %\begin{tabular}{cccccc}
    %\hline \hline
    %%t = 700, p = 0.0005, 6 \text{ qubits}
    %\multicolumn{3}{c}{}  & &\multicolumn{2}{c}{Postselection}\\ 
    %Code & Qubits & $P(\text{Logical error})$ &  & Acceptance  & $\mathbb{E}[time]$ \\
    %\hline
    %$k = 2$ & $25 \times 3 = \mathbf{75}$ & $.032$ & & $11.5\%$ & $1067 $ \\
    %$k = 4$& $\mathbf{50}$ & $.028$ & & $20.5\%$ & $725$\\
    %$k = 6$&  $\mathbf{25}$ & $\mathbf{.017}$ & & $35.5\%$ & $530$\\[.25 cm]
    %$d = 4$ & $150$ & $\mathbf{.009}$ & & $3.5\%$ & $2820$\\
    %$d = 3$ & $78$ & $.143$ & & $100\%$ &$300$\\
    %$d = 5$ & $246$ & $\mathbf{.016}$ &  & $100\%$ & $300 $\\
    %\hline \hline \\
    %\end{tabular}
    % \hspace{-30mm}
    \centering
    \includegraphics[width=.7\textwidth]{imagesChap5/introsummresults005}
\end{table}

Knill previously combined postselection with error correction, on concatenated distance-two codes, to show an impressive $3 \%$ fault tolerance threshold~\cite{Knill05nat}.  
We also combine postselection and error correction, but with distance-four codes.
 As~\tabref{f:disterrtab} indicates, distance-two codes can detect single errors and distance-three codes can correct them, meaning logical errors are due to second-order faults. Distance-three codes may alternatively be used to detect one or two errors, but then they lose the ability to correct and computations are very short-lived. We choose to use distance-four codes since they can simultaneously correct an error and detect two errors. Correcting some errors ensures restarts are less frequent, so longer computations can be run.  Since logical errors are caused only by third-order faults, logical error rates are very low.

\begin{table}
    \caption{Distance-four codes with postselection lead to $O(p^3)$ logical errors, much like distance-five codes. Even-distance codes require restarts, however, unlike odd-distance codes.
    \label{f:disterrtab}}
    \centering
    \includegraphics[width=.55\textwidth]{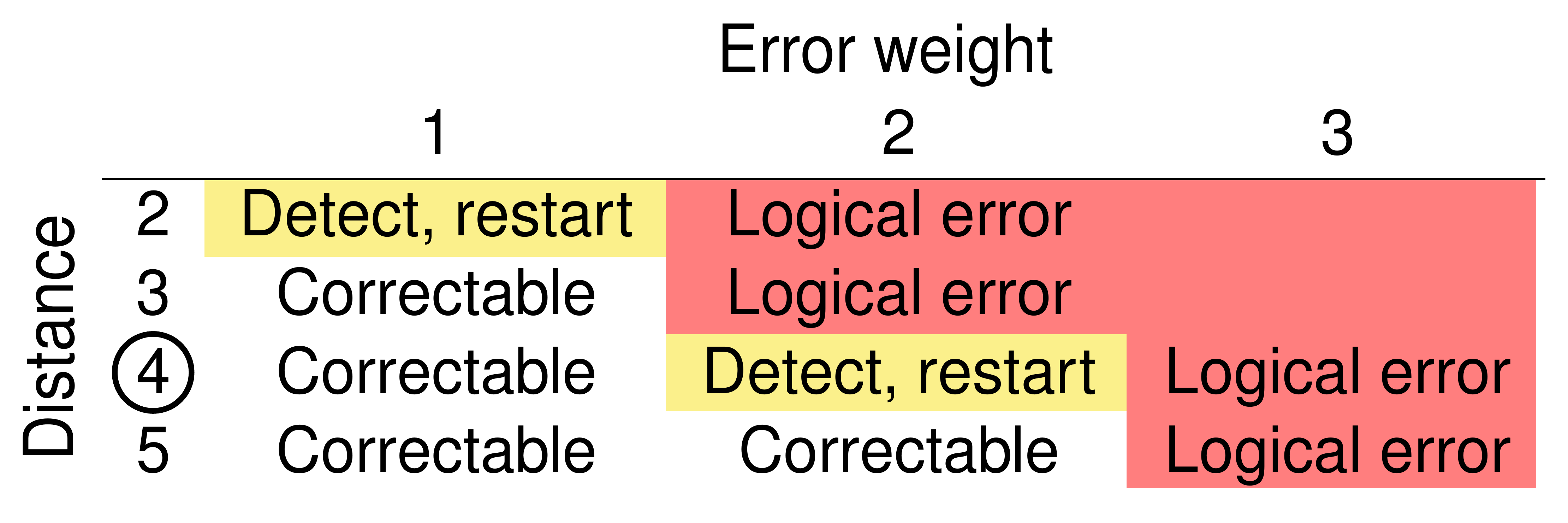}
\end{table}

\begin{figure}
\centering
\includegraphics[width=.99\textwidth]{imagesChap5/codes.pdf}
\caption{Codes considered in this paper, with associated distance-four fault-tolerant $Z$ or $X$ stabilizer measurement sequences.  (The last three codes are self-dual CSS.)  Time steps of parallel measurements are separated by ``$\mid$".  ($\ast$)~For the surface code, fault-tolerant $X$ and $Z$ error correction is carried out using a rolling window of four syndromes, each measured in two time steps.}
\label{f:codes}
\end{figure}

\begin{figure}
\centering
{\includegraphics[width=.35\textwidth]{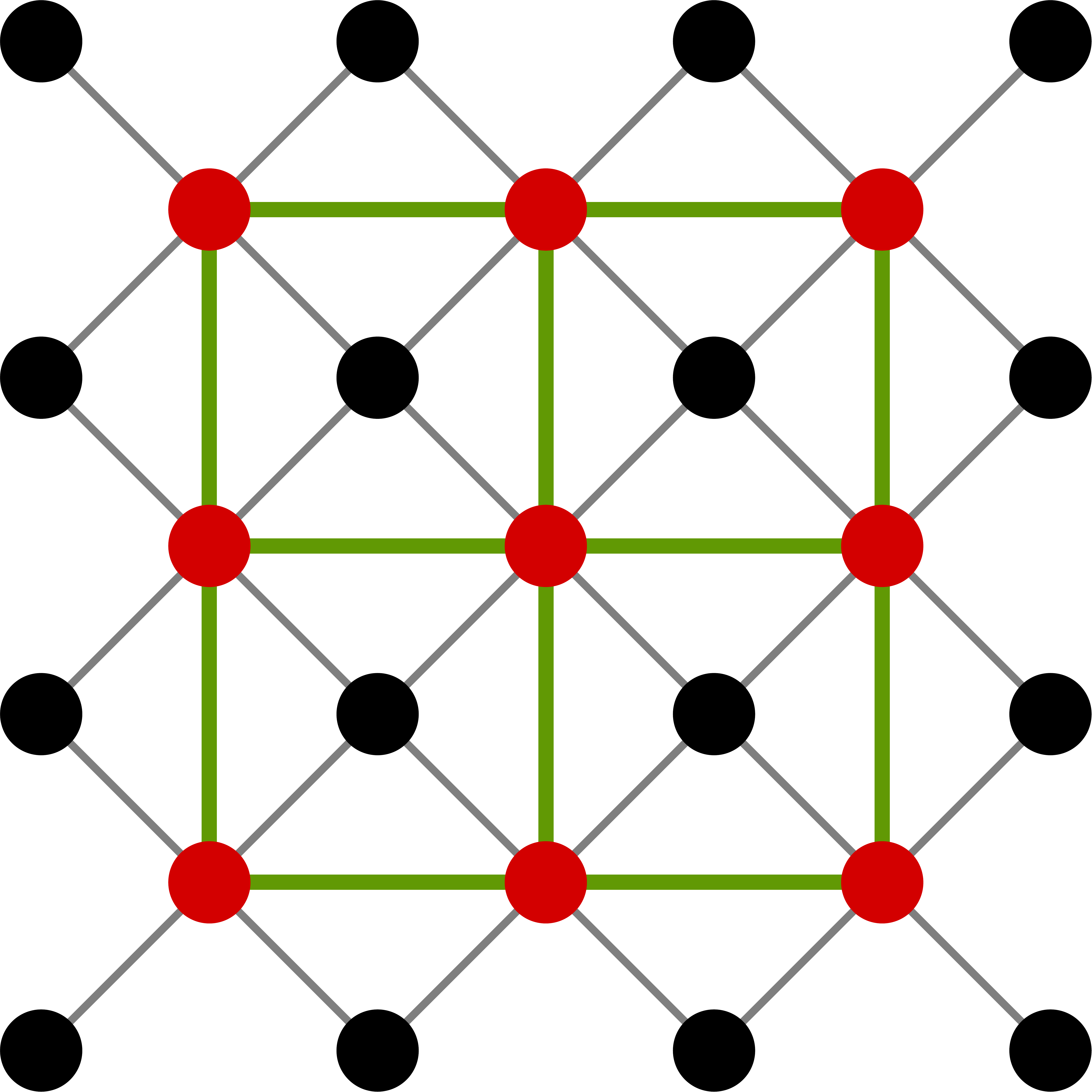}}
\caption{Planar layout of $16$ data and 9 ancilla qubits, in black and red respectively.  CNOT gates are allowed along the edges.  Grey edges are required for the surface code, and green edges between ancillas are required for the new codes in this paper.}
\label{f:layout}
% \vspace{-0.2cm}  %DEBUG
\end{figure}

\noindent
\textbf{Physical layout.}
In practice, it is difficult to build a quantum computer with native (fast, reliable) two-qubit gates between \textit{every} pair of qubits. Instead, qubits are placed on a one or two-dimensional lattice and two-qubit gates are mediated by local interactions, as in superconducting architectures and solid state systems. Current ion trap systems use long-range gates~\cite{Nigg14SteaneEC} and transport mechanisms~\cite{honeywell21steane} to connect all the qubits, but some degree of locality is required for larger systems.
In light of these connectivity constraints, it may be wiser to choose quantum codes that can be laid out on a lattice such that error correction requires the fewest number of local native gates.
The popular surface code has the attractive feature that it requires only nearest-neighbor interactions on a $2$D square lattice~\cite{Fowler12surface, CampbellroadstoFT2017}.  Similarly, error correction for topological codes has been investigated on sparser degree-three lattices~\cite{Chamberland20heavycodes, Chamberlandtriangularcodesflag2020, gidney2021faulttolerant}. 
But there is insufficient research on the performance gains of denser connected layouts.  

%\medskip

We suggest using $16$-qubit codes on the $25$-qubit rotated square lattice of~\figref{f:layout}, where ancilla qubits additionally interact with neighboring ancillas.  This allows the use of 
flag qubits for fault tolerance~\cite{ChamberlandBeverland17flagft, ChaoReichardt17errorcorrection, chao2019flag, Prabhu23}, in turn allowing measurement of large stabilizers. 
As shown in~\figref{f:codes}, we choose distance-four codes whose stabilizer generators are fairly local, with short Shor-style stabilizer measurement sequences that do not require any SWAP gates.
We consider two block codes and a color code that encode multiple qubits~\cite{delfosse2020short}, and the rotated surface code~\cite{BombinMartindelgado07surfaceoverhead} as a benchmark for postselection.  
In contrast, before the advent of topological codes, block codes were used for the simulation of $2$D local error correction~\cite{Svore07localSteane, Spedalieri09localBS, Lai14localKnill}.  These proposals performed Steane error correction on small distance-two and -three codes, and required many~swaps.

% \pagebreak %DEBUG

\medskip
\noindent
\textbf{Results.} 
We compare our $16$-qubit codes with the $25$-qubit, distance-five surface code.
 We show below in~\figref{f:scaling} that, with rejection, the normalized logical error rate of the proposed codes is less than that of the distance-five surface code by as much as one order of magnitude. The distance-four surface code actually achieves two orders of magnitude separation. 

However, the logical error rate per time step does not capture the drawback of restarts. Instead, a better metric is the cumulative probability of logical error.
\figureref{f:resultssumm} compares this metric between the different codes for short computations that do not restart too often (more information in~\figref{f:PLEt}).
For one logical qubit, the distance-four surface code vastly outperforms its distance-five counterpart, and the $k = 2$ and $k = 4$ codes achieve a good balance of low qubit overhead and low logical error rate.
 We also show that just $50$--$75$ physical qubits are sufficient for good protection of twelve logical qubits.  Overall, we obtain lower logical error rates with higher encoding rates, using postselection and multi-qubit codes.

In \appref{s:unenc}, we compare the storage error rate of unencoded qubits with the encodings in \figref{f:codes}. As expected, at error rates up to $10^{-3}$, fault-tolerant error correction is more robust than leaving qubits idle.

\begin{figure}
\centering
\includegraphics[width =0.99\textwidth]{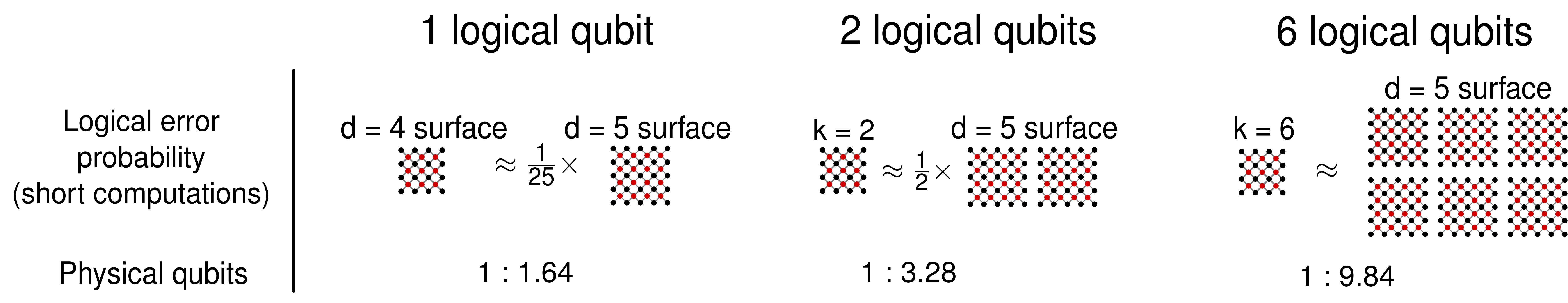}
\caption{Summary of results. For short computations, the probability of a logical error in the distance-$4$ rejection-based surface code is approximately $25$ times lower than that of the distance-$5$ variant. Further, for $6$ logical qubits, the $k=6$ code on one patch of $25$ qubits can match $6$ patches of the distance-$5$ surface code.}
\label{f:resultssumm}
\end{figure}

%\pagebreak. %DEBUG

\medskip
\noindent
\textbf{Future work.} 
In order to verify these results on current quantum systems, some work is required.  Dense qubit connectivity in ion trap systems may allow for simple measurement of high-weight stabilizers, but superconducting devices generally prefer low qubit degree due to high crosstalk errors.  It may be possible to modify the circuits in this work to allow maximum qubit degree at most five or six, such as in the IBM Tokyo device~\cite{Tan21optimality}.
Consequently, in~\figref{f:degree4k2} below, we show that error correction of the $k=2$ code is possible with degree-four connectivity, but requires many extra qubits.

We only show how to do fault-tolerant error correction, but the ultimate goal is to perform
quantum computation. 
Selective logical measurements could induce computation within a patch, and transversal gates between vertically stacked code patches could facilitate non-Clifford gates.  
If these operations introduce a low amount of error, it may be possible to execute relatively high-depth circuits. These tools can then be used to execute short NISQ and magic state distillation algorithms.  As an example, our results show that just $50$ physical qubits may be sufficient to demonstrate $10$-to-$2$ MEK distillation experimentally with $O(p^3)$ logical errors~\cite{Meier13MSD}.  

\medskip
\noindent
\textbf{Organization.}
In \secref{sec:d4codes} we provide more details about distance-four codes and the examples we choose in this paper. \secref{sec:d4FTEC} details the methods used for fault tolerance. In particular, stabilizer measurement circuits are dealt with in \secref{s:SMC} and sequences of stabilizer measurements are handled in \secref{s:SMS}. The noise model and results of simulations are contained in~\secref{sec:d4simresults}. \secref{sec:d4future} concludes with a discussion of future work and open questions.

\section{Codes}
\label{sec:d4codes}

We compare the error correction performance of six $\llbracket n,k,d \rrbracket$ stabilizer quantum codes, where $n$ is the number of physical data qubits and $k$ is the number of logical qubits.
A distance-$d$ quantum code should correct all errors of weight $j \leq t= \lfloor \frac{d-1}{2} \rfloor$, occurring at rate $O(p^j)$, for error rate $p$.
At low error rates, an unlikely error of weight-$(d-j)$ may be misidentified as the more likely weight-$j$ error, inducing a logical flip on recovery.
In even-distance codes, errors of weight $d/2$ 
can be detected, but applying a correction may induce a logical flip.
In this paper, we stop the computation instead of attempting to correct, ensuring logical flips only occur at rate $O(p^{t+2})$ and not $O(p^{t+1})$ as before.

For the distance-four codes shown in~\figref{f:codes}, we show that the logical error rate scales as $O(p^3)$ like a distance-five code.  
As a benchmark we first consider the rotated distance-four surface code~\cite{tomita14surface} on the layout of~\figref{f:layout}.  For a fair comparison of both the resource requirements and logical error rate, we consider additional benchmarks: the  distance-three and -five surface codes.  As in Ref.~\cite{tomita14surface}, each distance-$d$ surface code uses $d^2$ data qubits and $(d-1)^2$ ancilla qubits.  

The next three codes are the central focus of this work.  These self-dual CSS (Calderbank-Shor-Steane) codes were first considered in Ref.~\cite{delfosse2020short}, to show examples of codes that can be constructed to have single-shot sequences of stabilizer measurements. 
By fixing some of the logical operators, the $k=6$ code can be transformed into the $k=4$ and $k=2$ codes. Alternatively, puncturing the $k=6$ code yields the well-known $\llbracket 15,7,3 \rrbracket$ Hamming code.

\medskip
\noindent
\textbf{Improvements.}
Although these codes encode more logical qubits, they suffer from the difficult task of having to measure weight-eight stabilizers.  It is possible to construct a $\llbracket 16,2,4 \rrbracket$ subsystem code with only weight-four stabilizers and gauge operators.  Using the layout of~\figref{f:layout}, we compared  this code with the $k=2$ subspace code in this paper but found no significant improvements. This code is still useful, however, as we show in~\secref{sec:d4future}.

Many other codes can also be constructed with $16$ qubits. For a biased-noise system, a CSS code with two logical qubits can be constructed with $Z$-distance six and $X$-distance four.  For more logical qubits, a non-CSS $\llbracket 16,7,4\rrbracket$ code can be used~\cite{Grassl07codetable}.
Although its stabilizer generators are larger, flag-based measurement may still offer a low-overhead route to fault-tolerance.

\section{Fault-tolerant error correction}
\label{sec:d4FTEC}

A stabilizer measurement \textbf{circuit} is made fault tolerant to quantum errors by using extra physical qubits.  These ancillas are used to catch faults that may spread to high-weight errors.  In contrast, the bad faults in a syndrome extraction \textbf{sequence} flip syndrome bits.
Additional stabilizers are measured, essentially encoding the syndrome into a classical code.

%\pagebreak  %DEBUG

A circuit is fault-tolerant to distance $d$ if $j \leq t= \lfloor \frac{d-1}{2} \rfloor $ mid-circuit faults cause an output error of weight at most~$j$.  
Additionally for even distance fault tolerance, sets of $d/2$ faults spreading to weight $>d/2$ errors should be detected so the computation can be restarted.  
When these faults yield an error of weight $d/2$, the computation is restarted if the faults can be detected, else it is rejected in the next round of error correction.

\subsection{Stabilizer measurement circuits} \label{s:SMC}

Quantum error correction involves the measurement of a set of operators called stabilizers, to diagnose the location of errors. For fault-tolerant error correction, these stabilizers may be measured individually, as in Shor's scheme~\cite{Shor96}, or together, using Steane- or Knill-type syndrome extraction~\cite{Steane97, Knill05}.  

The flag method is a popular spin-off of Shor's scheme~\cite{ChamberlandBeverland17flagft, ChaoReichardt17errorcorrection, chao2019flag, Prabhu23}.  By connecting multiple data qubits to each flag qubit, large stabilizers can be measured with relatively low overhead.  In addition, flag circuits can be made fault-tolerant only up to a desired degree. For example, Shor-style measurement of a weight-$w$ stabilizer needs $w+1$ ancillas and is fault-tolerant to distance $w$, but we show a weight-eight stabilizer measurement circuit with six ancillas that is fault-tolerant to distance four.

For distance-three fault tolerance, we show in~\figref{f:smcrules} that one fault in the circuit should result in an error of $X$ and $Z$ weight at most one.  For distance four, if two faults occur and can be detected, the computation must be rejected and restarted.  If this detection is not possible, the circuit must be designed to ensure errors cannot spread to weight greater than two.  Note that a fault may alter the value of the measured syndrome bit; syndrome bit errors are dealt with in~\secref{s:SMS}.

\begin{figure}
    \centering
    \subfloat[\label{f:smc00} ]{\includegraphics[width=.24\textwidth]{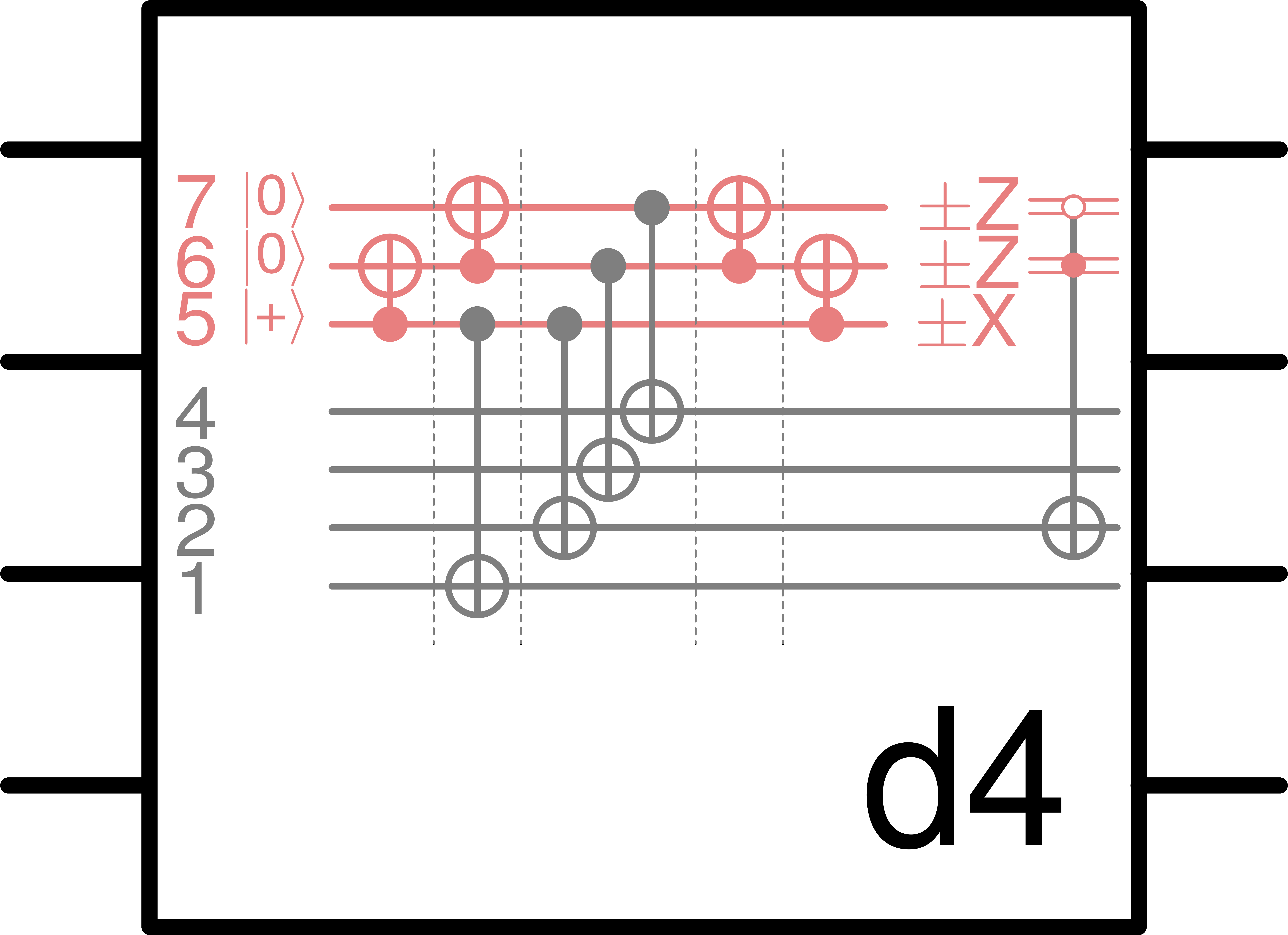}}
    \hspace{1cm}
    \subfloat[\label{f:smc0102} ]{\includegraphics[width=.62\textwidth]{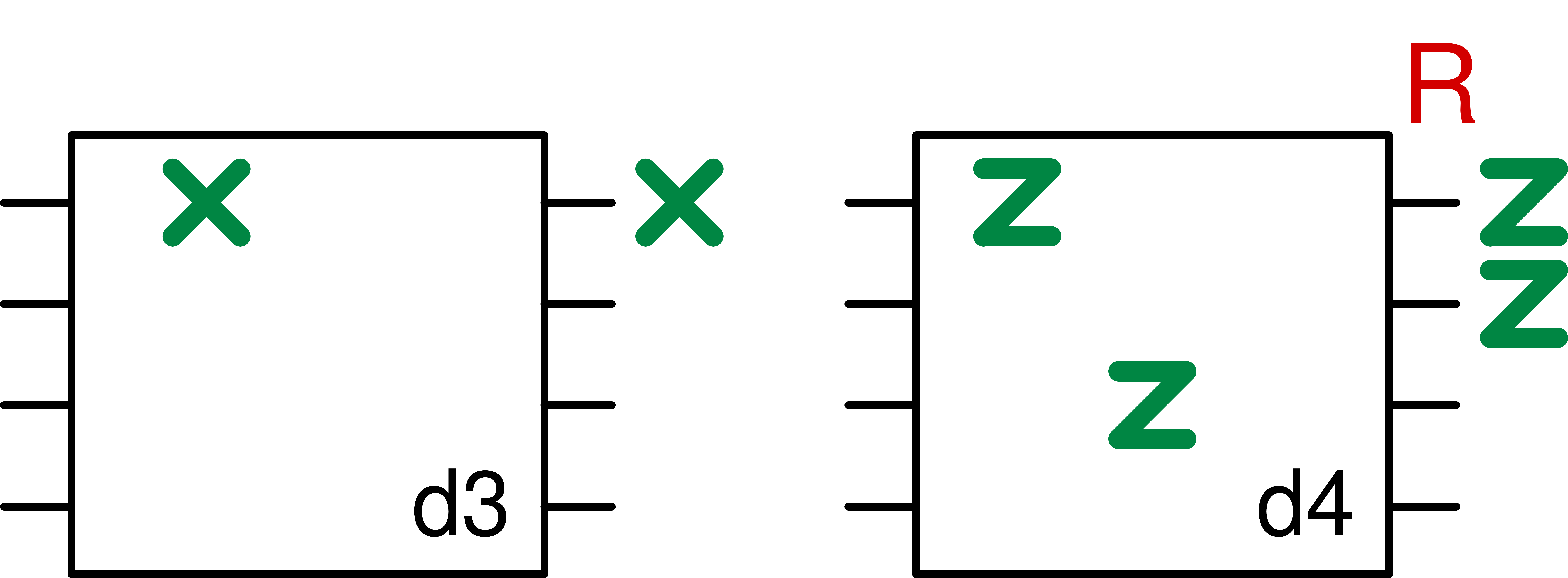}}
    \caption{(a)~A distance-$4$ stabilizer measurement circuit contains ancilla preparation, CNOTs, measurement and a recovery.  (b)~Rules for fault tolerance.  One fault should be corrected to an error of $X/Z$ weight at most one---this is sufficient for distance $3$.  Two faults should either be rejected (denoted by the red R) or result in an error of weight~two.}
    \label{f:smcrules}
\end{figure}

\begin{figure}
    \centering
    \subfloat[\label{f:wt4circ} ]{\includegraphics[width=.4\textwidth]{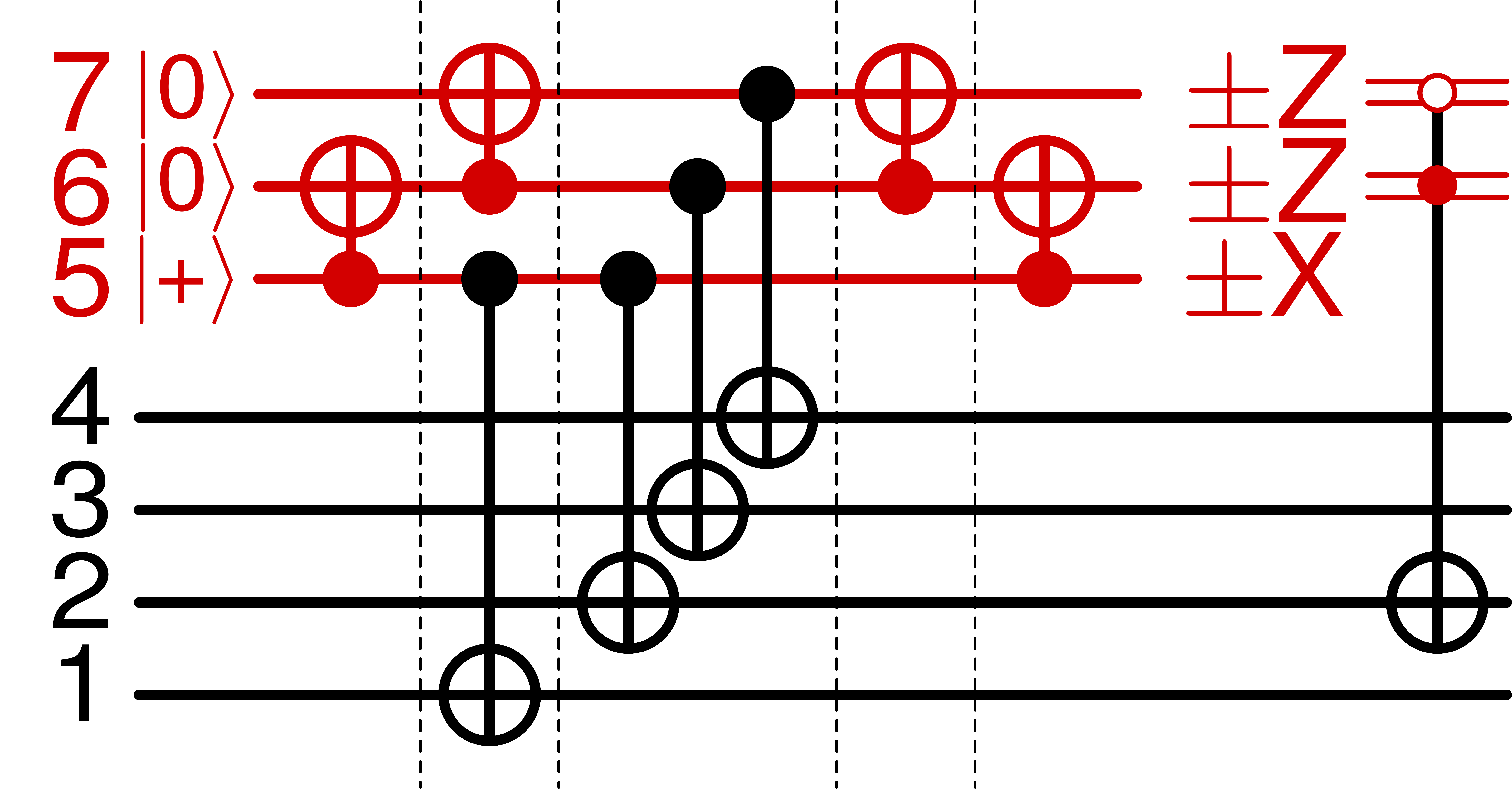}}
    \hspace{.8cm}
    \subfloat[\label{f:wt4layout} ]{\includegraphics[width=.51\textwidth]{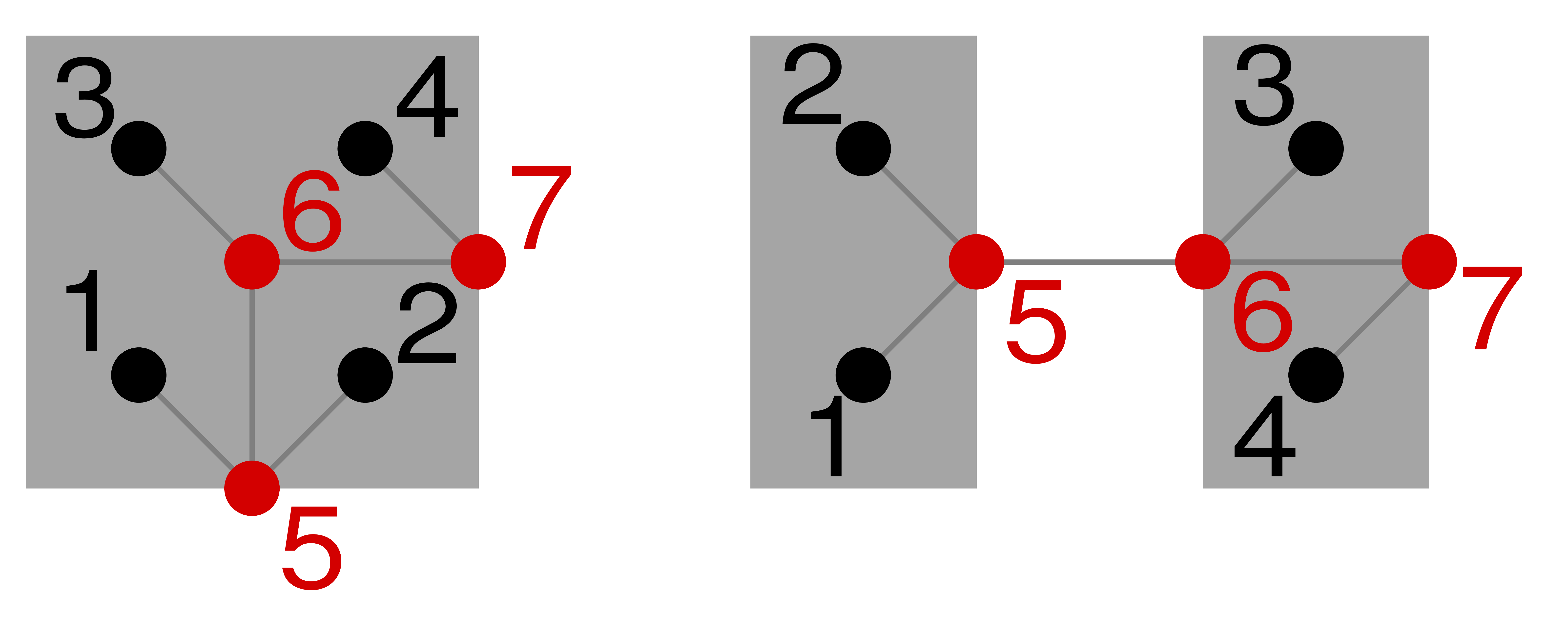}}
    \caption{(a) Circuit to measure a weight-four $X$ stabilizer fault-tolerantly to distance-four, satisfying the locality constraints in~(b).  The $\pm Z$ measurements are used to flag mid-circuit faults. Gates bunched together can be performed in parallel. (b)~Two layouts for measuring stabilizers in the sequences of~\figref{f:codes}.}
     \label{f:wt4circuit}
\end{figure}

\begin{figure}
    \centering
    \subfloat[\label{f:wt8circ} ]{\includegraphics[width=.7\textwidth]{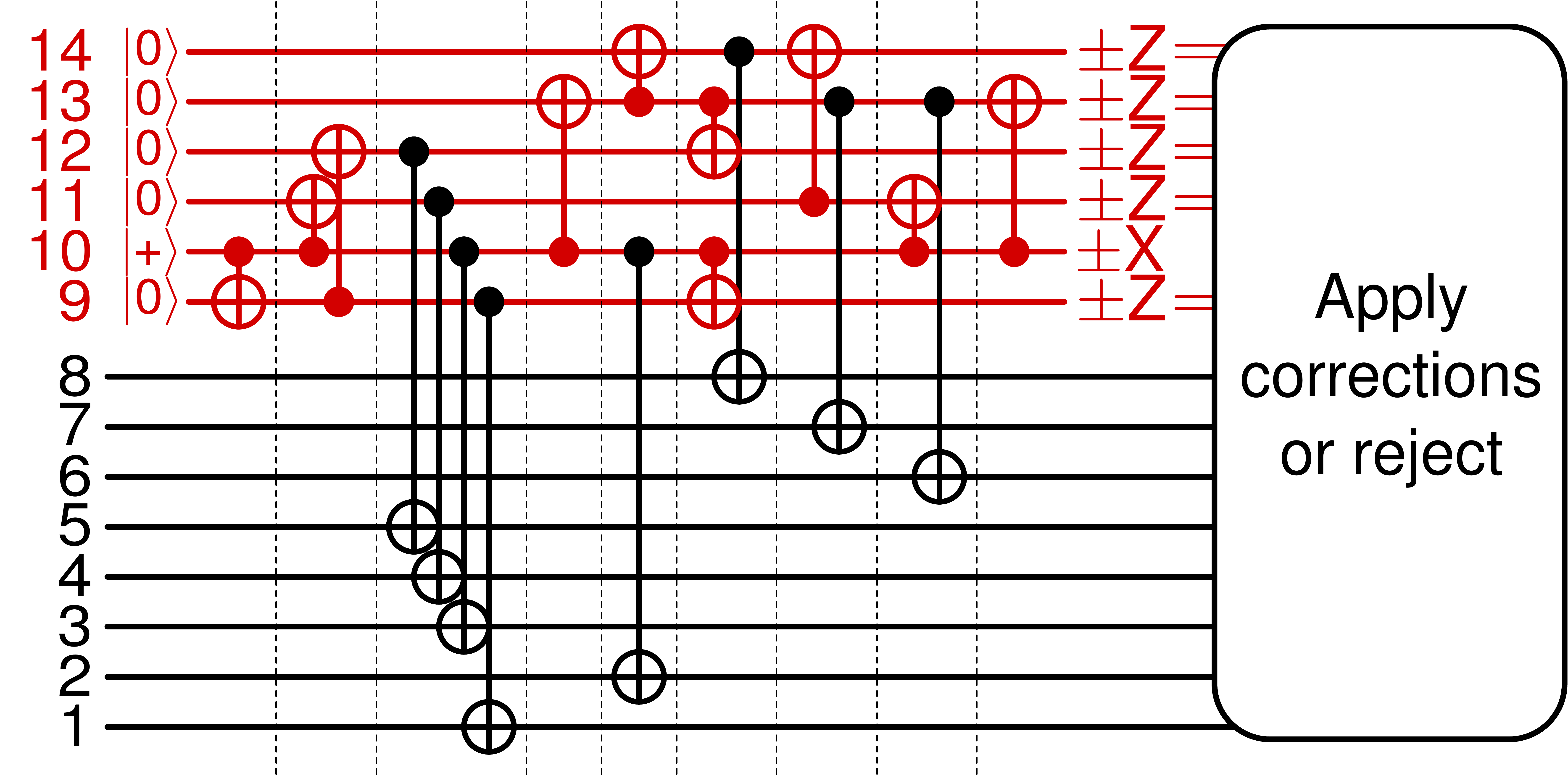}}
    \hspace{.8cm}
    \subfloat[\label{f:wt8layouta} ]{\includegraphics[width=.8\textwidth]{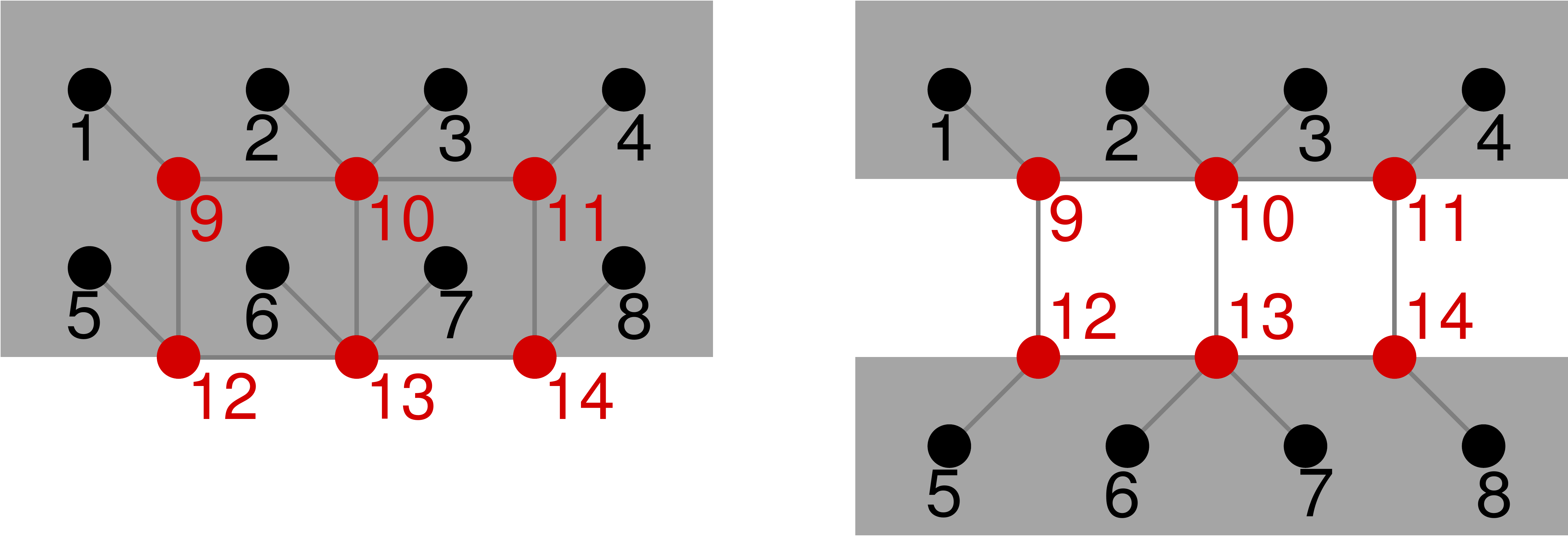}}
    \caption{(a) Circuit to measure a weight-eight $X$ stabilizer fault-tolerantly to distance-four, satisfying the locality constraints in~(b).  One fault is corrected to at most a weight-one error, but two or more faults may either be corrected, or detected resulting in rejection.  The resulting flag outcomes for corrections and rejection are tabulated in \appref{s:wt8corrs}.  (b)~Two layouts for measuring stabilizers in the sequences of~\figref{f:codes}.}
     \label{f:wt8circuit}
\end{figure}

We develop flag-based stabilizer measurement circuits.
For this, we use a randomized search algorithm constrained by the above fault tolerance rules and the geometric locality of~\figref{f:layout}.  With all six codes in this work, the stabilizers that are measured are of weight two, four and eight.  At the circuit level, the measurement of a weight-two stabilizer is automatically fault-tolerant (one fault causes an error of weight at most one).  
A weight-four stabilizer measured fault-tolerantly to distance three (i.e., one fault results in error of weight at most one) is automatically fault-tolerant to distance four, as two faults occurring in the circuit cannot create data errors with $X$ and $Z$ weight greater than two.  In~\figref{f:wt4circ}, the weight-four stabilizer measurement circuit applies a correction only for the $01$ ancilla measurement.
\figureref{f:wt8circ} shows a novel circuit to measure weight-eight stabilizers fault-tolerantly to distance four. This circuit uses different patterns of flag-qubit measurements to either correct an error, or reject---detecting an $O(p^2)$ fault event.  The flag patterns associated with corrections or with rejection have been tabulated in~\appref{s:wt8corrs}.  Figures~\ref{f:wt4layout} and~\ref{f:wt8layouta} show different ways of arranging the qubits to measure weight-four and weight-eight operators.

\medskip
\noindent
\textbf{Improvements.}
The benefit of measuring stabilizers individually is that error decoding is relatively simple. When stabilizers with overlapping support are measured in parallel, as in the surface code, more complicated decoding algorithms like minimum-weight perfect matching are required. However, we can still make small improvements for additional parallelism. In the $k=2$ and $k=4$ codes, only two of the corner weight-four stabilizers can be measured simultaneously, as each stabilizer requires three ancilla qubits. We conjecture that by sharing one ancilla qubit among all the corner stabilizers, it may be possible to fault-tolerantly measure all four of them using just nine ancilla qubits, like in Ref.~\cite{Reichardt18steane}.  Alternatively, in Steane-style syndrome extraction, subsets of stabilizers are measured in parallel using $n$-qubit resource states. Nine ancilla qubits are not sufficient for Steane's method, but Ref.~\cite{Huang21ShorSteane} shows that any subset of stabilizers can be jointly measured with specific resource states. If the fault-tolerant preparation of those resource states is possible on the nine-qubit ancilla sub-lattice of~\figref{f:layout}, it will be possible to develop faster and more efficient stabilizer measurement circuits.

The circuit of~\figref{f:wt8circ} uses six ancilla qubits for distance-four fault tolerance.  We found a distance-three fault-tolerant circuit using only four ancillas (qubits $9, 10,11,13$ in~\figref{f:wt8layouta}), which also requires fewer rounds of parallel gates. 
On the layout of~\figref{f:layout}, this may free up enough ancillas to measure two weight-eight stabilizers in each time step. The result is that more stabilizers can be measured faster and data qubits in an error correction block experience less idle noise.  Since these circuits are only fault-tolerant to distance three, a future avenue of research could use techniques in Ref.~\cite{ChaoReichardt17errorcorrection} to look at their performance in adaptive distance-four error correction. 

% \vspace{-0.3cm} %DEBUG
\subsection{Stabilizer measurement sequences}  \label{s:SMS}

% \noindent
% \textbf{Introduction.} 
The correction of errors in a quantum code requires a syndrome built from the measurement results of a sequence of stabilizers.
Since syndrome extraction is noisy, it is generally not sufficient to measure just a set of stabilizer generators, as shown in~\figref{f:smsintro1}.
Even one erroneous collected syndrome bit can result in an incorrect recovery, pushing the code into a state of logical error.  Instead, 
more stabilizers are redundantly measured to protect from quantum faults that cause syndrome bit flips.  The distance-$d$ surface code does this by measuring the stabilizer generators sequentially $\lceil \tfrac{d}{2} \rceil$ times, in a syndrome repetition code.  We may port this technique to the $k=2, 4$ and $6$ codes, but recent research has shown that these codes have very small stabilizer~\mbox{measurement sequences~\cite{delfosse2020short}.  }

\begin{figure}
\centering
\includegraphics[width =0.75\textwidth]{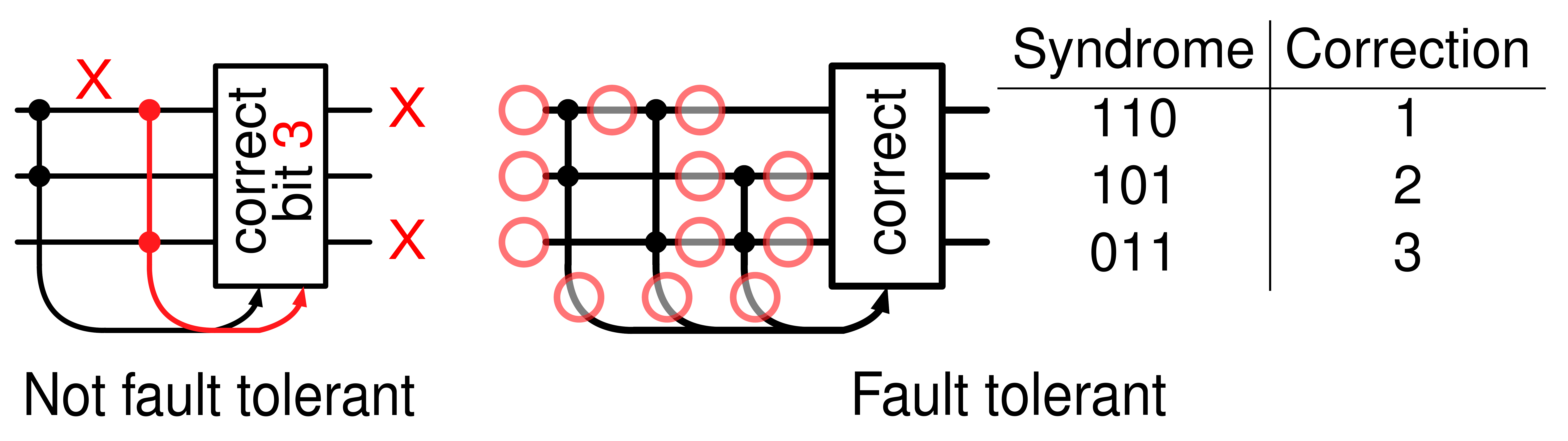}
\caption{Fault-tolerant error correction with the three bit repetition code $\{ 000, 111\}$ (adapted from~Fig. 1 of Ref.~\cite{delfosse2020short}).  It is not fault tolerant to correct errors based on the two parity measurements $1 \oplus 2$ and $1 \oplus 3$. An internal fault on bit $1$ can be mistaken for an input error on bit $3$, as they yield the same syndrome.  Errors can be corrected fault-tolerantly by adding another parity check, $2 \oplus 3$. Now for up to one fault at any of the circled locations, an input error is corrected, and an internal fault leaves an output error of weight $0$ or $1$.}
\label{f:smsintro1}
\end{figure}

%\pagebreak %DEBUG

The depth of a quantum circuit is generally calculated as the number of rounds of parallel two-qubit gates, since single-qubit gates are trivially short.  However, in current systems, the time needed for measurement dominates over the length of a CNOT
~\cite{Google2021, duke21BScode, honeywell21steane, riste20bitflipEC}.
Hence the focus shifts from minimizing CNOT depth to reducing the rounds of measurements needed for error correction.  We therefore denote by ``time step'' the time needed to measure a set of stabilizers in parallel, as shown in~\figref{f:smsintro}.
In addition to finding short fault-tolerant sequences of stabilizers, we carefully parallelize their measurement circuits to~\mbox{further speed up error correction.}

\begin{figure}
    \centering
    \subfloat[\label{f:smsintro} ]{\includegraphics[width=.55\textwidth]{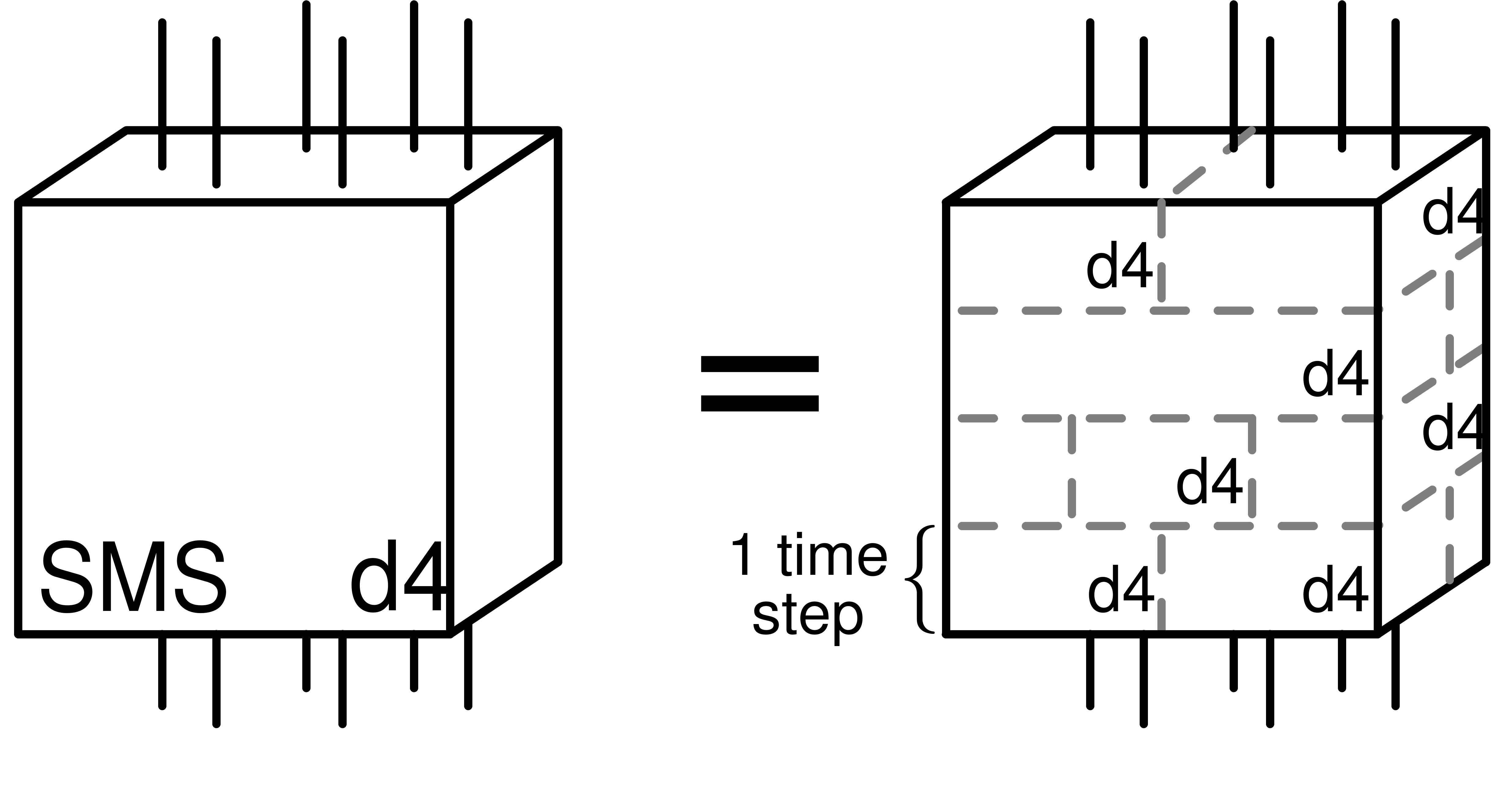}}
    \hspace{1cm}
    \subfloat[\label{f:smsrules} ]{\includegraphics[width=.85\textwidth]{imagesChap5/smsrules.pdf}}
    \caption{(a) A stabilizer measurement sequence (SMS) consists of multiple time steps of parallel stabilizer measurement circuits, where the end of a time step denotes the simultaneous measurement of all the ancilla qubits. (b) Rules for distance-$4$ fault tolerance---first two are sufficient for distance $3$: ($i$) An input $1$-qubit error must be corrected. ($ii$) $1$ internal fault must be corrected to an error of weight at most $1$. ($iii$) A $2$-qubit input error is rejected. ($iv$) $1$ input error and $1$ internal fault should be corrected to an error of weight at most $1$ or rejected. (v) $2$ internal faults must be rejected or propagate to an error of weight at most $2$.}
     \label{f:smspics}
\end{figure}

% \begin{figure}
%     \begin{subfigure}{0.45\linewidth}
%         \includegraphics[width=0.99\textwidth]{imagesChap5/smsintro.pdf}
%         \caption{\label{f:smsintro}}
%     \end{subfigure}
%     \hspace{8mm}
%     \begin{subfigure}{0.45\linewidth}
%         \includegraphics[width=0.99\textwidth]{imagesChap5/smsrules.pdf}
%         \caption{\label{f:smsrules}}
%     \end{subfigure}
% \caption{(a) A stabilizer measurement sequence (SMS) consists of multiple time steps of parallel stabilizer measurement circuits, where the end of a time step denotes the simultaneous measurement of all the ancilla qubits. (b) Rules for distance-$4$ fault tolerance---first two are sufficient for distance $3$: ($i$) An input $1$-qubit error must be corrected. ($ii$) $1$ internal fault must be corrected to an error of weight at most $1$. ($iii$) A $2$-qubit input error is rejected. ($iv$) $1$ input error and $1$ internal fault should be corrected to an error of weight at most $1$ or rejected. (v) $2$ internal faults must be rejected or propagate to an error of weight at most $2$.}
% \end{figure}

% \pagebreak %DEBUG
\medskip
\noindent
\textbf{Fault tolerance rules.}
We follow the `exRec' formalism of Ref.~\cite{AliferisGottesmanPreskill05} to determine rules for fault-tolerant error correction, as shown in \figref{f:smsrules}.  For distance-three fault tolerance, only two rules are needed.  If the input to an error correction block has a weight-one error and there are no internal faults, the syndrome must be sufficient to correct back to the codespace. This is actually the basic rule for an ideal error correction block.  If there are no input errors and one internal fault occurs, the weight of the output error after recovery should be at most one.

For distance-four fault tolerance, we must consider the effect of up to two input errors or internal faults. If the input error has weight two and there are no internal faults, then the stabilizer measurement sequence must detect the error and restart the computation.  
If there is a weight-one input error and an internal fault, either the computation is restarted, or the output of the error correction block must have error of $X$ and $Z$ weight at most one. 
Finally, if two internal faults occur with no input error, either the computation is restarted, or the output error must have weight at most two. 

(In the last four rules, the resulting syndrome must never be equivalent to a weight-one input error on a different qubit. This ensures that every weight-one input error can be reliably corrected.)

\medskip
\noindent
\textbf{Solutions.} To perform distance-four fault-tolerant error correction on the layout of~\figref{f:layout}, we consider measuring stabilizers only of the form given in~\secref{s:SMC}. The goal is then to devise short parallel stabilizer measurement sequences occupying the fewest time steps while satisfying the fault tolerance rules above.  For each of the newly proposed codes, the sequences in~\figref{f:codes} were found using randomized search and subsequent minor alterations.

The $k = 2$ code measures ten $X$ (or $Z$) stabilizers over five time steps, hence recovery occurs every ten time steps. On the other hand, the $k = 6$ code contains no stabilizers of weight less than eight, so parallelism is difficult. Here, seven $X$ (or $Z$) stabilizers are measured over seven time steps, for a total $14$ time steps between recoveries. The distance-four surface code measures all its stabilizer generators in two time steps. The first is used to measure the nine weight-four stabilizer generators using the nine ancilla qubits, and the second time step is used to measure the boundary weight-two stabilizers. Recovery occurs after two fresh syndrome layers are measured, at a frequency of four time steps. 

% Also talk about the heuristics for how this was developed.

\section{Simulation results}
\label{sec:d4simresults}

\noindent
\textbf{Noise model.} For simulation, we consider independent circuit-level noise, as described below:

\begin{itemize}[leftmargin=*]
\item With probability $p$, the preparation of $\ket 0$ is replaced by $\ket 1$ and vice versa---similarly $\ket +$ and $\ket -$.
\item With probability $p$, $\pm X$ or $\pm Z$ measurement on any qubit has its outcome flipped.
\item With probability $p$, a one-qubit gate is followed by a random Pauli error drawn uniformly from $\{ X, Y, Z\}$.
\item With probability $p$, the two-qubit CNOT gate is followed by a random two-qubit Pauli error drawn uniformly from $\{ I, X, Y, Z\}^{\otimes 2} \setminus \{I \otimes I \}$.
\item After each time step, with probability $p(1+m/10)$, each data qubit is acted upon by a random one-qubit Pauli error drawn uniformly from $\{X, Y, Z\}$.  (A time step denotes one round of parallel stabilizer measurements of maximum CNOT depth $m$, as in~\secref{s:SMS}.)
\end{itemize}

The rest error rate models the observed performance of current-day quantum systems, where the time taken to measure an ancilla qubit is long compared to the CNOT gate time.  We model the rest error rate during measurement as $p$, and during CNOT gates as $p/10$.  Even with dynamical decoupling~\cite{ViolaKnillLloyd99DD}, the error incurred by the idle data qubits can be quite high.  

\medskip
\noindent
\textbf{Normalized logical error rate.} 
The logical error rate of fault-tolerant storage can be estimated by checking for a logical error after each block of error correction.
However, different codes correct errors at different frequencies---once every four time steps for the surface code, but fourteen for the $k=6$ code. 
To compare the codes on a similar time scale, we normalize the logical error rates~\mbox{with respect to time step.}

We plot the logical error rate per time step in~\figref{f:scaling}, where we show that a distance-four surface code has a storage error rate of $O(10^{-9})$, for a CNOT gate error rate of just $10^{-4}$.  Even with the infidelity of current day CNOTs, $\sim 10^{-3}$, we show logical error rates approaching $10^{-6}$.  These results demonstrate the benefits of postselection.

\begin{figure}
\centering
\includegraphics[width =0.75\textwidth]{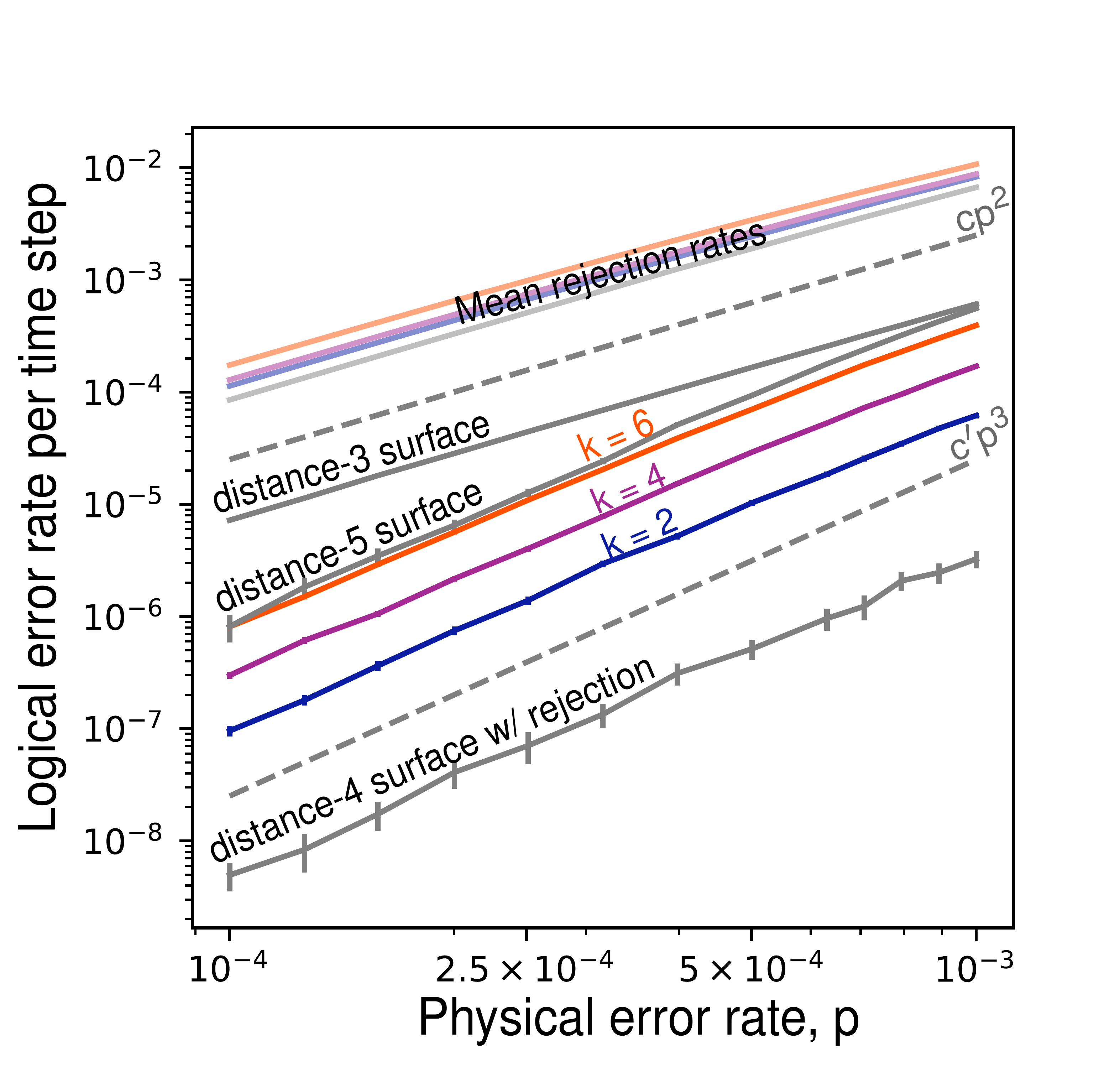}
\caption{$O(p^3)$ scaling of $X$ logical error rate and $O(p^2)$ scaling of rejection rate, with error bars, for the distance-four codes. The distance-three and distance-five surface codes are shown for comparison. The new codes have logical error rate per time step as low as $1/10$th the distance-five surface code.
The distance-four surface code is as low as $1/100$.}
\label{f:scaling}
\end{figure}

\medskip
\noindent
\textbf{Cumulative logical error probability.} 
The mean rejection rates in~\figref{f:scaling} provide a good comparison of how often the different codes reject, but do not accurately describe behavior for bounded-length computation.  Here, a more useful metric is the probability of acceptance, $P_a(t)$, which is how often a $t$-time-step computation completes.  This quantity can be estimated empirically by simulating the application of noisy error correction to an initial state for bounded time, which we denote as a simulation `run'.  If $R$ is the total number of executed runs and $R_{a}(t)$ is the number of runs that have not rejected until time step $t$, 
\begin{equation}
P_a(t) = \frac{R_{a}(t)}{R}  \, .
\end{equation}

Similar to the rejection rate, the logical error rate per time step is indicative of the frequency of logical errors, but does not help to understand the drawbacks of postselection. We again refer to a cumulative metric, the probability of a logical error after $t$ time steps of error correction, empirically given by
\begin{equation}
P_{L}(t)= \frac{R_{L}(t)}{R}  \, ,
\end{equation}
where $R_{L}(t)$ is the number of runs in a state of logical error at time step $t$.  For even distance codes, one must instead look at the probability of logical error conditioned on acceptance, which is calculated as
\begin{equation}
P_{L | a}(t) =  \frac{P_{L}(t)}{P_{a}(t)}  = \frac{R_{L}(t)}{R_{a}(t)}  =  \frac{R_{L}(t)}{R \, P_{a}(t)} \, .
\end{equation}
For even distance codes with postselection, $P_a(t) < 1$ and so $P_{L | a}(t) > P_{L}(t)$.  For odd distance codes that do not reject, $P_{L | a}(t) = P_{L}(t)$.

The above formula holds only for a single code patch.  The probability of logical error while using multiple patches can be upper bounded from the data for a single patch as 
\begin{equation}
P_{L | a}(t,c) \leq \frac{c \, R_{L}(t)} {R \, P_a^c(t)} \, ,
\end{equation}
where $c$ is the number of code patches used. Note that the number of logically incorrect runs grows linearly with the number of patches, but the probability of acceptance of multiple patches is the probability that every patch has accepted. 

\medskip
\noindent
\textbf{Discussion.} 
We simulated fault-tolerant error correction of the codes in~\figref{f:codes} for up to 12000 time steps at error rate $p \in \{ 0.001, 0.0005, 0.00025,$ $0.0001\}$. Using the empirical formulae above, we then plotted in~\figref{f:PLEt} the probability of $X$ logical error conditioned on acceptance and the probability of acceptance for one, two, six or twelve logical qubits.  Note that some plots look discontinuous. This is because we only check for logical errors and reject at the end of an error correction block.  Above the graphs of~\figref{f:PLEt}, we compare the number of physical qubits required for each code.

\begin{table}
\caption{\label{f:k1p0_001}
Error correction for $1$ logical qubit at $p = 0.001$. The probability of logical error and acceptance are shown for $80$ and $200$ time steps.  Each code uses one patch of qubits.  The distance-four surface code has the lowest logical error probability for short computations.  
}
%\setlength{\tabcolsep}{6.5pt}
%\begin{tabular}{cccccc}
%\hline \hline
%& &\multicolumn{2}{c}{$t = 80$} &\multicolumn{2}{c}{$t = 200$}\\ 
%Code & Qubits & $P_{L|a}$ & $P_a$ & $P_{L|a}$ & $P_a$ \\
%\hline
%$k = 2$ & $25$ & $.0037(2)$ & $53.9\%$ & $.0099(6)$ & $19.6\%$ \\
%$k = 4$& $25$ & $.0114(4)$ & $50.5\%$ & $.0283(11)$ & $17.2 \%$\\
%$k = 6$&  $25$ & $.0275(7)$ & $44.1\%$ & $.0687(21)$ & $11.1 \%$\\[.25 cm]
%$d = 4$ & $25$ & $.0001(1)$ & $59.4\%$  & $.0003(1)$ & $25.9 \%$\\
%$d = 3$ & $13$ & $.0242(6)$ & $100\%$ & $.0581(8)$ &$100 \%$\\
%$d = 5$ & $41$ & $.0062(3)$ & $100\%$ & $.0135(4)$ & $100 \%$\\
%\hline \hline
%\end{tabular}
% \vspace{-0.9mm}
% \hspace{-2.1mm}
\centering
\includegraphics[width=0.7\textwidth]{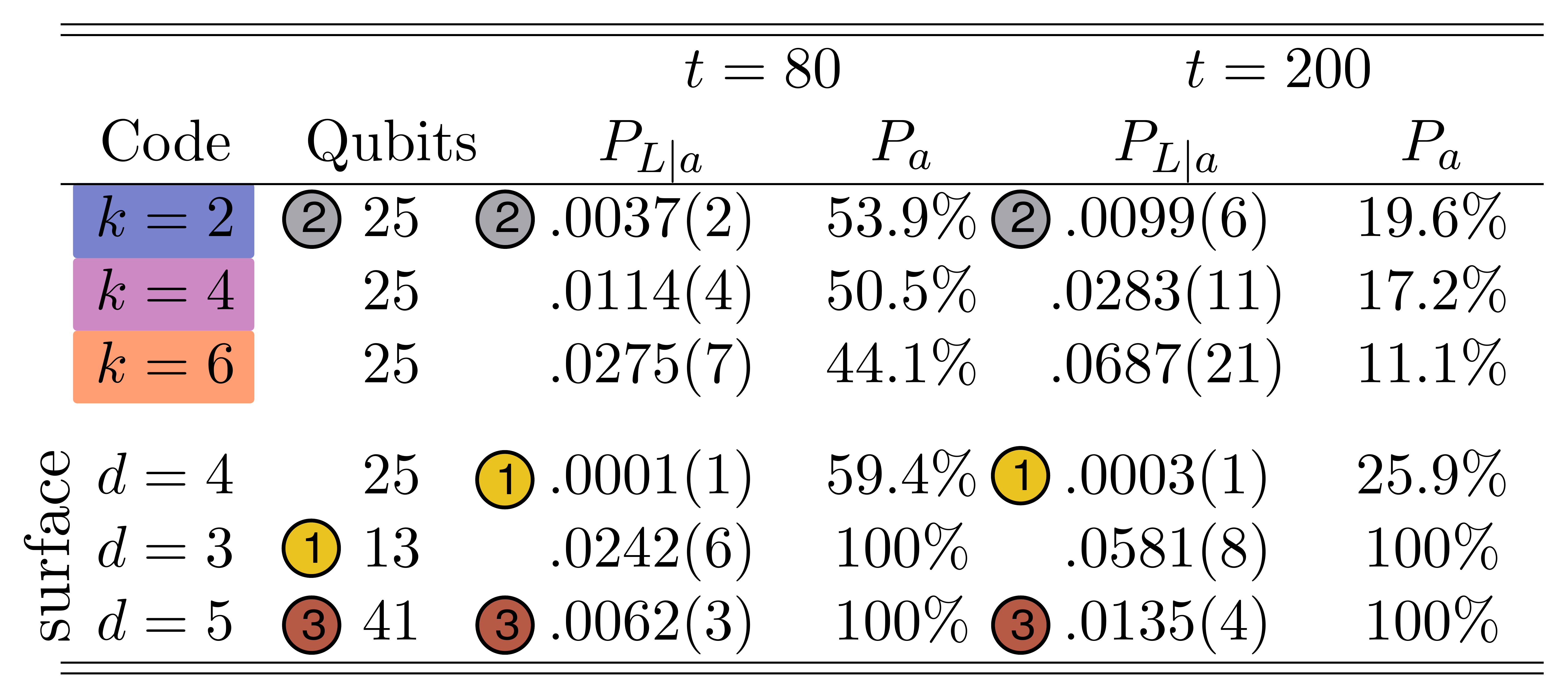}
% \vspace{-0.5cm} %DEBUG
\end{table}

There are many things to be learned from~\figref{f:PLEt}. To start, the first column of graphs shows how a single patch of each code fares against the others, for different error rates.  The $d = 4$ surface code with rejection boasts the lowest logical error probability overall and has the highest acceptance rates among all the even-distance codes.  The logical error probability of the $k = 2$ code actually matches the distance-five surface code, even though it encodes twice as much information.  This is also apparent from~\tabref{f:k1p0_001}, where we show the probability of acceptance and logical error for one logical qubit at $p = 10^{-3}$.

\begin{figure}
\centering
\includegraphics[width=.94\textwidth]{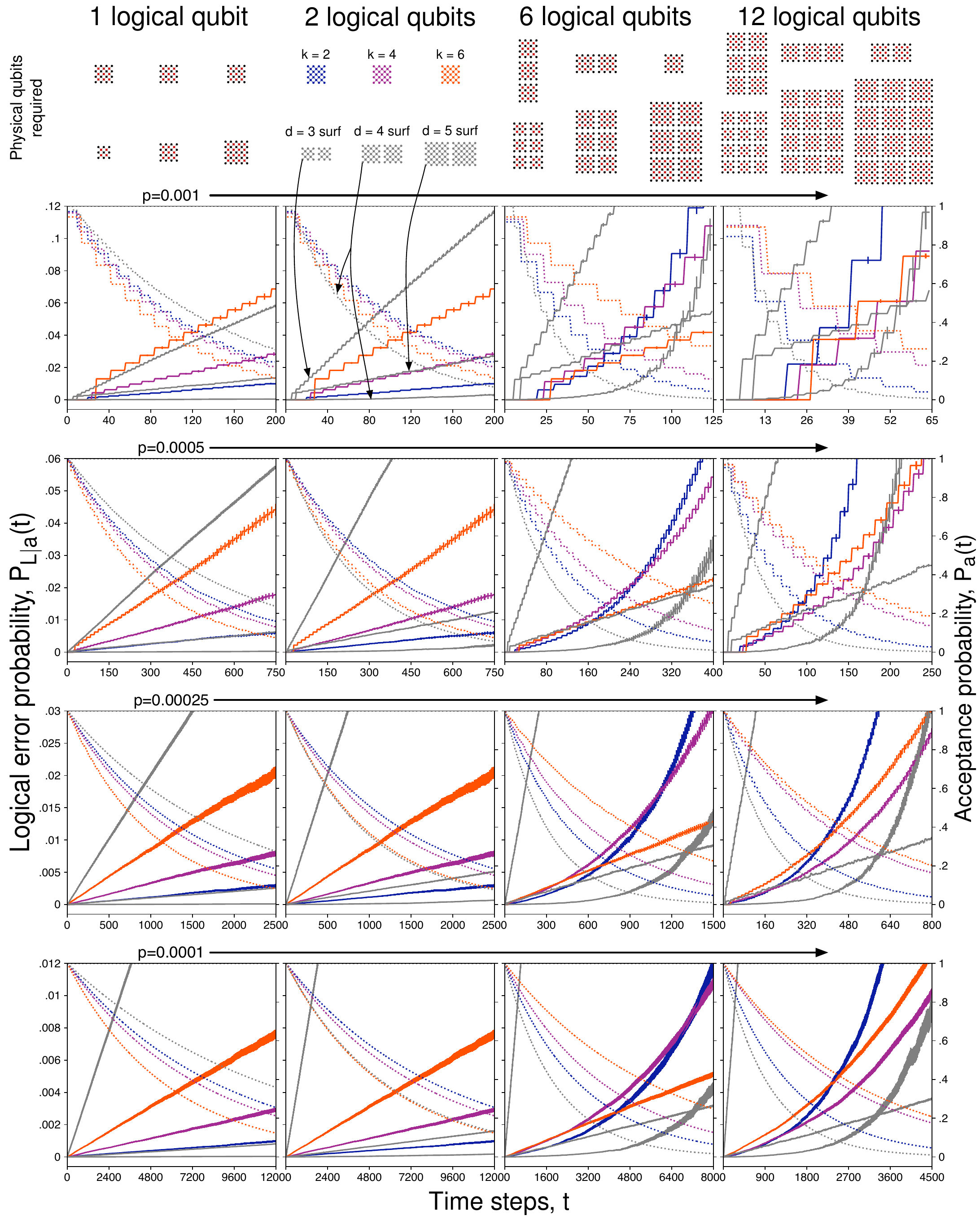}
\vspace{-.1cm} %DEBUG
\caption{Probability of $X$ logical error (solid) and acceptance (dotted) for $t$ time steps of error correction on six codes, as a function of physical error rate (row) and desired logical qubits (column).  
The three colored curves correspond to the $k = 2$, $ k = 4$ and $k = 6$ codes and the three gray curves are the surface codes.  The graphs for few time steps look like step functions because the code patches are checked for logical errors only after blocks of error correction, not time steps.  The top row compares the number of physical qubits required to achieve the desired number of logical qubits.}
\label{f:PLEt}
\end{figure}

\begin{table}
\caption{\label{f:k2p0_0005}
Error correction for $2$ logical qubits at $p = 0.0005$.  The probability of logical error and acceptance are shown~for $300$ and $750$ time steps.  The surface codes require more than one patch of physical qubits.  Among the new codes, the $k = 2$ color code has few large stabilizers and a fast sequence. These advantages help it achieve the lowest logical error probability at the highest acceptance rates.}
%\setlength{\tabcolsep}{6.5pt}
%\begin{tabular}{cccccc}
%\hline \hline
%& &\multicolumn{2}{c}{$t = 300$} &\multicolumn{2}{c}{$t = 750$}\\ 
%Code & Qubits & $P_{L|a}$ & $P_a$ & $P_{L|a}$ & $P_a$ \\
%\hline
%$k = 2$ & $25$ & $.0025(2)$ & $48.5\%$ & $.0060(5)$ & $16\%$ \\
%$k = 4$& $25$ & $.0066(3)$ & $45.4\%$ & $.0178(9)$ & $13.1 \%$\\
%$k = 6$&  $25$ & $.0172(6)$ & $35.7\%$ & $.0441(18)$ & $7.6 \%$\\[.25 cm]
%$d = 4$ & $50$ & $.0003(1)$ & $31.9\%$  & $.0021(4)$ & $5.6 \%$\\
%$d = 3$ & $26$ & $.0477(5)$ & $100\%$ & $.1145(8)$ &$100 \%$\\
%$d = 5$ & $82$ & $.0053(2)$ & $100\%$ & $.0125(3)$ & $100 \%$\\
%\hline \hline 
%\end{tabular}
% \vspace{-0.9mm}
% \hspace{-2.1mm}
\centering
\includegraphics[width=0.7\textwidth]{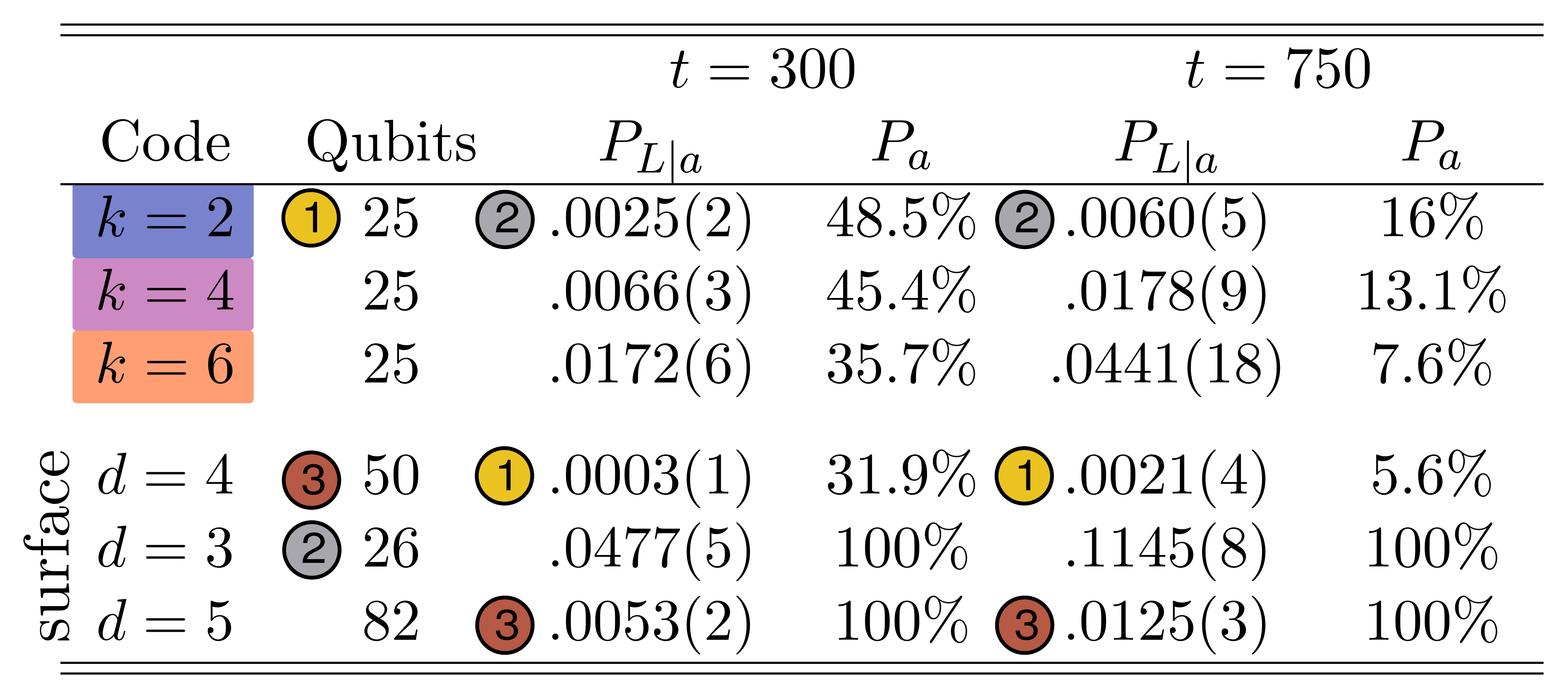}
% \vspace{-0.6cm}  %DEBUG
\end{table}

For two logical qubits, (second column of graphs), the surface codes need two patches of qubits, hence the probability of logical error doubles and the acceptance is squared. The distance-four surface code now has the lowest acceptance probability among the distance-four codes.  We keep the range of time steps consistent between the first and second columns to show that the curves for the multi-qubit codes are unchanged.  As shown in~\tabref{f:k2p0_0005}, for two logical qubits, the $k = 2$ code halves the logical error probability of the $d = 5$ surface code, using fewer than one-third as many physical qubits. 

Going further, we have analyzed error correction for six and twelve encoded qubits, as shown in the last two columns.
With only one-tenth of the physical overhead, a single $k = 6$ code patch rivals the performance of six patches of the $d = 5$ surface code.  The single patch of the $k=6$ code even outperforms the $k = 2$ and $k = 4$ codes, but this is precisely because only one code patch is used. When multiple code patches are used for many logical qubits, the acceptance rate of the distance-four codes drops exponentially.  This is also observed in~\tabref{f:k6p0_00025} and~\tabref{f:k12p0_0001}, as the acceptance probability of the distance-four surface code quickly approaches zero.

In the last column of graphs, we compare statistics for twelve logical qubits. Although current NISQ systems only protect one logical qubit, our results show that just $50$--$75$ physical qubits are sufficient for twelve logical qubits. 
In this regime, the $k = 4$ code achieves lower logical error probability than the $k=6$ code with only $50\%$ more overhead. Unfortunately at longer time scales, postselection sharply increases the logical error probability, rendering the distance-four codes much less useful.

All simulations in this paper, developed in Python, were executed on the USC Center for Advanced Research Computing (CARC) high-performance computing cluster. The simulations used over one million minutes of CPU core time on Intel Xeon processors operating at $2.4$ GHz.

%\medskip
\noindent
\textbf{Take-home message.} Postselection can play a crucial role in reducing logical error rates.  
However, when logical information is stored for too long, it is likely to be wiped and reset. 
This is okay for some algorithms: applications with low depth, like variational algorithms~\cite{cerezo2020variational, bharti2021noisy}, or those that are designed with rejection, like magic state distillation~\cite{Bravyi05}.
If only one or two qubits are required, the distance-four surface code and the $k=2$ code offer very low probability of logical error. For more qubits, we advise using the $k=4$ or $k=6$ codes, as they use far fewer physical resources to achieve competitively low logical error.  
We show that $50$--$75$ good physical qubits are sufficient to correct errors on twelve logical qubits. Even at a CNOT error rate as high as $5\times 10^{-4}$, error correction up to $100$ time steps can be run with error probability as low as $1\%$.

\begin{table}
\caption{\label{f:k6p0_00025}
Error correction for $6$ logical qubits  at $p = 0.00025$.  The probability of logical error and acceptance are shown~for $700$ and $1500$ time steps.  The $k = 6$ code requires one-tenth the physical qubits as the distance-$5$ surface code, while nearly matching the logical error probability.}
%\setlength{\tabcolsep}{6.5pt}
%\begin{tabular}{cccccc}
%\hline \hline
%& &\multicolumn{2}{c}{$t = 700$} &\multicolumn{2}{c}{$t = 1500$}\\ 
%Code & Qubits & $P_{L|a}$ & $P_a$ & $P_{L|a}$ & $P_a$ \\
%\hline
%$k = 2$ & $75$ & $.0061(3)$ & $24.5\%$ & $.0412(17)$ & $4.8\%$ \\
%$k = 4$& $50$ & $.0078(3)$ & $34.7\%$ & $.0306(10)$ & $10.4 \%$\\
%$k = 6$&  $25$ & $.0060(2)$ & $50.1\%$ & $.0129(5)$ & $22.3 \%$\\[.25 cm]
%$d = 4$ & $150$ & $.0001(1)$ & $11.5\%$  & $.0137(12)$ & $1 \%$\\
%$d = 3$ & $78$ & $.0847(5)$ & $100\%$ & $.1796(8)$ &$100 \%$\\
%$d = 5$ & $246$ & $.0043(1)$ & $100\%$ & $.0092(2)$ & $100 \%$\\
%\hline \hline 
%\end{tabular}
% \vspace{-0.9mm}
% \hspace{-2.1mm}
\centering
\includegraphics[width=0.7\textwidth]{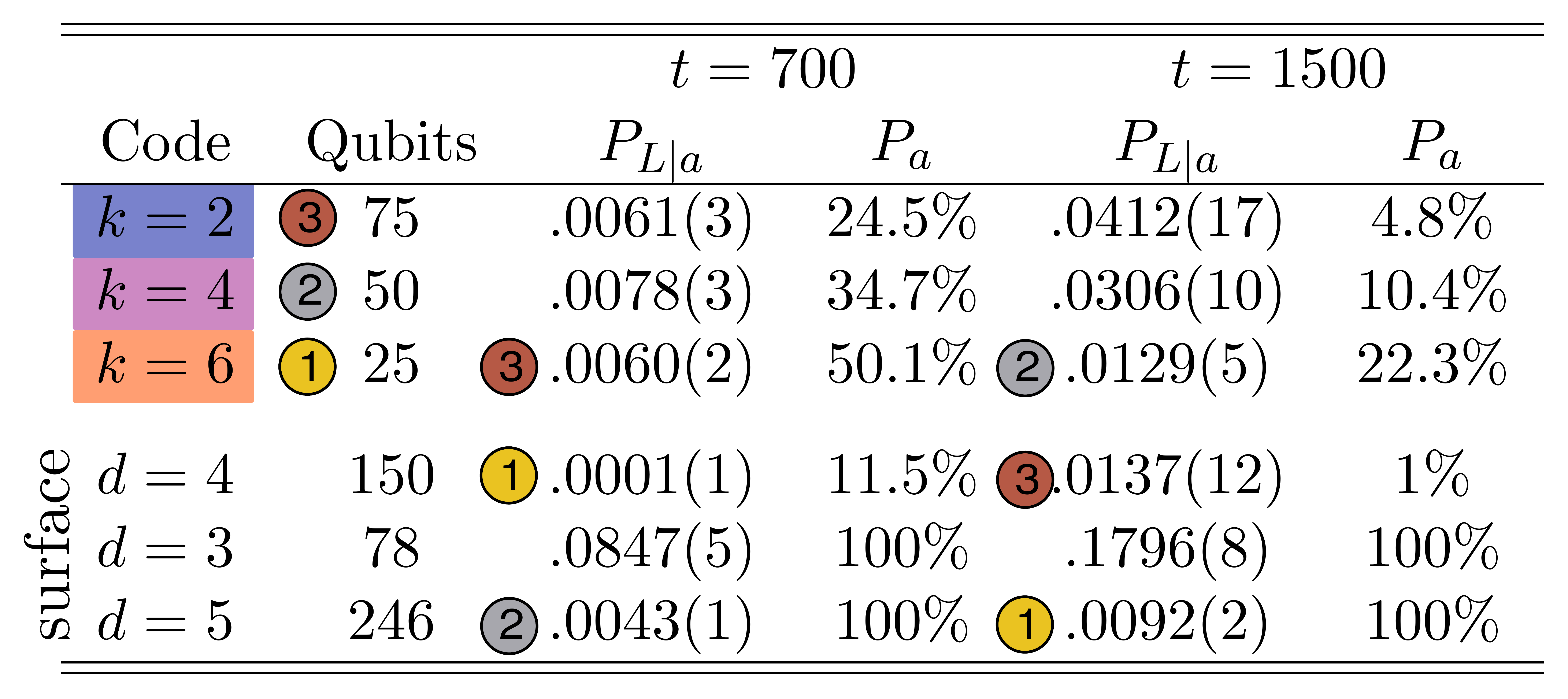}
% \vspace{-0.6cm} %DEBUG
\end{table}

\begin{table}
\caption{\label{f:k12p0_0001}Error correction for $12$ logical qubits at $p=0.0001$.  The probability of logical error and acceptance are shown~for $1800$ and $4500$ time steps.  The $k=4$ code is well-balanced, achieving competitive logical error rates with low~\mbox{qubit overhead.}}
%\setlength{\tabcolsep}{6.5pt}
%\begin{tabular}{cccccc}
%\hline \hline
%& &\multicolumn{2}{c}{$t = 1800$} &\multicolumn{2}{c}{$t = 4500$}\\ 
%Code & Qubits & $P_{L|a}$ & $P_a$ & $P_{L|a}$ & $P_a$ \\
%\hline
%$k = 2$ & $150$ & $.0025(1)$ & $29.3\%$ & $.0285(9)$ & $4.6\%$ \\
%$k = 4$& $75$ & $.0022(1)$ & $50\%$ & $.0101(3)$ & $17.5 \%$\\
%$k = 6$&  $50$ & $.0032(1)$ & $53.4\%$ & $.0127(3)$ & $20.8 \%$\\[.25 cm]
%$d = 4$ & $300$ & $.0003(1)$ & $15.8\%$  & $.0094(7)$ & $1 \%$\\
%$d = 3$ & $156$ & $.0708(2)$ & $100\%$ & $.1761(4)$ &$100 \%$\\
%$d = 5$ & $492$ & $.0013(1)$ & $100\%$ & $.0036(1)$ & $100 \%$\\
%\hline \hline 
%\end{tabular}
% \vspace{-0.9mm}
% \hspace{-2.1mm}
\centering
\includegraphics[width=0.7\textwidth]{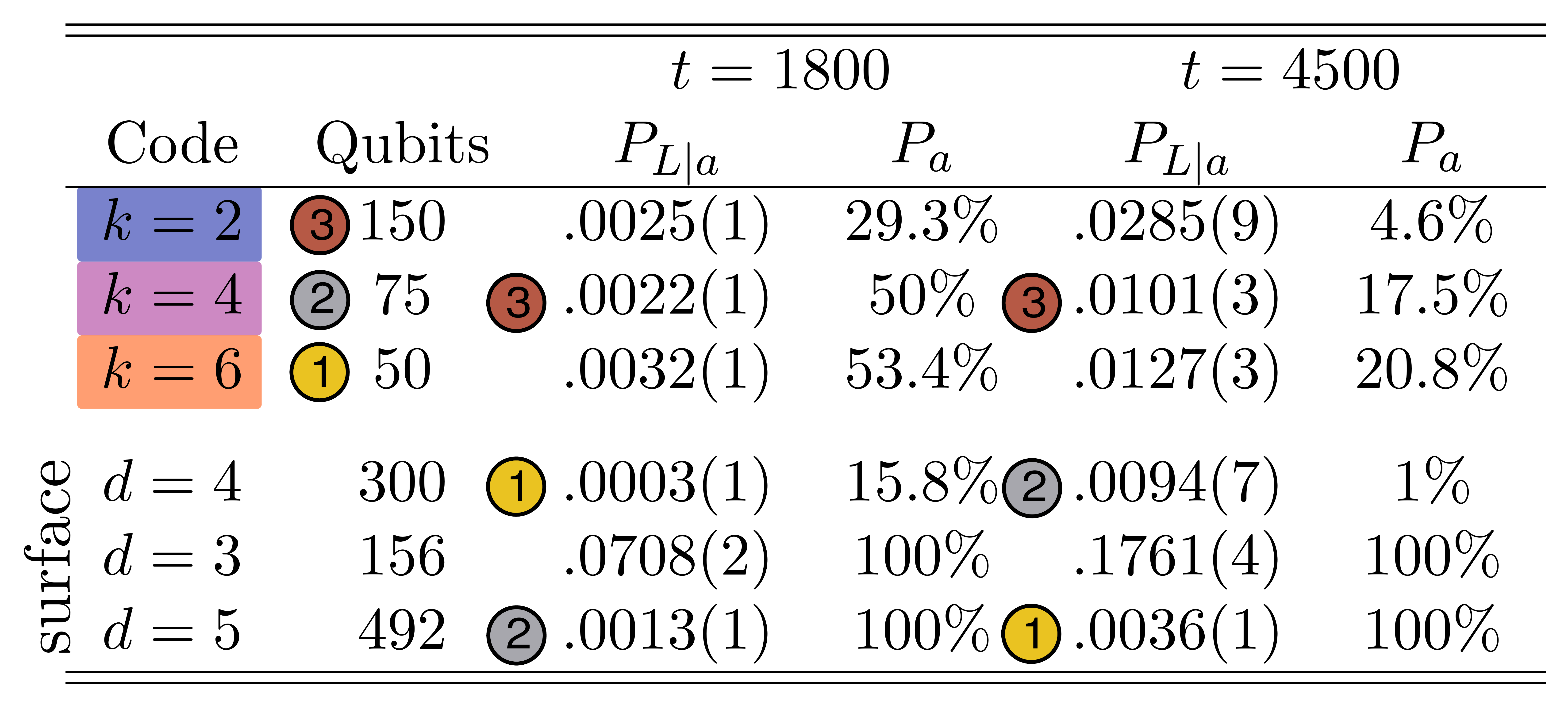}
%\vspace{-0.8cm} % DEBUG
\end{table}

\subsection{Comparing memory against unencoded qubits}
\label{s:unenc} 

In this section, we determine whether information protected by fault-tolerant error correction is more reliable than being stored in an unprotected qubit.
For the postselection codes, \figref{f:unencoded} plots the CNOT depth at which qubits have accumulated $1\%$ probability of logical error.
The unprotected qubit is modeled to accumulate errors only through rest noise, whereas logical errors in the encoded qubits are due to circuits for fault-tolerant error correction.  For the error rates we consider ($\leq 10^{-3}$), it is clear that the encoded qubits are better preserved for much longer than an unprotected qubit.

% \parbox[c]{1cm}{} %DEBUG

%\medskip
% \vspace{6cm} %DEBUG

\begin{figure}
% \hspace{-1cm}
\centering
\includegraphics[width=.65\textwidth]{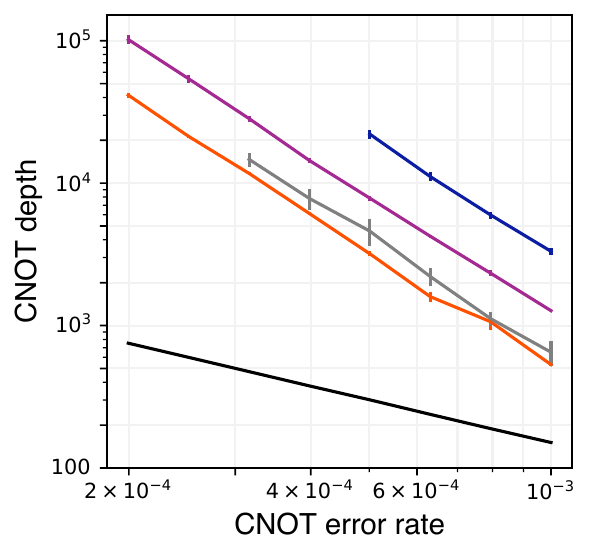}
\caption{CNOT depth at which each code has accumulated $1\%$ probability of $X$ logical error. In black, the depth is plotted for one unencoded qubit, at rest error rate one-tenth the CNOT error rate. Plots are shown for the $ k = 2$ code in blue, $k = 4$ code in purple, $k = 6$ code in orange and the $d=4$ surface code in grey. We assume the depth of ancilla qubit measurement is ten times the depth of a CNOT gate. The CNOT depth shown for the surface code is for $0.01\%$ probability of $X$ logical error.  All data points shown for the postselection-equipped codes have acceptance $> 5 \%$.}
\label{f:unencoded}
% \vspace{-0.3cm}
\end{figure}

\section{Potential future work}
\label{sec:d4future}

In this paper, we show how to perform fault-tolerant storage with $16$-qubit codes. There are two immediate roadblocks en route to universal fault-tolerant quantum computation. Currently, no devices exist with the layout of~\figref{f:layout}, so until they are fabricated, we turn to other layout improvements.
The middle ancilla qubit in~\figref{f:layout} is connected to eight neighboring qubits.  
However, careful analysis and modification of the stabilizer measurement routines may yield solutions that only require maximum qubit degree five or six.
This may not be interesting for densely-connected ion trap quantum computers, but is necessary in superconducting architectures to maintain low cross-talk.  Alternatively, if we are allowed extra ancilla qubits, we show that maximum degree four is possible, as in the Google Sycamore lattice of~\figref{f:degree4k2}.  The stabilizer measurement circuits are all fault-tolerant to distance four, but since all the stabilizer generators are measured simultaneously, error decoding will require new strategies.  The weight-four stabilizers can be measured using the circuit in~\figref{f:wt4circ}, but the weight-eight stabilizer requires a new circuit, as we detail in~\appref{s:wt8corrs}.  In~\figref{f:degree4k2}, the only qubits with degree-four connectivity are the ancillas used for measuring the weight-eight stabilizer. For systems with high crosstalk, qubits of degree three may be sufficient to correct errors on the $k=2$ subsystem code, since errors can be corrected by measuring only weight-four~operators. 

Conversely, we may consider subsystem codes to simplify the measurements performed for error correction. In \secref{sec:gaugecode}, we show a $\llbracket 16,4,2,4 \rrbracket$ code with gauge operators of weight four. We show a process for fault-tolerant error correction that only requires fast single-shot measurement of the gauge operators on a planar square lattice. We observe improved logical error rates, which can be attributed to the speed of error correction and low overhead.

On the theoretical front, we must develop encoding circuits and a universal logical gate set.  States may be prepared by either using flags for fault-tolerance, or by combining patches of distance-two code states into a distance-four state.  For fault-tolerant universal computation, one possible route is teleportation and logical measurements with distilled magic states.  In fact, logical measurements can be performed along with error correction~\cite{delfosse2020short}.  Another route to universality is to use transversal multi-qubit gates between vertically stacked code patches.  It may be possible for gates like the CCZ to induce magic~\cite{Paetznick13}, as the required $\llbracket 15,7,3 \rrbracket$ code can be obtained by puncturing the $\llbracket 16,6,4 \rrbracket$ code. If the error introduced by logical operations is kept low, many logical gates can be applied every time step, allowing high-depth logical circuits.

\begin{figure}
\centering
\includegraphics[width=.5\textwidth]{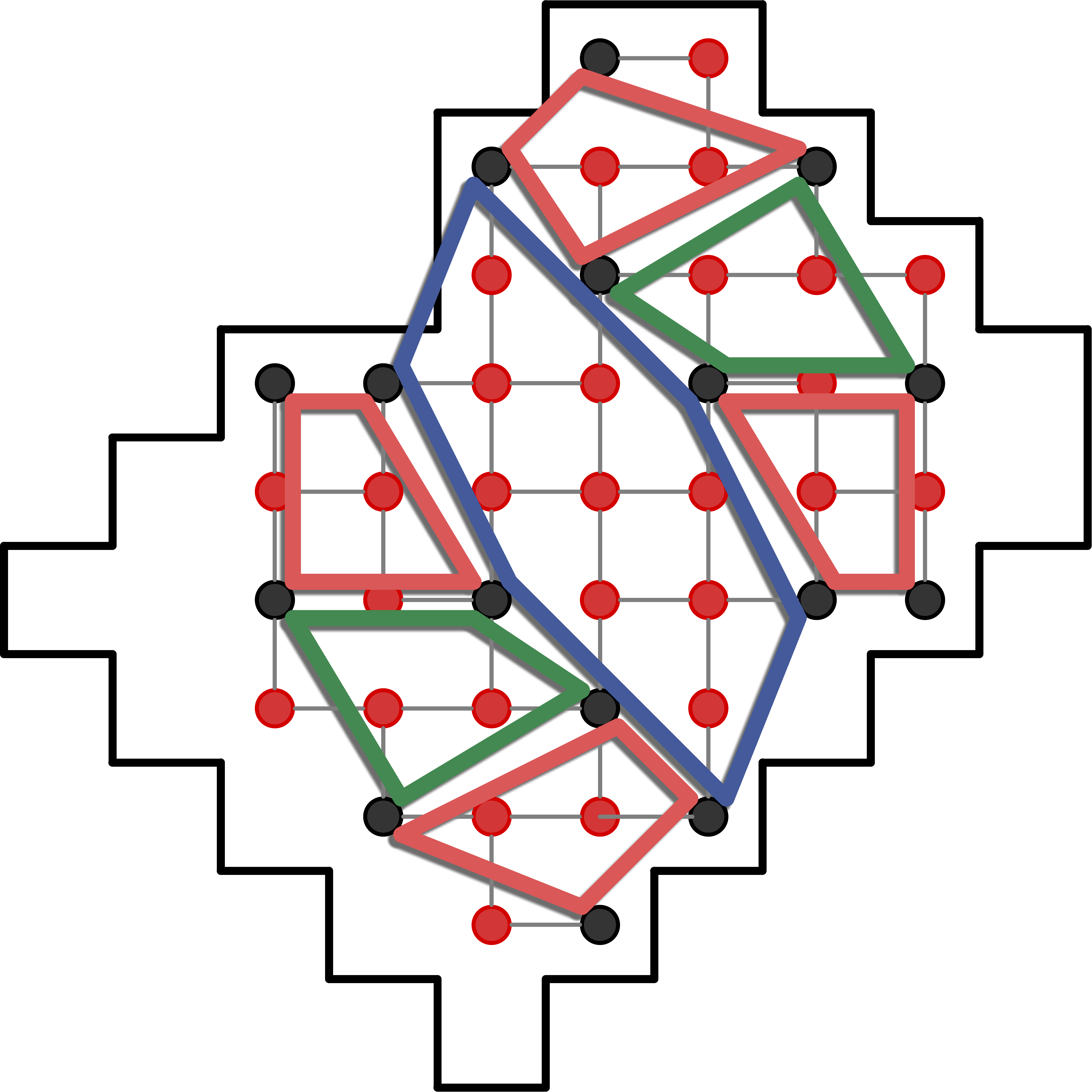}
\caption{A degree-four layout for flag-fault-tolerant error correction of the $k = 2$ code, using $43$ of the $53$ qubits on the Google Sycamore lattice. The stabilizer generators of the code are overlaid. Note that qubits have degree $4$ only in the ancillas measuring the weight-$8$ stabilizer, but elsewhere the maximum qubit degree is $3$. It may be possible for the $k=2$ subsystem code with only weight-$4$ stabilizers and gauges to fit on a layout of maximum degree $3$.} 
% \vspace{-0.4cm} %DEBUG
\label{f:degree4k2}
\end{figure}

Near the threshold of the odd distance surface codes, postselection on the distance-two and -four  variants shows reduced logical error rates.  With higher distance, the compounding effect is larger, meaning distance-eight or -ten surface codes may be sufficient for very precise computations.  At higher distance, rejections also become exceedingly rare, increasing the possible duration of computations.  
For larger patches on lattices like~\figref{f:layout}, more qubits can be encoded at high distance.

The biggest difficulty will then be in performing operations on or between different logical qubits in the same patch. 
Another avenue to pursue is concatenation. This technique can combine the low logical error rates of the surface code with the high encoding rates of block codes.  

% \section{Additional distance-$4$ codes and error correction sequences}
% \label{sec:addcodes}

% Talk about the other codes that we looked at and the alternate sequences of stabilizer measurements that we used to use. 

% \subsection{Codes}
% \label{subsec:d4codes}

% \subsection{Alternate error correction protocols}
% \label{subsec:d4altECP}

% \subsection{Alternate distance-$4$ error correction sequences}
% \label{subsec:d4altECS}

% \section{Distance-three error correction}
% \label{sec:d4d3}

\section{Surface code error correction with the union-find decoder}
\label{sec:d4surfcodes}

When performing error correction for a finite time interval~\cite{DennisKitaevLandahlPreskill01topological}, the classical error decoding algorithm is performed on one of two types of inputs. Either the entire syndrome history is used for one shot error decoding, or the process can be repeatedly applied on $\sim d$ rounds of syndromes, as and when they are collected. With the distance-four surface code, we consider repeated decoding with four rounds of syndromes to best model an experimental implementation. 

With the goal of generalizing to larger surface codes, we implemented a version of the union-find decoding algorithm~\cite{delfosse2017almostlinear}. Here, corrections are only applied to faults identified in the older layers of syndromes, as in the ``overlap recovery" method of Ref.~\cite{DennisKitaevLandahlPreskill01topological}. It is also important to ensure that the input syndrome graph contains horizontal boundary vertices for each layer (as shown in \figref{fig:Xsyndgraph}) and a temporal boundary vertex above the most recently collected layer of syndromes (as shown in \figref{fig:3Dsyndgraph}). All the boundary vertices may be identified together as one vertex, to make the syndrome graph simpler for the decoder.

\begin{figure}
    \centering
    \subfloat[\label{fig:Xsyndgraph}]{\includegraphics[width=.38\textwidth]{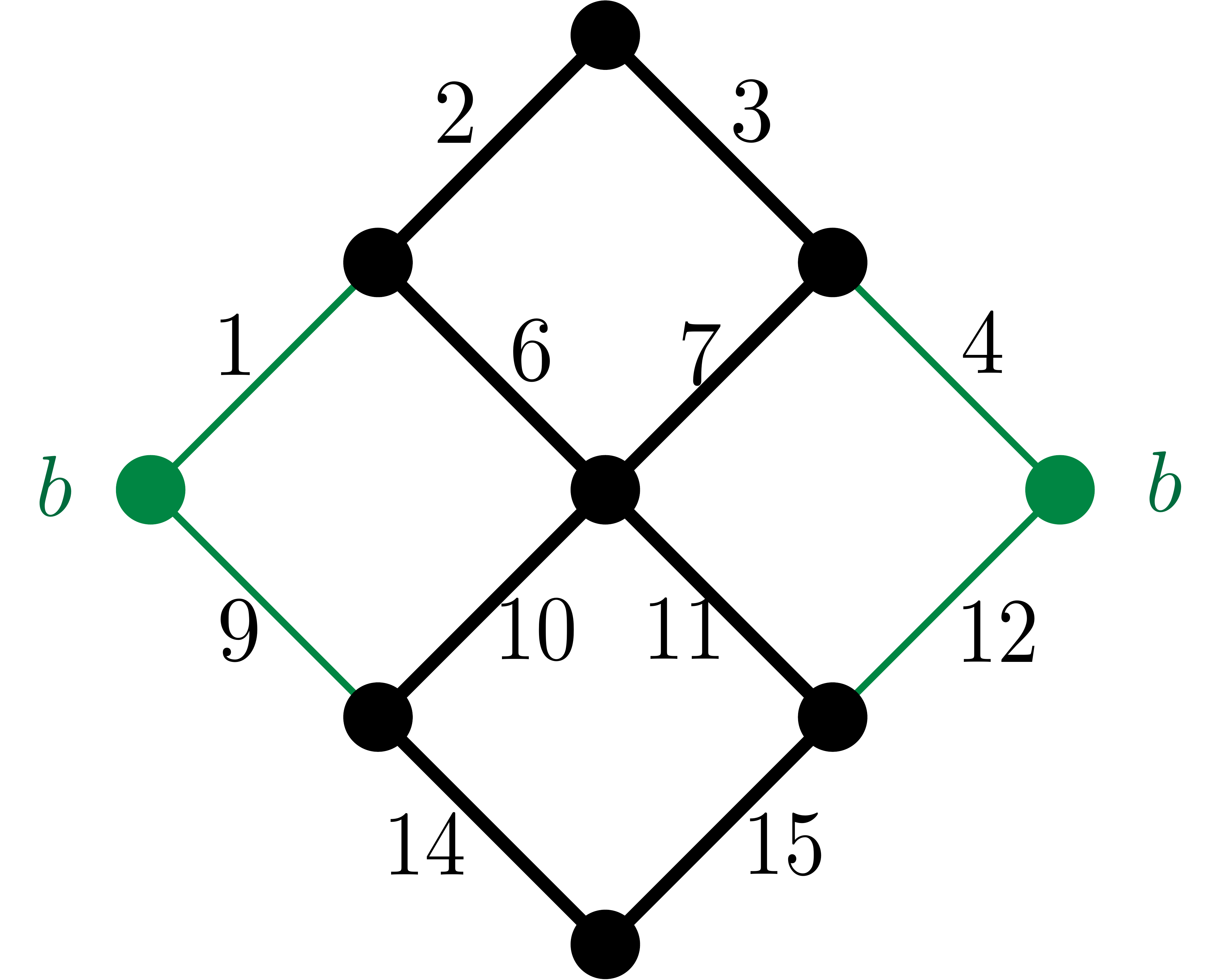}}
     % \hspace{1cm}
    \subfloat[\label{fig:Zsyndgraph}]{\includegraphics[width=.3\textwidth]{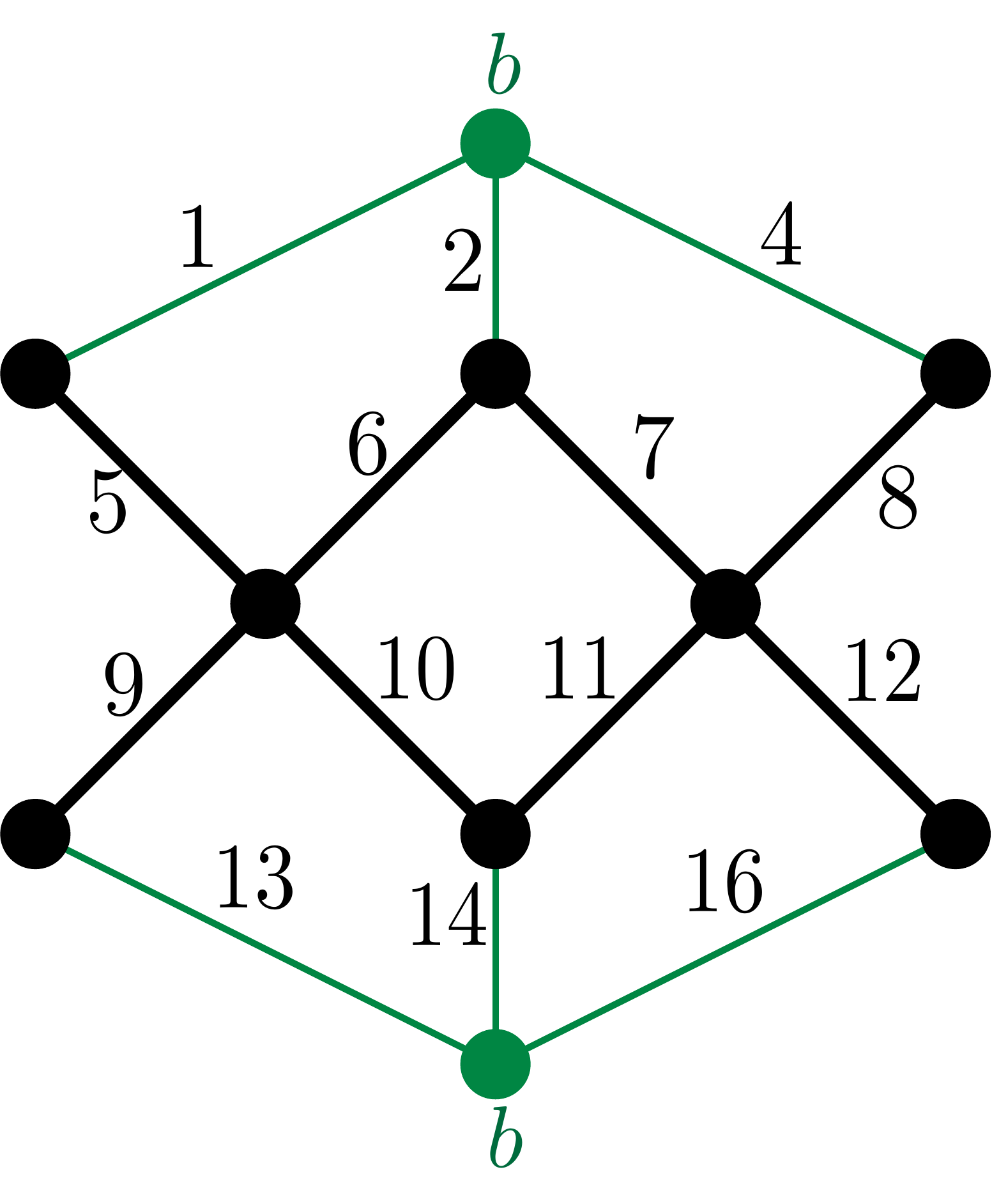}}
    \hspace{.2cm}
    \subfloat[\label{fig:3Dsyndgraph}]{\includegraphics[width=.65\textwidth]{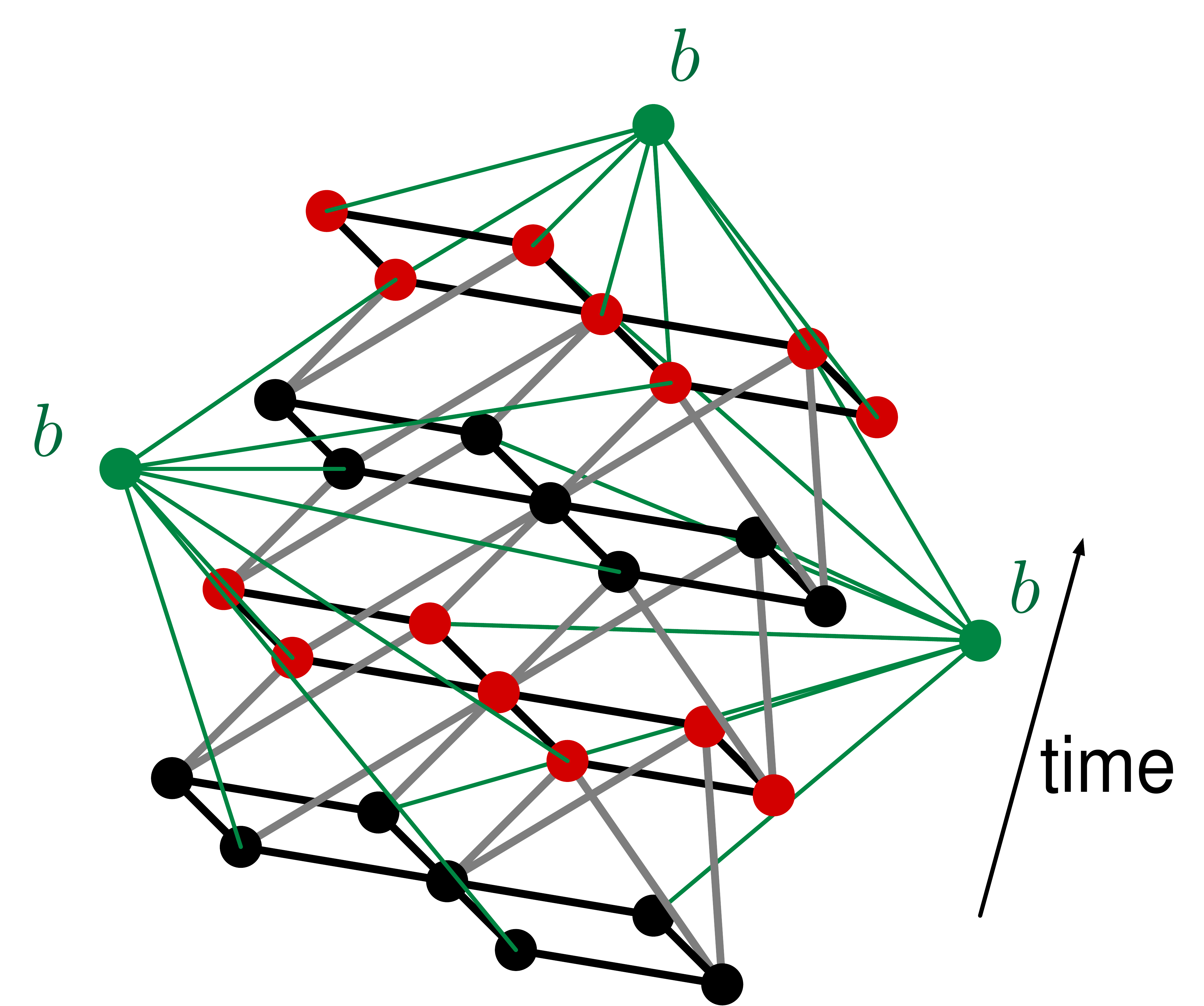}}
    \caption{Syndrome graphs passed to the union-find decoder. All the vertices labeled $b$ are identified as the same boundary vertex. (a)~Graph of vertices representing $X$ stabilizer measurement outcomes. The numbers indicate the index of the data qubit $Z$ correction for each edge of the graph. (b)~Syndrome graph for $Z$ stabilizers. (c)~($2+1$)-D syndrome graph, used for fault-tolerant decoding with circuit-level errors. We show four layers of syndromes in alternating colors red and black. Edges in the same syndrome layer are shown in black, with diagonal edges between layers shown in grey. For clarity, we do not show the vertical edges between syndrome vertices of the same stabilizer. Corrections are only applied for edges which have at least one vertex in the bottom two layers.}
    \label{fig:surfrej}
\end{figure}

Finally, we reject or apply corrections based on the rules for distance-four fault tolerance in \figref{f:smsrules}. The union-find decoding algorithm presents a set of edges in the syndrome graph of \figref{fig:3Dsyndgraph} denoting the presence of faults that caused $Z$-type errors.
 
\begin{procedure}
    \label{proc:surfcoderejdecoding}
    \textbf{Rejection decoding with the surface code: }
    \begin{enumerate}
        \item If $\abs{E} \geq 3$, reject and restart.
        \item If $\abs{E} = 2$ and the two faults are data errors that are sufficiently separated in time, correct the fault that occurred first. 
        \item If $\abs{E} = 2$ and $E$ contains vertical edges, remove them from $E$ and apply a correction to resolve any remaining edges. 
        \item If $\abs{E} = 1$, apply a correction if the edge has a syndrome vertex in either of the older two layers.
        \item If the removed/corrected edge has a syndrome vertex in the third layer, ensure syndrome vertices are flipped before the next round of decoding.
    \end{enumerate}
\end{procedure}

Note that there is a trade-off in choosing which cases to reject for and which cases to attempt a correction. We choose a separation that rejects as rare as possible while maintaining the $O(p^3)$ logical error rate. It is also possible to reduce the edge count at the start, by projecting measurement faults away before applying any of the rules. Evidently, there is a lot of scope for further improvements to the postselection rules. Research also needs to be done on how to generalize these rules to higher-distance surface codes.

\section{A subsystem code with four useful logical qubits}
\label{sec:gaugecode}

The new codes presented in \figref{f:codes} all suffer from one main caveat, the requirement that at least one weight-eight stabilizer must be measured. High-degree connectivity and extra ancillas are needed to measure these stabilizers fault-tolerantly to distance four. However, weight-four stabilizers are measured fault-tolerantly to distance three and only need three ancillas, with sparse connectivity. In \figref{fig:k4gaugecode}, we consider a $\llbracket 16, 4, 2, 4 \rrbracket$ code that is derived from the $\llbracket 16, 6, 4 \rrbracket$ code of \figref{f:codes}. We assign two of the logical qubits as gauge operators to simplify error correction. By measuring different weight-four representations of the gauge operators, we can construct syndromes that permit distance-four fault-tolerant error correction. Furthermore, performing only weight-four measurements simplifies the required connectivity and permits low-overhead implementation on a planar square lattice. Here, we show an implementation of error correction on a degree-four lattice requiring only $44$ qubits, with logical fidelity on par with the distance-five surface code needing $41$ or $49$ qubits.

In addition to low overhead and high-fidelity error correction, this code boasts more exciting properties. The $\llbracket 16,6,4\rrbracket$ code is a highly symmetric code, allowing for many types of simple fault-tolerant logical operations. These operations naturally pass down to the $\llbracket 16, 4, 2, 4\rrbracket$ code too. In \secref{subsec:k4gaugeloggates}, we show how to perform some Clifford operations. 
% Additionally, as we show in appendix, by switching to a $16,4,2$ code, we may also perform a non-Clifford CCCZ gate.

\begin{figure}
    \centering
    \subfloat[\label{fig:stabs}]{\includegraphics[width=.17\textwidth]{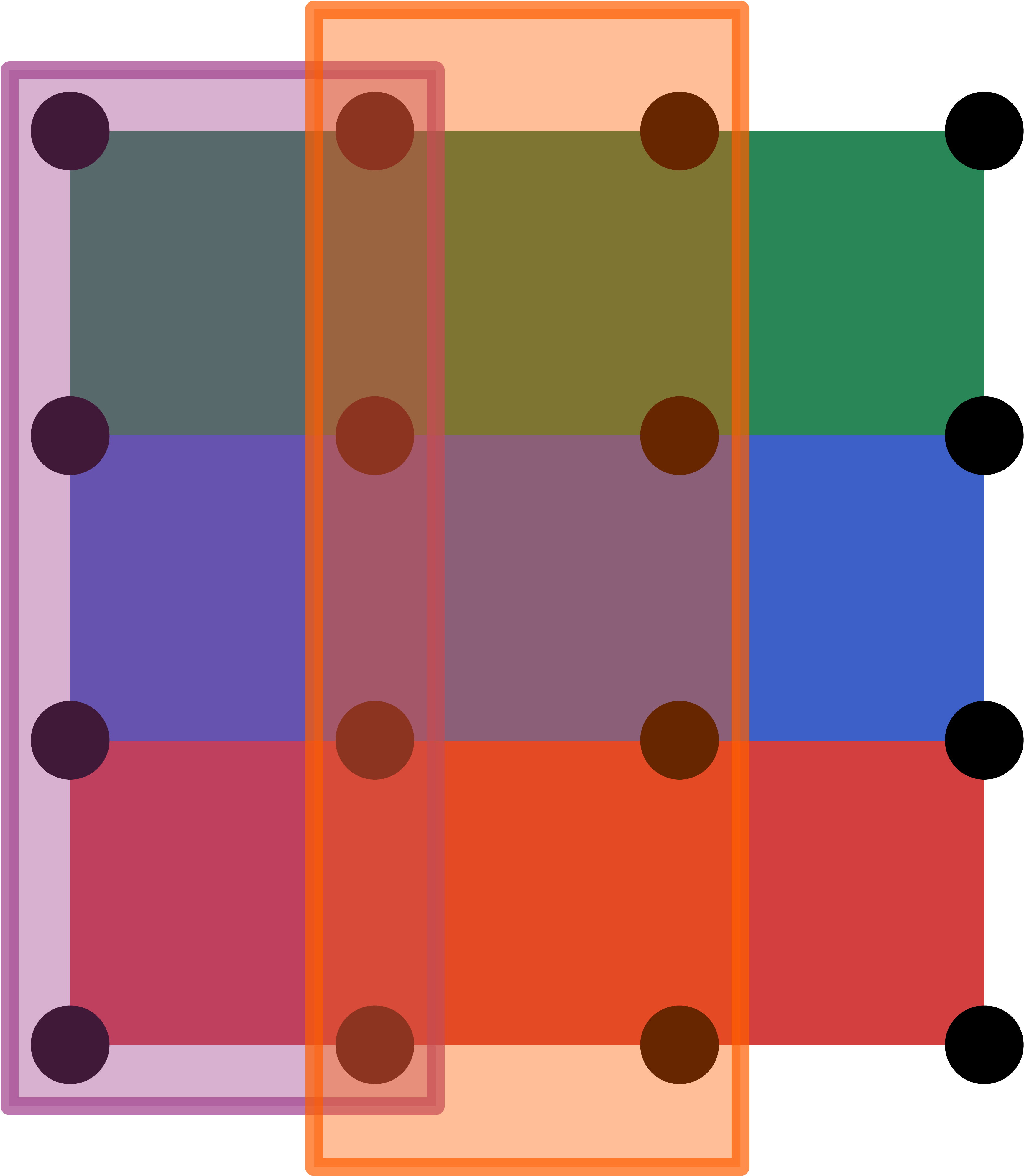}}
    \hspace{.4cm}
    \subfloat[\label{fig:gauges}]{\includegraphics[width=.235\textwidth]{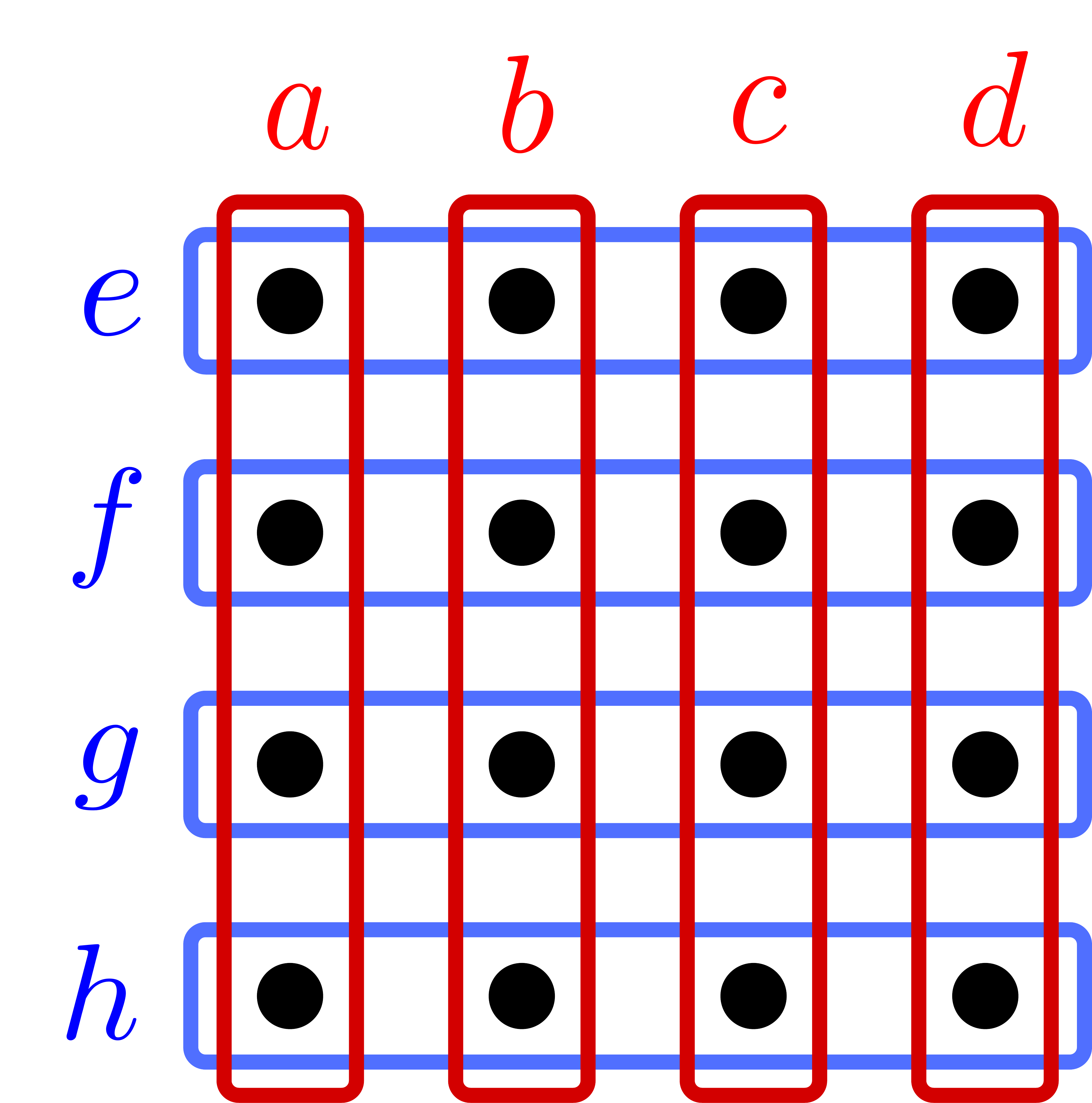}}
    \hspace{.2cm}
    \subfloat[\label{fig:logical1}]{\includegraphics[width=.45\textwidth]{imagesChap6/16qlogicals.pdf}}
    \caption{(a)~Stabilizers of the $\llbracket 16, 6, 4\rrbracket$ code. The $\llbracket 16, 4, 2, 4\rrbracket$ code is derived by assigning two of the logical qubits as gauge qubits. (b)~The two logical qubits chosen as gauges are the horizontal and vertical weight-four operators. By stabilizer equivalence, every column (row) of four qubits is a representation of the vertical (horizontal) gauge. These groups of data qubits are denoted by the letters $a$-$h$. (c)~$X$ and $Z$ operators specifying the four logical qubits.}
    \label{fig:k4gaugecode}
\end{figure}

\subsection{Fault-tolerant error correction and detection}

Fault-tolerant error correction is performed by measuring the four weight-four representations of each gauge operator. A deterministic routine takes four rounds of measurements. In the first round, measure all the $X$-type row operators. In the second round, $X$-type columns. In the third and fourth, measure the $Z$-type operators. Errors are detected or corrected according to the following rules. Consider $X$ error correction for example, 
\begin{enumerate}
    \item If either the row or column syndrome is of weight two, reject.
    \item For a weight-one (or equivalently, a weight-three) syndrome among the rows \textit{and} among the columns, apply a weight-one $X$ correction on the qubit indexed by the row and column with a $1$ (or $0$ in case the syndrome is weight-$3$).
    \item Ignore all other syndromes.
\end{enumerate} 

\begin{procedure}
\label{proc:noswaps}
Distance-$4$ fault-tolerant error correction with the $\llbracket 16,4,2,4 \rrbracket$ code.

\begin{enumerate}
    \item Measure $X_a$, $X_b$, $X_c$ and $X_d$ on the qubit layout of \cref{fig:exp5layout} using the stabilizer measurement circuits of \cref{fig:exp5smca} and \cref{fig:exp5smcb}. 
    \item Measure $X_e$, $X_f$, $X_g$ and $X_h$ with \cref{fig:exp5smca} and \cref{fig:exp5smcb}. This completes the circuits for $Z$ error correction. 
    \item For $X$ error correction, repeat the process with $Z$-type operators.
\end{enumerate} 
\end{procedure}

It may be possible to further reduce the number of rounds of measurements needed for error correction. For example, one could find a way to measure all eight operators of one type in one round. The difficulty is in the four intersecting regions of the ancilla circuits for the middle column and row gauge operator measurements. This could be made fault-tolerant using a shared flag ancilla scheme or a scheme that prepares ancilla states for Steane-style error correction~\cite{Huang21ShorSteane}. Alternatively, an adaptive routine may reduce the number of measurements needed.

\subsubsection{Numerical simulations}

\begin{figure}
    \centering
     \subfloat[\label{fig:exp5layout}]{\includegraphics[width=.37\textwidth]{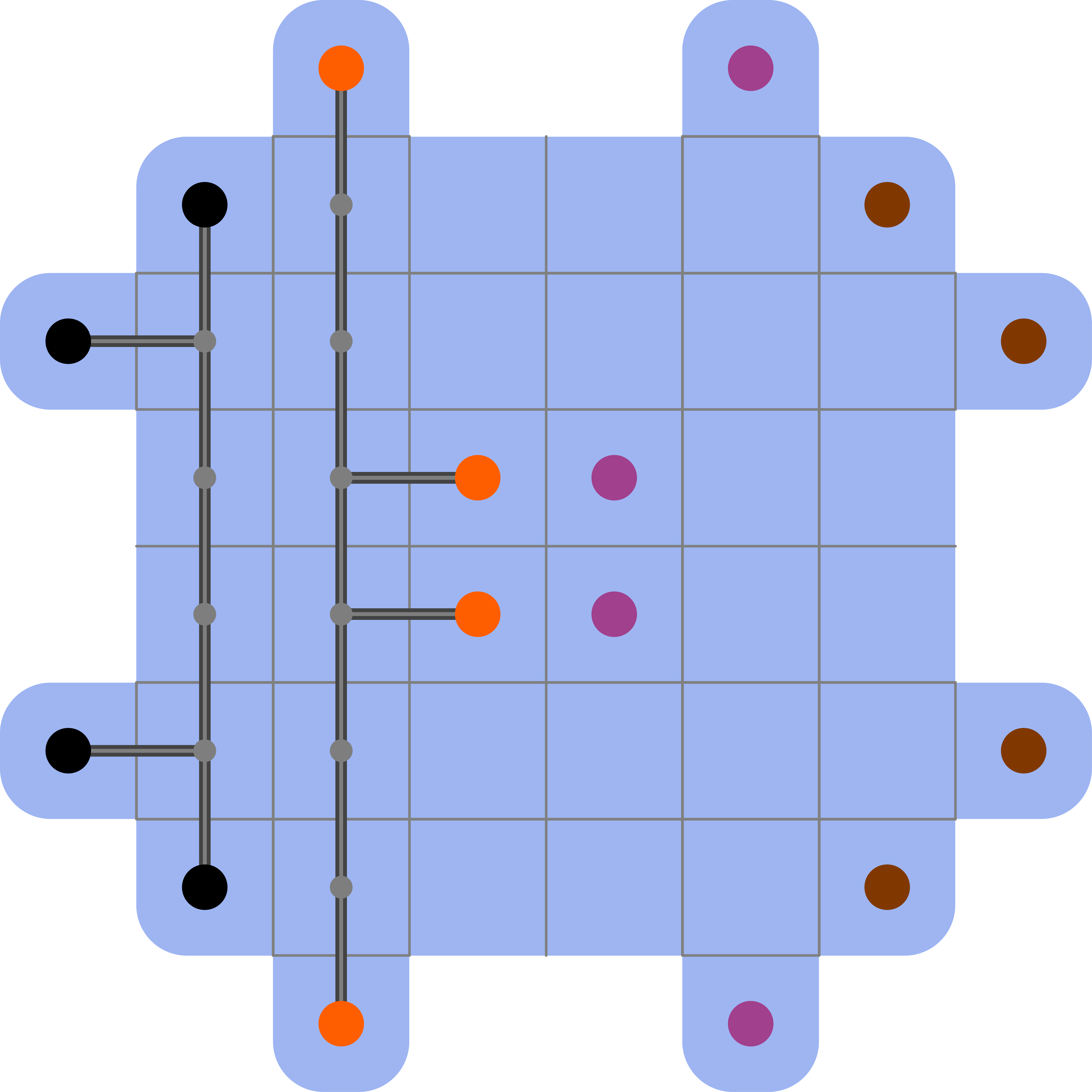}}
     \hspace{5cm}
     \subfloat[\label{fig:exp5smca}]{\includegraphics[height=.22\textwidth]{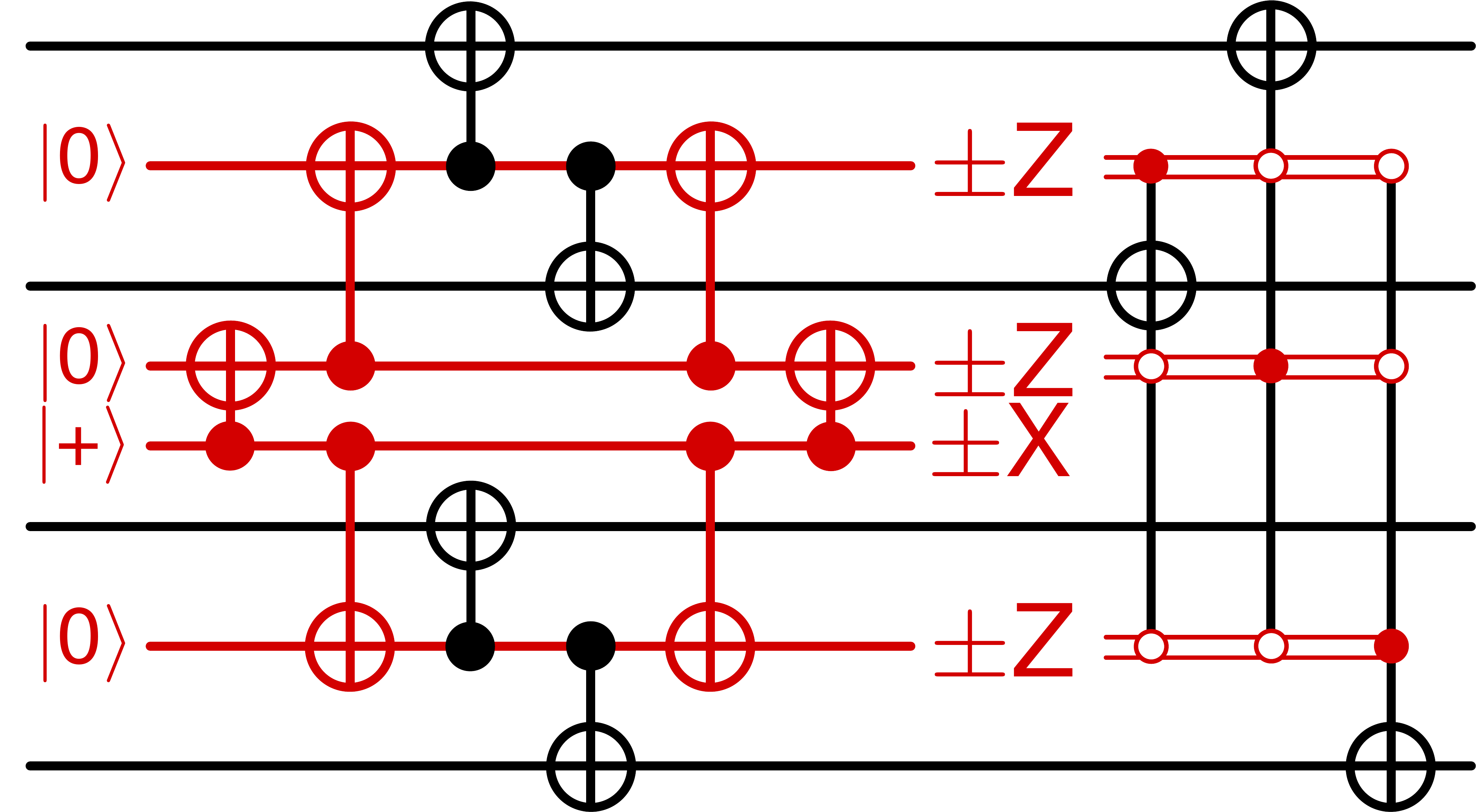}}
     \hspace{2.5cm}
     \subfloat[\label{fig:exp5smcb}]{\includegraphics[height=.22\textwidth]{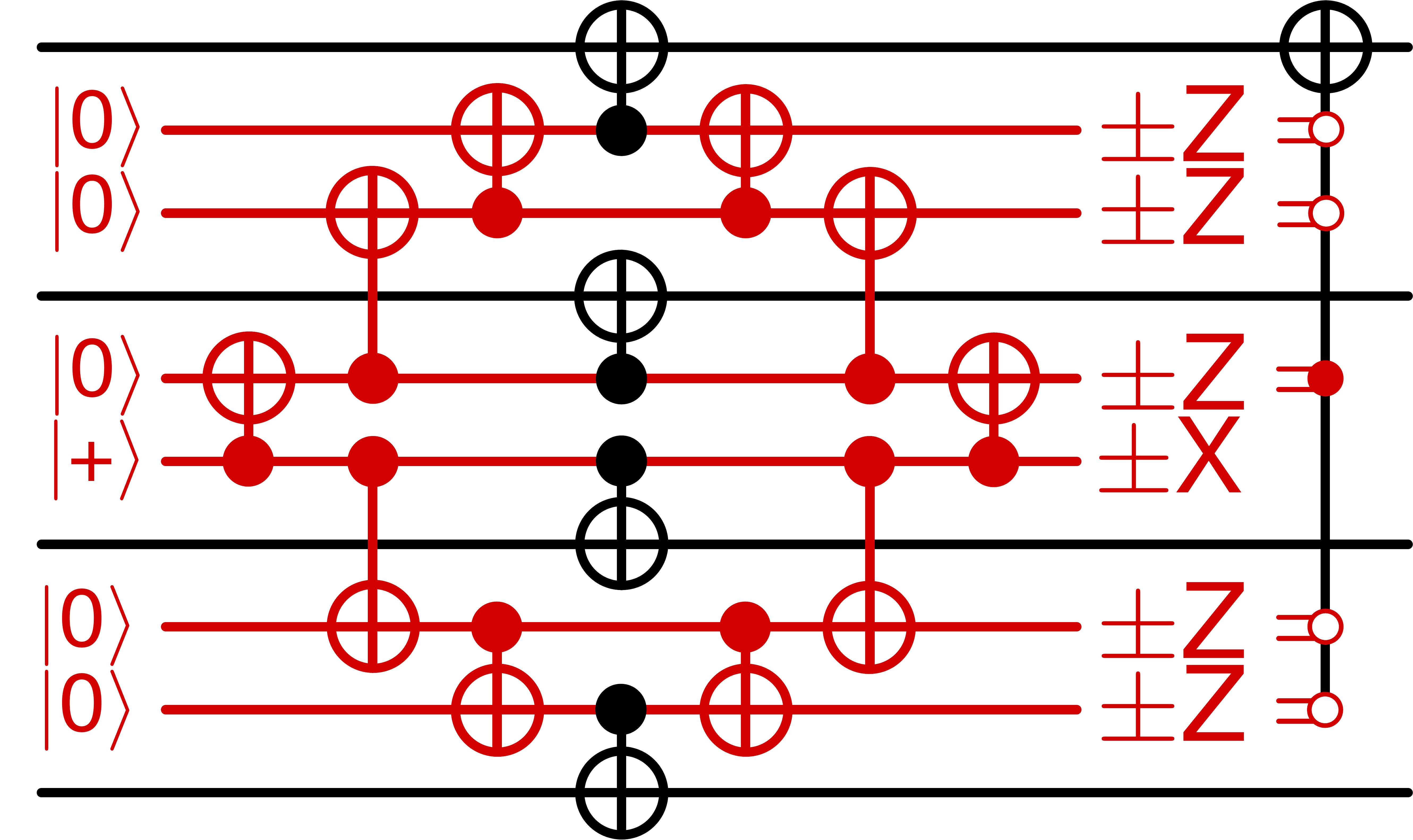}}
     \caption{Implementation of quantum error correction with a $\llbracket 16,4,2,4\rrbracket$ code on a square planar lattice of qubits. (a) Layout of qubits in \procref{proc:noswaps}. (b,c) Associated gauge operator measurement circuits.}
     \label{fig:exp5}
\end{figure}

We perform simulations of error correction with \procref{proc:noswaps} and compare the results with the distance-five and distance-four surface codes in \figref{f:scalingk4}. Note that while the logical error rate per time step for the $k=4$ subsystem code is larger than that of the $k=4$ stabilizer code discussed in \figref{f:codes}, the subsystem code only requires four time steps per round of error correction whereas the stabilizer code requires twelve. Hence the logical error rate per round of error correction will be much lower with the subsystem code. In practice, the subsystem code will be far more favorable.

% We use a noise model inspired by the Google device~\cite{Acharya23}. Note that the parameter $p$ in the actual device is $6 \times 10^{-3}$. The noise model is detailed below.

% \begin{itemize}
%     \item With probability $\frac{p}{6}$, a one-qubit gate is followed by a Pauli $X$, $Y$ or $Z$ error uniformly at random.
% 	\item With probability $p$, a two-qubit gate is followed by a two-qubit  Pauli error drawn uniformly at random from $\{I,X,Y,Z\}^{\otimes 2} \setminus \{ I \otimes I \}$.
% 	\item With probability $\frac{p}{6}$, the preparation of the $\ket{0}$ state is replaced by $\ket{1}=X\ket{0}$. Similarly for $\ket{+}$ and $\ket{-}$.
% 	\item With probability $3p$, a single-qubit $Z$ or $X$ basis measurement outcome is flipped.
% 	\item With probability $p/10$, each idle gate location is followed by a Pauli $X$, $Y$ or $Z$ error uniformly at random.
% 	\item With probability $4p$, each data qubit, during ancilla measurement, is acted upon by a Pauli $X$, $Y$ or $Z$ error chosen uniformly at random. This simulates noisy dynamical decoupling.
% \end{itemize}

\begin{figure}
\centering
\includegraphics[width =0.75\textwidth]{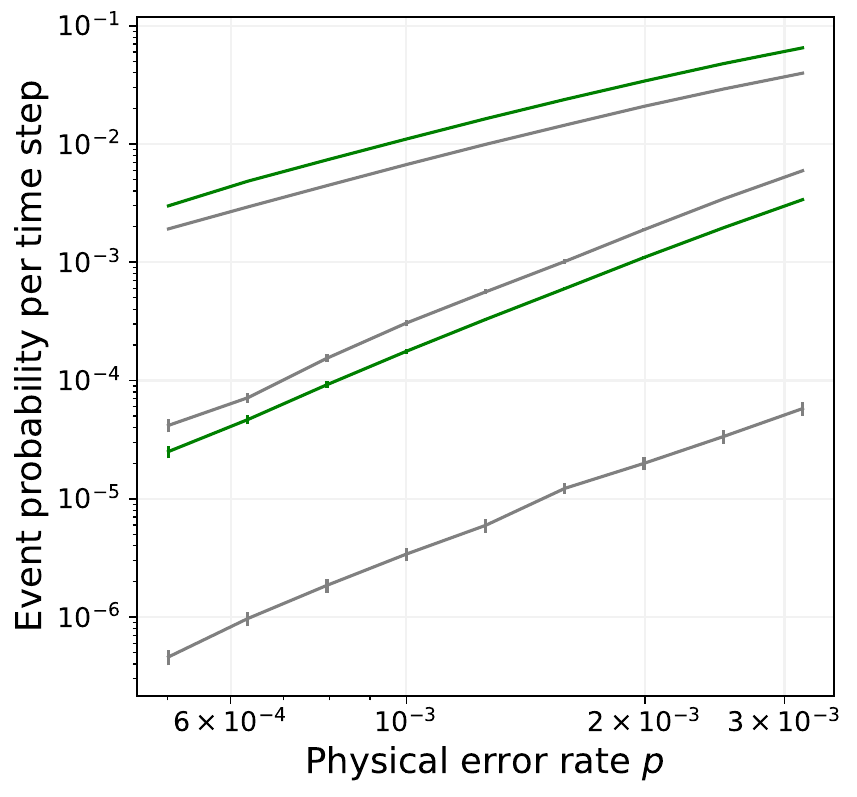}
\caption{Rejection rates (top two) and logical error rates (bottom three) per time step, when performing error correction with the $k=4$ subsystem code. The noise model is the same as that in \figref{f:scaling}. Here, the probability estimates are compared with the distance-four and distance-five surface codes, with the distance-five surface code exhibiting a larger error probability than distance four.}
\label{f:scalingk4}
\end{figure}

\subsection{Logical gates}
\label{subsec:k4gaugeloggates}

There exists a wide variety of logical Clifford operations that are relatively easily to implement with the $\llbracket 16, 4, 2, 4\rrbracket$ code. This includes transversal gates, permutation automorphisms, multi-qubit logical measurements, and CZ automorphisms.

For the logical qubits of the $k=6$ code, we show the logical Cliffords that result from different physical operations. We show simple examples for each case, leaving a full characterization of all the possible Cliffords to future work.

\begin{figure}
    \centering
    \hspace{-.8cm}
     \subfloat[\label{fig:nftswap}]{\includegraphics[height=.13\textwidth]{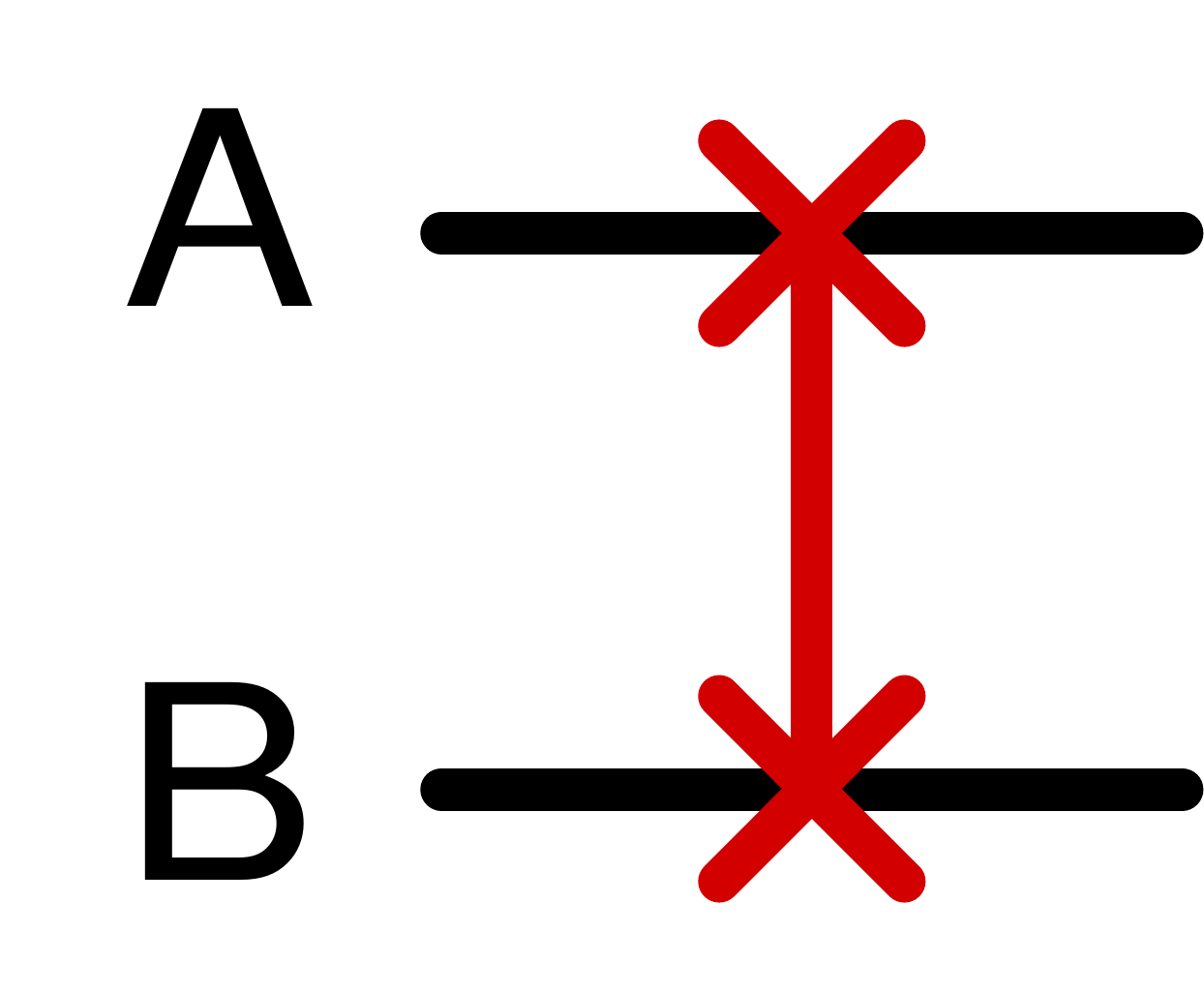}}
     \hspace{2.5cm}
     \subfloat[\label{fig:nftcz}]{\includegraphics[height=.13\textwidth]{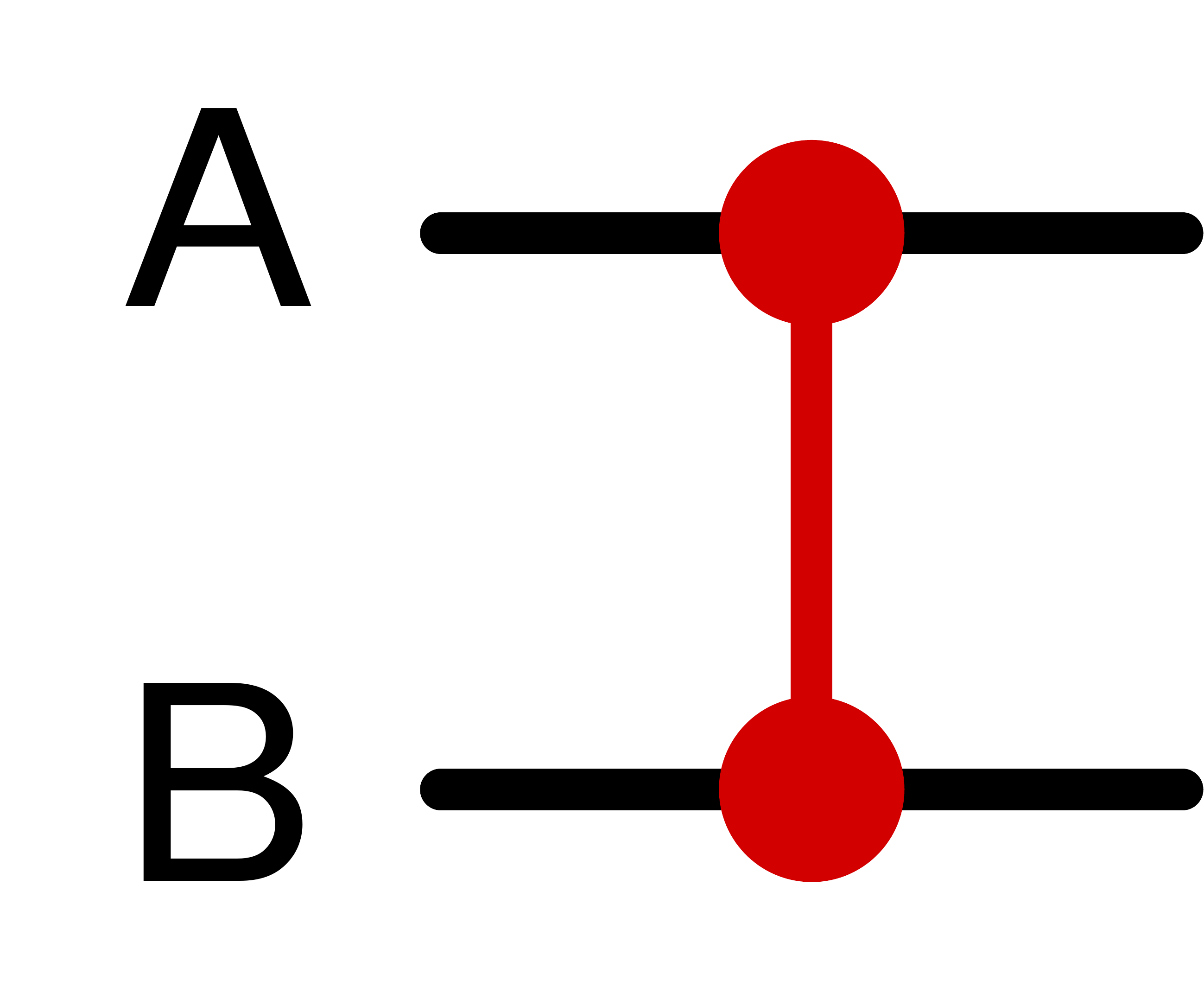}}
     \hspace{5cm}
     \subfloat[\label{fig:ftswap}]{\includegraphics[height=.2\textwidth]{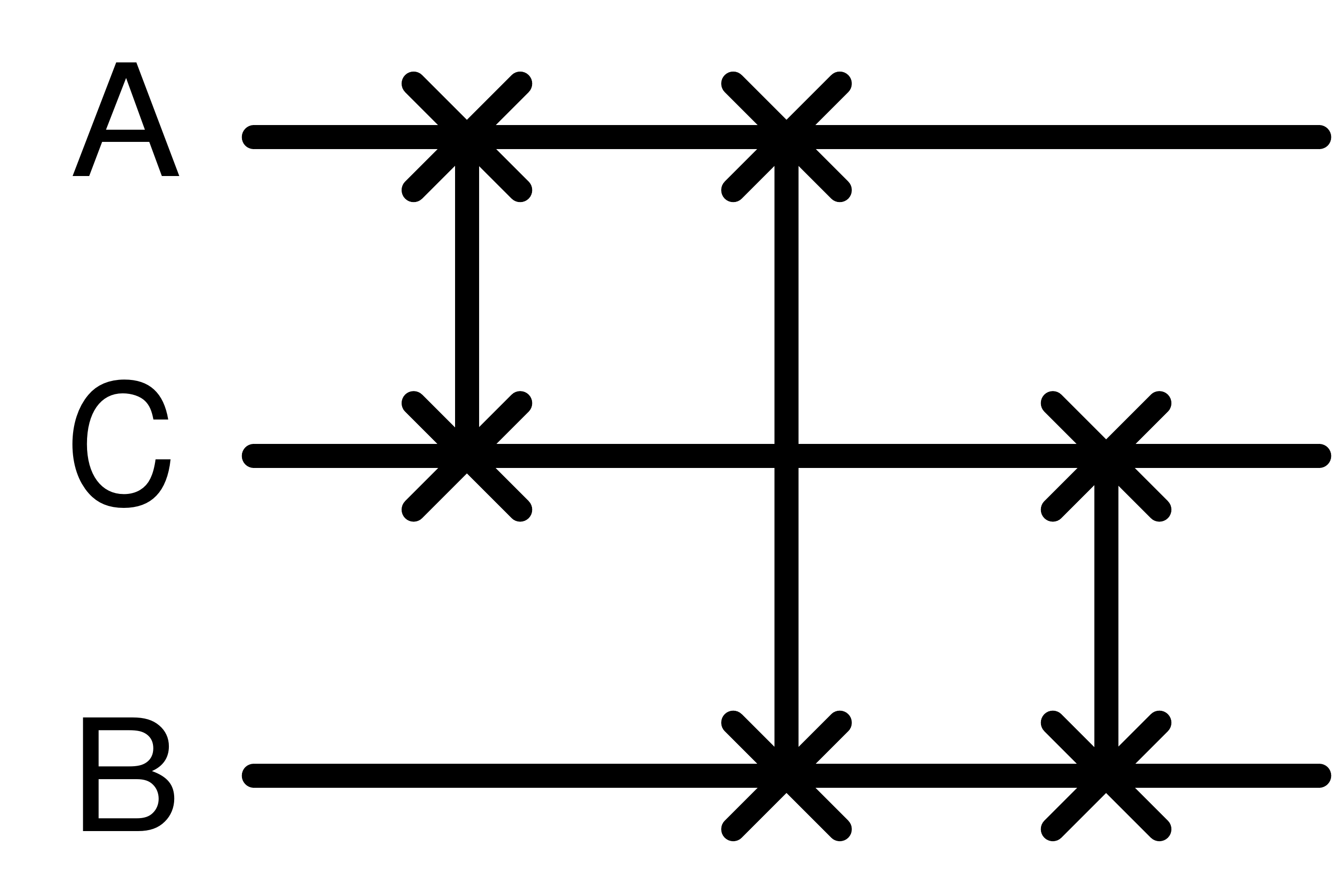}}
     \hspace{1cm}
     \subfloat[\label{fig:ftcz}]{\includegraphics[height=.2\textwidth]{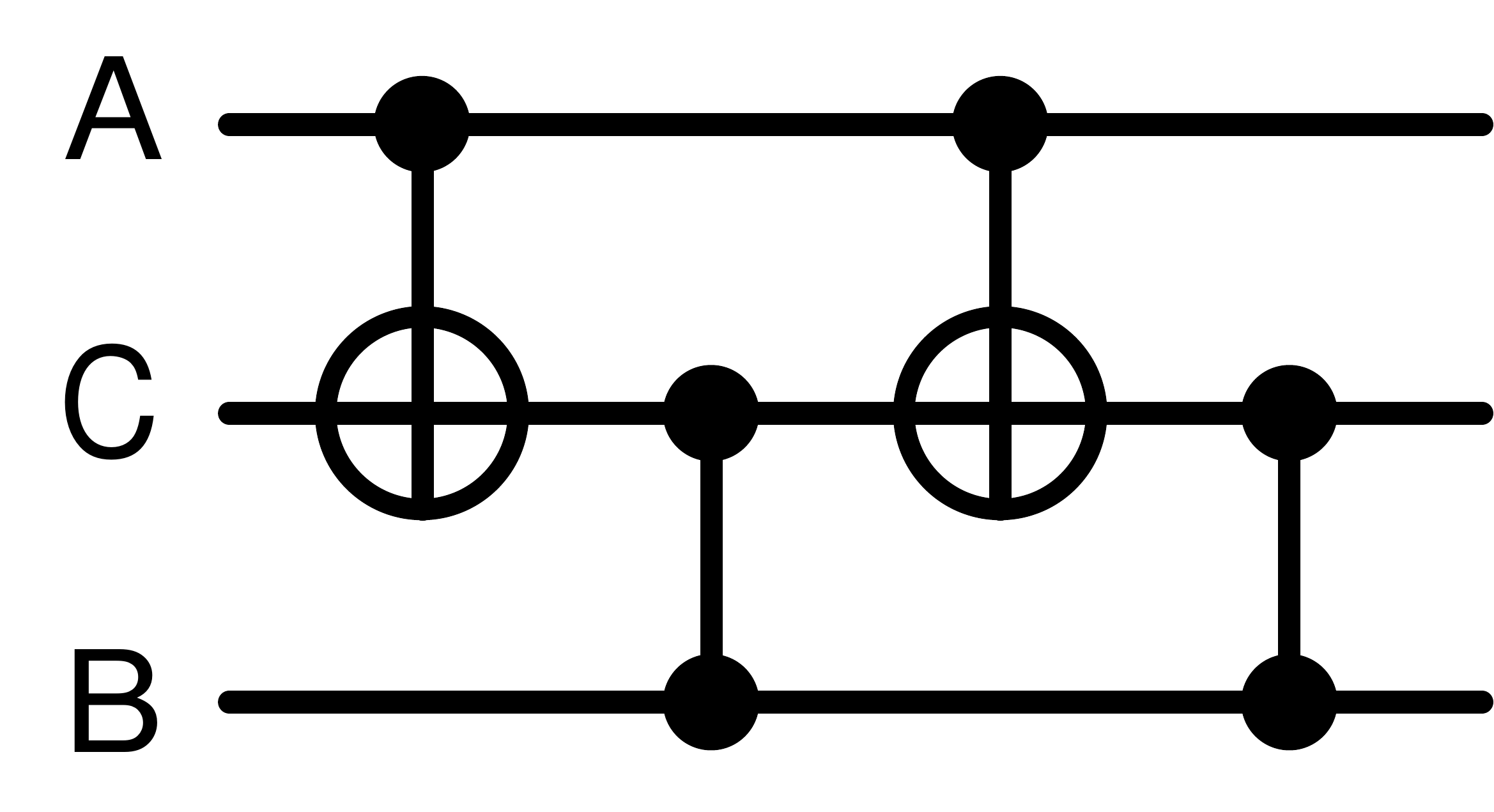}}
    \caption{Fault-tolerant versions of physical CZ and SWAP gates. (a)~Performing a SWAP between two qubits between qubits A and B is not fault-tolerant, even when broken into a sequence of three CNOTs. (c)~It can be made fault tolerant by swapping through an ancilla qubit, C. (b)~Similarly, a CZ is not fault tolerant, but can be made fault-tolerant using an ancilla, as in (d).}
    \label{fig:CZSWAP}
\end{figure}

\subsubsection{Transversal gates}

\begin{itemize}
    \item $\SWAP_{12} \SWAP_{34} \SWAP_{56} H^{\otimes 6}  := H^{\otimes 16}$.
    \item $\CZ_{12} \CZ_{34} \CZ_{56} := S^{\otimes 16}$.
\end{itemize}

\subsubsection{Permutation automorphisms}

Permutations of the physical qubits may leave the codespace invariant while modifying the logically encoded information. These operations are called permutation automorphisms. On the layout of \figref{fig:exp5layout}, permutations can be performed via $\SWAP$ gates. However, physical SWAPs between two data qubits of the code are not fault-tolerant. They must be mediated by an ancilla, as shown in \figref{fig:CZSWAP}.

Below, we note the logical effects of some permutation autmorphisms.

\begin{itemize}
    \item $\CNOT_{3,5} \CNOT_{6,4} := (1,5)(2,6)(3,7)(4,8)$
    \item $\SWAP_{3,6} \SWAP_{4,5} := (2,3)(6,7)(10,11)(14,15)$
    \item $\CNOT_{2,4} \CNOT_{2,5} \CNOT_{3,1} \CNOT_{6,1} := (1,5)(4,8)(9,13)(12,16)$
    \item $\CNOT_{2,4} \CNOT_{2,5} \CNOT_{3,1} \CNOT_{3,5}
    \CNOT_{6,1} \CNOT_{6,4} := $ \\
    $ (1,5)(4,8)(10,14)(11,15)$
    % \item $\CNOT_{1,3} \CNOT_{1,4} \CNOT_{1,5} \CNOT_{1,6}
    % \CNOT_{3,2} \CNOT_{4,2}
    % \CNOT_{5,2} \CNOT_{6,2} :=$ $ (1,4)(2,3)(13,16)(14,15)$
\end{itemize}

\subsubsection{Multi-qubit logical measurements}

In \tabref{tab:multqublogmeas}, we list the sequences of operators needed to measure the logical operators fault-tolerantly to distance four. These types of measurements allow targeted logical Cliffords, such as CNOTs or Hadamards. When assisted by magic states, they can also be used to perform non-Clifford gates (see, for example, \chapref{sec:MSDMSD}).

\begin{table}
\caption{\label{tab:multqublogmeas} Sequences of weight-four operators needed to measure different logical operators fault-tolerantly to distance four. Majority voting decides the measurement outcome. If the measured results are split with equal probability, the measurement result is rejected.}
\centering
\setlength{\tabcolsep}{4pt}
\begin{tabular}{c | c}
\hline \hline
Operator & Sequence for fault-tolerant measurement \\
\hline
\multirow{1}{*}{$X_3$} & \includegraphics[width=.4\textwidth]{imagesChap5/mX3.pdf} \\
\multirow{1}{*}{$X_4$} & \includegraphics[width=.4\textwidth]{imagesChap5/mX4.pdf} \\
\multirow{1}{*}{$X_5$} & \includegraphics[width=.4\textwidth]{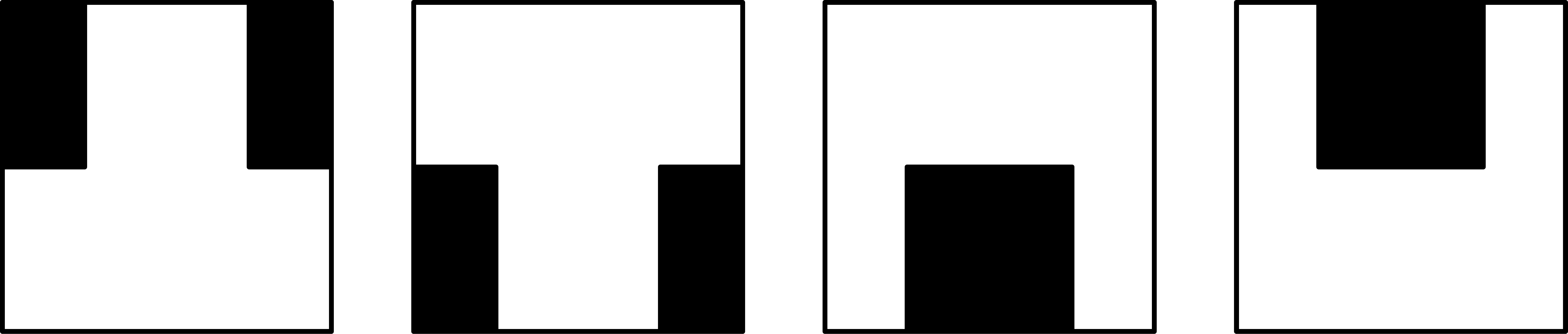} \\
\multirow{1}{*}{$X_6$} & \includegraphics[width=.4\textwidth]{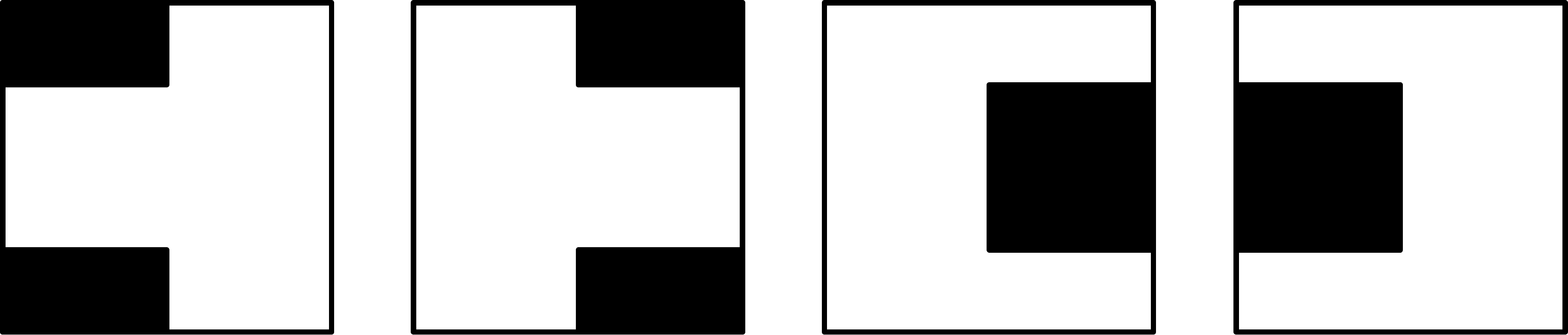} \\ 
\multirow{1}{*}{$X_1 X_3$} & \includegraphics[width=.4\textwidth]{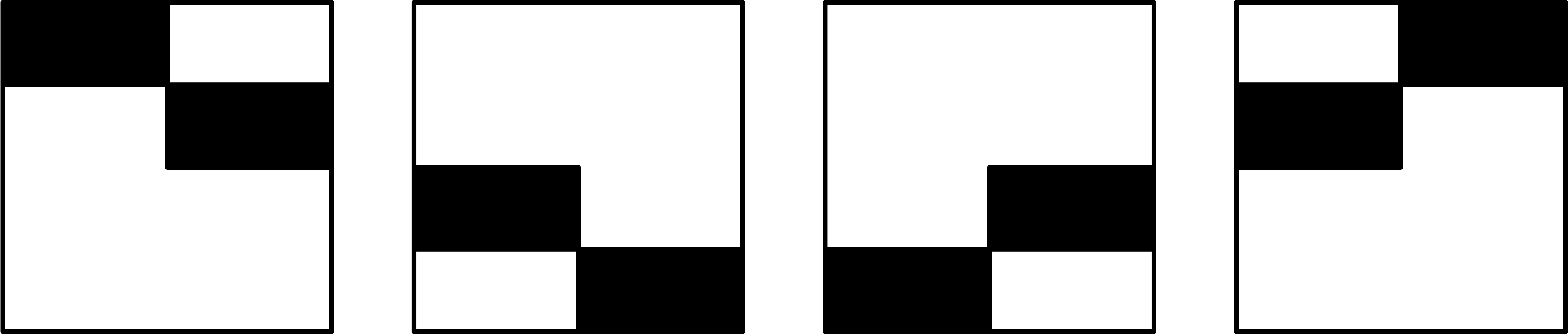} \\
\multirow{1}{*}{$X_1 X_3 X_5$} & \includegraphics[width=.4\textwidth]{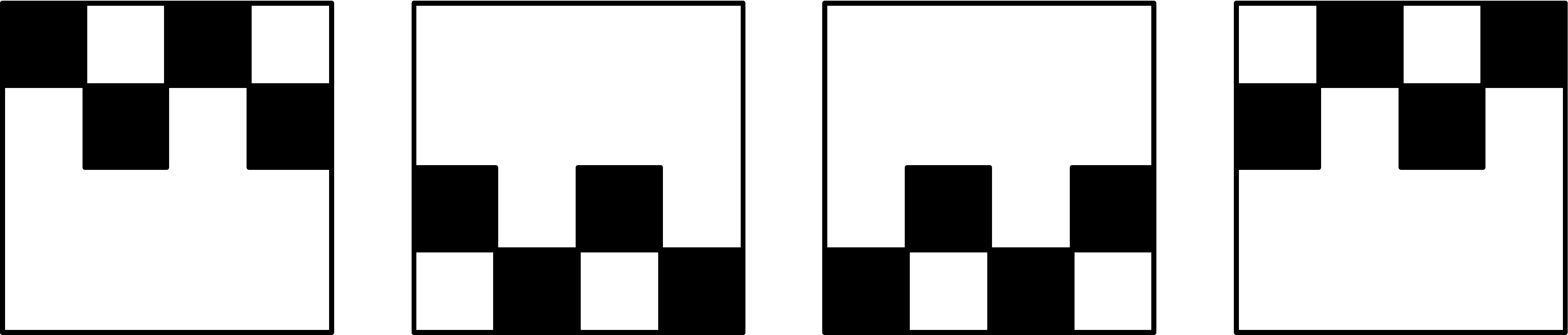} \\
\multirow{1}{*}{$X_3 X_4 X_5 X_6$} & \includegraphics[width=.4\textwidth]{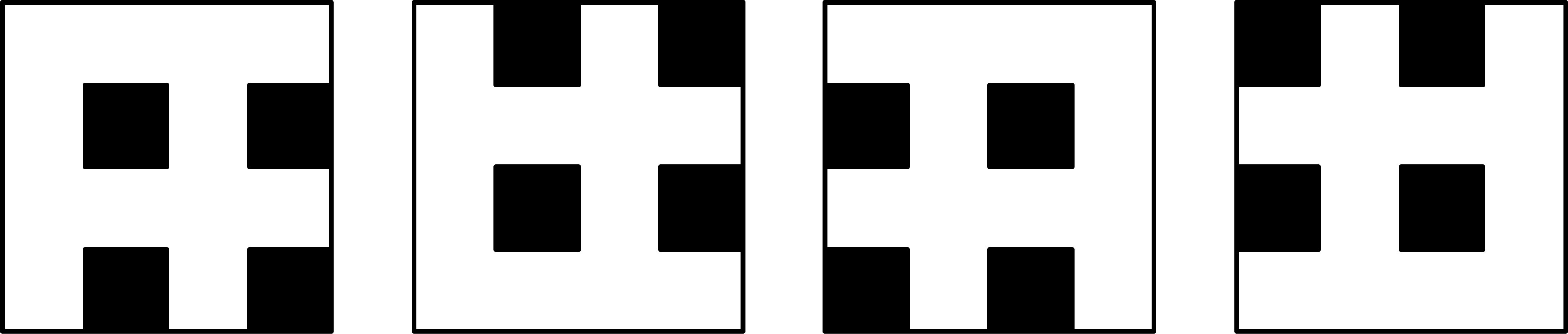} \\
\multirow{1}{*}{$X_1 X_3 X_4 X_5 X_6$} & \includegraphics[width=.4\textwidth]{imagesChap5/mX13456.pdf} \\
\hline \hline 
\end{tabular}
\end{table}

\subsubsection{CZ automorphisms}
Similar to permutation automorphisms, by applying CZ gates between specific pairs of physical qubits, we can perform logical CZ transformations while leaving the codespace invariant.

Up to a logical $\CZ_{1,2} \CZ_{3,4} \CZ_{5,6}$, which can be undone by the transversal gate $S^{\otimes 16}$, the following logical $\CZ$ transformations can be achieved.
\begin{itemize}
    \item $\CZ_{1,4} \CZ_{4,5} := \CZ_{1,2}\CZ_{5,6}\CZ_{3,7}\CZ_{4,8}\CZ_{9,10}\CZ_{13,14}\CZ_{11,15}\CZ_{12,16}$
    \item $\CZ_{3,5} \CZ_{1,6}  := \CZ_{1,5}\CZ_{2,3}\CZ_{4,8}\CZ_{6,10}\CZ_{7,11}\CZ_{9,13}\CZ_{14,15}\CZ_{12,16}$
    \item $\CZ_{1,4} \CZ_{2,4} \CZ_{4,5} \CZ_{4,6} := \CZ_{1,2}\CZ_{5,6}\CZ_{3,7}\CZ_{4,8}\CZ_{9,13}\CZ_{10,14}\CZ_{11,12}\CZ_{15,16}$
\end{itemize}

\chapter[Temporally encoded lattice surgery]{New magic state distillation factories optimized by temporally encoded lattice surgery}
\label{chap:MSD}

\begin{figure}
\hspace{-1.25cm}
\includegraphics[width=1.05\textwidth]{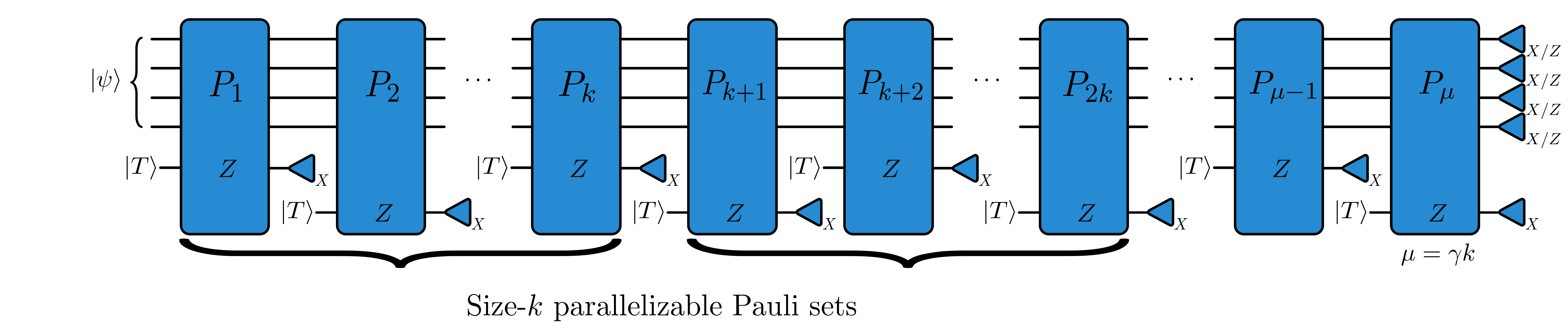}
\caption{General model of Pauli-based computation. A quantum algorithm can be written as a sequence of multi-qubit Pauli measurements which perform both Clifford and non-Clifford operations (here we show the implementation of non-Clifford gates). In general the multi-qubit Pauli operations can be ordered into sets of commuting Pauli operators, where Clifford corrections can be conjugated to the end of each set. Such sets are called parallelizable Pauli (PP) sets. A logical $T$ gate (which is non-Clifford and forms a universal gate set when combined with Clifford operations) can be implemented via a multi-qubit Pauli measurement acting on a set of data qubits and an ancillary magic state $\ket T = (\ket{0} + e^{i \pi / 4}\ket{1}) / \sqrt{2}$. }
\label{fig:PBC}
\end{figure}

\begin{figure}
    \centering
    \subfloat[\label{fig:TELSprotocol} ]{\includegraphics[width=.55\textwidth]{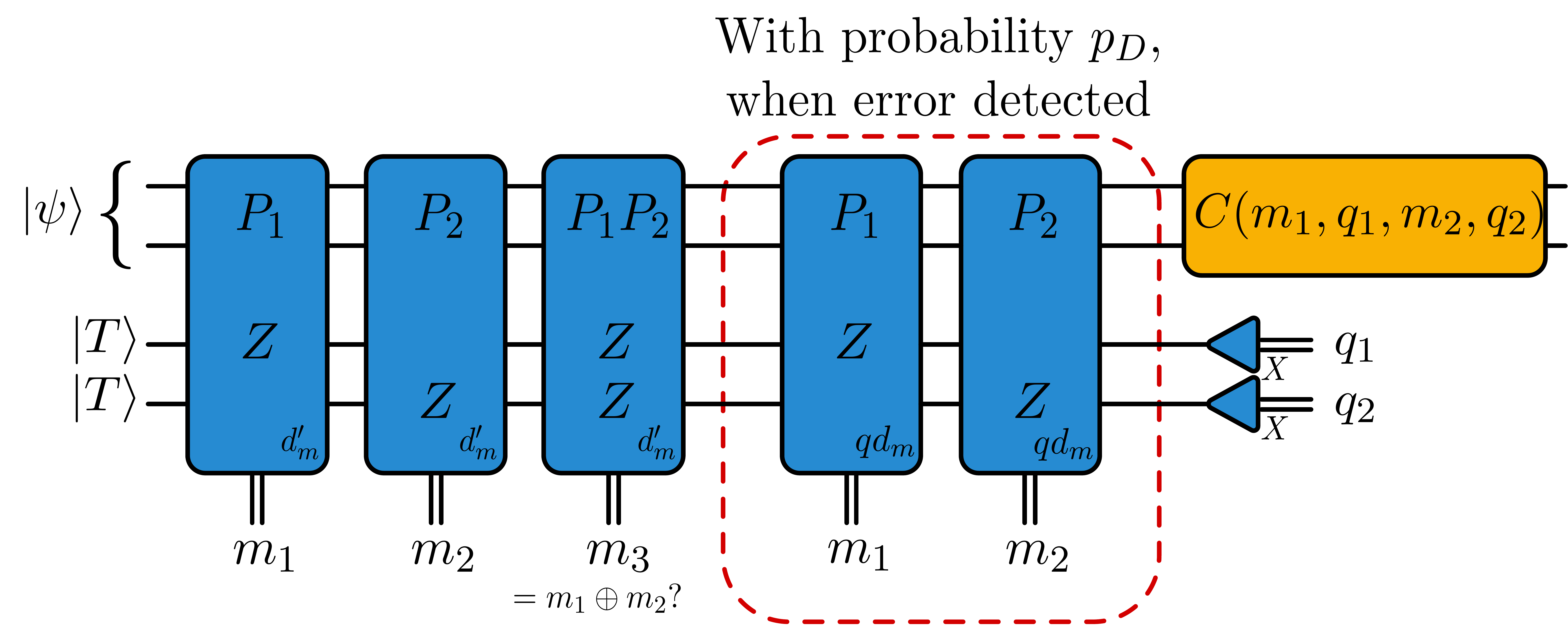}}
\hspace{0.3cm}
\subfloat[\label{fig:TELSprotocolnew} ]{\includegraphics[width=.415\textwidth]{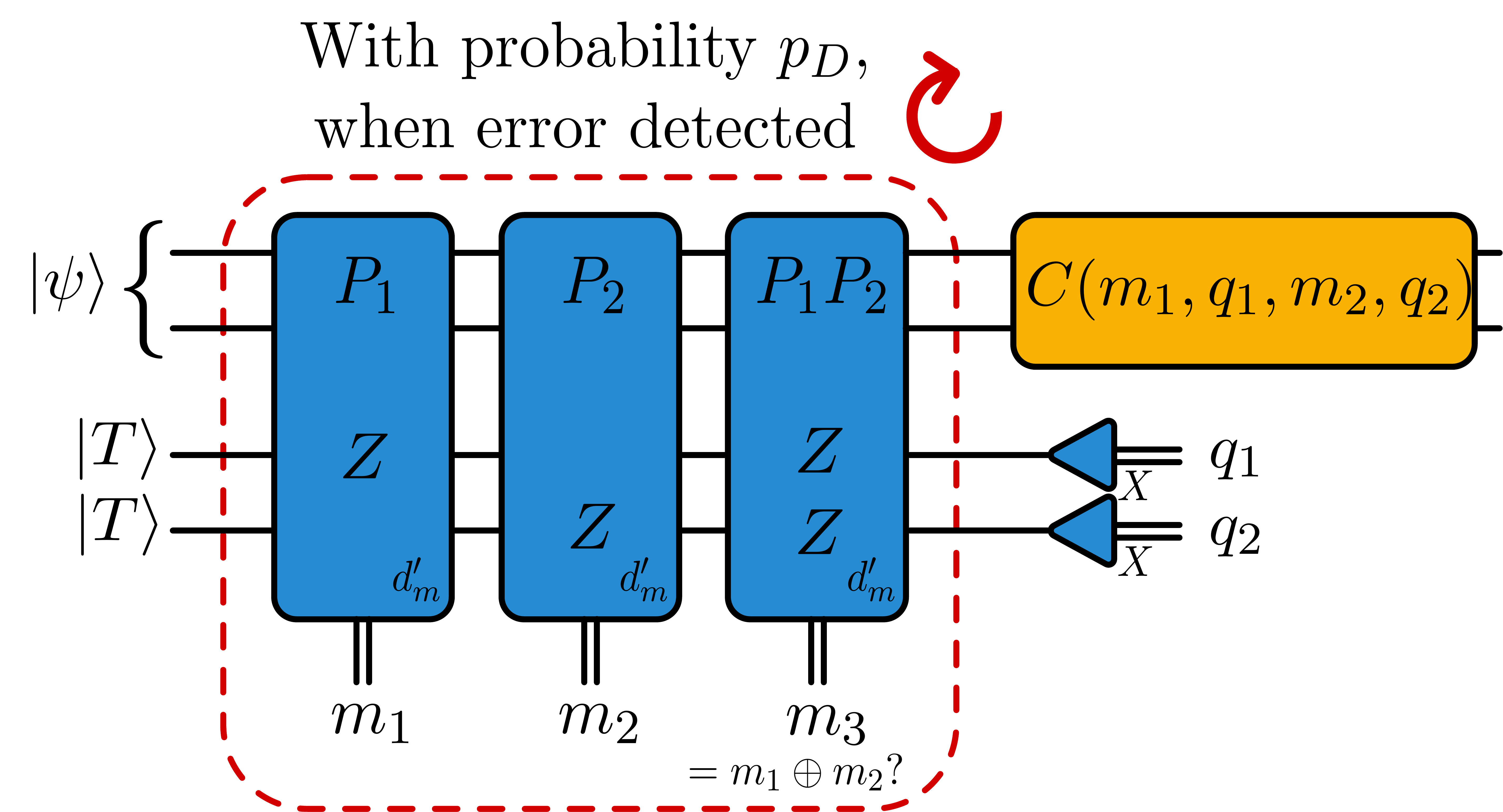}}
    \caption{(a) Old protocol for temporally encoded lattice surgery of a PP set of size-2, where $P_1P_2$ is a redundant measurement which is used to detect failures in the measurements of $P_1$ or $P_2$. If the measurement results of the three multi-qubit Pauli operators are inconsistent, the original multi-qubit Pauli operators $P_1$ and $P_2$ are measured again. Blue boxes correspond to multi-qubit Pauli measurements, and blue triangles correspond to logical single-qubit measurements. Orange boxes correspond to Clifford corrections. (b) New protocol for temporally encoded lattice surgery of a PP set of size-$2$. The operators $P_1$, $P_2$ and $P_1P_2$ are repeatedly measured until no logical timelike failures are detected. In \cref{subsec:MSDTELSnew} we show that such a scheme results in smaller average runtimes for the implementation of a PP set. Orange boxes denote Clifford corrections that result from applying non-Clifford gates.}
\end{figure}

Universal fault-tolerant quantum computers will be required in order to implement large scale quantum algorithms. However, both the space and time costs due to the implementation of fault-tolerant quantum error correction protocols can be quite high \cite{Fowler12surface,Bravyi12,JonesToffoli,Yoder2016,chamberland2017overhead,CJO17Reed,BevsUniversal,Chamberland2019Magic1,Litinski19,Litinski19magic,chamberland2020very,ChambsCat,Chamberland22,SC22Modulo}. For fault-tolerant quantum computers where qubits are encoded in topological quantum error-correcting codes, lattice surgery paired with magic state distillation protocols provides an efficient way to implement universal gate sets while being compatible with the locality requirements of two-dimensional planar hardware architectures \cite{Bravyi05,Bravyi12,Meier13MSD,BombinTopoDistill,Campbell2018magicstateparity,Chamberland2019Magic1,chamberland2020very,fowler2018low,Litinski2018latticesurgery,Litinski19,Litinski19magic,Chamberland22,Chamberland22b,PsiQuantumLatticeSurgery}. In addition to the extra space-costs associated with lattice surgery protocols, there are also additional time costs arising from the required protection against timelike failures (which can result in logical parity measurement failures) \cite{Chamberland22,Chamberland22b,Gidney2022stability}. The timelike distance of a lattice surgery protocol is given by the number of syndrome measurement rounds which need to be performed during the measurement of a multi-qubit Pauli operator. 

In Ref.~\cite{Chamberland22}, a protocol called temporal encoding of lattice surgery (TELS) was introduced in order to reduce the required timelike distance of lattice surgery protocols, thus reducing algorithm runtimes. The first step in a TELS protocol is to divide all multi-qubit Pauli measurements $\{ P_1, P_2, \cdots, P_{\mu} \}$ required to run a quantum algorithm into sequences of parallelizable Pauli (PP) sets. For a PP set of size $k$, the multi-qubit Pauli operators in a PP set $P_{[t+1,t+k]} = \{ P_{t+1}, P_{t+2}, \cdots, P_{t+k} \}$ commute, and any necessary Clifford corrections can be conjugated to the end of the sub-sequence. This general model of Pauli-based computation is shown in \cref{fig:PBC}. In this work we consider a universal gate set generated by $\langle T, H, S, \text{CNOT} \rangle$ where $T = \text{diag}(1, e^{i \pi / 4})$ and $H$ and $S$ are the Hadamard and phase gates. For such a universal gate set, a $T$ gate can be implemented using multi-qubit Pauli measurements and the resource magic state $\ket T = (\ket{0} + e^{i \pi / 4}\ket{1}) / \sqrt{2}$. The main idea behind TELS is that for a given PP set, one can measure a larger over-complete set of multi-qubit Pauli operators, where each multi-qubit Pauli operator in the over-complete set is a product of multi-qubit Pauli operators from the original PP set. As was shown in Ref.~\cite{Chamberland22}, each multi-qubit Pauli operator in the new over-complete set is associated with a codeword of a classical $[n,k,d]$ error-correcting code. The measurement results denote the $n$ bits of the classical code. Applying the parity check matrix of the code to these measurement outcomes enables the detection of logical timelike failures. This in turn allows fewer rounds of syndrome measurements for each multi-qubit Pauli, due to the extra protection offered by the overlaying classical~code.

In this work, we introduce new TELS protocols that further reduce the timelike distance required for lattice surgery protocols. In previous TELS protocols, if a lattice surgery failure was detected while measuring the over-complete set, the multi-qubit Pauli operators from the original PP set would be remeasured. We show that better speedups can be obtained if, during the re-measure step, the operators from the over-complete set are repeatedly measured until no logical timelike failures are detected. We also show that in some cases, even larger speedups can be achieved when using the classical error-correcting codes to correct a subset of errors of smaller weight and detect all possible errors of higher weight. We also consider a large number of classical error-correcting codes for various sizes of PP sets. Such considerations enable more efficient TELS protocols to be applied to a wider range of quantum algorithms, where the average PP set sizes depend on the particular algorithm being implemented. 

We then focus on implementing TELS protocols in the context of magic state distillation using Clifford frames. In doing so, we consider a biased circuit-level noise model, where the logical qubits are encoded using asymmetric surface codes. We consider asymmetric surface codes since such codes require fewer qubits to achieve a desired logical failure rate for physical error rates $p = 10^{-3}$ and $p = 10^{-4}$ compared to other topological codes such as $XZZX$ and $XY$ codes \cite{HiggottBP}. By developing new layouts for magic state distillation tiles which are adapted to TELS protocols, we show that the space-time costs of such distillation protocols can be reduced compared to magic state distillation tiles that do not use TELS.

The manuscript is structured as follows. In \cref{subsec:PauliBasedReview} we review the notion of Pauli-based computation implemented via lattice surgery and in \cref{subsec:TELSreview} we review the implementation of previously proposed TELS protocols. Next, in \cref{subsec:MSDTELSnew} we show how repeated encoded multi-qubit Pauli measurements using TELS can result in smaller algorithm runtimes. We then show in \cref{subsec:MSDTELScorrect} how the simultaneous correction and detection of errors  with the classical codes used in a TELS encoding can lead to reduced runtimes compared to a pure error detection scheme. In \cref{subsec:ManyCodes}, we present the best average runtimes of TELS protocols using a wide range of classical error-correcting codes. In \cref{app:malignantsets}, we provide different methods to count the number of malignant fault sets for a given classical code. This is useful in analyzing the performance of TELS protocols with various classical codes. In \cref{app:codeconstruction}, we show how to construct the different classical codes used in \cref{subsec:ManyCodes}. \cref{app:speedups} contains details about the performance of TELS in additional noise regimes and target logical failure rates. 

Next in \cref{sec:MSDMSD}, we show how to apply TELS protocols to magic state distillation factories. Specifically, \cref{subsec:MSDMSDCliff} contains a new circuit for executing lattice-surgery-based magic state distillation. In this protocol, magic states are distilled up to a known Clifford correction. In \cref{app:CliffpartTraceproof}, we discuss in detail how Clifford frames are incorporated in our TELS protocols when applied to magic state distillation schemes. In \cref{sec:MSDMSDdesign}, we analyze the space-time costs of various distillation protocols that use TELS with the distillation circuit of \cref{subsec:MSDMSDCliff}. In particular, for physical error rate $p=10^{-4}$, we use $15$-to-$1$ and $116$-to-$12$ distillation protocols to distill magic states with logical failure probabilities $10^{-10}$ and $10^{-15}$, where $p$ is the noise parameter of a biased circuit-level noise model described in \cref{subsec:PauliBasedReview}. For $p=10^{-3}$, we consider $114$-to-$14$ and $125$-to-$3$ distillation protocols to produce magic states with logical failure rate $10^{-10}$ and $10^{-15}$ respectively. For each distillation protocol, we compute the minimum distances and space time costs using different hardware layouts that are dependent on the chosen lattice surgery mechanism.  \cref{app:golaycodechoice} shows a specific choice of codewords derived from the classical Golay code for use in a $15$-to-$1$ distillation protocol with a TELS implementation. This classical code allows for very low average runtime per Pauli. In \cref{app:algoSTcosts}, we show the general procedure for determining the minimum space-like and time-like distances for different distillation protocols and layouts. \cref{app:constants} contains additional information regarding each of the layouts proposed in this work, which is used in conjunction with \cref{app:algoSTcosts} to determine the minimum distances and space-time costs. We conclude with a note on how to schedule the operation of a distillation factory which contains many distillation tiles. Using a round-robin scheduling algorithm, we optimize the number of tiles that are required to output enough distilled magic states as required for seamless operation of a quantum core.

\section{Review of Pauli-based computation and lattice surgery}
\label{sec:MSDreview}

\subsection{Pauli-based computation and lattice surgery}
\label{subsec:PauliBasedReview}

There exist many models of computing to execute an algorithm on a quantum computer. The most well-studied is the circuit model of quantum computation~\cite{Deutsch89,Bennett95}. Examples of some other models are adiabatic quantum computation~\cite{Farhi01}, measurement-based models~\cite{Gottesman99,Aliferis04,Briegel09}, and fusion-based computation~\cite{Bartolucci21} which may be more suitable for specific hardware architectures. For a universal quantum computing architecture which uses two-dimensional planar topological codes (such as the surface code), the most natural way to implement a quantum algorithm is to use the Pauli-based model of quantum computation~\cite{Bravyi16}. In a Pauli-based computation (PBC), all logical gates can be applied by measuring multi-qubit Pauli operators (potentially using additional ancillas), as shown in \cref{fig:PBC}. Quantum algorithms are then executed using a pool of purified magic states and a sequence of multi-qubit Pauli measurements. In the interest of speeding up algorithms, we first note that the sequence of multi-qubit measurements can be grouped into subsequences of mutually commuting measurements. In \cref{fig:PBC}, the set $\{ P_1,P_2, \mathellipsis P_k\}$ is such a subsequence, and is called a PP set. An algorithm executing $\mu$ $T$-gates with $T$-depth $\gamma$ can in general be broken down into $\gamma$ PPs of average size $k = \mu / \gamma$. In some algorithms of interest like quantum chemistry simulations~\cite{Kim22}, the size of a PP set is between $9$ and $14$. We also note that for magic state distillation protocols that are expressed as sequences of multi-qubit Pauli measurements, all the non-Clifford gates commute and thus form PP sets. For instance, the $15$-to-$1$ distillation protocol has a PP set of size $11$, and for the $125$-to-$3$ distillation protocol, the PP set is of size $99$.

\begin{figure}
    \centering
    \hspace{-1mm}
    \includegraphics[width=.65\textwidth]{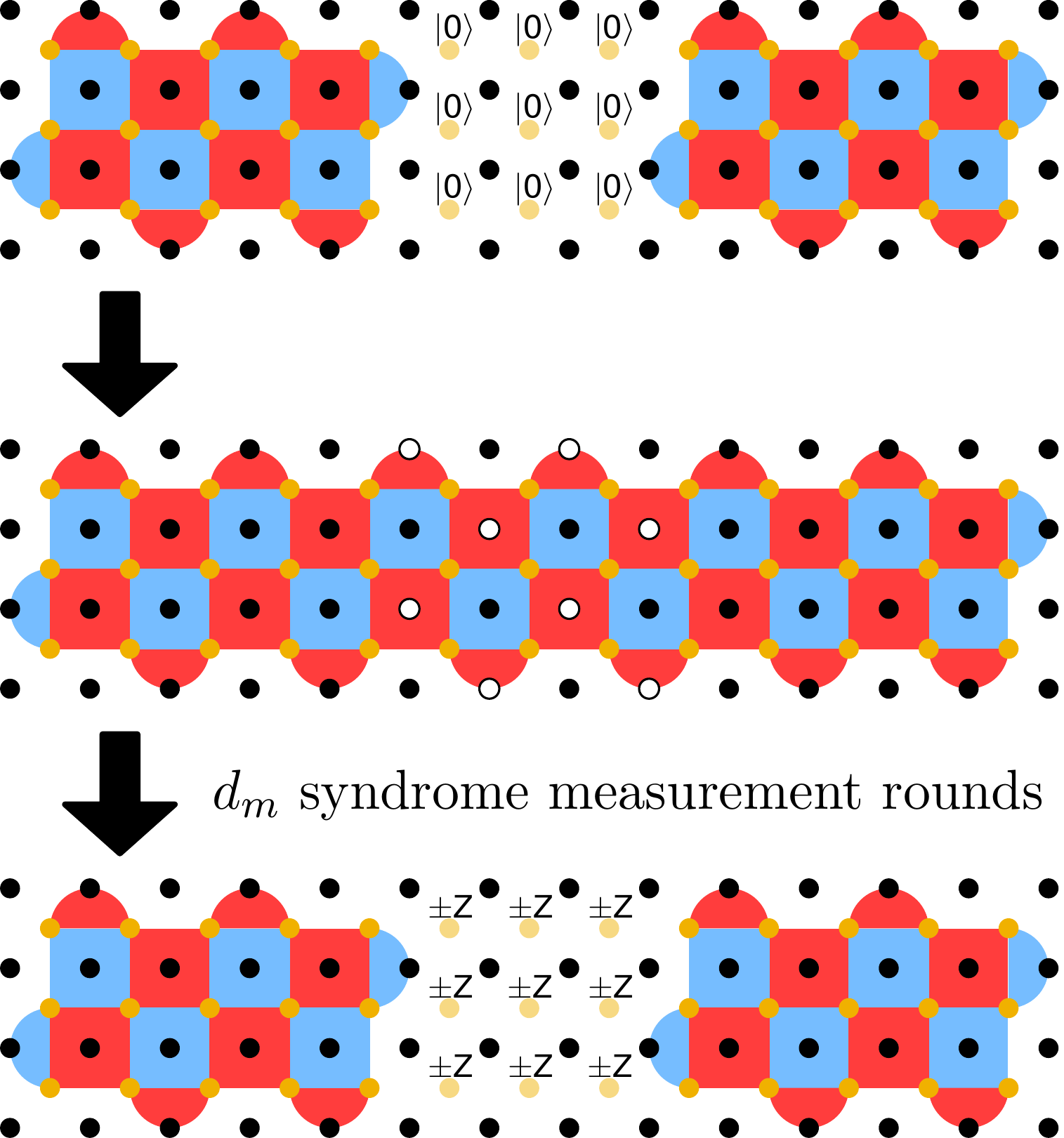}
    \caption{Lattice surgery implementation of an $X \otimes X$ measurement between two logical qubits encoded in $d_x=3$, $d_z=5$ surface code patches. Note that $X$ ($Z$) stabilizers are represented by red (blue) plaquettes. Prior to measuring $X \otimes X$, yellow data qubits in the routing region are prepared in the $\ket 0$ state. The $X \otimes X$ measurement outcome is then obtained by measuring the $X$-stabilizers (shown with white ancillas) in the routing space. The stabilizers of the merged surface code patch are measured for $d_m$ syndrome measurement rounds in order to correct timelike failures which can occur in the first round of the merge resulting in the wrong parity of $X \otimes X$. In the first syndrome measurement round of the merged patch, the individual measurement outcomes of $X$ stabilizers in the routing space region are random, but their product gives the result of the $X \otimes X$ measurement outcome. At the end of the $d_m$ syndrome measurement rounds, the data qubits in the routing space are measured in the $Z$ basis. Since measurement and reset of qubits typically takes a much longer time than the implementation of the physical CNOT gates used to measure the stabilizers, in this work we assume that the qubits in the routing space used to measure $X \otimes X$ are only available one syndrome measurement round after the split. Hence the merge/split operation takes a total of $d_m+1$ syndrome measurement rounds. }
    \label{fig:latsurg}
\end{figure}

In topological codes such as the surface code and color code, lattice surgery is the dominant mechanism used to perform multi-qubit Pauli measurements~\cite{Landahl14,Horsman12}. In \cref{fig:latsurg}, we show how to perform the $X\otimes X$ lattice surgery measurement between two logical qubits encoded in the surface code. Prior to measuring $X \otimes X$, the data qubits in the routing space are initialized in the $\ket 0$ state. Subsequently, the $X \otimes X$ measurement outcome is obtained by measuring $X$ operators in the routing space region shown in \cref{fig:latsurg} (note that a minimum-weight representative of the logical $X$ operator of the surface code is given by the product of $X$ operators along the vertical boundaries of the patch). The measurement outcomes of the $X$ stabilizers in the routing space are random (such stabilizers are illustrated with white ancillas in \cref{fig:latsurg}). However the product of all such stabilizers give the parity of the $X \otimes X$ measurement outcome. The stabilizers of the merged surface code patch are measured for $d_m$ rounds after which the surface code patches are split by measuring the qubits located in the routing space in the $Z$ basis. A logical timelike failure occurs during a lattice surgery protocol when the wrong parity of the multi-qubit Pauli measurement is obtained. Note that the logical timelike failure rate is exponentially suppressed with the number of syndrome measurement rounds $d_m$ (see Ref.~\cite{Chamberland22}).

All numerical results in this work are obtained using the following biased circuit-level noise model:
\begin{enumerate}
    \item Each single-qubit gate location is followed by a Pauli $Z$ error with probability $\frac{p}{3}$ and Pauli $X$ and $Y$ errors each with probability $\frac{p}{3 \eta}$.
	\item Each two-qubit gate is followed by a $\{ Z \otimes I, I \otimes Z, Z \otimes Z \}$ error with probability $p/15$ each, and a $\{ X \otimes I, I \otimes X, X \otimes X, Z \otimes X, Y \otimes I, Y \otimes X, I \otimes Y, Y \otimes Z, X \otimes Z, Z \otimes Y, X \otimes Y, Y \otimes Y \}$ each with probability $\frac{p}{15 \eta}$.
	\item With probability $\frac{2p}{3 \eta}$, the preparation of the $\ket{0}$ state is replaced by $\ket{1}=X\ket{0}$. Similarly, with probability $\frac{2p}{3}$, the preparation of the $\ket{+}$ state is replaced by $\ket{-}=Z\ket{+}$.
	\item With probability $\frac{2p}{3 \eta}$, a single-qubit $Z$ basis measurement outcome is flipped. With probability $\frac{2p}{3}$, a single-qubit $X$-basis measurement outcome is flipped.
	\item Lastly, each idle gate location is followed by a Pauli $Z$ with probability $\frac{p}{3}$, and a $\{ X, Y \}$ error each with probability $\frac{p}{3 \eta}$.
\end{enumerate}
We also set $\eta =100$. Using the above noise model and a minimum-weight perfect matching decoder, in Ref.~\cite{Chamberland22} it was shown that the timelike logical failure rate of an $X \otimes X$ multi-qubit Pauli measurement with $d_m$ syndrome measurement rounds is given by
\begin{equation}
    p_m(d_m) = 0.01634 A (21.93p)^{(d_m+1)/2},
    \label{eq:TimelikePoly}
\end{equation}
when $p$ is below threshold. In \cref{eq:TimelikePoly}, $A$ corresponds to the area of the routing space connecting the various surface code patches that take part in a multi-qubit Pauli measurement. In \cref{sec:NewTELS}, we set $A = 100$ in order to directly compare our results to those in Ref.~\cite{Chamberland22}. Further, we also use \cref{eq:TimelikePoly} when considering multi-qubit Pauli measurements containing $Z$ and $Y$ terms. 

Depending on the accuracy requirements of the algorithm, the target logical error rate per Pauli $\delta$ sets an upper bound on the maximum tolerable noise. This condition is $p_m < \delta$. To achieve low $p_m$, the measurement distance $d_m$ must be accordingly increased. Surface codes with $X$ and $Z$ boundaries can perform $X \otimes X$, $Z \otimes Z$ and $Z \otimes X$ measurements similar to \cref{fig:latsurg}. However to access the $Y$ boundary, twist defects are required. They have been studied extensively in surface codes~\cite{Horsman12,Litinski19,Chamberland22b} and in color codes~\cite{Gowda21}. We use the methods of Ref.~\cite{Chamberland22b} to implement $Y$-type measurements.
As shown in \cref{fig:latsurg}, lattice surgery requires routing space between different logical qubit patches. In the interest of reducing routing space for lattice surgery, Ref.~\cite{Herr19} considered extremely thin data busses that were only one qubit wide. The caveat of this method is that the measurements need to be performed for $d_m^2$ syndrome measurement rounds, instead of $d_m$ in the regular case. In this manuscript we consider routing space of dimensions which are functions of $d_x$ and $d_z$, which are the $X$ and $Z$ distances of the code respectively. For instance, the core of a quantum processor would have a layout given by Fig.~14 (e) in Ref.~\cite{Chamberland22}.

\subsection{Temporally encoded lattice surgery} 
\label{subsec:TELSreview}

It has been shown that encoding the measurement results of lattice surgery in an error-detecting code allows one to effectively reduce the measurement distance $d_m$ of each multi-qubit Pauli measurement~\cite{Chamberland22}. In particular, the Pauli operators of a given PP set $\mathcal{P} = \{ P_{t+1}, P_{t+2}, \cdots, P_{t+k} \}$ can be replaced by a new set $\mathcal{S} = \{Q[\mathbf{x}^1], Q[\mathbf{x}^2], \cdots, Q[\mathbf{x}^n] \}$ where 
\begin{align}
    Q[\mathbf{x}] = \prod_{j=1}^{k} P_{t+j}^{x_j} \,,
\end{align}
and $\mathbf{x}$ is a binary vector of length $k$. Such a replacement is allowed since $\langle \mathcal{S} \rangle = \langle \mathcal{P} \rangle$. In this encoding, the vectors $\{\mathbf{x}^1, \mathbf{x}^2, \cdots, \mathbf{x}^n \}$ form the columns of the generator matrix $G$ of some classical code $\mathcal{C}$, where the rows of $G$ are the codewords of $\mathcal{C}$. Note that the encoding of a TELS protocol takes place entirely in the time domain since additional multi-qubit Pauli measurements are performed without requiring additional qubits. By multiplying the measurement outcomes of all operators in $\mathcal{S}$ by the parity check matrix of $\mathcal{C}$, timelike lattice surgery failures will be detected if the result is equal to 1 instead of 0. 

In \cref{fig:TELSprotocol}, we show a TELS protocol used in Ref.~\cite{Chamberland22} for a PP set given by $\mathcal{P} = \{P_1, P_2 \}$, and with $\mathcal{S} = \{P_1, P_2, P_1P_2 \}$ (so that $k=2$ and $n=3$). The generator matrix for the set $\mathcal{S}$ is given by
\begin{align}
    G = \begin{bmatrix} 
101 \\
011
\end{bmatrix},
\end{align}
with the rows of $G$ generating codewords of the classical $[3,2,2]$ code. Note that in a general-purpose quantum computer, the algorithms that are executed may have different sizes of PP sets. Depending on the size of the PP set, a suitable classical code would be chosen (the classical code chosen for a PP set of size-2 in \cref{fig:TELSprotocol} was chosen due to its simplicity). 

When measuring the operators in $\mathcal{S}$, the multi-qubit Pauli operators are measured using $d_m'<d_m$ syndrome measurement rounds during the merge step of the lattice surgery protocol (since the ability to detect logical timelike failures allows for noisier lattice surgery operations). In the example of \cref{fig:TELSprotocol}, if the measurement outcome of $P_1P_2$ is inconsistent with the measurement outcome of $P_1$ and $P_2$, a logical timelike failure has been detected. If timelike logical failures are not detected, Clifford corrections are applied based on the measurement results of the $k$ original Paulis, which are obtained by decoding the $n$ measurement bits according to a classical code (here, the [3,2,2] code). Note however that if two logical timelike failures occur, the protocol in \cref{fig:TELSprotocol} will be unable to detect such a timelike failure (this can also be seen by noting that the distance of the classical code $\mathcal{C}$ is 2). 

If a logical timelike failure during the TELS protocol is detected with probability $p_D$, the original $k$ Paulis are measured with measurement distance $q d_m$, where $q$ is a constant that can be optimized to reduce the overall measurement distance. This offsets the fact that the remeasure round only occurs with probability $p_D$. Note that if the Paulis in the original set $\mathcal{P}$ are measured (due to the detection of a logical timelike failure), timelike failures cannot be detected since only the original $k$ Paulis are measured. The circuits used in the TELS protocol described in this section are performed in the Clifford frame, hence any Clifford corrections that may be required can be conjugated through to the end of the circuit.

To perform a TELS protocol involving Clifford or non-Clifford gates that use resource states,  hardware must be allocated to hold the resource states in memory for the entire TELS protocol. They cannot, in general, be measured out after each Pauli measurement. For example, in \cref{fig:TELSprotocol} the $\ket T$ states used to measure $P_1$ and $P_2$ need to be held in memory until the end of the protocol. Note that a common approach for designing the architecture for a quantum computer is to use a central processing unit (also referred to as a ``core'') and a set of distillation factories. If a quantum computer is built according to this model, the use of TELS results in a small additive factor to the total number of logical qubits needed instead of a multiplicative factor. This does not result in a substantial increase to the space-time cost of implementing algorithms on such an architecture, as in general the algorithms runtime is reduced by a multiplicative factor of between $2 \times$ and $5 \times$, as we show in this manuscript. In contrast, in Ref.~\cite{Litinski19}, Litinski showed that a runtime reduction of approximately $2 \times$ required a $6 \times$ increase in qubit overhead cost. 

We now calculate the total time taken by a TELS protocol to execute a PP set of size $k$, using an $[n,k,d]$ classical code $\mathcal C$. To preface, in a regular lattice surgery protocol, each measurement takes time $(d_m + 1)$, for a total time of $k(d_m + 1)$. Following the TELS protocol described in this section, the total time taken to measure all Paulis in a PP set using TELS is, on average,
\begin{equation}
T_{\text{old}} = n(d_m' + 1) + p_D k (\lceil q d_m \rceil + 1) \: ,
 \label{eq:OldTELStime}
\end{equation}
where the second term is due to the contribution from measuring the Paulis of $\mathcal{P}$ if timelike logical failures are detected. We use the subscript `old' to refer to the TELS protocol described in this section, since in \cref{sec:NewTELS}, the remeasure part of the TELS protocol has a different time cost. Similarly for this protocol, the logical error rate per Pauli is the sum of the logical error rate contributions of the temporally encoded set and the remeasure set,
\begin{align}
p_L = &\frac{1-p_D}{k} \sum_{i=d}^{n} l_i (p_m(d_m'))^i  (1-p_m(d_m'))^{n-i} \nonumber \\
 &\quad+ p_D p_m(\lceil q d_m \rceil) \nonumber \\
 \approx &\frac{1-p_D}{k} l_d (p_m(d_m'))^d (1-p_m(d_m'))^{n-d} \nonumber \\
 &\quad + p_D p_m(\lceil q d_m \rceil) \: ,
 \label{eq:PLfirst}
\end{align}
where $l_i$ is the number of weight-$i$ timelike logical failures that cause trivial syndromes when multiplying the lattice surgery measurement outcomes by the parity check matrix of the classical code $\mathcal{C}$. The variable $p_m(d_m')$ is the probability of a logical timelike failure of a single lattice surgery measurement with measurement distance $d_m'$ (obtained from \cref{eq:TimelikePoly}),
and $p_m(\lceil q d_m \rceil)$ is the probability of a logical timelike failure in the remeasure round, where the measurement distance is $\lceil q d_m \rceil$. Note that the timelike logical error rate is due to wrong Clifford corrections being applied after the non-Clifford gates. Hence if TELS fails without detecting a timelike error, the probability that there is a logical error is scaled by $1-p_D$.
There are various ways that the $l_i$ term in \cref{eq:PLfirst} can be calculated. In \cref{subapp:monte,subapp:bernoulli}, we show how $l_i$ is calculated using sampling methods, where the computational complexity of sampling grows with `$d$`. In \cref{subapp:macwilliams}, we show how the $l_i$ coefficients can be computed determinstically using MacWilliams identities. The advantage of using the MacWilliams identities is that the computational complexity only grows exponentially with $k$ as opposed to $d$.

Finally, the probability that an error is detected during the first stage is 
\begin{equation}
p_D \geq \sum_{i=1}^{d-1} \binom{n}{i} (p_m(d_m'))^i (1-p_m(d_m'))^{n-i}.
\label{eq:PDgreaterThan}
\end{equation}
Note that in \cref{eq:PDgreaterThan}, we use the $\geq$ sign since some sets of $\geq d$ logical timelike failures will also be detected.

\section{New TELS encoding protocol}
\label{sec:NewTELS}

\subsection{Improvements arising from repeated temporally encoded measurements}
\label{subsec:MSDTELSnew}

In this section we describe an improved TELS protocol compared to the one described in \cref{subsec:TELSreview}. An example of the application of the new protocol is given in \cref{fig:TELSprotocolnew}. The main difference is that if a logical timelike failure is detected when measuring the Paulis in the set $\mathcal{S}$, operators in $\mathcal{S}$ (instead of $\mathcal{P})$ are measured anew. In particular, operators in $\mathcal{S}$ are repeatedly measured until no logical timelike failures are detected. Only at this point can we determine the Clifford corrections that must be applied. 

Although it is clear that the protocol described above is different from the protocol described in \cref{subsec:TELSreview}, it is not clear that the new protocol takes fewer syndrome measurement rounds. First, let us determine the time taken by the protocol described above, which is given by
\begin{align}
    T_{\text{new}} = & n (d_m' + 1) + p_D T_{\text{new}} \nonumber \\
     = & \frac{n (d_m' + 1)}{1-p_D} \: ,
     \label{eq:TEnew}
\end{align}
since the lattice surgery implementation of each multi-qubit Pauli measurement in $\mathcal{S}$ is performed using $d_m'$ syndrome measurement rounds during the merge step.

To compare \cref{eq:TEnew} to the protocol described in \cref{subsec:TELSreview}, we can rewrite \cref{eq:OldTELStime} as
\begin{align}
     T_{\text{old}} = & n (d_m' + 1) + p_D u n (d_m' + 1) \nonumber \\
      = & n (d_m' + 1) (1 + p_D u) \: ,
\end{align}
where $u = \frac{k (\lceil q d_m \rceil + 1)}{n (d_m' + 1)}$.

If $\frac{1}{1-p_D} < (1 + p_D u)$, the revised protocol is more efficient than the one described in \cref{subsec:TELSreview}.
Simplifying this, we can get a constraint on $u$,
\begin{equation}
    u > \frac{1}{1-p_D} \: .
    \label{eq:ucheck}
\end{equation}

We computed the value of $u$ for all the classical codes considered in this paper that were used in a TELS protocol. If the condition in \cref{eq:ucheck} was satisfied (which was true for all classical codes considered in the work except the distance-$2$ Single Error Detect code (see \cref{subsubsec:SED})), we used the new protocol for TELS. An interesting note is that the time difference between the old protocol and the new is only exacerbated when $p_D$ is larger. 

The probability of a logical error in this new protocol is the probability that an undetectable set of measurement errors occurs during the last (or successful) iteration of temporally encoded lattice surgery. Let $p'_L \equiv \frac{1}{k} \sum_{i=d}^{n} l_i \, (p_m(d_m'))^i (1-p_m(d_m'))^{n-i}$, which corresponds to the probability (per Pauli) that a series of timelike logical failures during the execution of the measurements in $\mathcal{S}$ results in a trivial syndrome when multiplied by the parity check matrix of the classical code $\mathcal{C}$. The TELS protocol failure probability is then given by the following equation
\begin{align}
p_L &= (1-p_D)p'_L + p_D(1-p_D)p'_L + p_D^2(1-p_D)p'_L + \cdots \nonumber \\
 &= (1-p_D)p'_L(1 + p_D + p_D^2 + \cdots) \nonumber \\
 &= p'_L \nonumber \\
 &= \frac{1}{k} \sum_{i=d}^{n} l_i \, (p_m(d_m'))^i (1-p_m(d_m'))^{n-i} \: .
\end{align}

\subsection{Improvements from correcting classical errors}
\label{subsec:MSDTELScorrect}

In Ref.~\cite{Chamberland22}, the classical codes used in a TELS protocol were for a pure error detection scheme (as described in \cref{subsec:TELSreview}). Further, it was argued that performing error correction using TELS would result in worse speedups compared to a pure error detection scheme. In this section we show that performing a hybrid scheme using TELS, where errors of low weight are corrected and errors of higher weight are detected, can result in further performance improvements compared to a pure error detection scheme. The overall effect is to reduce the average time per Pauli measurement, while staying under the logical error rate threshold set by $\delta$. This effect is more dominant when using classical codes with high $k$ and $d$, and at higher tolerable noise rates $\delta$.  At larger $\delta$, it is possible to correct classical errors of higher weight than for smaller $\delta$ since there is more of an error budget and $p_L$ can be increased to match $\delta$. In addition, $p_D$ tends to be higher when $d_m$ is small, and $d_m$ is smallest at large delta. As we show in \cref{tab:ECED}, the benefit of using error correction is that  $p_D$ can be made smaller and that in turn allows smaller average runtimes per Pauli.

When using a distance-$d$ classical code in an error detection scheme, an error is detected with probability $p_D \approx O(p_m(d_m'))$. When using classical codes of high distance, the target logical timelike failure rate of a lattice surgery protocol may be achieved using measurement distances $d_m$ close to 1. However, with such small measurement distances, $p_m(d_m')$ is inevitably large, and so~is~$p_D$.

% \pagebreak %DEBUG

If we instead use the classical code to correct all errors up to some weight $c$, a detection event is only triggered by errors of weight at least $c+1$. In other words, $p_D \approx O\big ((p_m(d_m'))^{c+1} \big)$. A lower value of $p_D$ thus requires fewer repeated measurements of the Paulis in the set $\mathcal{S}$. Although time is now saved by performing fewer remeasure rounds, the logical error rate increases relative to a pure error detection scheme. In a pure error detection scheme, the logical error rate scales as $O \big ( (p_m(d_m'))^d \big )$. However, some errors of weight $d-c$ will have the same syndrome as errors of weight-$c$, since they may differ by a logical operator. The most probable correction for this syndrome is the weight-$c$ error, but applying this correction yields a logical error with probability $O\big ((p_m(d_m'))^{d-c}\big )$. Consequently, using the classical code in a TELS protocol to correct low-weight errors leads to a tradeoff between the logical error rate of the encoded measurements and the time saved due to fewer remeasure rounds. 
Such tradeoffs have also been considered for magic state distillation schemes (see for instance Ref.~\cite{Haah18}).

By incorporating the ability to correct errors up to weight $c$, the probability of observing an uncorrectable but detectable error becomes
\begin{equation}
p_D \geq \sum_{i=c+1}^{d-c-1} \binom{n}{i} (p_m(d_m'))^i (1-p_m(d_m'))^{n-i} \: ,
\end{equation}
which is significantly smaller than in \cref{eq:PDgreaterThan}.
The probability of a logical error per Pauli due to the TELS protocol described in this section is
\begin{align}
p_L =& \frac{1}{k} ( \sum_{i=d-c}^{d-1} l_i ( p_m(d_m'))^i (1-p_m(d_m'))^{n-i} + \nonumber \\ 
& \sum_{i=d}^{n} l_i (p_m(d_m'))^i (1-p_m(d_m'))^{n-i}) \nonumber \\
 \approx & \frac{1}{k} l_{d-c} ( p_m(d_m'))^{d-c} (1-p_m(d_m'))^{n-d+c} \: .
\label{eq:p_LtotECED}
\end{align}
Note that in \cref{eq:p_LtotECED}, $l_i$ in the second term of the sum includes contributions from both undetectable errors and errors of weight greater than $d$ which have the same syndrome as correctable errors. We only include the leading order term since higher order terms have very small contributions. Note however that for classical codes with very large values of $d$, a larger value of $p_m(d_m')$ may be tolerated. In such cases, including higher order terms in \cref{eq:p_LtotECED} may be required.

To show the benefits of including error correction in a TELS protocol, we consider TELS with a classical $[127,92,11]$ BCH code (see \cref{subsubsec:BCH}) at a physical error rate of $10^{-4}$ and a target logical error rate of $\delta = 10^{-15}$ per Pauli. In \cref{tab:ECED}, we show that by correcting errors up to weight three, it is possible to reduce the average time per Pauli measurement by $36 \%$ .
 
\setlength{\tabcolsep}{8pt}
\begin{table}
    \centering
    \begin{tabular}{c c c c c c}
    $d_m$ & $c$ & $p_{L}$ & $p_D$ & $T$  & $T/k$ \\
    \hline
    $1$ & $0$ & $1 \times 10^{-24}$ & $0.37$ & $400.71$  & $4.36$  \\[.2cm]
    $1$ & $1$ & $2.5 \times 10^{-21}$ & $0.076$ & $275.08$ & $3$  \\
    $1$ & $2$ & $1 \times 10^{-18}$ & $0.011$ & $256.83$ & $2.79$  \\
    $1$ & $3$ & $2.5 \times 10^{-16}$ & $0.0012$ & $254.3$  & $2.76$  \\[.2cm]
    $11$ & $-$ & $1.8 \times 10^{-16}$ & $-$ & $1104$ & $12$
    \end{tabular}
    \caption{Comparison between the performance of a TELS protocol implemented using pure error detection versus protocols implemented using combined error detection and correction with the $[127,92,11]$ BCH code. The last line of the table shows the performance of unencoded lattice surgery. Here, $d_m$ is the measurement used when measuring multi-qubit Pauli operators using lattice surgery, and $c$ is the maximum weight of errors that can be corrected by the classical code used in the TELS protocol. The objective is to minimize the average time taken per Pauli measurement (last column) while ensuring the logical error rate is less than $10^{-15}$ per Pauli. The logical error rate per Pauli $p_L$ is calculated using \cref{eq:p_LtotECED}, where the routing space area is $A=100$. The results in the first row are for the TELS protocol implemented using pure error detection. By correcting weight-one, -two and -three errors, the average measurement time is reduced to $2.76$ syndrome measurement rounds, as opposed to $4.36$ when using TELS with pure error detection, or $12$ without TELS. Note that a pure error detection scheme with $d_m \ge 3$ results in a larger runtime than those obtained in the first four rows of this table since the total number of syndrome measurement rounds will be at least $127 \times 4 = 508$.}
    \label{tab:ECED}
\end{table}

% \pagebreak %DEBUG
\subsection{Protocols for PP sets of size up to one hundred}
\label{subsec:ManyCodes}

When performing temporally encoded lattice surgery, it is unclear which classical $[n,k,d]$ code will give the best speedup. \cref{eq:TEnew} shows that the total time to measure all the multi-qubit Pauli operators using a TELS protocol is clearly proportional to $n$. Also, the time is inversely proportional to the distance of the classical code used since a higher value of $d$ results in a lower value of $d_m'$. However the contribution arising from $p_D$ needs to be determined, as well as the tradeoffs between a pure error detection scheme and an correction + detection scheme. 

Ref.~\cite{Chamberland22} evaluated the speedups arising from using TELS protocols at a physical error rate of $p=10^{-3}$, $\delta = 10^{-15}$ and $A=100$ (see \cref{eq:TimelikePoly}). Three classical codes were considered, and the speedups arising from each code were computed. In the analysis, the distance-$4$ Extended Hamming code resulted in the best performance improvements.
For the TELS protocols considered in this work, it was unclear at the outset which classical codes gave the best performance improvements. As such, we collected results for an expanded list of classical codes, which are listed in \cref{tab:codes}. Note that in our analysis, we only consider odd values of $d_m'$, as \cref{eq:TimelikePoly} only applies to odd measurement distances. Since the calculation of the number of malignant fault sets is computationally expensive, we could not consider classical codes of distances higher than those shown in \cref{tab:codes}. The code constructions for the codes in \cref{tab:codes} are provided in \cref{app:codeconstruction}.

\begin{table}
    \centering
    \begin{tabular}{c c c}
        Code / Code family  &  Distances \\
        \hline
        Single Error Detect(SED)  & $2$\\
        Hamming (Hamm) & $3$ \\
        Concatenated SED (CSED) & $4$ \\
        Extended Hamming (EHamm) & $4$ \\
        Golay (Gola) & $7$ \\
        Extended Golay (EGol) & $8$ \\
        Doubly Concatenated SED (DCSED) & $8$ \\
        Reed-Muller (RM) & $8$ \\
        Polar (Pol) & $4,8$ \\
        Zetterberg (Zett) & $5,6$ \\
        Bose–Chaudhuri–Hocquenghem (BCH) & $3,5,7,9,11$ 
    \end{tabular}
    \caption{Classical codes (and their associated distances) used in the TELS protocols considered in this work. In different noise regimes, some codes perform better than others, as we show in \cref{fig:delt10}, \cref{fig:p3delt152025} and \cref{fig:p4delt1520}. We provide explicit constructions of the best performing codes in \cref{app:codeconstruction}.}
    \label{tab:codes}
\end{table}

\begin{figure}
    \centering
    \includegraphics[width=.98\textwidth]{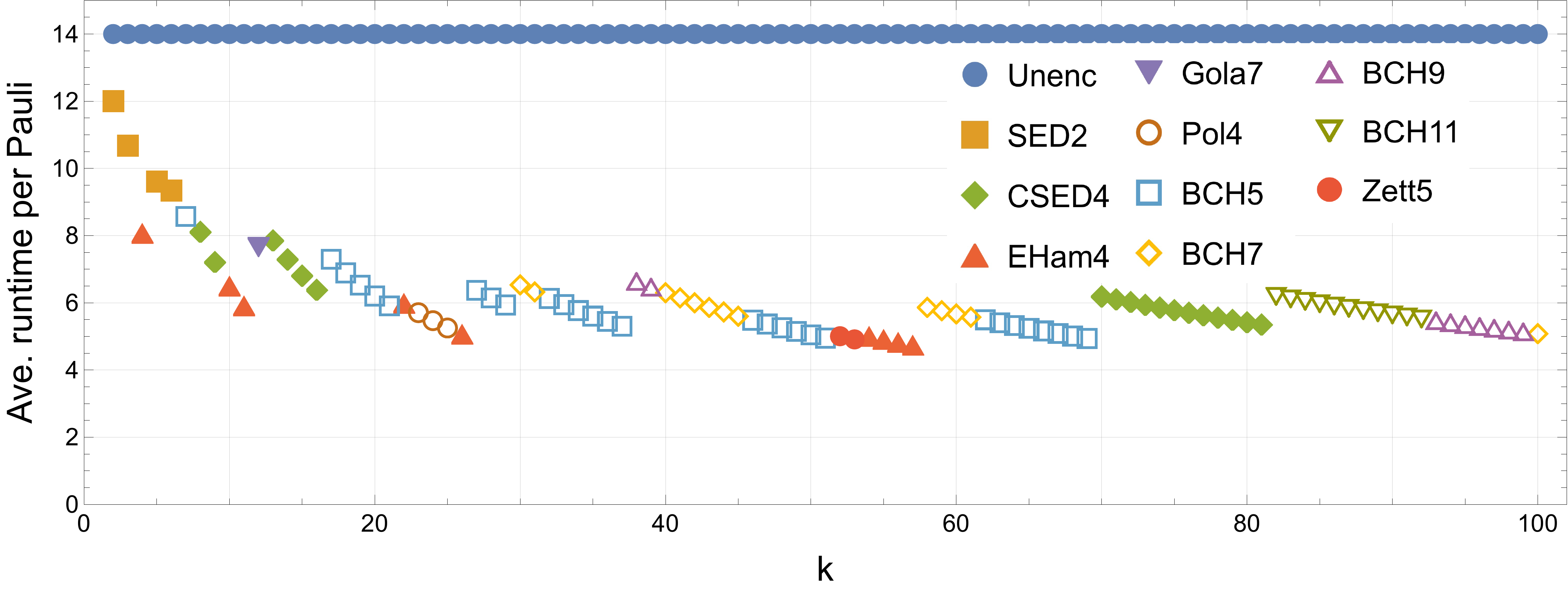}
    \caption{The best average runtime per Pauli (in units of syndrome measurement rounds) for all classical codes considered in this work when using the TELS protocol of \cref{sec:NewTELS} for PP sets of size $k \in \{ 2,3, \mathellipsis, 100\}$ at $p=10^{-3}$, $\delta = 10^{-10}$ and with maximum routing space $A = 100$. For example, for PP sets of size $k=20$, a distance-$5$ BCH code achieves the lowest average runtime per Pauli among all the classical codes considered. To calculate $p_L$, we used \cref{eq:p_LtotECED}, with $p_m$ given by \cref{eq:TimelikePoly} and $d_m$ chosen to minimize the runtime while keeping $p_L<\delta $. We compare the results of TELS with un-encoded lattice surgery, which is shown here to take $14$ syndrome measurement rounds per Pauli. The legend labels correspond to codes from \cref{tab:codes}. Low distance codes perform better for small values of $k$, whereas for larger values of $k$, the high rate of larger-distance codes enables smaller measurement distances.}
    \label{fig:delt10}
\end{figure}

For $k \in \{ 2,3, \mathellipsis 100\}$, the lowest achievable average time per Pauli measurement for TELS protocols described in this section and which use the codes of \cref{tab:codes} are shown in \cref{fig:delt10}. The average runtime per Pauli is $T_{\text{new}} / k$ where $T_{\text{new}}$ is given by \cref{eq:TEnew}. Such a calculation allows us to compare the performance of a TELS protocol relative to the time taken to measure multi-qubit Pauli operators in the original PP set $\mathcal{P}$ without using TELS. Our results are obtained using the parameters $p=10^{-3}$ and $\delta = 10^{-10}$. Other regimes are considered in \cref{app:speedups}. In particular, results are obtained for the parameters $\delta \in \{ 10^{-15}, 10^{-20}, 10^{-25}\}$ with $p=10^{-3}$ and $\delta \in \{  10^{-15}, 10^{-20}\}$ with $p=10^{-4}$. 
We also show the average runtimes per Pauli for unencoded lattice surgery protocols. 

With a classical code that is defined to encode $k$ logical bits, we can implement TELS protocols for PP sets of any size \textbf{up to} $k$. This allows us to use good classical codes for many different sizes of PP sets. If the PP set is of size $k-j$, the Paulis associated with the remaining $j$ logical bits of the code are simply set to the identity. When decoding the temporally encoded measurement results, information corresponding to the extra logical bits is thrown away.

The biggest insight from our findings is that small codes with low distances perform well at small values of $k$, whereas at larger values of $k$, the larger and higher distance codes perform better. Moreover, we notice that for values of $k$ higher than $30$, codes from the BCH family generally give the largest speedups. Of course, the list of codes considered in \cref{tab:codes} is not all-encompassing. There may be many codes that perform better than the codes considered here. To find codes that work well for a specific value of $k$, the primary target should be to ensure that the rate of the code ($\tfrac{k}{n}$) is not too low. In our observations, codes with rate less than $1/2$ generally did not give the best speedups. The second biggest consideration is the distance of the classical code. High-distance classical codes admit very low values of $d_m'$ at the cost of much higher probabilities of detecting a logical timelike failure. Error correction for smaller weight errors can be used to reduce the detection rate as long as the logical failure rate per Pauli is below the target rate $\delta$.

% \subsection{Time analysis for parallelizable Pauli sets of size up to 100}
% \label{subsec:MSDTELsk100}

\chapter{TELS for improving magic state fidelity}
\label{sec:MSDMSD}

Quantum gates can be classified into those that can be efficiently simulated by classical computers (such as Clifford gates), and those that cannot~\cite{Aaronson04}. As discussed in the introduction to \cref{chap:MSD}, a universal gate set can be achieved by combining Clifford gates with at least one non-Clifford gate. However, the implementation of logical non-Clifford gates with topological codes is not as straightforward as implementing gates from the Clifford group. Examples of non-Clifford gates include rotations about one of the Bloch sphere axes by angles of $j\pi/8$ where $j\in \{ 1,1/2,1/4, \mathellipsis\}$  or multiply-controlled $Z$ gates $C^jZ$ where $j\geq 2$. One common method used to implement non-Clifford gates is to use magic states as resource states along with stabilizer operations to perform gate teleportation.

A logical magic state can be created by first preparing a physical magic state, followed by a gauge fixing step which encodes the state into a logical qubit patch such as the surface code~\cite{Vuillot19}. However such operations are not fault-tolerant and lead to encoded magic states with unacceptably high physical noise rates. To get high fidelity magic states, the prepared magic states are then injected into a magic state distillation protocol where a quantum error-correcting code uses stabilizer operations to detect logical failures present on the injected magic states. Such protocols can be concatenated to achieve any desired target logical failure rate.

\subsection*{Magic state injection}
\label{subsec:MSDMSDinjection}

Here we briefly describe various magic state injection methods used to prepare magic states prior to performing a distillation protocol. Ref.~\cite{Lao22} considers a state injection protocol on the rotated surface code afflicted by a standard depolarizing noise model. On the other hand, Ref.~\cite{Singh22} considers an injection protocol under a biased depolarizing noise model with $\eta = 10^3$ and $\eta = 10^4$. Further, the magic states are prepared in the $XZZX$ code rather than the rotated surface code. In the implementation of the protocol, the magic states are initialized into an effective two-qubit repetition code using low-error $ZZ$ rotations, and the stabilizers of this code are measured twice. Afterwards, the error detecting code is merged into the final $XZZX$ code, where $d_m$ syndrome measurement rounds are performed. The prepared magic states are then shown to have failure probability $O(p^2)$. In this work, we conjecture that a similar injection protocol can be devised using rectangular surface codes in the presence of biased noise. However, we consider the entire protocol to require only two syndrome measurement rounds since the $d_m$ syndrome measurement used after merging the error detecting code with the final code can be part of the lattice surgery operations used in the magic state distillation schemes described in \cref{subsec:MagicDist}. 

Recall that the noise model used in this work has bias $\eta = 10^2$, as opposed to $\eta = 10^3$ or $\eta = 10^4$ used in Ref.~\cite{Singh22}. Further, in Fig. 3 of Ref.~\cite{Singh22}, it can be seen that for values of the physical noise rate parameter $p \le 10^{-3}$, the injected magic states have logical error rate less than $p$ (with a quadratic scaling as a function of $p$). Since we use a rectangular surface code with bias $\eta = 10^2$, we take the injected magic states to have a logical error rate 
\begin{align}
    \epsilon_{\text{L,X}} = \epsilon_{\text{L,Y}} = \frac{p}{3 \eta},
    \epsilon_{\text{L,Z}} = \frac{p}{3},
    \label{eq:InjectProbs}
\end{align}
where $\epsilon_{\text{L,P}}$ is the probability of a logical Pauli $P$ error when injecting the prepared magic state. Note that the previous expression may be optimistic depending on the hardware implementation of the $ZZ(\theta)$ rotation used in Ref.~\cite{Singh22} and given that our noise bias $\eta = 100$ is lower than what was considered in Ref.~\cite{Singh22}. On the other hand, if the hardware architecture allows for the implementation of the $ZZ(\theta)$ as in Ref.~\cite{Singh22} and values of $p$ in the range $10^{-4} \le p \le 10^{-3}$ are below threshold at $\eta = 100$, then the expressions for $\epsilon_{\text{L,P}}$ in \cref{eq:InjectProbs} can be quite pessimistic. Note that if the underlying hardware architectures enables the implementation of the methods described in Refs.~\cite{chamberland2020very,SC22Modulo}, \cref{eq:InjectProbs} may be improved even further\footnote{Implementing the schemes described in Refs.~\cite{chamberland2020very,SC22Modulo} would result in a higher space-time cost compared to the other injection protocols described in this section, since color code patches would be used, and $\mathcal{O}(d)$ rounds of error detection would be required. However the injected magic states would have failure probabilities that scale as $\mathcal{O}(p^{(d+1)/2})$ instead of $\mathcal{O}(p)$ or $\mathcal{O}(p^2)$. Such low failure rates could then result in much smaller magic state distillation factories (see for instance Ref.~\cite{ChambsCat}). }. In what follows, we define 
\begin{align}
    \epsilon_L = \epsilon_{\text{L,X}} + \epsilon_{\text{L,Y}} + \epsilon_{\text{L,Z}} = \frac{p}{3} + \frac{2p}{3 \eta}.
    \label{eq:epsL}
\end{align}

Using the results of Fig.~$3$ of Ref.~\cite{Singh22}, we approximate the success probability of the state injection protocol to be greater than $99 \%$ for $p \leq 10^{-3}$. Hence, we take the time taken to prepare a magic state to be approximately $T_{\text{inj}} = 2/0.99$ syndrome measurement rounds for $p=10^{-3}$ and $T_{\text{inj}} = 2/0.995$ syndrome measurement rounds for $p=10^{-4}$. It is especially important to consider the time taken for injection since in TELS protocols, the time for each multi-qubit Pauli measurement can be as low as two syndrome measurement rounds (if we include the time required for ancilla reset). Note that the precise success probabilities require a more careful analysis, however we don't expect such analysis to have a significant effect~on~$T_{\text{inj}}$.

\begin{figure*}
    \centering
    \includegraphics[width=0.98\textwidth]{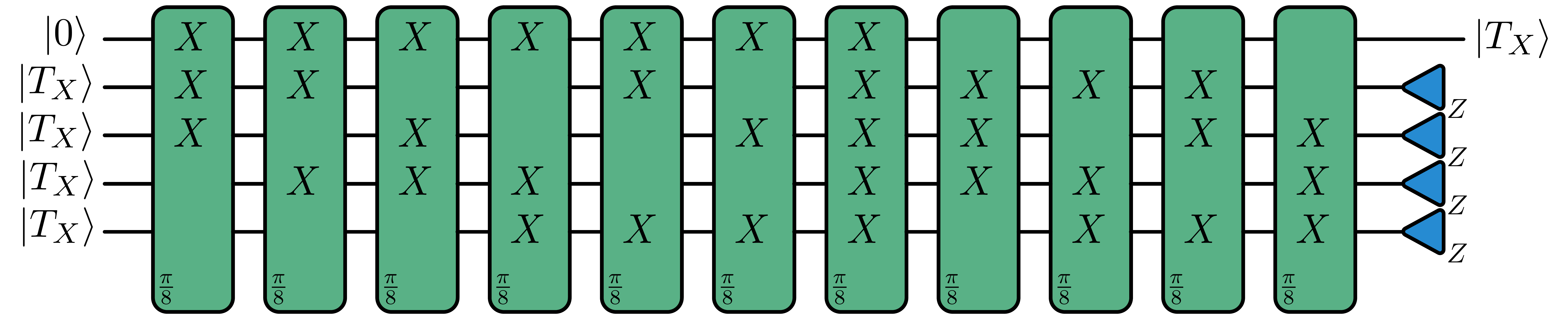}
    \caption{Circuit used in a $15$-to-$1$ magic state distillation protocol expressed as a sequence of multi-qubit non-Clifford gates. The circuit above is the Hadamard-transformed version of Fig.~15 of Ref.~\cite{Litinski19}, which produces the state $H \ket T$. This state can be used in the same way as the regular magic state, with only a change of Pauli basis while measuring.}
    \label{fig:Haddist}
\end{figure*}

\section{Magic state distillation}
\label{subsec:MagicDist}

A magic state distillation protocol takes several noisy magic states as input (which are encoded in some code such as the surface code), and by performing stabilizer operations which are part of an error detection protocol, yields fewer magic states of much higher fidelity. The output yield and magic state error probability depend entirely on the quantum code used for the error detection protocol. Bravyi and Haah suggested a class of distance-2 quantum codes for distilling $k$ magic states from $3k+8$ input magic states~\cite{Bravyi12}. These codes offer minimal protection, but have good yield. One of the protocols we consider in this work is the popular $15$-to-$1$ distillation protocol~\cite{Bravyi05}, which can be transformed into a series of $11$ commuting multi-qubit Pauli measurements~\cite{Litinski19}. These Pauli measurements form a PP set of size $11$, and thus a TELS protocol can be used to reduce the runtime required to measure each multi-qubit Pauli operator. In search of magic state distillation protocols with good output-to-input ratio, we found infinite families of protocols with near constant rate, where the focus is not on concatenating protocols but instead on measuring the stabilizers of an outer code using operations that are transversal for an inner code~\cite{Haah17}. Other techniques construct more complex codes for distillation by generalizing triorthogonal codes, or by puncturing Reed-Muller codes~\cite{Haah18}. In addition to the $\llbracket 15,1,3\rrbracket$ code, in this work we also consider several triorthogonal quantum codes that are constructed by puncturing a $[128,29,32]$ classical Reed Muller code, as described in Ref.~\cite{Haah18}. In particular, the distillation protocols we consider are derived from $\llbracket 114,14,3\rrbracket$, $\llbracket 116,12,4\rrbracket$ and a $\llbracket 125,3,5\rrbracket$ quantum codes. It is beneficial to consider quantum codes of different distances as this allows distilling magic states across a range of target logical error~probabilities.

Since our focus of applying TELS to magic state distillation is to reduce space-time costs, we comment on some previous work by Litinski~\cite{Litinski19,Litinski19magic}. First it was shown that the time costs of distillation algorithms can be reduced, however this leads to a disproportionate space increase, leading to overall larger space-time costs~\cite{Litinski19}. In further work, it was shown how to reduce space-time costs by performing the distillation using surface code patches of reduced distance, and using faulty $T$ measurements instead of $T$ state injection~\cite{Litinski19magic}. 
Improvements due to TELS protocols may also be applied to the work in Ref.~\cite{Litinski19magic}, allowing even smaller space-time costs than shown in \cref{tab:STcosts} of \cref{sec:MSDMSDdesign} (where we apply TELS to various distillation protocols and compute the space-time costs). The main objective of the results given in \cref{tab:STcosts} is to show that space-time costs can be reduced when using TELS as opposed to when there is no temporal encoding. In addition, we make a careful assessment of the time and space required for magic state injection, which makes the results in \cref{tab:STcosts} look more pessimistic when compared to results such as in Ref.~\cite{Litinski19,Litinski19magic}.

Traditional distillation algorithms, such as those in Ref.~\cite{Litinski19} perform only pure multi-qubit Pauli $Z$ measurements. In the remainder of this work, we apply Hadamard transformations to the distillation circuits to produce circuits consisting of pure multi-qubit $X$ measurements. Such a transformation reduces the space-time costs of the distillation factories, as only the shorter logical $X$ boundary of the asymmetric surface code patch will need to be accessed. The Hadamard-transformed version of the $15$-to-$1$ distillation circuit of Ref.~\cite{Litinski19} is shown in \cref{fig:Haddist}.

\subsection{Distillation in the Clifford frame}
\label{subsec:MSDMSDCliff}

To obtain the space-time costs of various magic state distillation protocols implemented with TELS, we first design an appropriate distillation protocol which outputs the desired magic state up to a Clifford correction. For simplicity, the general protocol can be separated into two steps. First the non-Clifford gates are applied using temporally encoded lattice surgery. If a non-trivial lattice surgery measurement failure is detected, all the physical qubits in the distillation tile are reset and the protocol restarts. We allow for the TELS protocol to also use the developments of \cref{subsec:MSDTELScorrect}, where classical errors of low weight may be corrected before signaling a lattice surgery measurement failure. Alternatively, if we follow the TELS protocol of \cref{subsec:MSDTELSnew}, more magic states would need to be simultaneously held in memory, and hence the hardware requirements would be larger. If TELS was successful, we are left with a distilled magic state up to a Clifford frame, prior to performing the single-qubit measurements. In the second part of the protocol, the Clifford frame is conjugated through the final single-qubit measurements. This changes the single-qubit measurements into multi-qubit $\pi/2$ Pauli measurements implemented via lattice surgery. Note that these measurements may now be tensor products of arbitrary Paulis and not just $\Id$ and $X$. These multi-qubit measurements may also be sped up using TELS, since they all commute. In particular, we use the TELS protocol of \cref{subsec:MSDTELSnew} for these final multi-qubit Pauli measurements.

After the Clifford frame is conjugated through the single-qubit measurements, the output distilled magic states are correct up to a Clifford correction. In fact, for the example in \cref{fig:Haddist}, the resulting state is exactly one of $\ket{T_X},X_{\pi/4}\ket{T_X},X_{\pi/2}\ket{T_X}$ ,or $X_{3\pi/4}\ket{T_X}$, as we prove in \cref{app:CliffpartTraceproof}.  When using the distilled state in an algorithm, the final magic state measurement axis is modified depending on the Clifford frame. If magic states were prepared in the Pauli frame (see \cref{fig:PauliAutocorr}) rather than Clifford frames, such magic states would be measured in the $Z$ basis using transversal single-qubit measurements (see \cref{fig:Pnoncliff}). However, the measurement basis may now be $-Y, -Z, Y$ given that
\begin{align}
    X_{\pi/4} Z & X^{\dagger}_{\pi/4} = -Y \: , \nonumber \\
    X_{\pi/2} Z & X^{\dagger}_{\pi/2} = -Z \: , \nonumber \\
    X_{3\pi/4} Z & X^{\dagger}_{3\pi/4} = Y \: .
\end{align}
For protocols that distill multiple magic states, the Clifford frame of the distilled states may contain multi-qubit operators. After the distilled states are used by non-Clifford gates in a core of a quantum computer, the Clifford frame must be further conjugated through the remaining single-qubit $Z$ measurements. This may result in further multi-qubit Pauli measurements and additional routing space area in order to access the $Y$ and $Z$ logical boundary of the distilled states. A caveat from using the Clifford frame is that $Y$ basis measurements may require an extra ancilla (depending on the chosen hardware implementation). As such, the design of magic state distillation factories may require additional routing space to store the ancillas needed for $Y$ measurements.

\begin{figure}
    \centering
    \subfloat[\label{fig:Pnoncliff} ]{\includegraphics[width=.46\textwidth]{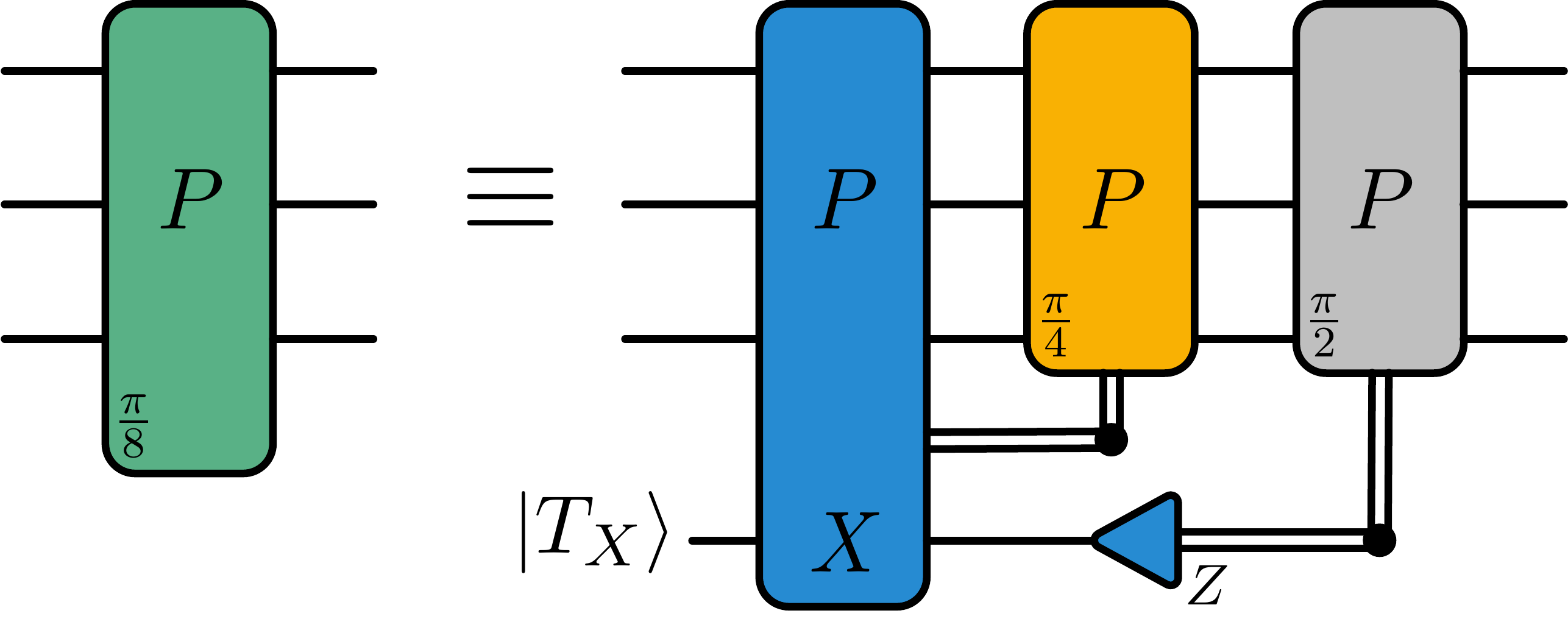}}
    \hspace{0.2cm}
\subfloat[\label{fig:Pcliff} ]{\includegraphics[width=.42\textwidth]{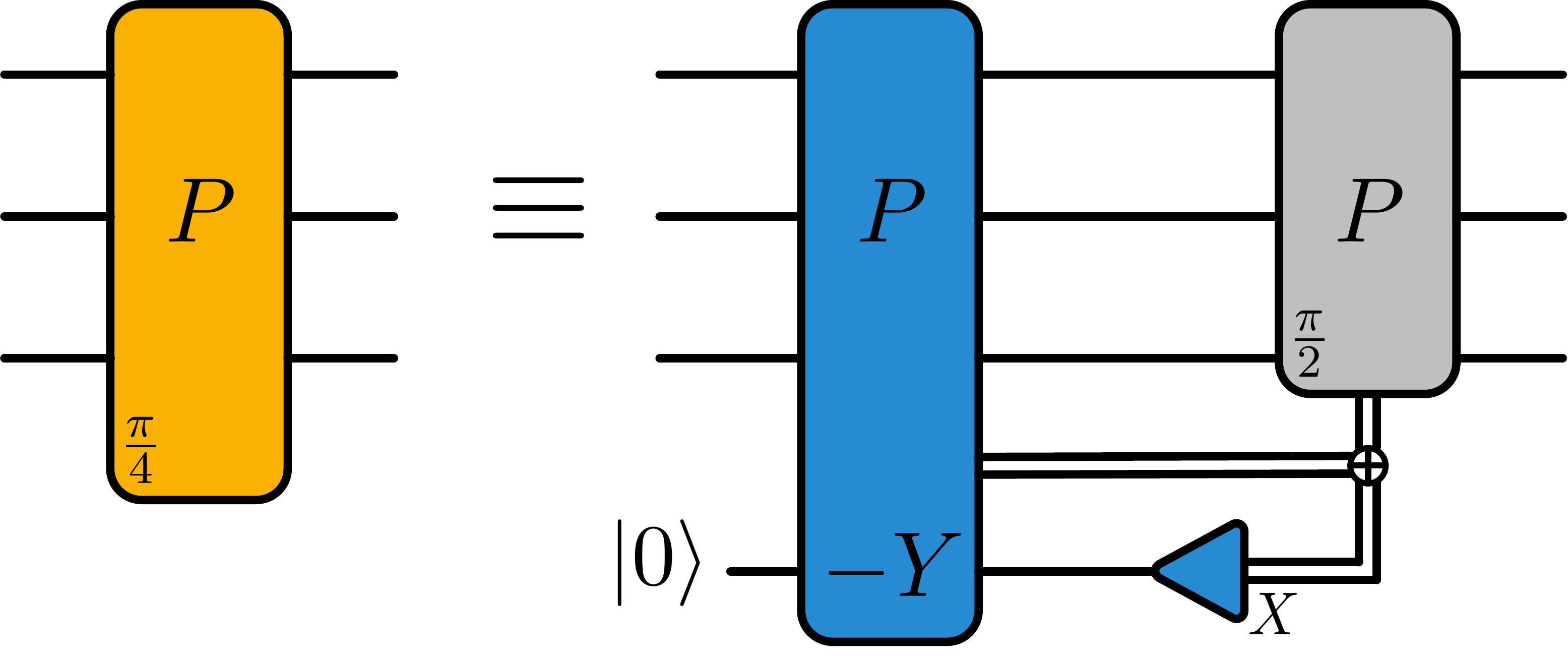}}
    \caption{(a) Circuit for performing a $\pi / 8$ multi-qubit Pauli measurement. The circuit requires a $\ket{T_X} = H \ket{T}$ resource state, and a Clifford correction may be required depending on the $P \otimes X$ measurement outcome. (b) Circuit for performing a Clifford gate using an ancilla prepared in $\ket{0}$. Both circuits are adapted from Ref.~\cite{Litinski19}.}
    \label{fig:Cliffgadgets}
\end{figure}

We now address the additive space cost of TELS and how it may be minimized in a distillation protocol. Distillation tiles are essential building blocks of fault-tolerant universal quantum computers, so it is worth finding the smallest, most optimal qubit layouts for them. Consider the distillation circuit in \cref{fig:Haddist}. Each of the $11$ non-Clifford gates requires an input $\ket{T_X}$ state, as shown in the non-Clifford circuit gadget of \cref{fig:Pnoncliff}. If we use TELS to perform the $11$ measurements, $11$ magic states will need to be held in memory, as indicated in \cref{fig:TELSprotocolnew}. In contrast, when performing lattice surgery without temporal encoding, only one cell is assigned to repeatedly prepare magic states for non-Clifford gates~\cite{Litinski19, Litinski19magic}. Upon close inspection of the multi-qubit Pauli measurements being performed in a TELS protocol, we notice that a magic state is stored only for as long as the Pauli it was associated with from the original PP set $\mathcal P$ appears in the sequence of new measurements $\mathcal S$.  Consequently, magic states do not need to be stored for the entire protocol. Consider performing TELS for a PP set of size $3$, where each Pauli measurement consumes a magic state. We use the $[4,3,2]$ error-detect code with the cyclic codeword matrix
\begin{equation}
    G = \begin{bmatrix} 
1100\\
0110 \\
0011
\end{bmatrix}.
\label{eq:Gmat432}
\end{equation}
Here, the choice of a cyclic representation is what allows us to reduce space requirements. We may read off the new multi-qubit Pauli measurements from  $\mathcal S$ as $\{ P_1, P_1 P_2, P_2 P_3, P_3\}$. Notice that after the measurement $P_1 P_2$, the magic state associated with $P_1$ does not need to be accessed again. At this point the  hardware holding the magic state used for the $P_1$ multi-qubit Pauli measurement can be reset to prepare the magic state required for the $P_3$ measurements. Since $P_1$ does not appear in any of the measurements after the first occurrence of $P_3$, the magic states associated with $P_1$ and $P_3$ are never simultaneously accessed. Continuing with this argument, it can be seen that space for only two magic states is required to perform all the measurements given by \cref{eq:Gmat432}.

% \pagebreak %DEBUG
Hence for every code, a particular choice and ordering of codewords can result in a smaller quantity of $\ket{T_X}$ states that need to be stored. For classical codes that are cyclic, a cyclic description of the codeword generator matrix allows for a reduced number of magic states needed to be held in memory. To determine exactly how many, note that for any column in this matrix (corresponding to a Pauli measurement from $\mathcal S$), the number of rows between the first $1$ and the last $1$ denotes the number of magic states that need to be held in memory for that Pauli measurement. The maximum number of $\ket{T_X}$ states required for any of the columns of the codeword generator matrix is the maximum for the entire PP set. For codes that do not have a natural cyclic set of codeword generators, the codewords must be chosen and ordered very carefully. In general, finding a sequence that minimizes the space requirements of magic states is an $\NP$ problem as there are exponentially many orderings of codewords. For instance, in the 15-to-1 distillation protocol of \cref{fig:Haddist}, one of the TELS protocols we considered (see \cref{sec:MSDMSDdesign}) uses the classical Golay code. For this code, there is a natural cyclic representation of codewords as we show in \cref{subsubsec:Golay}. According to this construction, $12$ magic states will need to be held in memory. However in \cref{app:golaycodechoice}, we show a specific choice of codewords that can minimize the space-requirements to $10$ magic states.

To execute the non-Clifford and Clifford gates in \cref{fig:Haddist}, we use the gate gadgets in \cref{fig:Cliffgadgets}. Since we conjugate the Clifford frame through to the final single-qubit measurements, we do not actually need to perform any $\pi/4$ Clifford gates, since such gates are converted to $\pi/2$ multi-qubit Pauli measurements~\cite{Litinski19}.

\subsubsection{Time-cost analysis}
\label{subsubsec:MSDMSDClifftime}

The time taken to successfully complete one round of distillation is calculated as follows. Let $T_1$ be the time taken to implement TELS on the non-Clifford gates and $T_2$ be the variable time associated with the final multi-qubit $\pi/2$ Pauli measurements. $T_2=0$ if there are no updates to the Clifford frame. If a lattice surgery logical timelike failure is detected during TELS with probability $p_D$,
\begin{align}
    T_1 = & T_{\text{inj}} + n (d_m'+1) + p_D \Big( T_{\text{inj}} + n (d_m'+1) \Big) + \nonumber \\
    &p_D^2 \Big(T_{\text{inj}} + n (d_m'+1) \Big) + p_D^3 \mathellipsis \nonumber \\
    =& (T_{\text{inj}} + n (d_m'+1))(1 + p_D +p_D^2 + \cdots) \nonumber \\
    =& \frac{T_{\text{inj}} + n (d_m'+1)}{1-p_D} \: .
\end{align}
Here, $T_{\text{inj}}$ is the time taken to inject a magic state into a cell, calculated using the analysis in \cref{subsec:MSDMSDinjection}. Note that in the above equation, we assume that all the Pauli measurements after the first one do not wait any extra time for newly injected magic states. This assumption is validated by the fact that, for $p=10^{-3}$, magic state injection requires two syndrome measurement rounds (with probability $99\%$), or four (with probability $99.99\%$). However for all the distillation protocols we consider, TELS requires at least four syndrome measurement rounds per Pauli measurement.

The final multi-qubit $\pi/2$ Pauli measurements are non-deterministic, implying we may need to perform $k'$ Pauli measurements, where $0 \leq k' \leq \kappa$.  In the previous inequality, $\kappa$ is the number of single-qubit measurements on the input $\ket{T_X}$ states of the distillation protocol (not the ones used for the $\pi / 8$ measurements) before the Clifford frame is conjugated through (for instance, in \cref{fig:Haddist}, $\kappa=4$). To execute these Pauli measurements, we perform TELS according to the method of \cref{sec:NewTELS}. If there are $k'$ measurements to perform, this takes time $T_2 = k' (d_m + 1)$ without TELS. If we use TELS with measurement distance $d_m''$ and an $[n', k', d']$ code with detection probability $p_D'$,
\begin{align}
    T_2 &= n' (d_{m}''+1) + p_{D}' (n' (d_{m}''+1) ) + (p'_{D})^2 ... \nonumber \\
    &= \frac{ n' (d_{m}''+1) }{1-p_{D}' } \: .
\end{align}
Note that when measuring the final multi-qubit Paulis, if a detection event is observed, the entire protocol does not need to be restarted. Instead, it is sufficient to just redo the Pauli measurements associated with the TELS protocol, as is done in \cref{fig:TELSprotocolnew}. 

The time to successfully distill the magic state also relies on whether the distillation protocol itself detected an error in any of the input magic states. This is modeled by the probability that the magic state protocol detects an error on an input magic state, which we denote $p_D^{(M)}$. Thus the total time required to successfully distill a magic state is
\begin{align}
\label{eq:timecostCliff}
    T &= T_1 + T_2 + p_D^{(M)} T \nonumber \\
      &= \frac{T_1 + T_2}{1- p_D^{(M)}} \: .
\end{align}

\subsection{Challenges of extending TELS protocols to Pauli frames}
\label{sec:pauliMSD}

\begin{figure}
    \centering
    \hspace{-0.2cm}
    \includegraphics[width=0.48\textwidth]{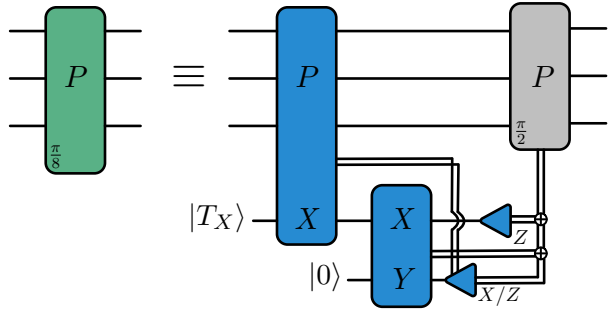}
    \caption{Circuit gadget for an auto-corrected non-Clifford gate. The circuit does not require the application of conditional Clifford gates to the logical data qubits. However, an extra ancilla prepared in $\ket{0}$ is required. }
    \label{fig:PauliAutocorr}
\end{figure}

The multi-qubit Paulis associated with the encoded TELS measurements in a Clifford-frame distillation circuit are performed using the circuit shown in \cref{fig:Pnoncliff}. However this results in a Clifford frame which eventually must be implemented using the Clifford gate gadgets in \cref{fig:Pcliff}. Keeping track of Clifford frames can be avoided by using the auto-corrected $T$ gadgets shown in \cref{fig:PauliAutocorr}. In such an implementation, the time associated with the Clifford correction can be traded for the extra space used by the additional $\ket 0$ ancilla. However, using auto-corrected $T$ gadgets in a TELS protocol leads to additional challenges. When the $k$ $P \otimes X$ measurements are performed using TELS, there will be an additional space cost associated with holding some magic state cells in memory. In order to benefit from the time speedups provided by TELS, a TELS protocol may need to also be performed on the $\ket 0$ ancilla states. Such considerations would result in the space cost being tripled (relative to the Clifford frame scheme). To see this, note that the $P \otimes X$ and $X \otimes Y$ measurements occur simultaneously since they have the same measurement distance and both $X$ boundaries of the $\ket{T_X}$ state can be accessed simultaneously. In such a protocol, the number of $\ket{0}$ and $\ket{T_X}$ states are identical. Furthermore, each $\ket{0}$ state requires an additional cell in order to access its $Y$ boundary. 

If instead, we do not perform TELS on the $X \otimes Y$ measurements, there are two options. The lattice surgery operations (with large measurement distance) can either be performed sequentially, which will result in a speed mismatch between the $X \otimes Y$ measurements and the $P \otimes X$ measurements, leading to a backlog of $\ket{T_X}$ states and $\ket{0}$ states that will need to be held in memory (and so there would be no time improvement due to TELS). Another option is to perform the slow lattice surgery operations in parallel, but this also admits an additional space cost to hold all the $\ket 0$ cells.

Given the above considerations, and the challenges associated with the design of distillation tiles for the inclusion of $\ket{0}$ ancillas, we leave the analysis of TELS protocols applied to magic state distillation protocols in the Pauli frame to future work.

\section{Precise design of distillation tiles}
\label{sec:MSDMSDdesign}

\setlength{\tabcolsep}{6pt}
\begin{table*}
    \centering
    \vspace{0.5cm} %DEBUG
    \begin{tabular}{c c c c c c c c c}
        $p$ & $\delta^{(M)}$ & Distillation  &  Circuit type & $d_x$ & $d_z$ & Space & Time & Space-time cost \\
         &  & code  &  &  &  &   & (NS)  & ($\#$ qubits $\times$ NS)\\
        \hline
        $10^{-4}$  & $10^{-10}$ & $\llbracket 15, 1, 3 \rrbracket$  & No encoding & $7$ & $9$ & $1360$ & $110.06$ & $1.5 \times 10^5$ \\
         &  &  &No enc., Par & $7$ & $9$ & $3300$ & $60.03$ & $1.98 \times 10^5$ \\
         &  &  & \textbf{Cliff-SED}& $7$ & $9$  & $1120$ & $104.05$ & $\mathbf{1.17 \times 10^5}$ \\
         &  &  & Cliff-SED, Par& $7$ & $9$ & $2352$ & $68.03$ & $1.6 \times 10^5$ \\
         &  &   &  Cliff-BCH& $7$ & $9$ & $1344$ & $92.08$ & $1.24 \times 10^5$ \\
         &  &  &  Cliff-BCH, Par& $7$ & $9$ & $2688$ & $64.06$ & $1.72 \times 10^5$ \\
         &  &  &  Cliff-Golay& $7$ & $9$ & $1792$ & $124.06$ & $ 2.23 \times 10^5$ \\
         &  &  &  Cliff-Golay, Par& $7$ & $9$ & $3024$ & $80.04$ & $2.42 \times 10^5$ \\[0.25cm]
         
        $10^{-4}$ &  $ 10^{-15}$ & $\llbracket 116,12,4 \rrbracket$  & No encoding & $9$ & $15$ & $9440$ & $1391.48$ & $1.31 \times 10^7$  \\
         &  &   & No enc., Par & $9$ & $15$ & $14934$ & $702.77$ & $1.05 \times 10^7$ \\
         &  &  & Cliff-BCH9, Par & $9$ & $15$ & $22800$ & $431.82$ & $9.85 \times 10^6$ \\
         &  &   & \textbf{Cliff-Zett5, Par} & $9$ & $15$ & $18600$ & $442.41$ & $\mathbf{8.23 \times 10^6}$ \\[.25cm]
         
        $10^{-3}$ & $10^{-10}$ & $\llbracket 114,14,3 \rrbracket$  &  No encoding & $9$ & $19$ & $11750$ & $1646.61$ & $1.93 \times 10^7$  \\
         &  &   & No enc., Par &  $9$ & $17$ & $16836$ & $831.62$ & $1.4 \times 10^7$ \\
         &  & & Cliff-BCH7, Par &  $9$ & $17$ & $22400$ & $595.31$ & $1.33 \times 10^7$ \\
         &  &  & \textbf{Cliff-Zett5, Par} &  $9$ & $17$ & $20480$ & $623.27$ & $\mathbf{1.28 \times 10^7}$ \\[.25cm]
         
        $10^{-3}$  & $10^{-15}$ & $\llbracket 125, 3, 5 \rrbracket$  & No encoding & $13$ & $25$  & $17514$ & $2479.17$ & $4.34 \times 10^7$  \\
         &  &  & No enc., Par &  $13$ & $25$ & $29548$ & $1252.11$ & $3.7 \times 10^7$ \\
         &  &  & Cliff-BCH7, Par&  $13$ & $25$ & $38640$ & $859.8$ & $3.32 \times 10^7$ \\
         &  &  & \textbf{Cliff-BCH9, Par}&  $13$ & $25$ & $42504$ & $729.66$ & $\mathbf{3.1 \times 10^7}$  
    \end{tabular}
    \vspace{0.3cm} %DEBUG
    \caption{Space-time costs of different distillation protocols on a biased-noise planar surface code. $\delta^{(M)}$ is the target logical error rate per output magic state. TELS protocols are labeled ``Cliff-xxx'', with ``Par'' implying that measurements are performed two at a time (i.e., with lattice surgery measurements which can access the two $X$ logical boundaries of surface code patches simultaneously). The number of physical qubits is two times the space cost, since the space cost counts only the number of data qubits of the surface code. The probability that a distillation algorithm rejects due to an error in an injected magic state is  $p_D^{(M)} = 1-(1-\epsilon_L)^n$ where $\epsilon_L$ is given by \cref{eq:epsL}. For the $15$-to-$1$ distillation protocol, the space time cost of a protocol using TELS is approximately $30\%$ smaller ($1.17 \times 10^5$) than a protocol that does not use TELS ($1.5 \times 10^5$). For the $125$-to-$3$ distillation protocol, the space time cost is decreased by approximately $20\%$ with TELS. The label NS refers to the number of syndrome measurement rounds required for the entire distillation protocol.}
    \label{tab:STcosts} 
\end{table*}

In this section we analyze space-time costs of various distillation protocols in different noise regimes. At a physical error rate of $p=10^{-4}$, one round of $15$-to-$1$ distillation with robust lattice surgery operations is sufficient to distill magic states with final error probability $\delta^{(M)} \le 10^{-10}$. For $\delta^{(M)}= 10^{-15}$ (which is relevant for larger algorithms), or for distillation protocols with $p=10^{-3}$, we considered $100+$ qubit quantum codes, as suggested in Ref.~\cite{Haah18}. To the best of our knowledge, our work is the first to analyze space-time costs of distillation protocols using these larger codes. For the noise rate regimes $p=10^{-4}$ and $p=10^{-3}$, we estimate the space-time costs of the various distillation protocols using different implementations of lattice surgery. We first consider protocols that do not use TELS. These protocols will execute non-Clifford gates using the auto-corrected non-Clifford gates of \cref{fig:PauliAutocorr}. Subsequently, we consider distillation protocols that perform TELS, using the methods developed in \cref{subsec:MSDMSDCliff}. In contrast to Ref.~\cite{Litinski19}, the distillation tiles developed in this paper are all rectangular, minimizing wasted space when tiled on a $2$D grid of qubits. Note also that in Ref.~\cite{Litinski19magic}, the logical qubits are designed to have different space-like distances $d_x$ and $d_z$ even without a biased noise model. This improvement is permitted due to the specific function of each qubit in the distillation protocol. When applying these improvements to the distillation protocols and layouts in this work, the space-time costs may be further reduced. 

The time-like distance of lattice surgery $d_m$ and the space-like distances $d_x$ and $d_z$ of the logical qubits and routing regions are computed using a procedure detailed in \cref{app:algoSTcosts}. Essentially, a set of distances $\{ d_x,d_z,d_m\}$ must be determined that minimize the overall space-time cost, while ensuring the output magic states have logical errors with probability at most $\delta^{(M)}$. In solving \cref{eq:conditiondelta}, we must first determine certain constants related to each hardware layout. This includes the area of the distillation tile used in each protocol (which we denote as the space cost), the number of logical qubits used in a tile $N$, the worst-case routing space area $A$ and maximum area used during any lattice surgery measurement. We develop $20$ different layouts in this work, and the associated constants for each of them are tabulated in \cref{app:constants}. 

In \cref{tab:STcosts}, we display the space-time costs of the various distillation protocols considered in this work. Using TELS protocols, it is possible to achieve lower space-time costs than using protocols without any temporal encoding. Although there are only minor improvements to the space-time cost of TELS-assisted distillation tiles, many of these tiles will be needed in each distillation factory. As a result, the improvements add up and the quantum computer as a whole will have a lower space-time volume. Moreover with reduced time costs, fewer distillation tiles may be required altogether, as we show in \cref{subsec:MSDMSDscheduling}. This in turn further reduces the space-time cost of distillation factories. Interestingly, for the $15$-to-$1$ distillation protocol, circuits that perform TELS do not produce tiles that have smaller time costs. Instead, TELS-assisted tiles require fewer qubits, and this can be attributed to the use of non-Clifford gate gadgets as shown in \cref{fig:Pnoncliff}. Below, we show layouts for the $15$-to-$1$ distillation routine. In \appref{sec:additionaldistillation}, we show layouts for the $100+$ qubit codes.

\subsection{15-to-1 distillation}
\label{sec:15to1}

The $15$-to-$1$ magic state distillation protocol is one of the most widely known protocols for distilling $\ket T$ states (see for instance Refs.~\cite{Bravyi05, BevsUniversal, fowler2018low, Litinski19, Litinski19magic}). The protocol originates from a $\llbracket 15,1,3\rrbracket$ triorthogonal CSS quantum code. The code has the property that the application of $T$ gates on all of the physical qubits of the code implements a logical $T^{\dagger}$ gate. Since we perform distillation with $\ket{T_X}$ states, we define this code to contain $1$ logical qubit,  $10$ $X$-type stabilizers and $4$ $Z$-type stabilizers. A distilled magic state is produced by encoding a logical $\ket{0}$ state, applying the transversal $T$ gates, decoding  and then performing measurements. By propagating the the Clifford gates past the transversal $T$ gates and removing the redundant parts of the  circuit, we are left with the circuit in \cref{fig:Haddist} (see Ref.~\cite{Litinski19} for a more detailed derivation). The circuit contains $11$ commuting Pauli measurements on $5$ logical qubits (four of which are logical $\ket{T_X}$ states). Since the $11$ Pauli measurements commute, they form a size-$11$ PP set. There exist many choices of classical codes to be used in TELS protocols with size-$11$ PP sets. In this paper we will focus on using a Single Error Detect code of distance $2$ ($[12,11,2]$), a BCH code of distance $3$ ($[15,11,3]$) and the Golay code of distance $7$ ($[23,12,7]$).

We will consider using this distillation protocol in a regime where the physical error rate is $p=10^{-4}$ and the target logical error rate per magic state is  $\delta^{(M)}=10^{-10}$. Using the noise model described in \cref{subsec:MSDMSDinjection} for the injection of magic states, we apply the analysis in Ref.~\cite{Litinski19magic} to determine the logical failure probability per output magic state for one round of a $15$-to-$1$ distillation scheme which is given by
\begin{align}
p^{(M)}_L = &35 \Big ( (\epsilon_{\text{L,Z}})^3+\frac{1}{2} 6 (\epsilon_{\text{L,Z}})^2 \epsilon_{\text{L,X}} \nonumber \\
&+\frac{1}{4} 12 \epsilon_{\text{L,Z}} (\epsilon_{\text{L,X}})^2 + \frac{1}{8} 8 (\epsilon_{\text{L,X}})^3  \Big ) \nonumber \\
=& \frac{35(1+\eta)^3}{27 \eta^3} p^3 \: .
\end{align}

For $p=10^{-4}$ and $\eta=100$, the probability that the distillation succeeds is $1 - p_D^{(M)} = (1-\epsilon_L)^{15} = 0.999$ and $p_L^{(M)} = 1.33 \times 10^{-12}$. As this is sufficiently below $\delta^{(M)}$, the lattice surgery measurements used to execute the distillation protocol must be modeled with measurement distance large enough to allow for distilled magic states of logical error rate at most $\delta^{(M)}$. Using the procedure in \cref{app:algoSTcosts}, we determined that the minimum spacelike distances required are $d_x=7$ and $d_z=9$. 

In the subsequent subsections, we detail the specifics of the hardware layouts that are used for the various distillation protocols, both with and without TELS. Arranging the logical qubits according to these layouts minimizes the space requirements of distillation blocks. In addition, TELS is used to minimize the time costs. Overall we observe that protocols that use TELS can achieve lower space-time costs than those that do not.

\begin{figure}
    \centering
\subfloat[\label{fig:15ue} ]{\includegraphics[width=.19\textwidth]{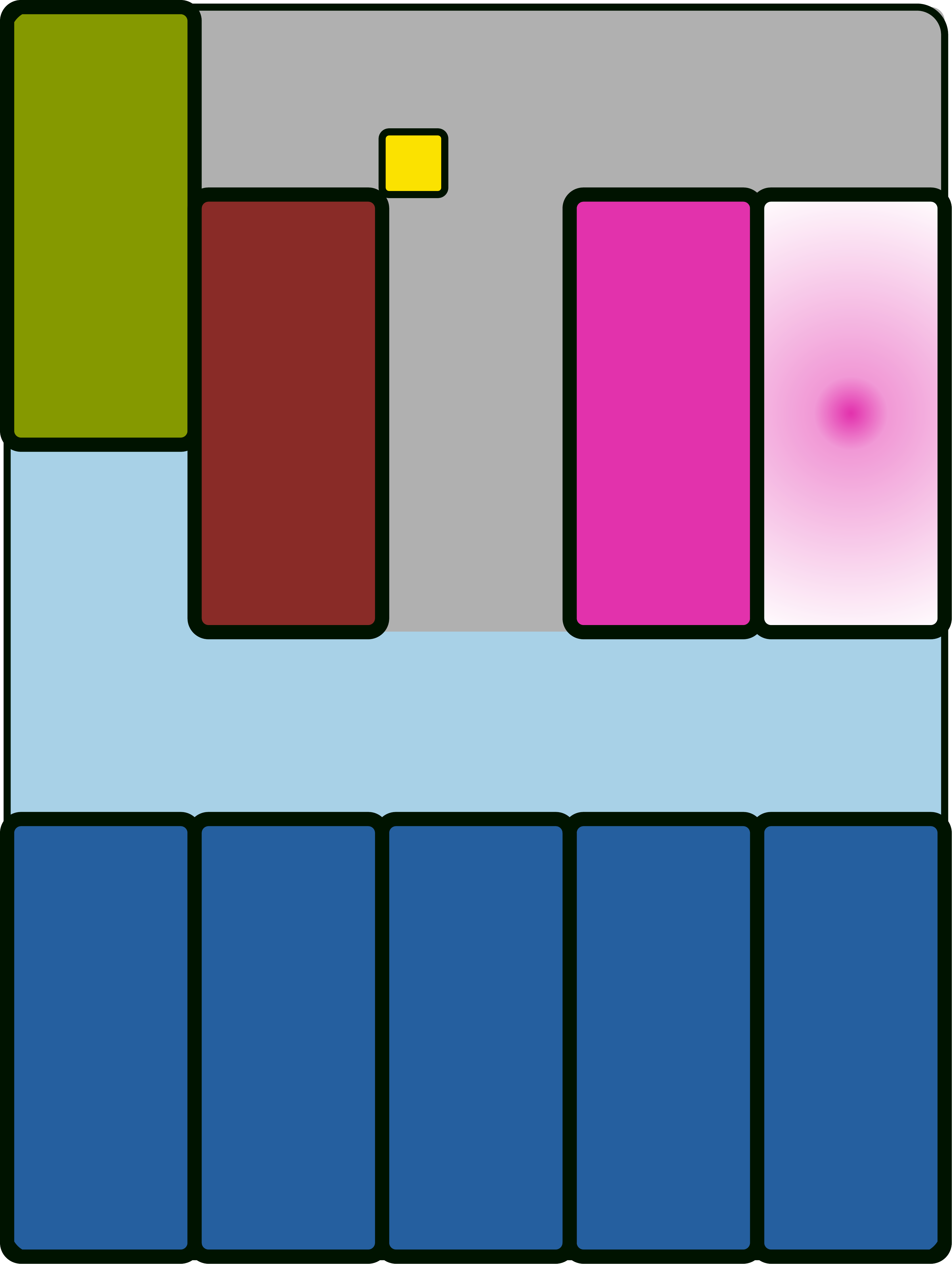}}
    \hspace{0.25cm}
    \subfloat[\label{fig:15sed} ]{\includegraphics[width=.19\textwidth]{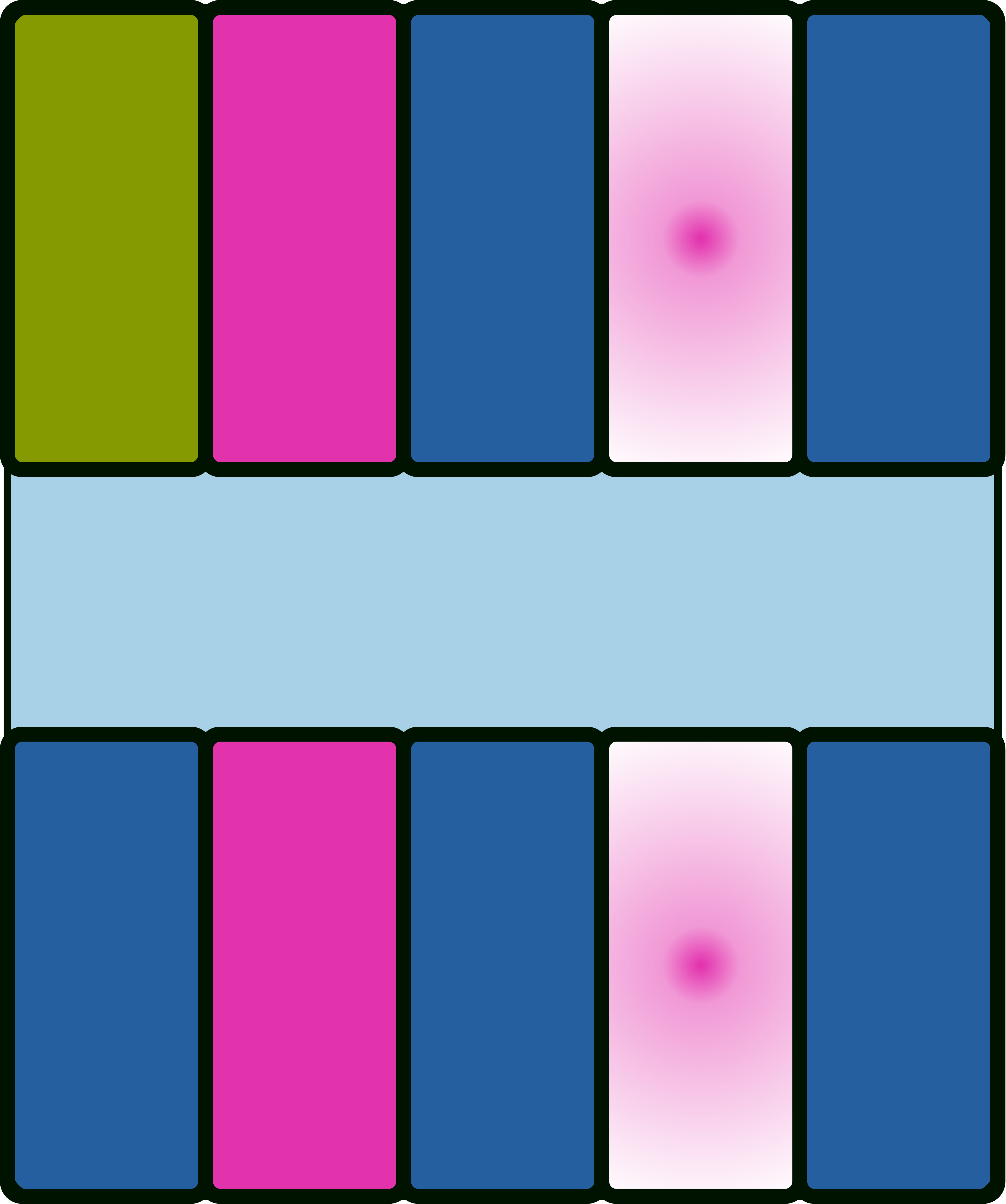}}
    \hspace{.25cm}
    \subfloat[\label{fig:15bch3} ]{\includegraphics[width=.227\textwidth]{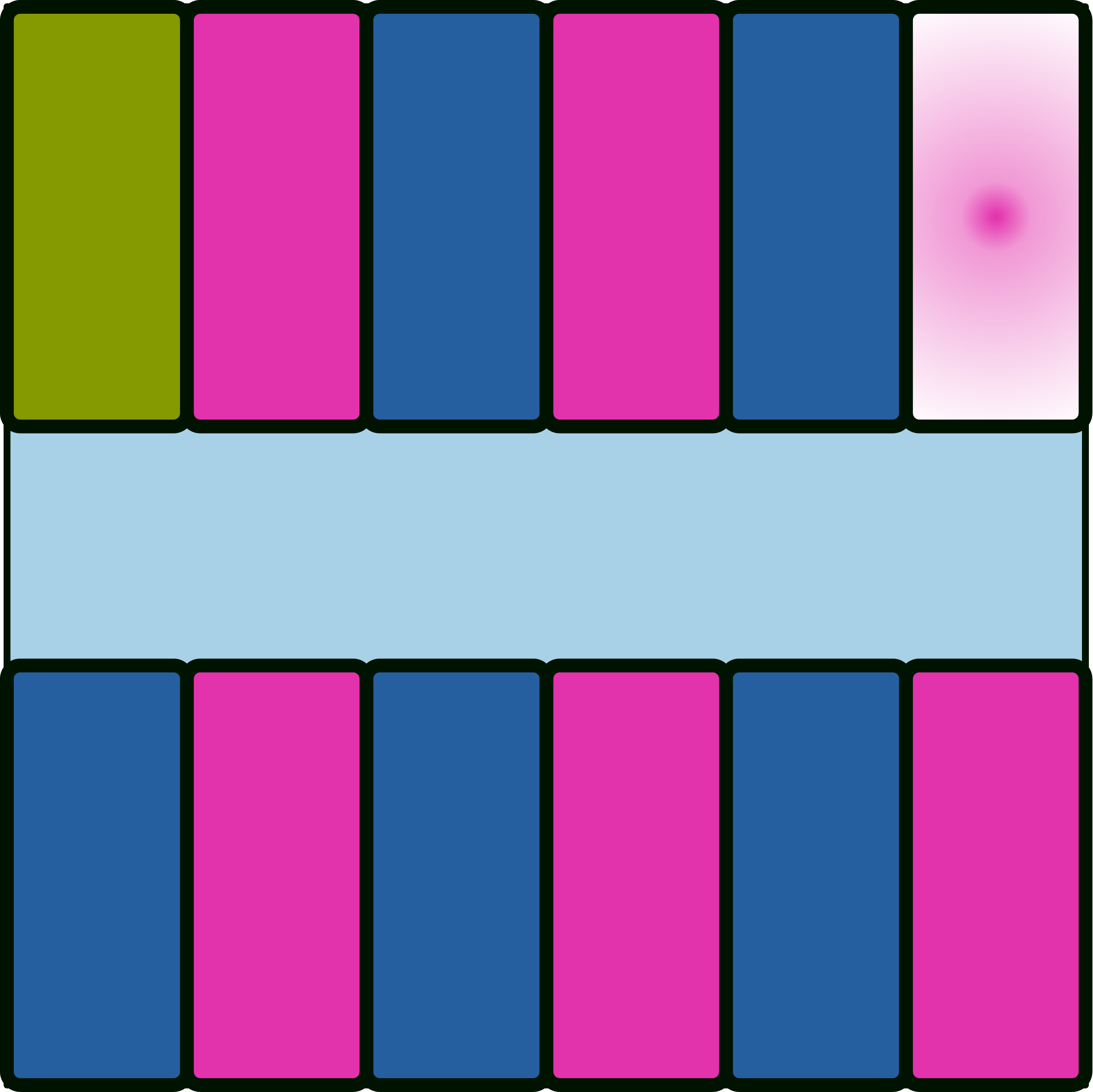}}
    \hspace{0.25cm}
    \subfloat[\label{fig:15golay} ]{\includegraphics[width=.302\textwidth]{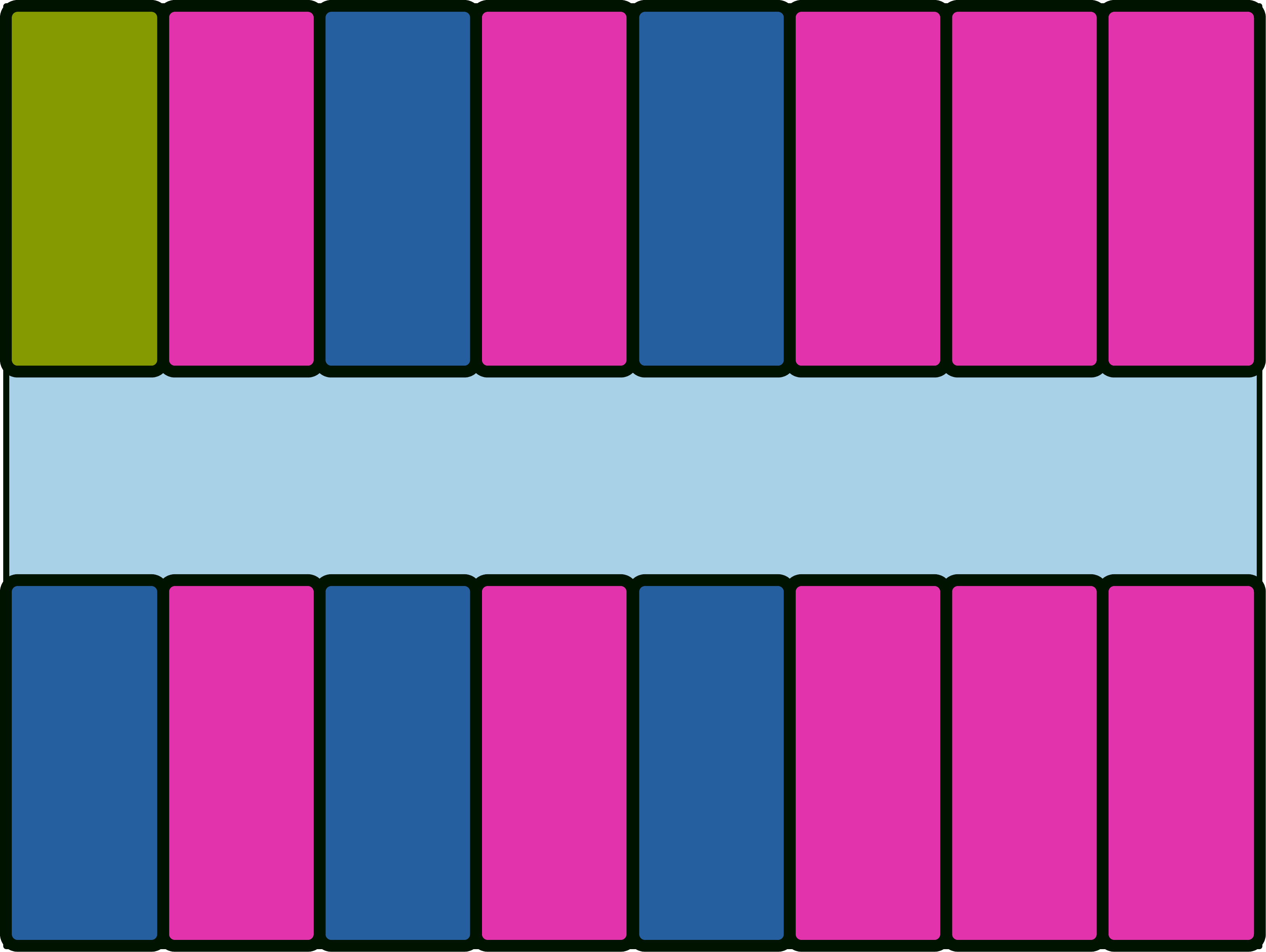}}
    \caption{Layouts of logical qubits for TELS-assisted $15$-to-$1$ state distillation. Data qubits are placed in blue cells. Magic states are in pink cells, where cells with a radial shading are extra cells used to prepare new magic states in parallel with the Pauli measurements. $\ket 0$ ancillas for autocorrected gadgets are placed in the brown cells adjacent to the yellow squares used for twists. Green cells are used to store distilled magic states for use by the core while the next round of distillation occurs. Additional green cells may be required if a distillation tile produces magic states faster than the core consumes them (alternatively, the magic states can be transported to additional tiles surrounding the core). Routing regions between cells are split into grey and blue to show that the relevant lattice surgery operations will not clash. (a) Layout for un-encoded lattice surgery using autocorrected non-Clifford gate gadgets of \cref{fig:PauliAutocorr}. The grey routing region handles the $X \otimes Y$ measurements and the blue routing regions performs $X$-boundary measurements between different logical qubits. (b) Layout for $15$-to-$1$ distillation with TELS, using the $[12,11,2]$ Single Error Detect code. Note that we only need one radial pink cell. However given the geometry of the entire tile, we use the remaining space for another pink radial tile. (c) Layout using the $[15,11,3]$ BCH code, and, (d)~using the $[23,12,7]$ Golay code.}
    \label{fig:15to1layouts}
\end{figure}

\begin{figure}
    \centering
    \subfloat[\label{fig:15uepar} ]{\includegraphics[width=.17\textwidth]{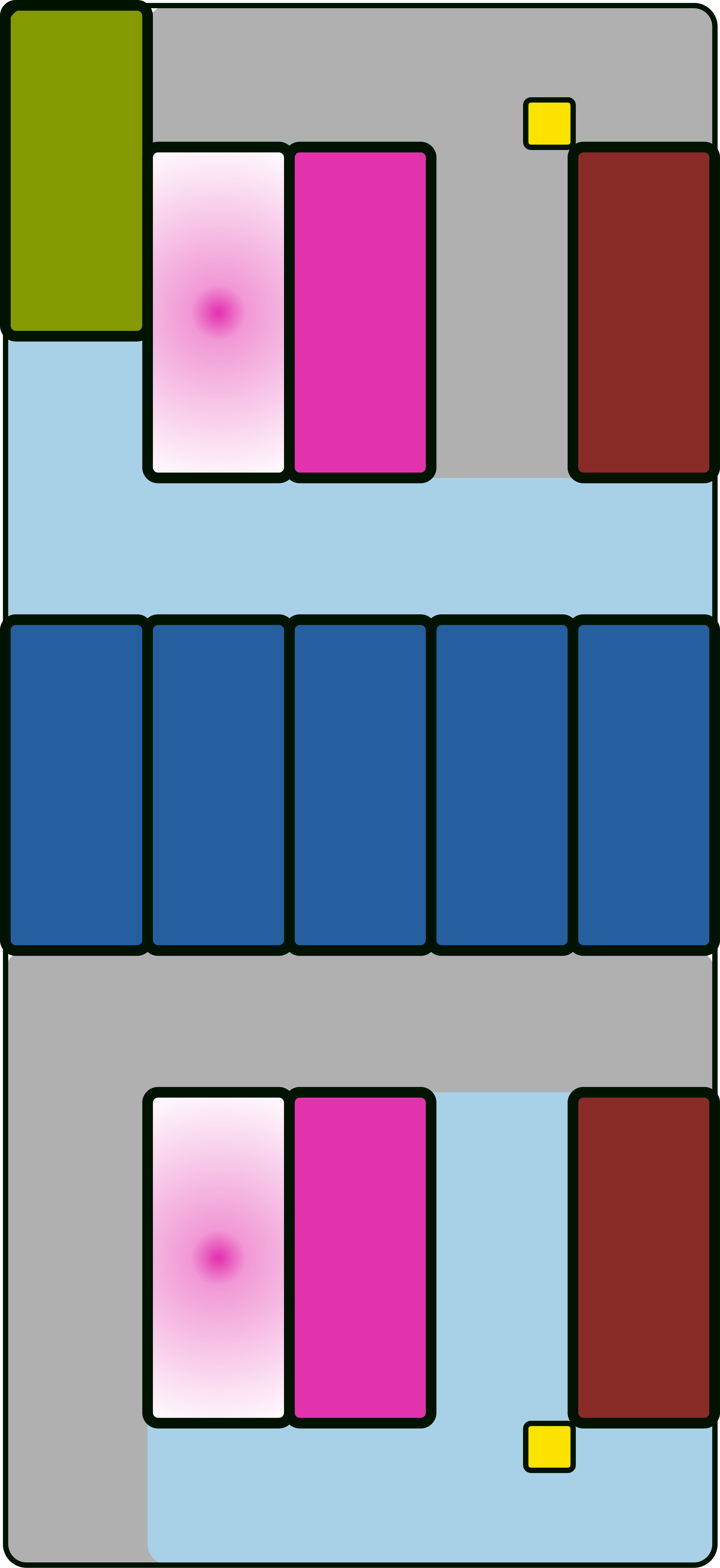}}
    \hspace{25cm}
    \subfloat[\label{fig:15sedparstag1}
    ]{\includegraphics[height=.23\textwidth]{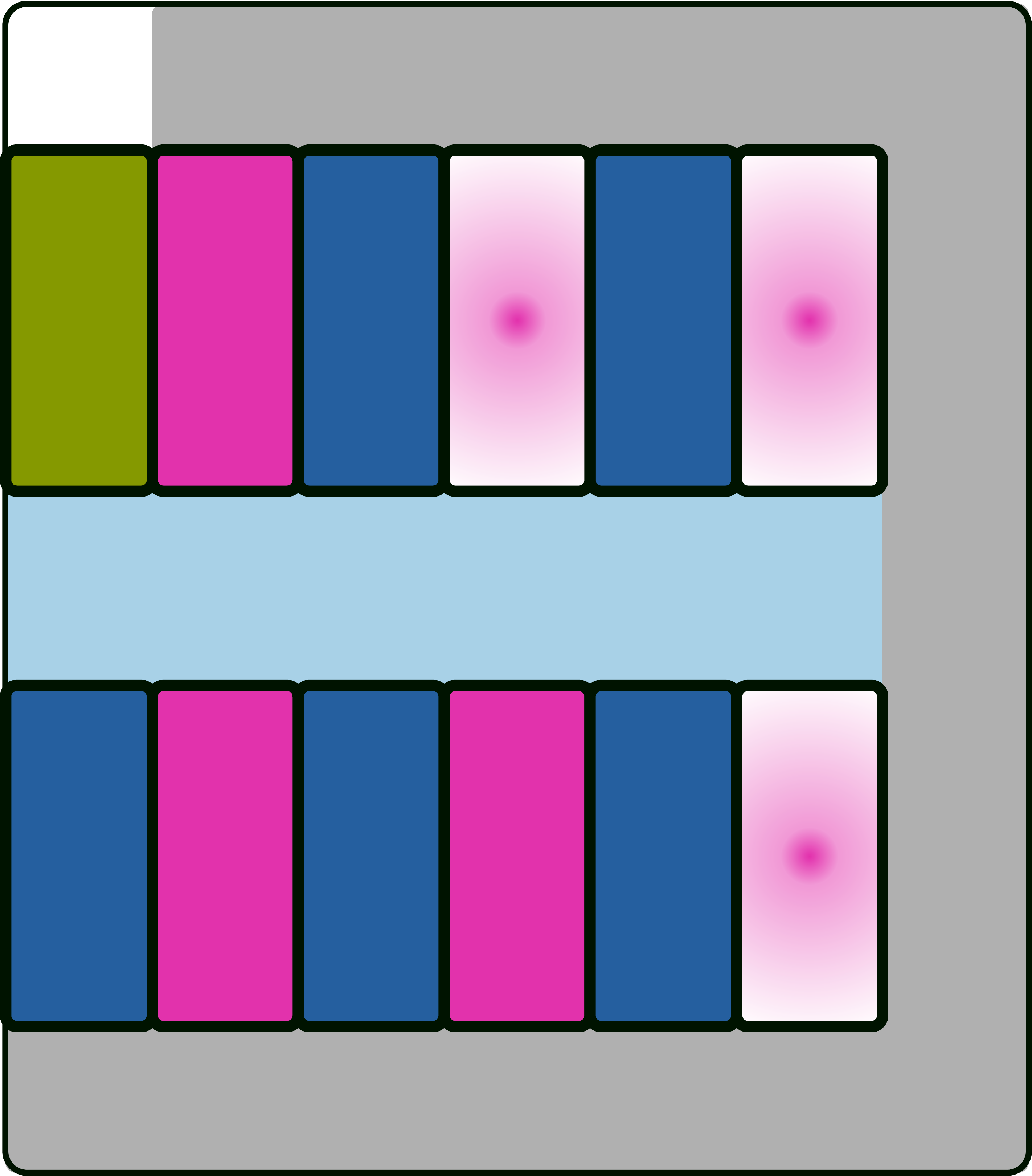}}
    \hspace{0.1cm}
    \subfloat[\label{fig:15sedparstage2} ]{\includegraphics[height=.23\textwidth]{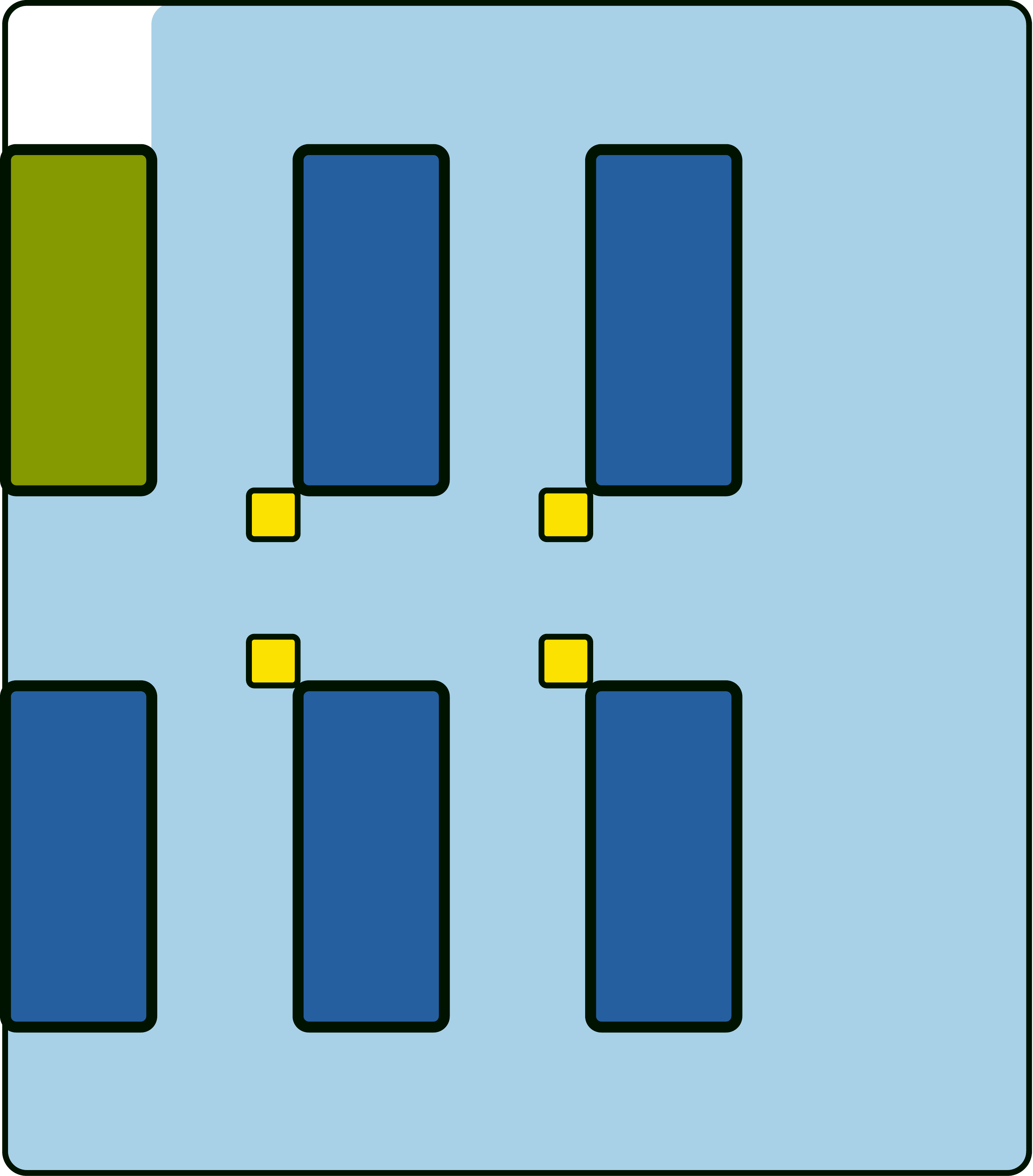}}
    \hspace{.1cm}
    \subfloat[\label{fig:15bch3par} ]{\includegraphics[height=.23\textwidth]{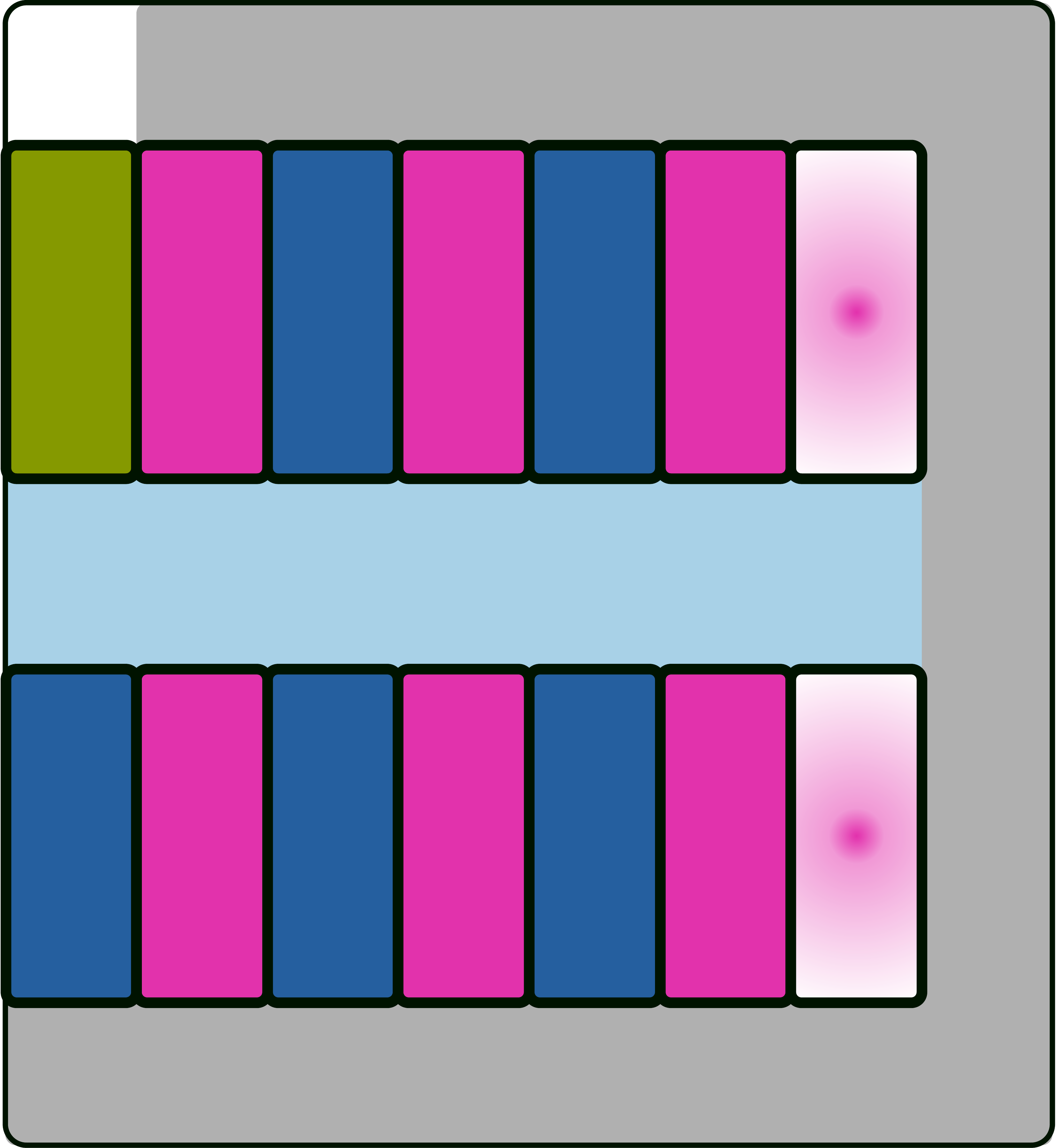}}
    \hspace{0.1cm}
    \subfloat[\label{fig:15golaypar} ]{\includegraphics[height=.23\textwidth]{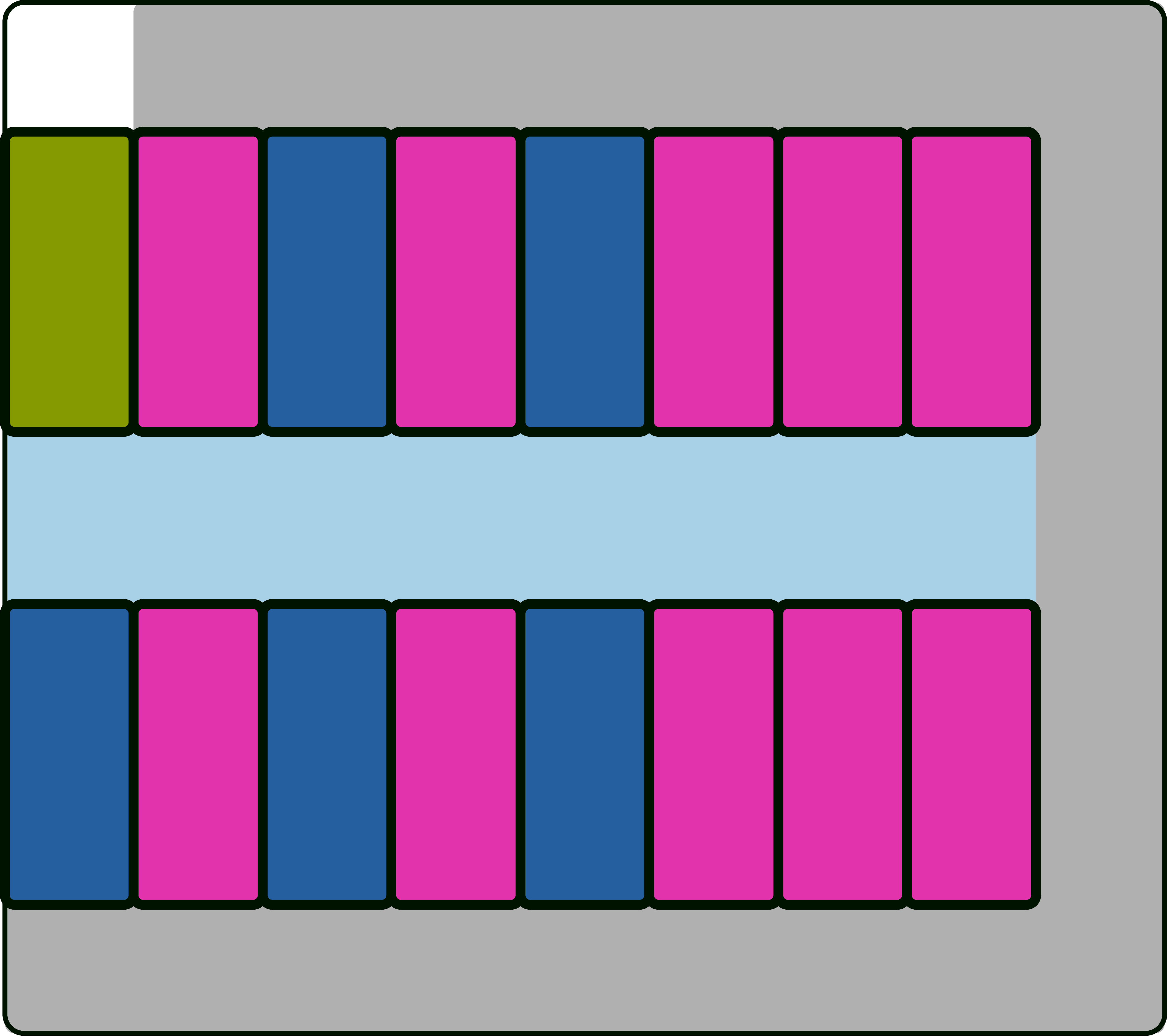}}
    \caption{Layouts of logical qubits for parallelized TELS-based $15$-to-$1$ state distillation protocols. The meaning of each color is described in the caption of \cref{fig:15to1layouts}. (a) Layout for un-encoded lattice surgery, with two routing regions, each accessing one $X$ boundary of the logical qubits. Each routing region has access to a separate magic state and a $\ket 0$ ancilla used in the circuit of \cref{fig:PauliAutocorr}. (b)~Layout for distillation with TELS, using the $[12,11,2]$ Single Error Detect code. Three magic state tiles are held in memory for each pair of parallel Pauli measurements. Then two are discarded and two prepared magic states on other pink cells are used in the following round. (c)~Layout and routing region used for the final multi-qubit Pauli measurements in the Clifford frame distillation protocol. (d)~Layout for parallelized distillation with TELS, using the $[15,11,3]$ BCH code. (e)~Layout for parallelized distillation with TELS using the $[23,12,7]$ Golay code. For (d) and (e), the layouts used to perform the final multi-qubit Pauli operations required by the  Clifford frame can be found in an analogous way from going from  (b) to (c). }
    \label{fig:15to1layoutsparallelized}
\end{figure}

\subsubsection{No temporal encoding} 

\begin{figure*}
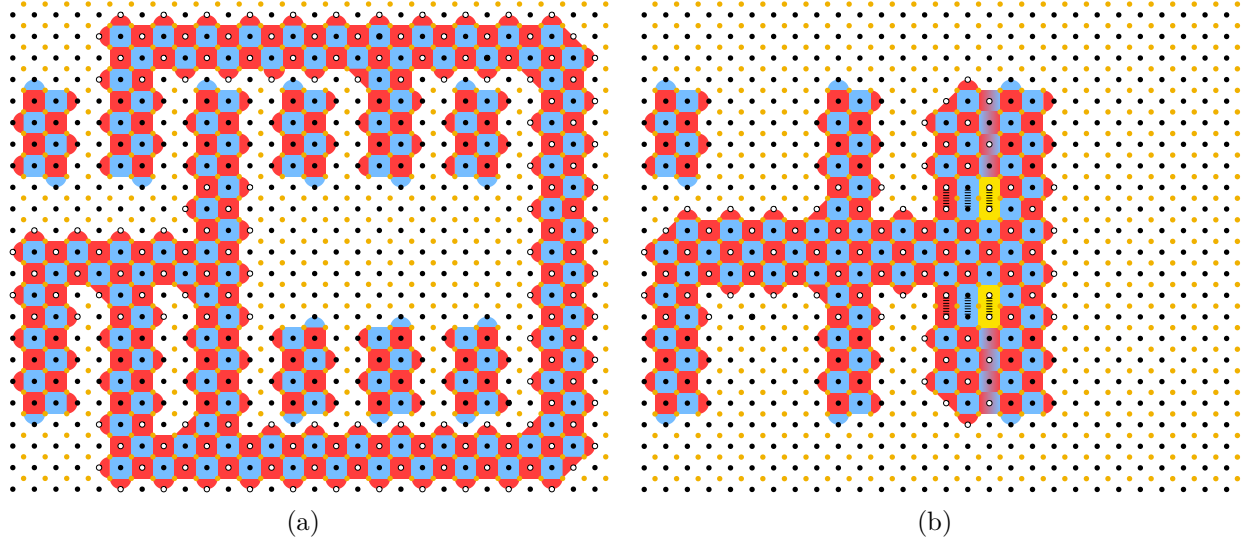

    \centering
    \subfloat[\label{fig:layoutSpec} ]{\includegraphics[width=0.48\textwidth]{imagesChap7/latsurg1.pdf}}
    \hspace{0.3cm}
    \subfloat[\label{fig:layoutSpecPaulis} ]{\includegraphics[width=0.48\textwidth]{imagesChap7/latsurg2.pdf}}
    \caption{(a) On the layout of \cref{fig:15sedparstag1}, we show how two separate routing spaces can be used to perform parallel lattice surgery measurements. The logical measurements are $X_1 \otimes X_2 \otimes X_3 \otimes X_{T_{X,1}}$ (in the equatorial routing space) and $X_3 \otimes X_4 \otimes X_{T_{X,1}} \otimes X_{T_{X,2}}$ (in the circumferential routing space). These are the first and second measurements respectively when performing TELS-assisted distillation using the $[12,11,2]$ SED code (see \cref{eq:CyclicG12} of \cref{app:codeconstruction} for the codeword generator matrix). Alternatively, they correspond to the first measurement and the product of the first and second measurements from \cref{fig:Haddist}. The logical patches have code distances $d_x=3$, $d_z=5$. $X$ stabilizers are in red, and $Z$ stabilizers are in blue. The product of the $X$ stabilizers indicated by white vertices gives the parity for the multi-qubit Pauli measurement outcomes. (b) On the same layout, we show how to perform Pauli measurements which are tensor products of $X,Y$, or $Z$ on the data qubits. These measurements are performed after the non-Clifford gates of the distillation protocol. The example in the figure measures $X \otimes Y \otimes X \otimes X \otimes Y$ on the five data qubits. The yellow stabilizers are twist defects that are used to access $Y$ boundaries of logical qubits that are originally defined with only $X$ and $Z$ boundaries, using the techniques shown in Ref.~\cite{Chamberland22b}. Note that the size of the routing space area separating the top and bottom rows of data qubits is taken to be large enough to allow for $Y$ measurements requiring twists. }
\end{figure*}

First, we calculate the space-time cost of the distillation circuit of \cref{fig:Haddist} without temporally encoded lattice surgery. For this, we consider a modified version of the layout used by Litisnki (see Fig. $18$ of Ref.~\cite{Litinski19}), as shown in \cref{fig:15ue}. In this figure, the five blue cells at the bottom correspond to the data qubits of \cref{fig:Haddist}; the pink cells are used to store magic states for performing the $\pi / 8$ multi-qubit Pauli measurements; radial pink cells are used to inject new magic states for subsequent Pauli measurements (thus preventing time delays due to state injection); brown cells with an adjacent yellow square for twists are $\ket 0$ ancillas used in the autocorrected non-Clifford gate gadget of \cref{fig:PauliAutocorr}; and green cells are used to store the distilled magic state from the previous round of distillation, so that it may be accessed by the core of the quantum computer.  

There are two cells assigned for magic states. Without TELS, only one magic state is used per non-Clifford gate (note that one $X$ boundary has access to the data qubits, and the other boundary has access to the $Y$ boundary of the $\ket 0$ ancilla, hence these lattice surgery measurements may be performed in parallel). However, at any given time, one magic state cell will take part in a non-Clifford gate, and the other will be used to prepare a noisy magic state for the subsequent non-Clifford gate. This is required as it takes a non-trivial amount of time to prepare a noisy magic state (roughly 2 syndrome measurement rounds as described in \cref{subsec:MSDMSDinjection}). In \cref{fig:15to1layouts,fig:15to1layoutsparallelized,fig:125to1layouts,fig:116to12layouts,fig:114to14layouts}, we use the radial-shaded pink cell to denote the extra cells needed for this simultaneous magic state preparation.  

Since we assume classical processing is instantaneous, only one qubit cell is assigned for the $\ket 0$ ancilla used in the autocorrected gadget. Note however, that in Ref.~\cite{Litinski19} and Ref.~\cite{EarlParallel}, it was shown that for finite decoding times, additional $\ket{0}$ ancilla qubits can be used to offset the extra time cost associated with decoding all syndrome measurement rounds associated with the previous lattice surgery operations. However with classical parallelization and pre-decoders \cite{ChambsLocalNN22,BrownLocalPre22,EarlParallel,ChaoPrallel22}, such additional ancillas may be unnecessary. 

Note that we do not need to shuttle the distilled magic state from a blue cell to a green cell. We design the distillation tile such that the output magic state is in the right most blue cell. In the next round of distillation, the layout is mirrored about the equator and the magic state cell now becomes a green cell with core access. In this way, the distillation protocol may be restarted without any shuttling delays. Using the procedure detailed in \cref{app:algoSTcosts}, we determined that the lattice surgery measurement distance must be $d_m=9$ to obtain magic states with logical failure rate at most $10^{-10}$.

For every data qubit cell, there are two accessible $X$ boundaries. We can almost trivially speed up the protocol by a factor of two by assigning new routing space and ancilla cells that access the second $X$ boundary of the data qubits. This new hardware layout allows us to perform two multi-qubit Pauli measurements in parallel. This idea was originally proposed in Ref.~\cite{Litinski19}. We show a layout that performs lattice surgery measurements two at a time without temporal encoding in \cref{fig:15uepar}. Note that with this layout, a distilled magic state present on a blue data cell must be shuttled to a green storage cell. There is an extra time cost associated with the shuttling operation. We do not include the time cost of shuttling in \cref{tab:STcosts} as it may still be possible to eliminate shuttling using a more clever layout. In any case, the layout without parallel measurements yields a smaller space-time cost.

\subsubsection{TELS-Single Error Detect $\mathbf{[12,11,2]}$} Next, we consider a distillation protocol that uses TELS to execute the non-Clifford gates, with the protocol described in \cref{subsec:MSDMSDCliff}. We first determine the space-time cost using the Single Error Detect $[12,11,2]$ code. The codeword generator matrix for this code is given in \cref{eq:CyclicG12} of \cref{app:codeconstruction}. If the measurements are performed sequentially, as in the layout of \cref{fig:15sed}, one routing space with access to the $X$ boundaries of all the qubits will suffice. For a faster distillation tile that performs measurements two at a time, the measurements may be performed using extra routing space as shown in \cref{fig:15sedparstag1}. Note that we now use the non-Clifford gadget of \cref{fig:Pnoncliff} to perform the $\pi / 8$ rotations since we perform distillation in the Clifford frame. This frees up space, as we do not need to allocate qubits for a $\ket{0}$ ancilla with $Y$ boundary access. This can reduce the routing space and the number of logical qubit cells needed. On the other hand, TELS incurs a larger space cost as more magic states need to be held in memory.

\begin{figure*}
    \centering
    \subfloat[\label{fig:timedynamicscircuit} ]{\includegraphics[width=0.93\textwidth]{imagesChap6/timedynamicsa2.pdf}}
    \hspace{0.2mm}
    \subfloat[\label{fig:timedynamicslatticesurgery} ]{\includegraphics[width=0.93\textwidth]{imagesChap6/timedynamicsb3.pdf}}
    \caption{Time dynamics of a TELS-assisted distillation factory. Here we consider a 15-to-1 distillation protocol with the $[12,11,2]$ code protecting temporally encoded lattice surgery of the $11$ non-Clifford gates, followed by the $4$ Pauli measurements protected with the $[5,4,2]$ code.  (a) The first part of the circuit consists of the Pauli measurements corresponding to the non-Clifford gates. We assume that the non-Clifford measurement results yield $110000000011$ and the eleven $\ket{T_X}$ state measurements yield $10000000001$.  The Clifford corrections, derived from \cref{fig:Pnoncliff}, are then conjugated through the final Pauli measurements. (b)~Sequence of lattice surgery measurements when both TELS-assisted protocols are combined.}
    \label{fig:timedynamics}
\end{figure*}

In \cref{fig:15sedparstag1}, we separate the routing space into grey and light blue regions to show non-intersecting routing areas for the two parallel multi-qubit Pauli measurements. Each of these routing spaces has access to the $X$ boundaries of all the data and magic state cells involved in the distillation. In \cref{fig:layoutSpec}, we show the routing space regions that are used to perform the lattice surgery measurements corresponding to the first two parallelizable Paulis when the $[12,11,2]$ code is used for TELS. In the first syndrome measurement round of the lattice surgery measurement, stabilizers with white vertices yield random outcomes due to the gauge fixing step~\cite{Vuillot19}. The lattice surgery measurement outcomes are then the error-corrected measurement values corresponding to the $X$ stabilizers in the respective routing regions. Using \cref{app:algoSTcosts}, we determined that $d_m = 5$ is sufficient for $\delta^{(M)} = 10^{-10}$ when using the $[12,11,2]$ code for TELS.

Access to only the $X$ boundaries of the data cells is sufficient for the first stage of the protocol, which is the temporally encoded measurements for the non-Clifford gates. In the second stage, we must perform multi-qubit Pauli measurements which are tensor products of $X$, $Y$ and $Z$. In \cref{fig:15sedparstage2}, we show how to perform these measurements on the same layout without the shuffling around of surface code patches. On the distilled magic state (bottom left blue cell), the multi-qubit measurements only need $X$ boundary access. The remaining data qubits will need at least one accessible $Z$ and $Y$ boundary. For the $15$-to-$1$ distillation protocol, there are at most $4$ multi-qubit $\pi/2$ Pauli measurements. Since these measurement may require access to different types of boundaries on each data cell, they cannot in general be performed in parallel. Hence the entire routing space (in light blue) is used to perform these measurements. As shown in \cref{subsec:MSDMSDCliff}, this set of measurements corresponds to the second PP set. Since there are four single-qubit measurements at the end of \cref{fig:Haddist}, there are at most $4$ measurements in this second PP set. We perform these measurements using TELS with a $[5,4,2]$ Single Error Detect code and with measurement distance $d_m = 5$.

When using the $[12,11,2]$ code for the lattice surgery measurements of the non-Clifford gates, we define $G$ as a cyclic code where at most two magic states need to be accessed simultaneously for each Pauli measurement. After each measurement, a magic state cell can be reset and reused for a future non-Clifford gate. If Pauli measurements are performed two at a time, three magic states will need to be concurrently held in memory. These are the solid pink tiles of \cref{fig:15sedparstag1}. In addition, for each subsequent pair of measurements, at least two injected magic states are required, which is why there is not just one additional magic state tile associated for injection but three. One of them may be removed, but then distillation tiles will either not be rectangular or will contain wasted physical qubits. In any case the extra cell for injection can ensure there is always a magic state injected and ready for a non-Clifford gate. Note however that since keeping track of Clifford frames requires occasionally performing $Y$ measurements when using the distilled magic states in the core, the extra pink radial cell could also be used to store the ancilla needed to perform the $Y$ measurement\footnote{Performing a $Y$ measurement on a surface code patch can be done in various ways. For instance, one can perform a logical phase gate, followed by measuring all the data qubits in the $X$ basis. However, performing a logical phase gate on a two-dimensional planar architecture with the surface code requires additional routing space and measurements involving twists (see for instance Fig. 23 of Ref.~\cite{PsiQuantumLatticeSurgery}). Alternatively, one could use an ancilla prepared in the logical $\ket{0}$ state, and perform a $Y \otimes Z$ measurement to get the parity of the $Y$ measurement outcome.}.

\subsubsection{Other TELS protocols}

In addition to the Single Error Detect code, we considered distillation with a distance-$3$ BCH code and the distance-$7$ classical Golay code. In \cref{fig:15bch3}, we show a layout for a distillation tile that performs TELS with a classical $[15,11,3]$ BCH code. The time cost can be decreased by performing the lattice surgery measurements corresponding to the non-Clifford gates two at a time. A layout with sufficient routing space for this is shown in \cref{fig:15bch3par}. Note that since there are two disjoint routing spaces (in blue and grey), the set of $n=15$ Pauli measurements can be performed in the time required for $8$ sequential measurements. To obtain the time cost shown in \cref{tab:STcosts}, we used $d_m = 3$.  Similarly, in \cref{fig:15golay,fig:15golaypar}, we show layouts for $15$-to-$1$ distillation tiles that perform TELS with a classical $[23,12,7]$ Golay code. Here, we used the parameters $d_m = 3$ and $c=2$ (where $c$ is the maximum weight of classical errors that are corrected in a TELS protocol) to obtain the time costs shown in \cref{tab:STcosts}. For the $15$-to-$1$ distillation protocol, implementing TELS using the classical Golay code does not allow for smaller space or time costs. However, for a small enough physical error rate, it is sufficient to consider a measurement distance $d_m=1$, which allows for a smaller time cost than any other lattice surgery protocol.

The space requirements of all the above layouts are described as functions of $d_x$ and $d_z$ in \cref{app:constants}. Additional constants in \cref{app:constants} can be used with the procedure of \cref{app:algoSTcosts} to determine all the minimum distances (spacelike and timelike) for the distillation protocols.

\subsection{Scheduling distillation tiles in a factory}
\label{subsec:MSDMSDscheduling}

\begin{figure*}
    \centering
    \includegraphics[width=0.95\textwidth]{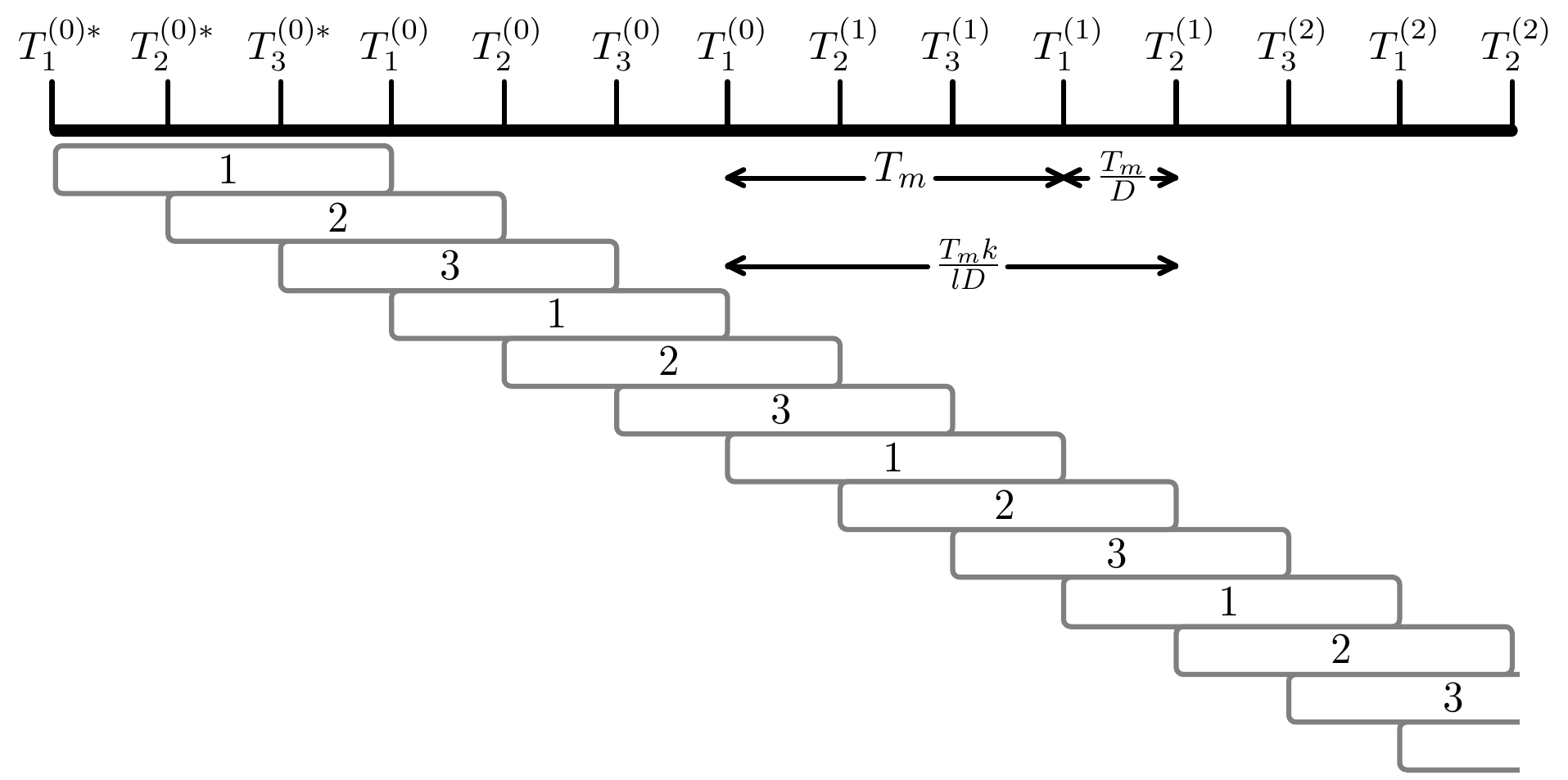}
    \caption{Round robin scheduling of deterministic-time distillation tiles in a factory. This scheduling method allows minimizing core wait time if a distillation tile rejects. The horizontal axis is time and the labels above are timestamps. Timestamp $T_j(i)$ indicates the end of the $j$th distillation tile while accumulating magic states for the $i$th PP set executed in the core. A * indicates the start of the distillation tile for the first time. In this example, there are $D=3$ distillation tiles, each producing $l=2$ distilled magic states in time $T_m$, for a core that executes PP sets of size $k=8$.}
    \label{fig:factoryscheduling}
\end{figure*}

Quantum computer architectures that perform Pauli-based computation generally contain two parts: a core and a magic state distillation factory. The core contains the data qubits taking part in the logical computation. It is known that TELS can speed up the runtime of PP sets executed in the core~\cite{Chamberland22}. In this paper, we also applied TELS to distillation circuits and observed reduced space-time costs. However distillation tiles are just modules that are used to construct a complete distillation factory, where many distillation tiles must be arranged with sufficient routing space to access the core. When merging a factory with a core, we are faced with additional scheduling and layout challenges. In this paper, we tackle the problem of scheduling, applying our speedups from TELS. We leave the design of the layouts of distillation factories for future work.

When executing an algorithm in the core using TELS on PP sets of size $k$, we denote the time to execute the algorithmic PP set as $T_{\text{PBC}}$. The distillation factory will simultaneously be working to distill at least $k$ new magic states for the next algorithmic PP set. It does this in time $T_{\text{magic}}$. General quantum algorithms operate on timescales much longer than $T_{\text{magic}}$, and so it is important to ensure the core is never idle and waiting for magic states. This situation, where the core is idling, is called a magic-state bottleneck. To avoid this bottleneck, we would like to ensure
\begin{equation}
\label{Cliff-bottleneck}
    T_{\text{magic}} \leq  T_{\text{PBC}} .
\end{equation}
The best case scenario is when $T_{\text{magic}} \approx  T_{\text{PBC}}$. As discussed in Ref.~\cite{Chamberland22}, this yields the smallest space-time cost for running algorithms.

We now wish to determine how many distillation tiles are required to satisfy the above condition. This can vary depending on the size of the PP set executed in the core and its relative speed-up. The number of tiles required also depends on the scheduling algorithm used in the factory, especially since some distillation tiles may detect errors and have to restart before producing magic states that can be used by the core. We first show how a simple algorithm can lead to large time costs when a distillation tile fails. Next we show how to use round robin scheduling to minimize additional time costs due to distillation tile failures. Note that the round robin scheduling algorithm can also be used when distilling lower level magic states in a concatenated distillation protocol.

Consider a situation where \cref{Cliff-bottleneck} is satisfied and $k$ is greater than or equal to the number of magic state storage cells (green cells in \cref{fig:15to1layouts}) in the factory. This implies each distillation tile takes less time to produce magic states than $T_{\text{PBC}}$. Hence a simple factory schedule would be to  start distillation on all tiles when a new PP set is beginning to be executed in the core. However, if the core is only marginally slower than the factory and a distillation tile rejects, the core will pause and wait for new magic states to be produced. Such a situation is undesirable since the core would now need to wait by a time $T_m$, where $T_m$ is the worst-case time needed to produce distilled magic states by a distillation tile.

In an attempt to reduce the core waiting time due to the rejection of a distillation tile (either due to TELS or the distillation algorithm), we suggest a round-robin approach. We assume that we have a distillation tile where, in the case where no errors are detected during the magic state distillation protocol, the tile produces $l$ magic states in time $T_m$ using a deterministic algorithm (TELS distillation is adaptive, but we consider the worst case time). The probability that a tile detects an error on an input magic state, or that the TELS protocol detects a timelike failure during lattice surgery is $p_D$.  If $D$ distillation tiles are used with round robin scheduling, the average time to distill $k$ magic states is
% \begin{align}
%     T_{\text{magic}} = & \frac{T_m k}{l D} + {k/l \choose 1} p_D (1-p_D) ^{k/l-1} \frac{T_m}{D}  \nonumber \\
%     & \quad+  {{k/l+1} \choose 2} p_D^2 (1-p_D)^{k/l-1} \frac{2 T_m}{D} + ... \nonumber \\
%     = & \frac{T_m k}{l D} + \bigg (\frac{p_D (1-p_D)^{k/l-1} T_m}{D} \bigg )\nonumber \\
%     & \quad \times \sum_{j=1}^{\infty} j {{k/l+j-1} \choose j}  p_D^{j-1} \nonumber \\
%     = & \frac{T_m k}{l D} + \bigg (\frac{p_D (1-p_D)^{k/l-1} T_m}{D}\bigg ) \frac{k}{l} (1-p_D)^{-k/l-1} \nonumber \\
%     = & \frac{T_m k (1-p_D+p_D^2)}{l D (1-p_D)^2} .
% \end{align}
\begin{align}
    T_{\text{magic}} = & \frac{T_m k}{l D} +\binom{k/l}{1}  p_D (1-p_D) ^{k/l-1} \frac{T_m}{D}   \quad+  \binom{k/l+1}{2} p_D^2 (1-p_D)^{k/l-1} \frac{2 T_m}{D} + ... \nonumber \\
    = & \frac{T_m k}{l D} + \bigg (\frac{p_D (1-p_D)^{k/l-1} T_m}{D} \bigg ) \times \sum_{j=1}^{\infty} j \binom{k/l+j-1}{j}  p_D^{j-1} \nonumber \\
    = & \frac{T_m k}{l D} + \bigg (\frac{p_D (1-p_D)^{k/l-1} T_m}{D}\bigg ) \frac{k}{l} (1-p_D)^{-k/l-1} \nonumber \\
    = & \frac{T_m k (1-p_D+p_D^2)}{l D (1-p_D)^2} .
\end{align}
In \cref{fig:factoryscheduling}, we show an example of the round robin scheduling algorithm with $3$ distillation tiles that each produce $2$ distilled magic states in time $T_m$. If the core executes PP sets with size $8$, the $8$ required magic states are distilled and prepared in time $\frac{T_m k}{l D}$, if no errors are detected.

From this we may solve for the number of distillation tiles required when $T_{\text{magic}} <  T_{\text{PBC}}$ with $D$ distillation tiles given by
\begin{equation}
D= \frac{T_m k (1-p_D+p_D^2)}{l T_{\text{magic}} (1-p_D)^2} .
\end{equation}

The shortcoming of this calculation is that it only applies to constant-time distillation tiles. If the tile takes adaptive time, such as in the magic state distillation protocol of \cref{subsec:MSDMSDCliff}, where there are a non-trivial number of extra measurements, it is unclear what the most efficient scheduling algorithm is. In this case, we may still upper bound the total time $T_m$ of a magic state distillation tile thus making the round robin scheduling algorithm applicable. However, adapting a scheduling algorithm to the type of distillation tile used could allow for a more precise calculation of the time cost of the factory, which could possibly reduce the required number of distillation tiles.

\backmatter
\phantomsection % To create a proper link in the table of contents
\begin{singlespace}
\bibliography{q}
\bibliographystyle{halpha-abbrv}
\end{singlespace}
\appendix
\chapter{Appendices}
\label{chap:appendices}
% \appendix

\section{Post-selective distance-three fault-tolerant cat state preparation}
\label{sec:SMcatPCSP}

% Since we look at postselection, distance $3$ here implies that we should be able to detect not just one fault, but also two faults. Remember that the definition of distance-$3$ error detection is that all faults of weight up to $d-1 = 3$ are detected.

Shor's method for fault-tolerant stabilizer measurement relies on the fault-tolerant preparation of a cat state by postselection.  In \figref{f:Shorcard}, the cat state is prepared fault-tolerantly to distance-two; it detects one fault.  For postselected distance-three fault tolerance, any one or two faults in the circuit must result in an error of weight at most one or two respectively, else the state must be rejected. In \figref{f:d3errordetection} we show how to prepare a weight-$12$ cat state fault-tolerantly to distance three---detecting up to two faults.

%\pagebreak%DEBUG

\begin{theorem} \label{t:errordetcatd3}
One ancilla qubit measured $m \geq 2$ times, can be used to prepare a cat state on $w$ qubits fault-tolerantly to distance three, detecting up to two faults,~for
\begin{equation}
 w \leq 3 \cdot 2^{m - 1} \, .
 \enspace 
\end{equation}
\end{theorem}

\begin{proof}
We explain the proof using the circuit in \figref{f:d3errordetection}.  The circuit passes with acceptable weight-one or weight-two errors when all the flag qubits are measured as $0$.  If one $X$ fault occurs on the $\ket +$ qubit during the preparation of the cat state, it may spread to a data error of weight~$> 1$. However the red flag qubit is triggered and the fault is detected.  If two $X$ faults occur on the $\ket +$ qubit, the red flag qubit may not catch it, yet a data error of weight~$> 2$ can exist on the cat state.  Since this scenario only arises from two faults, it suffices to check the parities between every third qubit of the cat state, as an error on two consecutive qubits is acceptable.  Higher-weight errors, such as the weight-seven error $X_2X_3 \mathellipsis X_8$ in \figref{f:d3errordetection} may not be detected by parity checks that have an even number of erroneous qubits. However these errors are always caught by other parity checks.

To check for errors of weight greater than two, we perform parity checks similar to that in \thmref{t:catstated3}. Instead of the flag sequence from \lemref{t:slowresetdistance3flagsequences}, the Gray code from \lemref{t:graycode} is used. Now the parities are computed between qubits $3j - 1$ for $j \in \{1, 2, \ldots, 2^{m - 1}\}$.  The first and the last qubits are not checked for errors and so with $m$ flags, the maximum cat state weight achieved is $3 \cdot 2^{m - 1}$.  
\end{proof}

\begin{figure}
    \centering
    \includegraphics[width=.5\linewidth]{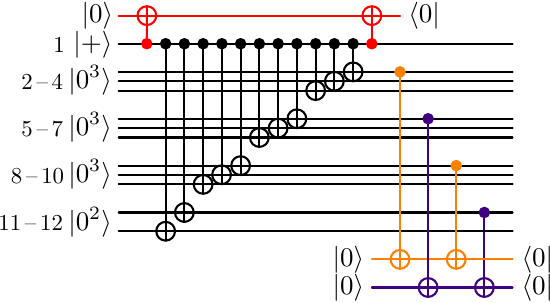}
    \caption{Two-error-detecting fault-tolerant circuit for the preparation of a weight-$12$ cat state. The state is only accepted when all flag qubits are measured as $0$.  Note that with fast reset, only one ancilla qubit is required.}
    \label{f:d3errordetection}
\end{figure}

% We also look at $d = 4$.

\section{Low-depth fault-tolerant cat state preparation}
\label{sec:catparallelized}

So far, we have focused on fault-tolerant preparation circuits with depth linear in the cat state weight.  In this section, we detail how to prepare cat states fault-tolerantly in logarithmic depth.

\begin{figure}
    \centering
    \subfloat[\label{f:parallelcsp8}]{
        \includegraphics[width = 0.25\textwidth]{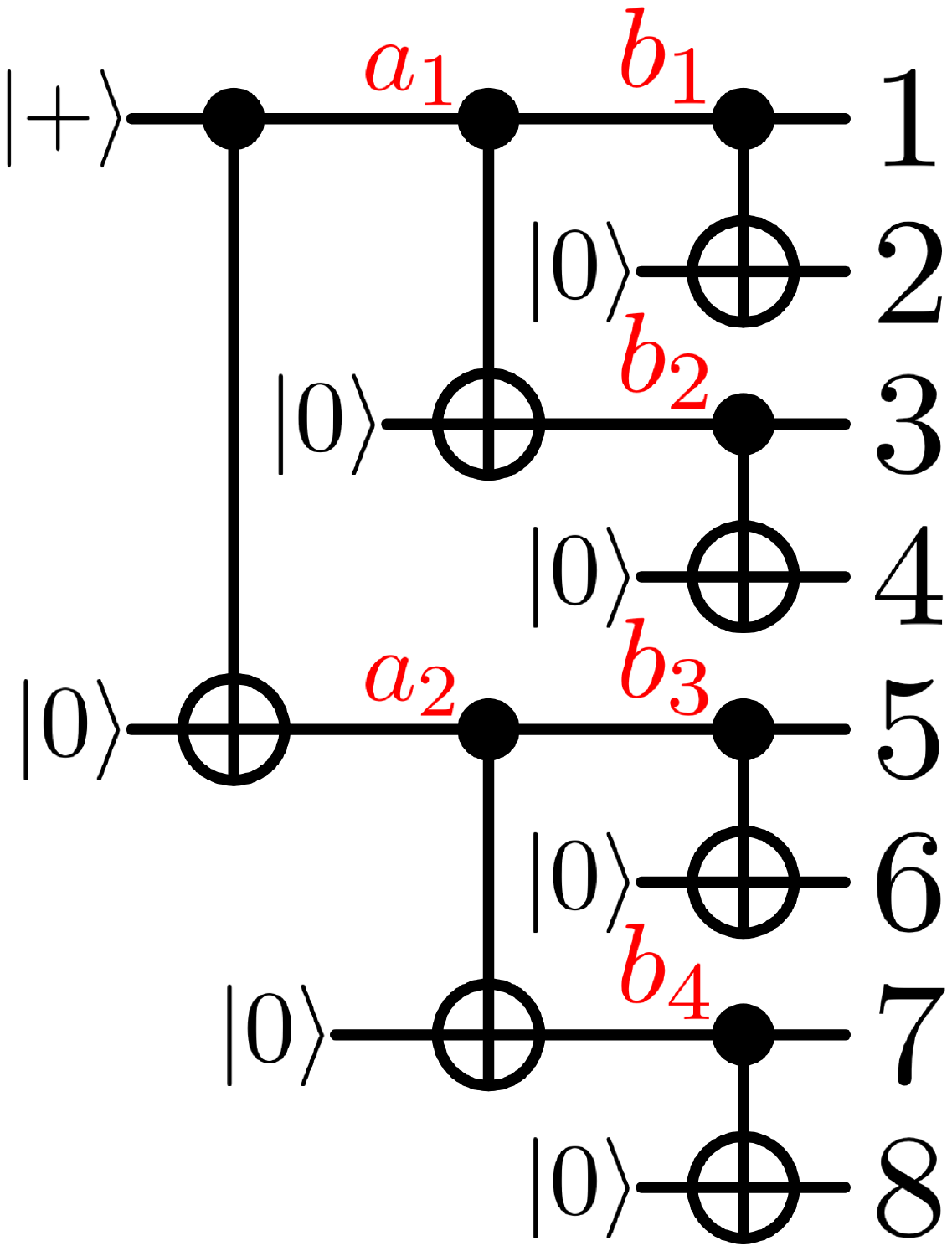}}
    \hspace{0.5cm}
    \subfloat[\label{f:parallelcsp8graph}]{
        \includegraphics[width = 0.33\textwidth]{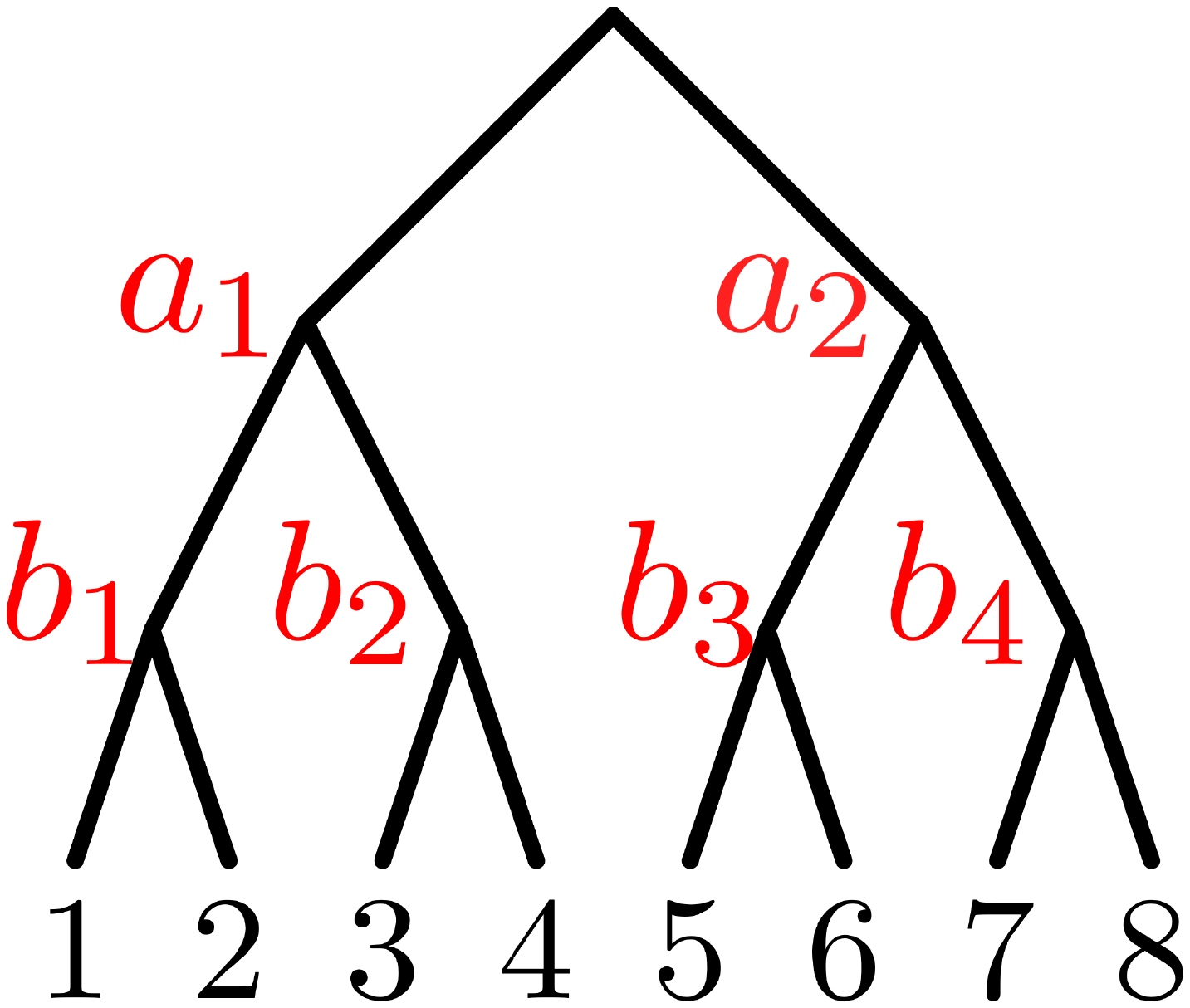}}
    \caption{(a)~Logarithmic-depth preparation of an eight-qubit cat state shows there are six possible locations for $X$ faults that create errors of weight at least two.  Parity checks need to be chosen to find corrections that leave the cat state with error of weight less than two. (b)~The circuit on the left can be represented as a graph, where a CNOT gate is represented by the splitting of an edge.} 
\end{figure}

In~\figref{f:parallelcsp8} an eight-qubit cat state is prepared in three rounds of CNOT gates.  There are six locations (marked in red) where an $X$ fault may cause an error of weight at least two.  These faults result in data errors with a different structure from the linear-depth protocols of \secref{sec:SMcatFDCSP}, hence different parity checks are required.  It is simpler to determine these parity checks if the circuit is viewed as a binary tree, as in~\figref{f:parallelcsp8graph}.  Here time flows down and every CNOT onto a fresh $\ket 0$ qubit is denoted by the splitting of an edge.  An $X$ fault at a marked location results in an $X$ error on all the leaf nodes directly under the location.  Note that a fault at the root cannot cause a bad error.  

We use only two-qubit parity checks, however larger parity checks may be used at the expense of increased depth.  If a parity check checks qubit $x$, it provides information on whether a fault occurred anywhere in the lineage: $l(x) = \{x, \pa(x), \pa(\pa(x)), \ldots , \ro\} $.  Therefore, if a parity check $(x,y)$ is triggered, a fault at one of the locations $ l(x) \cup l(y) \cup \{\spam \}$ has occurred, where $\{\spam \}$ is the set of faults during state preparation or measurement of the parity-check qubit. 

Using the parity checks $(1,5), (2,7), (3,6), (4,8)$, it is possible to separate the five distinct weight at least two errors (since the error due to $a_1$ and $a_2$ is the same up to the cat state's $X^{\otimes w}$ stabilizer) into distinct triggered flag~patterns:

\begin{center}
    \begin{tabular}{c @{\hspace{0.5cm}} c c c c}
     & $(1,5)$ & $(2,7)$ & $(3,6)$ & $(4,8)$\\
     \hline 
    $a_1$, $a_2$ & $\bullet$ &  $\bullet$ &  $\bullet$ & $\bullet$ \\
    $b_1$ &  $\bullet$ &  $\bullet$ & $\circ$ & $\circ$\\
    $b_2$ & $\circ$ & $\circ$ & $\bullet$  &  $\bullet$\\
    $b_3$ &  $\bullet$ & $\circ$ &  $\bullet$ & $\circ$\\
    $b_4$ & $\circ$ &  $\bullet$ & $\circ$ & $\bullet$ \\
    \end{tabular}
\end{center}

Note that a fault at any of the above locations requires a multi-qubit data correction.  We ensure that each of them is detected by at least two parity checks, as one faulty parity check must not induce corrections of weight greater than one.

\begin{theorem} \label{t:parallelcatd3s}
Using parallelized circuits, a $w$-qubit cat state can be prepared fault-tolerantly to distance three using $\frac{w}{2}$ parity checks, where $\frac{w}{2}= 2^j, \, j \in \mathbb{N}$.  The depth of the circuit is $2+ \log_2 w$.
\end{theorem}

\begin{proof}

For parity check $i \in \{1,2, \mathellipsis , \frac{w}{4} \}$, the cat state qubits checked are $(i, \frac{w}{2} + 2i - 1)$.  For the remaining parity checks  $i \in \{ \frac{w}{4} +1, \frac{w}{4} +2, \mathellipsis ,  \frac{w}{2} \}$, the qubits checked are $(i, 2i)$.  As in \figref{f:parallelcsp8graph}, faults at the ${\color{red}a}$ level (depth-one) locations trigger all the parity checks, since each parity check is executed on one cat state qubit from the first half, and one from the second.  The correction $X^{\otimes w/2}$ on either half of the qubits works for both faults as $(X^{\otimes w/2} \otimes \Id^{\otimes w/2})( \Id^{\otimes w/2} \otimes X^{\otimes w/2}) = X^{\otimes w}$ is a stabilizer of the cat state.  Faults at the ${\color{red}b}$ level (depth-two) trigger distinct sets of $\frac{w}{2^2}$ parity checks, where the correction is on all the leaf nodes under the uniquely identified fault.  The same holds for faults at depth-$k$, which trigger distinct sets of $\frac{w}{2^k}$ parity checks.

One faulty parity check leads to a weight-one flag pattern, for which we do not apply corrections, as the error is restricted to at most one cat state qubit.  
\end{proof}

% \section{Multiple cat state preparation}
% \label{sec:SMcatSDSM7}

% \section{Magic state Distillation by stabilizer measurements}
% \label{sec:magicbymeas}

% \section{Fault-tolerant puncturing and uses}
% \label{sec:punturing}

\section{Corrections and rejections for weight-eight stabilizer measurements}
\label{s:wt8corrs}

\begin{figure}
    \centering
    \subfloat[\label{f:degree4k2wt8circ} ]{\includegraphics[width=.7\textwidth]{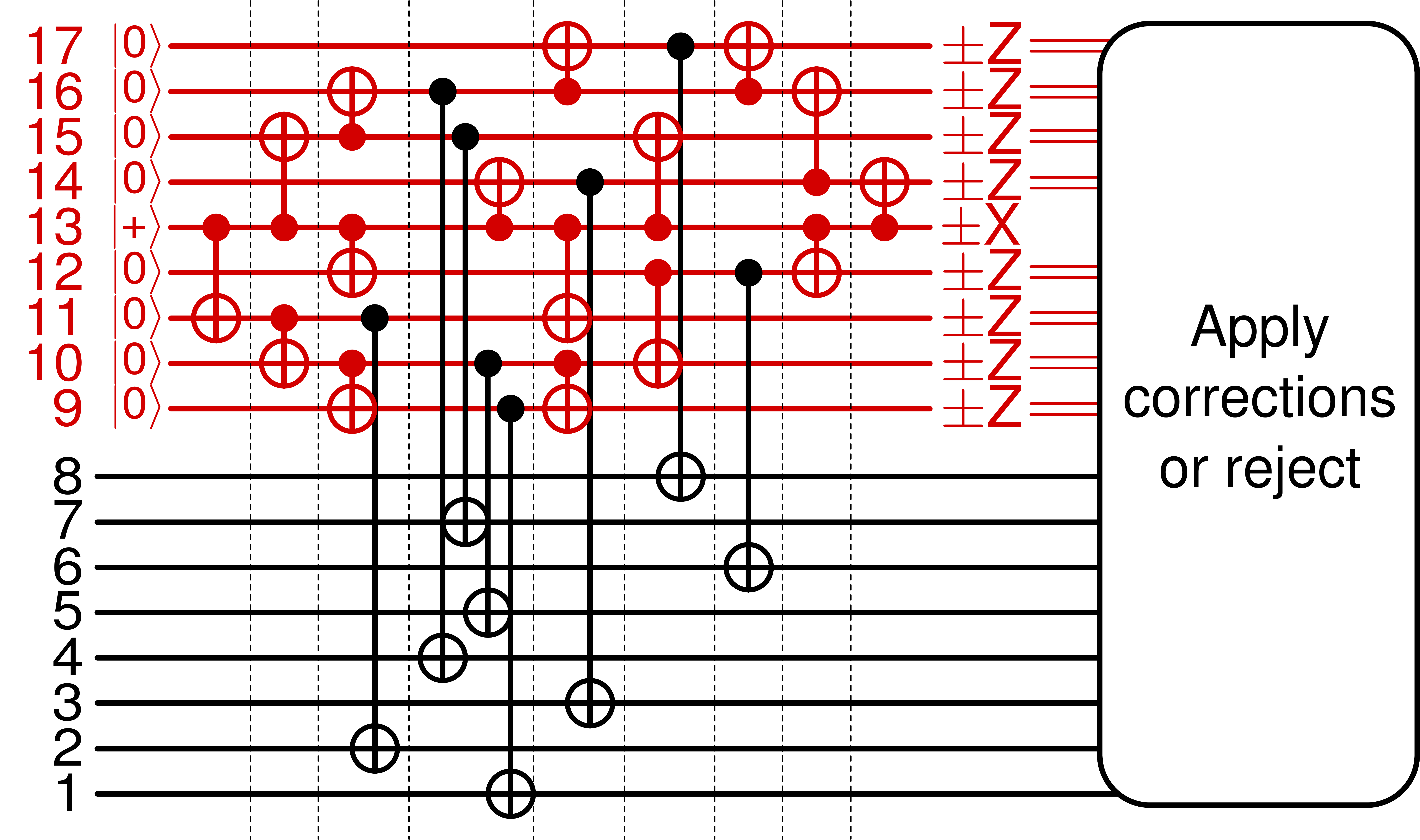}}
    \hspace{.2cm}
    \subfloat[\label{f:degree4k2wt8layout}
    ]{\includegraphics[width=.25\textwidth]{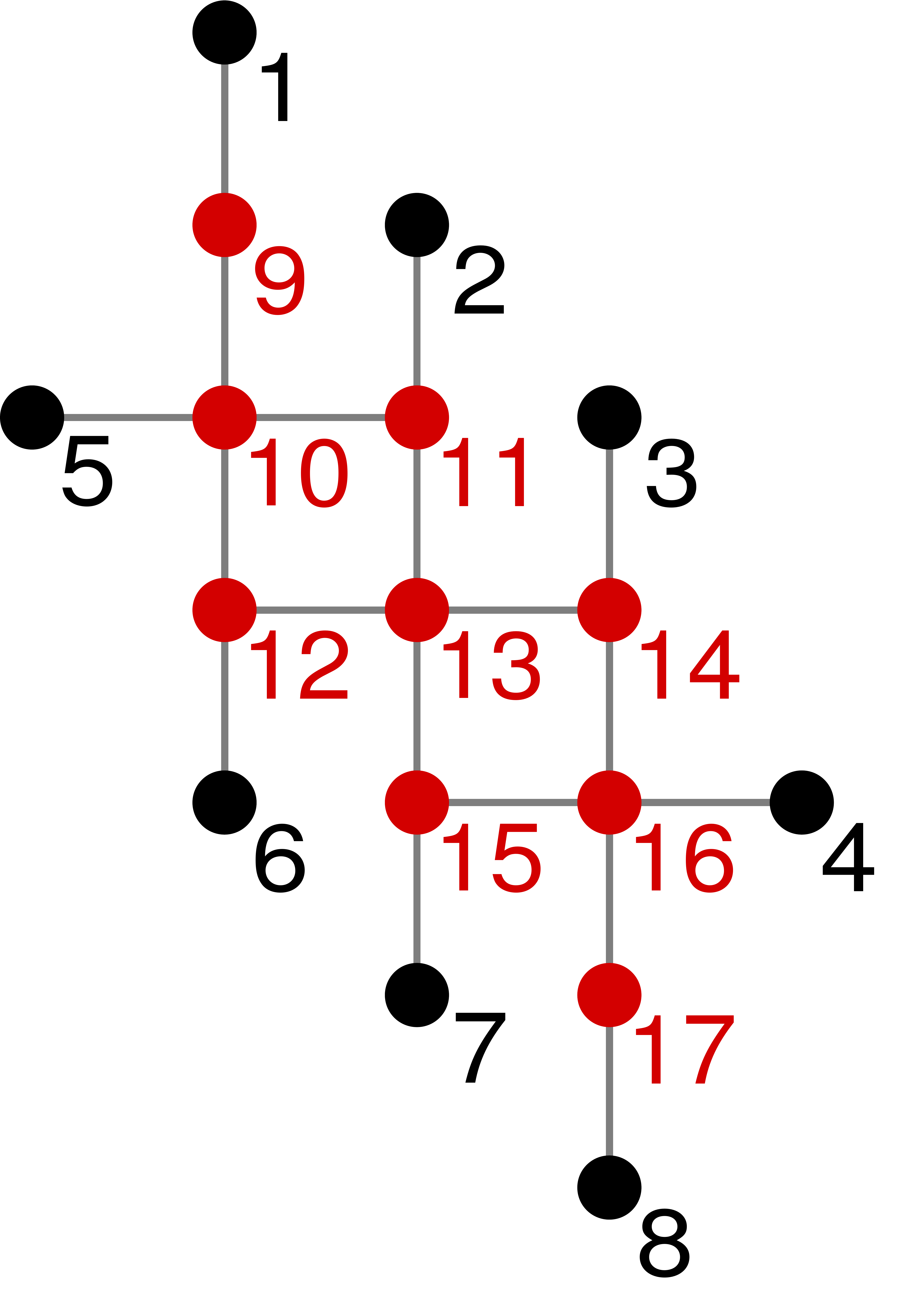}}
\caption{(a)~Distance-4 fault-tolerant circuit for measuring a weight-8 stabilizer on a square lattice layout, as arranged in~(b).}
    \label{f:degree4k2wt8}
\end{figure}

The corrections and rejections for the fault-tolerant weight-eight stabilizer measurement circuit in~\figref{f:wt8circ} are shown below.

\noindent
\setlength{\tabcolsep}{5pt}
\begin{center}
\begin{tabular}{cc|cc}
Raised flags & Correction & Raised flags & Correction  \\
\hline 
\{$13$\} & \{$6$\} & \{$12,13$\} & \{$6$\} \\
\{$9,12$\} & \{$1,5$\} & \{$12,13,14$\} & \{$6,7,8$\} \\
\{$9,11,12,14$\} & \{$1,4,5$\} & $$ & $$ \\ 
\end{tabular} 
\end{center}

%\vspace{0.1cm} %DEBUG

\noindent
\begin{center}
\begin{tabular}{c}
Rejections \\
\hline
\multicolumn{1}{p{0.97\linewidth}}{\{$9,11$\}, \{$9,13$\}, \{$9,14$\}, \{$11,12$\}, \{$12,14$\}, \{$13,14$\}, \{$9,11,12$\}, \{$9,11,14$\}, \{$9,12,13$\}, \{$9,12,14$\}, \{$9,13,14$\}, \{$11,12,13$\}, \{$11,12,14$\}, \{$11,13,14$\}, \{$9,11,12,13$\}, \{$9,11,13,14$\}, \{$9,12,13,14$\},  \{$11,12,13,14$\}, \{$9,11,12,13,14$\} }
\end{tabular}
\end{center}

\vspace{0.25cm} %DEBUG

For the layout described in~\figref{f:degree4k2}, the weight-eight stabilizer is measured fault-tolerantly with the circuit in~\figref{f:degree4k2wt8}, with associated corrections and rejections tabulated below.

\vspace{0.25cm} %DEBUG

\noindent
\setlength{\tabcolsep}{5pt}
\begin{center}
\begin{tabular}{cc|cc}
Raised flags & Correction & Raised flags & Correction  \\
\hline 
\{$10$\} & \{$1$\} & \{$11$\} & \{$2$\} \\
\{$12$\} & \{$6$\} & \{$14$\} & \{$8$\} \\
\{$15$\} & \{$7$\} & \{$16$\} & \{$8$\} \\
\{$10,11$\} & \{$1,2$\} & \{$10,12$\} & \{$1,6$\} \\
\{$15,16$\} & \{$7,8$\} & \{$10,11,15,16$\} & \{$3$\}  \\ 
\end{tabular} 
\end{center}

%\vspace{1mm} %DEBUG
\noindent
\begin{center}
\begin{tabular}{c}
Rejections \\
\hline
\multicolumn{1}{p{0.97\linewidth}}{
 \{$9, 11$\},  \{$ 9, 12$\},  \{$9, 15$\},  \{$9, 16$\},  \{$9, 17$\},  \{$10, 14$\},  \{$10, 15$\},  \{$10,  
16$\},  \{$10, 17$\},  \{$11, 12$\},  \{$11, 15$\},  \{$11, 16$\},  \{$11, 17$\},  \{$12, 15$\},  \{$12,  
16$\},  \{$12, 17$\},  \{$14, 15$\},  \{$15, 17$\},  \{$9, 10, 11$\},  \{$9, 10, 12$\},  \{$9, 10,  
15$\},  \{$9, 10, 16$\},  \{$9, 10, 17$\},  \{$9, 14, 16$\},  \{$9, 15, 16$\},  \{$10, 11,  
12$\},  \{$10, 11, 14$\},  \{$10, 11, 15$\},  \{$10, 11, 16$\},  \{$10, 11, 17$\},  \{$10, 12,  
14$\},  \{$10, 12, 15$\},  \{$10, 12, 16$\},  \{$10, 12, 17$\},  \{$10, 14, 16$\},  \{$10, 15,  
16$\},  \{$10, 16, 17$\},  \{$11, 12, 15$\},  \{$11, 12, 16$\},  \{$11, 14, 16$\},  \{$11, 15,  
16$\},  \{$12, 14, 16$\},  \{$12, 15, 16$\},  \{$14, 15, 16$\},  \{$14, 16, 17$\},  \{$15, 16,  
17$\},  \{$9, 10, 14, 16$\},  \{$9, 10, 15, 16$\},  \{$9, 11, 15, 16$\},  \{$10, 11, 12,  
14$\},  \{$10, 11, 14, 15$\},  \{$10, 11, 14, 16$\},  \{$10, 11, 15, 17$\},  \{$10, 11,  
16, 17$\},  \{$10, 12, 14, 15$\},  \{$10, 12, 14, 16$\},  \{$10, 12, 15, 16$\},  \{$11,  
12, 14, 16$\},  \{$12, 14, 15, 16$\},  \{$9, 10, 11, 15, 16$\},  \{$9, 11, 12, 15,  
16$\},  \{$10, 11, 12, 14, 15$\},  \{$10, 11, 12, 14, 16$\},  \{$10, 11, 12, 15,  
16$\},  \{$10, 11, 14, 15, 16$\},  \{$10, 11, 15, 16, 17$\},  \{$11, 12, 14, 15,  
16$\},  \{$11, 12, 15, 16, 17$\},  \{$9, 10, 11, 12, 15, 16$\},  \{$10, 11, 12, 14,  
15, 16$\} }
\end{tabular}
\end{center}

\section{Malignant set counting}
\label{app:malignantsets}

An $[n,k,d]$ binary classical error-correcting code encodes $k$ logical bits of information into $n \geq k$ physical bits, with distance $d$. During error detection/correction, all errors of weight less than $d$ are detected. However \textbf{some} of the weight-$d$ errors are not detected. These errors are called malignant sets as they can cause erroneous flips of the logical bits. The task of computing how many of the $\binom{n}{d}$ weight-$d$ bit strings are malignant is computationally hard. The deterministic method is to evaluate the weights of all $\binom{n}{d}$ bit strings. But this takes time that is exponential in the problem size. 

For larger codes, we searched for faster methods to estimate the number of malignant fault sets. The first was a Monte Carlo simulation. The second method modelled the malignancy of weight-$d$ errors using a Bernoulli random variable. Finally, a third method used the MacWilliams identity. 

\subsection{Monte-Carlo sampling}
\label{subapp:monte}

For physical bit error rate $p$, the logical bit error rate of an $[n,k,d]$ code is $p_L = \sum_{j=d}^{n-d} l_j p^j (1-p)^{n-j}$, where $l_j$ is the number of malignant sets of weight $j$. At sufficiently low $p$, $p_L$ is approximately the first term of the polynomial, $l_d p^d (1-p)^{n-d}$. 

We can estimate $l_d$ using Monte Carlo simulations in two steps:
\begin{enumerate}
    \item For different, small values of $p$, compute $p_L$ by sampling errors and evaluating the fraction of them that are malignant. An $n$-bit error sample $e$ is obtained by sampling each bit from a Bernoulli random variable with probability $p$. The error $e$ is malignant if $H e = 0$, where $H$ is the parity check matrix of the code.
    \item Perform a least squares fit of the obtained values with the polynomial $p^d (1-p)^{n-d}$. The coefficient of the fit is the Monte Carlo approximation of $l_d$. 
\end{enumerate} 

At sufficiently low $p$, many of the error samples will be trivial. Hence a lot of time is wasted evaluating these samples. For large $d$, this problem becomes worse. The probability of observing a weight-$d$ error scales as $p^d$, implying that the errors that actually may be malignant are rarely ever observed.

\subsection{Modelling malignancy with the Bernoulli distribution}
\label{subapp:bernoulli}

Since we only care to check whether a weight-$d$ error is malignant or not, it is faster to sample only from the set of weight-$d$ errors. Let an error be sampled by choosing $d$ out of $n$ locations at random without replacement. We can now model the malignancy of a weight-$d$ error using a Bernoulli random variable: a weight-$d$ error sample is malignant with probability $p$. We can estimate $p$ with high confidence by checking for malignancy on many samples. Finally, $l_d = p \binom{n}{d}$.

\subsection{MacWilliams identity}
\label{subapp:macwilliams}

For a code $C$ with $k$ codeword generators, there exist $n-k$ vectors spanning the nullspace (kernel), $C^{\perp}$.  If the weights of the $\abs {C^\perp}= 2^{n-k}$ bit strings can be enumerated, then the weights of the codewords of $C$ can be evaluated using the MacWilliams identity:
\begin{equation}
    W_j^{C} = \frac{1}{\abs C^\perp} \sum_{i=0}^n W_i^{C^\perp} K_j(i,n),
\end{equation}
for $j= 0,1,\mathellipsis , n$. Here $W_i^C$ is the number of codewords of $C$ of weight $i$ and $K_j(i,n)$ is the Krawtchouk polynomial 
\begin{equation}
    K_j(i,n) = \sum_{l=0}^j (-1)^l \binom{i}{l} \binom{n-i}{j-l}.
\end{equation}

Then $W_d^C$ is the number of malignant fault sets of weight-$d$.

A short Mathematica script can enumerate the weights of  $2^{36} = 68{,}719{,}476{,}736$ codewords in just under $24$ hours.

\section{Construction of classical codes}
\label{app:codeconstruction}

\subsection{Cyclic codes defined using polynomials}

Binary cyclic codes can be constructed using cyclic shifts of polynomials defined over finite fields. To understand why, first note the isomorphic map between the field $\mathbb{F}_2^n$ and polynomials of degree $<$ $n$ with coefficients in $\mathbb{F}_2$.  For example, in $\mathbb{F}_2^4$, the polynomial $x^3 + 0 x^2 + x + 1$ corresponds to the bit string $1011$, where the bit at position $i\in \{ 0,1,\mathellipsis, n-1 \}$ (right-to-left) is the coefficient of the term $x^i$. An $[n,k,d]$ code is cyclic if the $k$ codewords can be generated by cyclic shifts of a generator polynomial, $g(x)$. Cyclic shifts of $g(x)$ are obtained by multiplying $g(x)$ with $\{ 1,x,x^2, \mathellipsis , x^{k-1} \}$.

\subsubsection{Single Error Detect code}
\label{subsubsec:SED}

The Single Error Detect code with parameters $[\alpha +1, \alpha, 2]$ can be generated by taking cyclic shifts of the generating polynomial $g(x)= x+1$ over the field $\mathbb{F}_2^{\alpha+1}$. For example the codewords of the $[4,3,2]$ code are the rows of $G$ below.
\begin{equation}
    G = \begin{bmatrix} 
0011 \\
0110 \\ 
1100 \\
\end{bmatrix}.
\end{equation}
The parity check matrix $H$ is the nullspace of $G$, i.e. the span of all vectors in $\mathbb{F}_2^n$ that are orthogonal to elements of $G$. For the above code,
\begin{equation}
     H = \begin{bmatrix} 
1111 \\
\end{bmatrix}.
\end{equation}

We also display the codeword generator matrix for the $[12,11,2]$ code which we use in \cref{sec:15to1} for magic state distillation.
\begin{equation}
    G = \begin{bmatrix} 
000000000011 \\
000000000110 \\
000000001100 \\
000000011000 \\
000000110000 \\
000001100000 \\
000011000000 \\
000110000000 \\
001100000000 \\
011000000000 \\
110000000000 
\end{bmatrix}.
\label{eq:CyclicG12}
\end{equation}

\begin{table*}
    \centering
    \begin{tabular}{c c}
         Code & Generator polynomial \\
         \hline
         $[7,4,3]$ & $1011 = x^3+x+1$ \\
         $[15,11,3]$ & $10011$ \\
         $[31,26,3]$ & $ 100101 $ \\
         $[43,36,3]$ & $ 10101011 $ \\
         $[49,43,3]$ & $ 1000011 $ \\
         $[63,57,3]$ & $ 1000011 $ \\
         $[85,77,3]$ & $ 100011101 $ \\
         $[127,120,3]$ & $ 10000011 $ \\[.25cm]
         
         $[15,7,5]$ & $ 111010001 $ \\
         $[31,21,5]$ & $ 11101101001 $ \\
         $[43,29,5]$ & $ 100111110100011 $ \\
         $[49,37,5]$ & $ 1010100111001 $ \\
         $[63,51,5]$ & $ 1010100111001 $ \\
         $[85,69,5]$ & $ 10110111101100011 $ \\
         $[127,113,5]$ & $ 101010001111101 $ \\[.25cm]
         
         $[15,5,7]$ & $ 10100110111$ \\
         $[31,16,7]$ & $ 1000111110101111 $ \\
         $[43,22,7]$ & $ 1010010100110010100001 $ \\
         $[49,31,7]$ & $ 1111000001011001111 $ \\
         $[63,45,7]$ & $ 1111000001011001111 $ \\
         $[85,61,7]$ & $ 1101110111010000110110101 $ \\
         $[127,106,7]$ & $ 1010010011000000011011 $ \\[.25cm]
         
         $[49,25,9]$ & $ 1110110110010011101110111 $ \\
         $[63,39,9]$ & $ 1110110110010011101110111 $ \\
         $[85,53,9]$ & $ 111101110010110110100001011111101$ \\
         $[127,99,9]$ & $ 11000101001010111100100111111 $ \\[.25cm]
         
         $[43,15,10]$ & $ 11111110001001100100000101011 $ \\[.25cm]
         
         $[31,11,11]$ & $ 101100010011011010101 $ \\
         $[49,22,11]$ & $ 1000011011101000000100010011 $ \\
         $[63,36,11]$ & $ 1000011011101000000100010011 $ \\
         $[85,45,11]$ & $ 10011001101111101110100111010110100010001$ \\
         $[127,92,11]$ & $ 111000010001110010101001101101010111 $ \\
    \end{tabular}
    \caption{BCH codes and associated generator polynomials. The codewords generators are  cyclic shifts of the generator polynomial.}
    \label{tab:BCHpolynomials}
\end{table*}

\begin{table*}
    \centering
    \begin{tabular}{c c c}
         $u$ & Code & Generator Polynomial  \\
         \hline
        $ 3$ & $[9,2,6]$ & $ 10111101 = x^7+x^5+x^4+x^3+x^2+1$ \\
        $ 4$ & $[17,9,5]$ & $ 100111001$ \\
        $ 5$ & $[33,22,6]$ & $ 101001100101 $ \\
        $ 6$ & $[65,53,5]$ & $ 1000111110001 $ \\
        $ 7$ & $[129,114,6]$ & $ 1001010000101001$ \\
    \end{tabular}
    \caption{Zetterberg codes with associated generator polynomials.}
    \label{tab:Zettpolynomials}
\end{table*}

\subsubsection{Golay code}
\label{subsubsec:Golay}

The $[23,12,7]$ Golay code is a cyclic code generated by the polynomial $x^{11}+x^9+x^7+x^6+x^5+x+1$ over $\mathbb{F}_2^{23}$. 

\begin{equation}
G = \begin{bmatrix} 

000000000001 0 1 0 1 1 1 0 0 0 1 1 \\
00000000001 0 1 0 1 1 1 0 0 0 1 1 0 \\
0000000001 0 1 0 1 1 1 0 0 0 1 1 00 \\
000000001 0 1 0 1 1 1 0 0 0 1 1 000 \\
00000001 0 1 0 1 1 1 0 0 0 1 1 0000 \\
0000001 0 1 0 1 1 1 0 0 0 1 1 00000 \\
000001 0 1 0 1 1 1 0 0 0 1 1 000000 \\
00001 0 1 0 1 1 1 0 0 0 1 1 0000000 \\
0001 0 1 0 1 1 1 0 0 0 1 1 00000000 \\
001 0 1 0 1 1 1 0 0 0 1 1 000000000 \\
01 0 1 0 1 1 1 0 0 0 1 1 0000000000 \\
1  0 1 0 1 1 1 0 0 0 1 1 00000000000 \\
\end{bmatrix}.
\end{equation}

\pagebreak %DEBUG
\subsubsection{BCH codes}
\label{subsubsec:BCH}

Bose-Chaudhuri–Hocquenghem (BCH) codes are a well-studied family of classical cyclic codes constructed using polynomials over finite fields. Due to the flexible nature of the construction of these codes, codes of different distances can be defined for the same code size. In \cref{tab:BCHpolynomials}, we show the generating polynomials for the BCH codes that were considered in this paper.

\subsubsection{Zetterberg codes}

Zetterberg codes are binary cyclic codes defined as $[2^u+1,2^u+1-2u,5 \leq d \leq 6]$ codes for even $u$. For odd $u$, we obtain the parameters $[2^u+1,2^u-2u,6]$. These codes are quasi-perfect: the distance between two codewords is $5 \leq d \leq 6$.

In this paper we consider Zetterberg codes for $u\in \{3,4,5,6,7\}$. The codewords are defined by taking cyclic shifts of polynomials shown in \cref{tab:Zettpolynomials}. Note that all the polynomials are palindromic. For a chosen $u$, other codes with the same parameters may be defined using different palindromic polynomials of the same degree. For a more detailed description of the construction, consider~\cite{Jing10}.

\subsection{Reed-Muller and Polar codes}

Binary Reed-Muller codes are $[ 2^m, k ,2^{m-r} ]$ codes for $r \leq m$ where 
\begin{equation}
    k= 2^m-\sum_{i=0}^{m-r-1} \binom{m}{i} = \sum_{i=0}^{r} \binom{m}{i}.
\end{equation}
To determine the codewords of the $(r,m)$-Reed-Muller code, start with the m-fold tensor product of the generator matrix $\begin{bmatrix} 1 & 1 \\ 0 & 1\end{bmatrix}$.
Remove the $\sum_{i=0}^{m-r-1} \binom{m}{i}$ rows with fewer than $d$ $1$'s. The $k$ remaining rows denote the codewords. In this paper, we look at the family of $m=r+1$ Single Error Detect codes, $m=r+2$ distance-$4$ Extended Hamming codes, and $m=r+3$ distance-$8$ codes.

Polar codes are $2^m$-bit codes that were initially developed for communication systems to tackle analog noise. The binary codes constructed using this formalism can be used against discrete noise too. The method of construction is the same as that of the Reed-Muller codes, but allows for codes with fewer encoded bits. After removing low-weight codewords to create a Reed-Muller code, remove extra codewords (lowest weight first) until there are exactly as many encoded bits as required.

\section{Speedups offered by different codes}
\label{app:speedups}

In \cref{fig:p3delt152025}, we show the lowest average runtime per Pauli for temporally encoded lattice surgery of $k \in \{ 2,3, \mathellipsis, 100\}$ measurements for $p=10^{-3}$. We also indicate which code achieves the lowest average runtime per Pauli for each $k$. Note that this paper only considered a limited number of classical codes for TELS. It may be possible for other codes to perform better than the ones outlined here.
\cref{fig:p4delt1520} shows the best classical codes for $p=10^{-4}$ and $\delta = 10^{-15}$ and $\delta = 10^{-20}$ respectively.

In \cref{tab:k2to50}, we show the best average speedup due to a TELS code for $k \in \{ 2,3, \mathellipsis, 100\}$. These speedups are computed with respect to performing the $k$ measurements sequentially at the regular measurement distance $d_m$. These speedups are computed for various regimes: $p = 10^{-3}$, $\delta \in \{ 10^{-10}, 10^{-15}, 10^{-20}, 10^{-25}\}$ and $p = 10^{-4}$, $\delta \in \{ 10^{-15}, 10^{-20}\}$.

\begin{figure*}
    \centering
    \hspace{-1mm}
    \subfloat[\label{fig:delt15}]{\includegraphics[width=.98\textwidth]{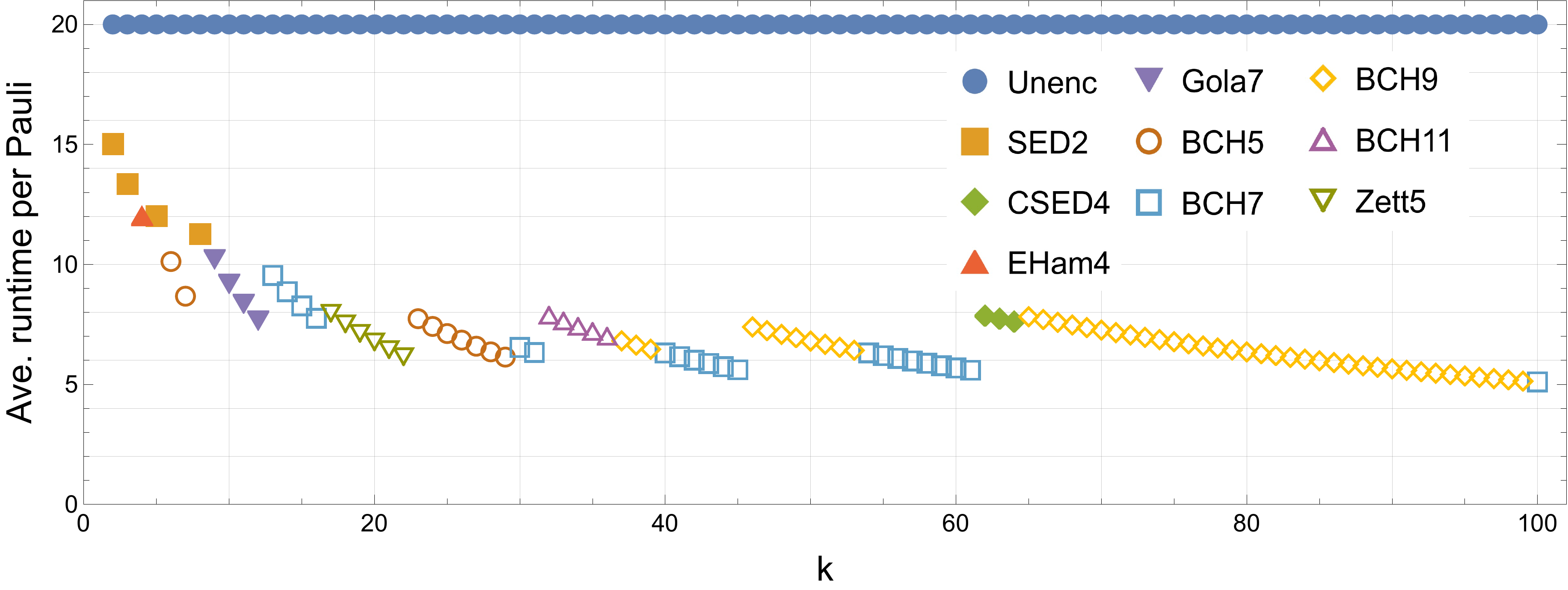}}
    \hspace{0.1cm}
    \subfloat[\label{fig:delt20}]{\includegraphics[width=.98\textwidth]{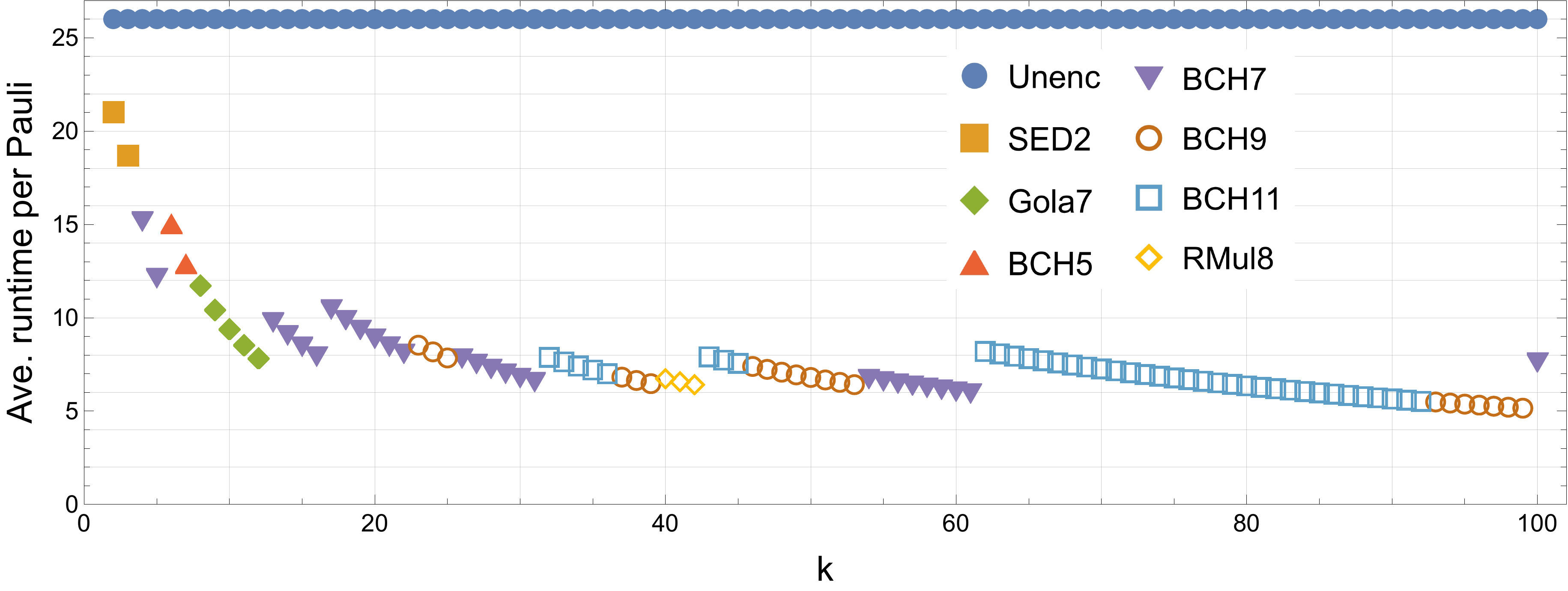}}
    \hspace{0.1cm}
    \subfloat[\label{fig:delt25}]{\includegraphics[width=.98\textwidth]{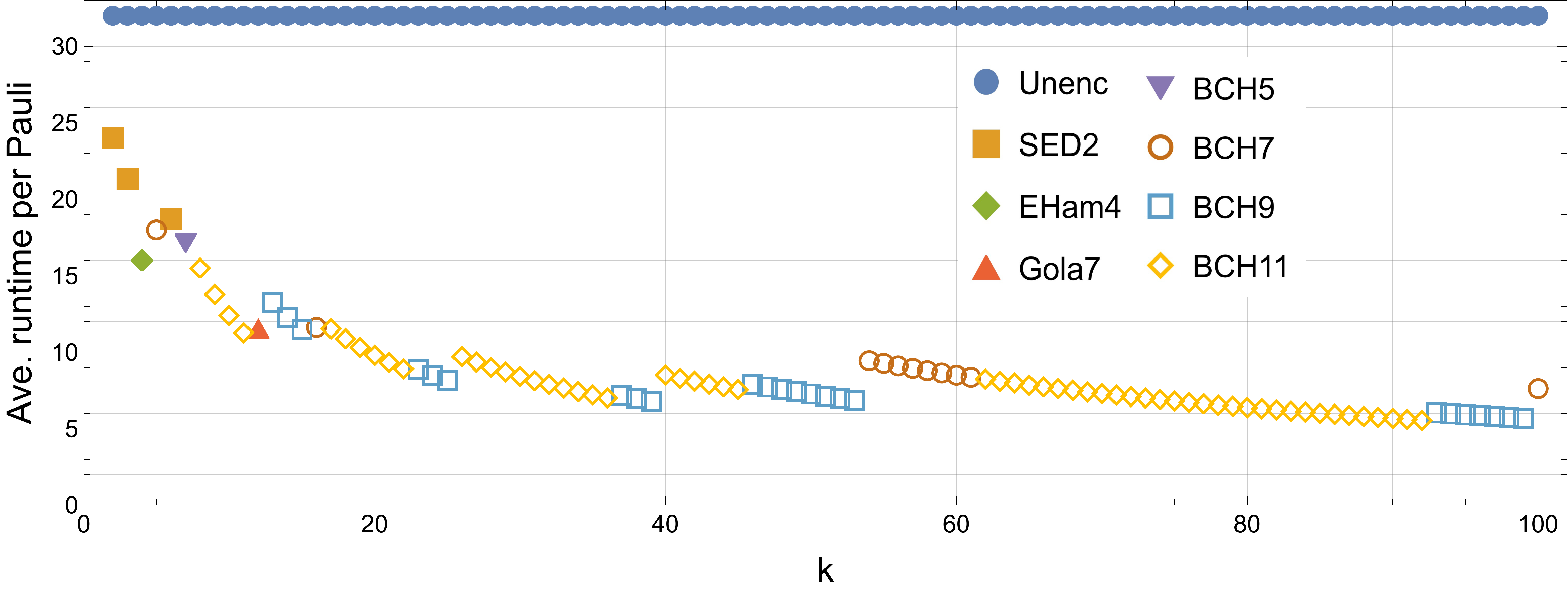}}
    \caption{We show the classical codes achieving the lowest average runtime per Pauli for $k \in \{ 2,3, \mathellipsis, 100\}$ at $p=10^{-3}$ and for (a)~$\delta= 10^{-15}$, (b)~$\delta= 10^{-20}$ and (c)~$\delta= 10^{-25}$. We set the routing space area $A = 100$. }
    \label{fig:p3delt152025}
\end{figure*}

\begin{figure*}
    \centering
    \hspace{-1mm}
    \subfloat[\label{fig:delt15p4}]{\includegraphics[width=.98\textwidth]{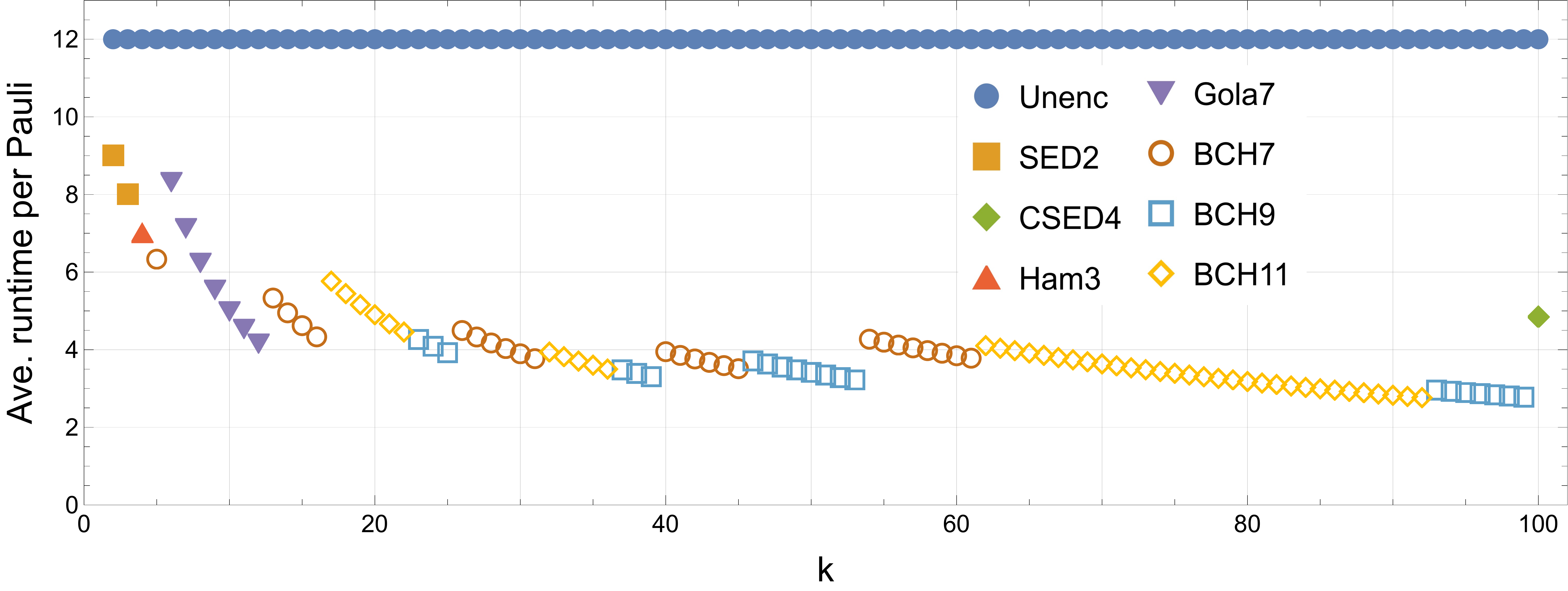}}
    \hspace{0.1cm}
    \subfloat[\label{fig:delt20p4}]{
    \includegraphics[width=.98\textwidth]{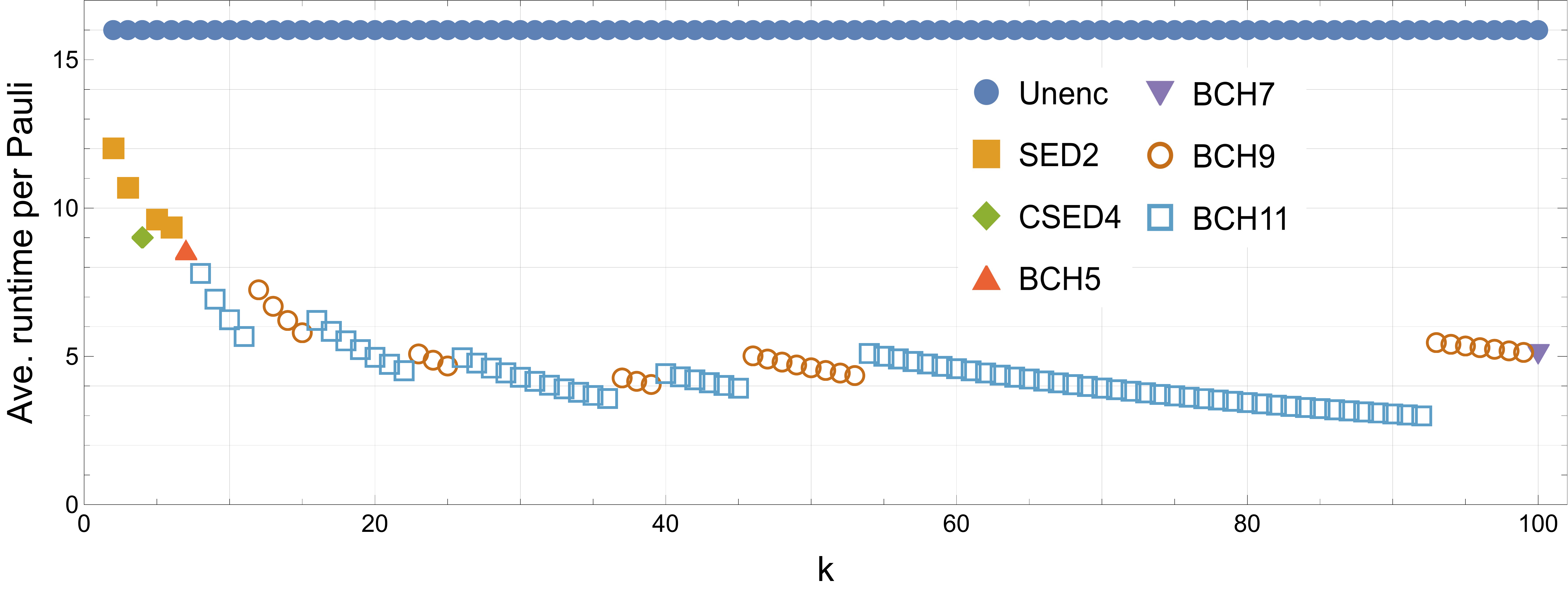}}
    \caption{We show the classical codes achieving the lowest average runtime per Pauli for $k \in \{ 2,3, \mathellipsis, 100\}$ at $p=10^{-4}$ and for (a)~$\delta= 10^{-15}$ and (b)~$\delta= 10^{-20}$. We set the routing space area $A = 100$.}
    \label{fig:p4delt1520}
\end{figure*}

\begin{table*}
    \centering
    \caption{The best lattice surgery speedup for $k \in \{ 1,2, \mathellipsis 100\}$, and associated classical code achieving it, in different noise regimes $p$, and for different target logical error rates $\delta$. Table continues on subsequent pages.}
    \begin{tabular}{c c c c c c c}
        $k$ & $p=10^{-3}$ & $p=10^{-3}$& $p=10^{-3}$ & $p=10^{-3}$ & $p=10^{-4}$ & $p=10^{-4}$ \\
        &  $\delta = 10^{-10}$& $\delta = 10^{-15}$& $\delta = 10^{-20}$& $\delta = 10^{-25}$& $\delta = 10^{-15}$& $\delta = 10^{-20}$ \\
        \hline
        $2$ & SED2, $1.167$ & SED2, $1.333$ & SED2, $1.238$ & SED2, $1.333$ & SED2, $1.333$  & SED2, $1.333$ \\
        $3$ & SED2, $1.312$ & SED2, $1.5  $ & SED2, $1.393$   & SED2, $1.5$ & SED2, $1.5$ & SED2, $1.5$ \\
        $4$ & EHam4, $1.739$ & EHam4, $1.666$ & BCH7, $1.713$& EHam4, $2$ & Ham3, $1.714$ & CSED4, $1.778$ \\
        $5$ & SED2, $1.458$ & SED2, $1.667$ & BCH7, $2.141$   & BCH7, $1.778$ & BCH7, $1.895$ & SED2, $1.667$ \\
        $6$ & SED2, $1.5$ & BCH5, $1.977$ & BCH5, $1.733$     & SED2, $1.714$ & Gola7, $1.441$ & SED2, $1.714$ \\
        $7$ & BCH5, $1.633$ & BCH5, $2.306$ & BCH5, $2.022$ & BCH5, $1.867$ & Gola7, $1.681$ & BCH5, $1.866$ \\
        $8$ & CSED4, $1.728$ & SED2, $1.778$ & Gola7, $2.22$ & BCH11, $2.065$ & Gola7, $1.922$ & BCH11, $2.053$ \\
        $9$ & CSED4, $1.944$ & Gola7, $1.956$ & Gola7, $2.498$ & BCH11, $2.323$ & Gola7, $2.162$ & BCH11, $2.31$ \\
        $10$ & EHam4, $2.16$ & Gola7, $2.174$ & Gola7, $2.775$ & BCH11, $2.581$ & Gola7, $2.402$ &  BCH11, $2.566$  \\
        $11$ & EHam4, $2.376$ & Gola7, $2.391$ & Gola7, $3.053$ & BCH11, $2.839$ & Gola7, $2.642$ & BCH11, $2.823$  \\
        $12$ & Gola7, $1.826$ & Gola7, $2.608$ & Gola7, $3.331$ & Gola7, $2.783$ & Gola7, $2.882$ & BCH9, $2.209$  \\
        $13$ & CSED4, $1.785$ & BCH7, $2.096$ & BCH7, $2.66 $ & BCH9, $2.417$ & BCH7, $2.251$ & BCH9, $2.393$ \\
        $14$ & CSED4, $1.922$ & BCH7, $2.257$ & BCH7, $2.865$ & BCH9, $2.603$ & BCH7, $2.424$ & BCH9, $2.577$ \\
        $15$ & CSED4, $2.059$ & BCH7, $2.419$ & BCH7, $3.069$ & BCH9, $2.789$ & BCH7, $2.597$ & BCH9, $2.761$ \\
        $16$ & CSED4, $2.196$ & BCH7, $2.58$ & BCH7, $3.274$ & BCH7, $2.753$ & BCH7, $2.771$ & BCH11, $2.577$ \\
        $17$ & BCH5, $1.919$ & Zett5, $2.51$ & BCH7, $2.484$ & BCH11, $2.775$ & BCH11, $2.082$ & BCH11, $2.738$ \\
        $18$ & BCH5, $2.032$ & Zett5, $2.657$ & BCH7, $2.63$ & BCH11, $2.939$ & BCH11, $2.204$ & BCH11, $2.899$ \\
        $19$ & BCH5, $2.145$ & Zett5, $2.805$ & BCH7, $2.777$ & BCH11, $3.102$ & BCH11, $2.326$ & BCH11, $3.06$ \\
        $20$ & BCH5, $2.257$ & Zett5, $2.953$ & BCH7, $2.923$ & BCH11, $3.265$ & BCH11, $2.449$ & BCH11, $3.221$ \\
        $21$ & BCH5, $2.37$ & Zett5, $3.1$ & BCH7, $3.069$ & BCH11, $3.429$ & BCH11, $2.571$ & BCH11, $3.382$ \\
        $22$ & EHam4, $2.346$ & Zett5, $3.248$ & BCH7, $3.215$ & BCH11, $3.592$ & BCH11, $2.694$ & BCH11, $3.543$ \\
        $23$ & Pol4, $2.453$ & BCH5, $2.586$ & BCH9, $3.049$ & BCH9, $3.613$ & BCH9, $2.814$ & BCH9, $3.149$ \\
        $24$ & Pol4, $2.56$ & BCH5, $2.698$ & BCH9, $3.181$ & BCH9, $3.77$ & BCH9, $2.937$ & BCH9, $3.286$ \\
        $25$ & Pol4, $2.666$ & BCH5, $2.81$ & BCH9, $3.314$ & BCH9, $3.927$ & BCH9, $3.059$ & BCH9, $3.423$ \\
        $26$ & EHam4, $2.773$ & BCH5, $2.923$ & BCH7, $3.319$ & BCH11, $3.298$ & BCH7, $2.67$ & BCH11, $3.23$ \\
        $27$ & BCH5, $2.196$ & BCH5, $3.035$ & BCH7, $3.446$ & BCH11, $3.425$ & BCH7, $2.773$ & BCH11, $3.354$ \\
        $28$ & BCH5, $2.278$ & BCH5, $3.148$ & BCH7, $3.574$ & BCH11, $3.551$ & BCH7, $2.876$ & BCH11, $3.478$ \\
        $29$ & BCH5, $2.359$ & BCH5, $3.26$ & BCH7, $3.702$ & BCH11, $3.678$ & BCH7, $2.978$ & BCH11, $3.603$ \\
        $30$ & BCH7, $2.143$ & BCH7, $3.059$ & BCH7, $3.829$ & BCH11, $3.805$ & BCH7, $3.081$ & BCH11, $3.727$ \\
        $31$ & BCH7, $2.214$ & BCH7, $3.161$ & BCH7, $3.957$ & BCH11, $3.932$ & BCH7, $3.184$ & BCH11, $3.851$ \\
        $32$ & BCH5, $2.284$ & BCH11, $2.54$ & BCH11, $3.302$ & BCH11, $4.059$ & BCH11, $3.047$ & BCH11, $3.975$ \\
        $33$ & BCH5, $2.355$ & BCH11, $2.619$ & BCH11, $3.405$ & BCH11, $4.186$ & BCH11, $3.143$ & BCH11, $4.1$ \\
        $34$ & BCH5, $2.427$ & BCH11, $2.698$ & BCH11, $3.508$ & BCH11, $4.312$ & BCH11, $3.238$ & BCH11, $4.224$ \\
        $35$ & BCH5, $2.498$ & BCH11, $2.778$ & BCH11, $3.611$ & BCH11, $4.439$ & BCH11, $3.333$ & BCH11, $4.348$ \\
    \end{tabular}
    \label{tab:k2to50}
\end{table*}

\begin{table*}
    \centering
    \begin{tabular}{c c c c c c c}
        $36$ & BCH5, $2.57$ & BCH11, $2.857$ & BCH11, $3.714$ & BCH11, $4.566$ & BCH11, $3.428$ & BCH11, $4.472$ \\
        $37$ & BCH5, $2.641$ & BCH9, $2.937$ & BCH9, $3.813$ & BCH9, $4.471$ & BCH9, $3.447$ & BCH9, $3.747$ \\
        $38$ & BCH9, $2.111$ & BCH9, $3.016$ & BCH9, $3.916$ & BCH9, $4.592$ & BCH9, $3.541$ & BCH9, $3.849$ \\
        $39$ & BCH9, $2.167$ & BCH9, $3.095$ & BCH9, $4.019$ & BCH9, $4.713$ & BCH9, $3.634$ & BCH9, $3.95$ \\
        $40$ & BCH7, $2.222$ & BCH7, $3.171$ & RMul8, $3.863$ & BCH11, $3.765$ & BCH7, $3.038$ &  BCH11, $3.623$\\
        $41$ & BCH7, $2.278$ & BCH7, $3.25$ & RMul8, $3.96$ & BCH11, $3.859$ & BCH7, $3.114$ & BCH11, $3.713$ \\
        $42$ & BCH7, $2.333$ & BCH7, $3.329$ & RMul8, $4.059$ & BCH11, $3.953$ & BCH7, $3.19$ & BCH11, $3.804$ \\
        $43$ & BCH7, $2.389$ & BCH7, $3.409$ & BCH11, $3.288$ & BCH11, $4.047$ & BCH7, $3.266$ & BCH11, $3.895$ \\
        $44$ & BCH7, $2.444$ & BCH7, $3.488$ & BCH11, $3.365$ & BCH11, $4.141$ & BCH7, $3.342$ & BCH11, $3.985$ \\
        $45$ & BCH7, $2.5$ & BCH7, $3.567$ & BCH11, $3.441$ & BCH11, $4.235$ & BCH7, $3.418$ & BCH11, $4.076$ \\
        $46$ & BCH5, $2.553$ & BCH9, $2.706$ & BCH9, $3.517$ & BCH9, $4.05$ & BCH9, $3.235$ & BCH9, $3.191$ \\
        $47$ & BCH5, $2.608$ & BCH9, $2.765$ & BCH9, $3.594$ & BCH9, $4.138$ & BCH9, $3.306$ & BCH9, $3.26$ \\
        $48$ & BCH5, $2.664$ & BCH9, $2.824$ & BCH9, $3.67$ & BCH9, $4.226$ & BCH9, $3.376$ & BCH9, $3.33$ \\
        $49$ & BCH5, $2.719$ & BCH9, $2.882$ & BCH9, $3.747$ & BCH9, $4.314$ & BCH9, $3.446$ & BCH9, $3.399$ \\
        $50$ & BCH5, $2.775$ & BCH9, $2.941$ & BCH9, $3.823$ & BCH9, $4.402$ & BCH9, $3.517$ & BCH9, $3.468$ \\
        $51$ & BCH5, $2.83$ & BCH9, $3$ & BCH9, $3.9$ & BCH9, $4.49$ & BCH9, $3.587$  & BCH9, $3.538$ \\
        $52$ & Zett5, $2.797$ & BCH9, $3.059$ & BCH9, $3.976$ & BCH9, $4.578$ & BCH9, $3.657$ &BCH9, $3.607$ \\
        $53$ & Zett5, $2.85$ & BCH9, $3.118$ & BCH9, $4.053$ & BCH9, $4.666$ & BCH9, $3.728$ & BCH9, $3.676$ \\
        $54$ & EHam4, $2.808$ & BCH7, $3.17$ & BCH7, $3.862$ & BCH7, $3.388$ & BCH7, $2.809$ & BCH11, $3.141$ \\
        $55$ & EHam4, $2.86$ & BCH7, $3.228$ & BCH7, $3.934$ & BCH7, $3.451$ & BCH7, $2.861$ & BCH11, $3.199$ \\
        $56$ & EHam4, $2.912$ & BCH7, $3.287$ & BCH7, $4.006$ & BCH7, $3.514$ & BCH7, $2.913$ & BCH11, $3.257$ \\
        $57$ & EHam4, $2.964$ & BCH7, $3.346$ & BCH7, $4.077$ & BCH7, $3.576$ & BCH7, $2.965$ & BCH11, $3.315$ \\
        $58$ & BCH7, $2.388$ & BCH7, $3.405$  & BCH7, $4.149$ & BCH7, $3.639$ & BCH7, $3.017$ & BCH11, $3.374$ \\
        $59$ & BCH7, $2.429$ & BCH7, $3.463$  & BCH7, $4.22 $ & BCH7, $3.702$ & BCH7, $3.069$ & BCH11, $3.432$ \\
        $60$ & BCH7, $2.47$ & BCH7, $3.522$ &   BCH7, $4.292$ & BCH7, $3.765$ & BCH7, $3.122$ & BCH11, $3.49$ \\
        $61$ & BCH7, $2.512$ & BCH7, $3.581$ &  BCH7, $4.363$ & BCH7, $3.827$ & BCH7, $3.174$ & BCH11, $3.548$ \\
        $62$ & BCH5, $2.548$ & CSED4, $2.548$ & BCH11, $3.173$ & BCH11, $3.887$ & BCH11, $2.926$ & BCH11, $3.606$ \\
        $63$ & BCH5, $2.589$ & CSED4, $2.589$ & BCH11, $3.224$ & BCH11, $3.95 $ & BCH11, $2.973$ & BCH11, $3.664$ \\
        $64$ & BCH5, $2.63$ & CSED4, $2.63 $ &  BCH11, $3.276$ & BCH11, $4.013$ & BCH11, $3.02 $ & BCH11, $3.723$ \\
        $65$ & BCH5, $2.671$ & BCH9, $2.559$  & BCH11, $3.327$ & BCH11, $4.076$ & BCH11, $3.067$ & BCH11, $3.781$ \\
        $66$ & BCH5, $2.712$ & BCH9, $2.598$ & BCH11, $3.378$ & BCH11, $4.138$ & BCH11, $3.114$ & BCH11, $3.839$ \\
        $67$ & BCH5, $2.753$ & BCH9, $2.638$ & BCH11, $3.429$ & BCH11, $4.201$ & BCH11, $3.162$ & BCH11, $3.897$ \\
        $68$ & BCH5, $2.794$ & BCH9, $2.677$ & BCH11, $3.48 $ & BCH11, $4.264$ & BCH11, $3.209$ & BCH11, $3.955$ \\
        $69$ & BCH5, $2.835$ & BCH9, $2.717$ & BCH11, $3.531$ & BCH11, $4.326$ & BCH11, $3.256$ & BCH11, $4.013$ \\
        $70$ & CSED4, $2.265$ & BCH9, $2.756$ & BCH11, $3.583$ & BCH11, $4.389$ & BCH11, $3.303$ & BCH11, $4.072$ \\
    \end{tabular}
    \label{tab:k51to100}
\end{table*}

\begin{table*}
    \centering
    \begin{tabular}{c c c c c c c}
        $71$ & CSED4, $2.297$ & BCH9, $2.795$ & BCH11, $3.634$ & BCH11, $4.452$ & BCH11, $3.35$ &  BCH11, $4.13$ \\
        $72$ & CSED4, $2.329$ & BCH9, $2.835$ & BCH11, $3.685$ & BCH11, $4.514$ & BCH11, $3.397$ & BCH11, $4.188$ \\
        $73$ & CSED4, $2.362$ & BCH9, $2.874$ & BCH11, $3.736$ & BCH11, $4.577$ & BCH11, $3.445$ & BCH11, $4.246$ \\
        $74$ & CSED4, $2.394$ & BCH9, $2.913$ & BCH11, $3.787$ & BCH11, $4.64 $ & BCH11, $3.492$ & BCH11, $4.304$ \\
        $75$ & CSED4, $2.427$ & BCH9, $2.953$ & BCH11, $3.839$ & BCH11, $4.703$ & BCH11, $3.539$ & BCH11, $4.362$ \\
        $76$ & CSED4, $2.459$ & BCH9, $2.992$ & BCH11, $3.89 $ & BCH11, $4.765$ & BCH11, $3.586$ & BCH11, $4.421$ \\
        $77$ & CSED4, $2.491$ & BCH9, $3.031$ & BCH11, $3.941$ & BCH11, $4.828$ & BCH11, $3.633$ & BCH11, $4.479$ \\
        $78$ & CSED4, $2.524$ & BCH9, $3.071$ & BCH11, $3.992$ & BCH11, $4.891$ & BCH11, $3.681$ & BCH11, $4.537$ \\
        $79$ & CSED4, $2.556$ & BCH9, $3.11 $ & BCH11, $4.043$ & BCH11, $4.953$ & BCH11, $3.728$ & BCH11, $4.595$ \\
        $80$ & CSED4, $2.588$ & BCH9, $3.15 $ & BCH11, $4.094$ & BCH11, $5.016$ & BCH11, $3.775$ & BCH11, $4.653$ \\
        $81$ & CSED4, $2.621$ & BCH9, $3.189$ & BCH11, $4.146$ & BCH11, $5.079$ & BCH11, $3.822$ & BCH11, $4.711$ \\
        $82$ & BCH11, $2.26 $ & BCH9, $3.228$ & BCH11, $4.197$ & BCH11, $5.141$ & BCH11, $3.869$ & BCH11, $4.77$ \\
        $83$ & BCH11, $2.287$ & BCH9, $3.268$ & BCH11, $4.248$ & BCH11, $5.204$ & BCH11, $3.917$ & BCH11, $4.828$ \\
        $84$ & BCH11, $2.315$ & BCH9, $3.307$ & BCH11, $4.299$ & BCH11, $5.267$ & BCH11, $3.964$ & BCH11, $4.886$ \\
        $85$ & BCH11, $2.343$ & BCH9, $3.346$ & BCH11, $4.35 $ & BCH11, $5.33 $ & BCH11, $4.011$ & BCH11, $4.944$ \\
        $86$ & BCH11, $2.37 $ & BCH9, $3.386$ & BCH11, $4.402$ & BCH11, $5.392$ & BCH11, $4.058$ & BCH11, $5.002$ \\
        $87$ & BCH11, $2.398$ & BCH9, $3.425$ & BCH11, $4.453$ & BCH11, $5.455$ & BCH11, $4.105$ & BCH11, $5.06$  \\
        $88$ & BCH11, $2.425$ & BCH9, $3.465$ & BCH11, $4.504$ & BCH11, $5.518$ & BCH11, $4.152$ & BCH11, $5.119$  \\
        $89$ & BCH11, $2.453$ & BCH9, $3.504$ & BCH11, $4.555$ & BCH11, $5.58 $ & BCH11, $4.2. $ & BCH11, $5.177$  \\
        $90$ & BCH11, $2.48 $ & BCH9, $3.543$ & BCH11, $4.606$ & BCH11, $5.643$ & BCH11, $4.247$ & BCH11, $5.235$  \\
        $91$ & BCH11, $2.508$ & BCH9, $3.583$ & BCH11, $4.657$ & BCH11, $5.706$ & BCH11, $4.294$ & BCH11, $5.293$  \\
        $92$ & BCH11, $2.535$ & BCH9, $3.622$ & BCH11, $4.709$ & BCH11, $5.768$ & BCH11, $4.341$ & BCH11, $5.351$  \\
        $93$ & BCH9, $2.563$ & BCH9, $3.661$ & BCH9, $4.738$ & BCH9, $5.302$ & BCH9, $4.057$ & BCH9, $2.929$  \\
        $94$ & BCH9, $2.591$ & BCH9, $3.701 $ & BCH9, $4.789$ & BCH9, $5.359$ & BCH9, $4.101$ & BCH9, $2.961$  \\
        $95$ & BCH9, $2.618$ & BCH9, $3.74 $ & BCH9, $4.84$ & BCH9, $5.416$ & BCH9, $4.144$ & BCH9, $2.992$  \\
        $96$ & BCH9, $2.646$ & BCH9, $3.78$ & BCH9, $4.891$ & BCH9, $5.473$ & BCH9, $4.188$ & BCH9, $3.024$  \\
        $97$ & BCH9, $2.673$ & BCH9, $3.819$ & BCH9, $4.942$ & BCH9, $5.53$ & BCH9, $4.232$ & BCH9, $3.055$ \\
        $98$ & BCH9, $2.701$ & BCH9, $3.858$ & BCH9, $4.993$ & BCH9, $5.587$ & BCH9, $4.275$ & BCH9, $3.087$ \\
        $99$ & BCH9, $2.728$ & BCH9, $3.898$ & BCH9, $5.043$ & BCH9, $5.644$ & BCH9, $4.319$ & BCH9, $3.118$ \\
        $100$ & BCH7, $2.755$ & BCH7, $3.919$ & BCH7, $3.412$ & BCH7, $4.199$ & CSED4, $2.477$ & BCH7, $3.15$ \\
    \end{tabular}
    \label{tab:k76to100}
\end{table*}

\section{Clifford frames of distilled magic states}
\label{app:CliffpartTraceproof}

In circuits like \cref{fig:Haddist}, where non-Clifford gates are implemented using Pauli measurements, we prove that the Clifford frame of the distilled magic states are powers of $X_{\pi/4}$.
After the temporally encoded non-Clifford gates, the Clifford frame consists of a sequence of Clifford rotations which are tensor products of $X$ and $\Id$. We first observe the effect of an $(X \otimes X)_{\pi/4}$ gate on an input $\ket{T_X} \otimes \ket{\psi}$ state, where $\ket{\psi} = \alpha \ket 0 + \beta \ket 1$ is some arbitrary state. If we can determine the effect of the $(X \otimes X)_{\pi/4}$ rotation on the subsystem of the distilled magic state, we can determine the final Clifford frame of the distilled magic state.
\begin{align}
    (X \otimes &X)_{\pi/4} \ket{T_X} \otimes \ket{\psi} =   \nonumber \\
    & \begin{bmatrix} 
 \frac{1}{\sqrt{2}} &0 &0 &\frac{-i}{\sqrt{2}}  \\
 0 &\frac{1}{\sqrt{2}} &\frac{-i}{\sqrt{2}} &0  \\
 0 &\frac{-i}{\sqrt{2}} &\frac{1}{\sqrt{2}} &0  \\
 \frac{-i}{\sqrt{2}} &0 &0 &\frac{1}{\sqrt{2}}
\end{bmatrix} \cdot \frac{1}{\sqrt{2}}
\begin{bmatrix} 
 (1+ e^{\frac{i \pi}{4}}) \alpha \\
 (1+ e^{\frac{i \pi}{4}}) \beta  \\
 (1- e^{\frac{i \pi}{4}}) \alpha \\
 (1- e^{\frac{i \pi}{4}}) \beta
\end{bmatrix}.
\end{align}
When we trace out the subsystem that started as $\ket {\psi}$, we obtain the following state on the subsystem of the magic state (up to a global phase)
\begin{align}
    tr_{\ket {\psi}}((X \otimes &X)_{\pi/4} \ket{T_X} \otimes \ket{\psi}) =   \nonumber \\
& \begin{bmatrix} 
 1+ e^{\frac{i \pi}{4}} - i (1- e^{\frac{i \pi}{4}}) \\
 1- e^{\frac{i \pi}{4}} - i (1+ e^{\frac{i \pi}{4}})
\end{bmatrix}.
\end{align}
This is equivalent to $X_{\pi/4} \ket{T_X}$, which is 
\begin{equation}
    X_{\pi/4} \ket{T_X} = 
    \begin{bmatrix} 
 1+ e^{\frac{i \pi}{4}} - i (1- e^{\frac{i \pi}{4}}) \\
 1- e^{\frac{i \pi}{4}} - i (1+ e^{\frac{i \pi}{4}})
\end{bmatrix}. 
\end{equation}

This shows that if there is one Clifford operator in the Clifford frame with $X$ support on the magic state qubit, it essentially results in an $X_{\pi/4}$ Clifford frame update on the magic state. However with more than one Clifford correction in the Clifford frame, the total rotation accumulated on the magic state qubit is the product of $X_{\pi/4}$ rotations for all the Cliffords in the Clifford frame that contain an $X$ operator on the support of the magic state qubit.

\section{Choice of codewords for the Golay code for TELS of a 15-to-1 distillation protocol}
\label{app:golaycodechoice}

In the $15$-to-$1$ distillation protocol of \cref{sec:15to1}, we implemented a TELS protocol for the non-Clifford measurements using the $[23,12,7]$ Golay code. We remove one codeword, since we only need to perform $11$ Pauli measurements in the PP set, and permute some columns (reordering the resulting Pauli measurements) to get the following codeword generator matrix
\begin{equation}
    G = \begin{bmatrix} 

 1 1 1 1 1 1 1 0 0 0 0 0 0 0 0 0 0 0 0 0 0 0 0 \\
 0 0 0 1 1 0 1 1 1 1 1 0 0 0 0 0 0 0 0 0 0 0 0 \\
 0 1 1 0 1 0 0 0 0 1 1 1 1 0 0 0 0 0 0 0 0 0 0 \\
 0 0 0 1 0 0 0 0 1 1 0 1 1 1 1 0 0 0 0 0 0 0 0 \\
 0 0 1 0 1 1 0 0 0 0 0 1 0 1 1 1 0 0 0 0 0 0 0 \\
 0 0 0 1 0 1 1 0 0 1 0 0 0 1 0 1 1 0 0 0 0 0 0 \\
 0 0 0 0 1 1 1 0 0 0 1 1 0 0 0 0 1 1 0 0 0 0 0 \\
 0 0 0 0 0 0 1 0 0 1 1 0 1 1 0 0 0 1 1 0 0 0 0 \\
 0 0 0 0 0 1 0 0 0 0 1 1 1 0 1 0 0 0 1 1 0 0 0 \\
 0 0 0 0 0 0 1 0 0 0 0 0 1 1 1 1 0 0 0 1 1 0 0 \\
 0 0 0 0 0 0 0 0 0 0 0 0 0 1 1 0 1 1 0 1 1 1 1
\end{bmatrix}.
\end{equation}
Now, after the seventh measurement, the cell holding the magic state associated with the first row of $G$ can be reset and used to inject a new magic state that will only be required for the $14$th Pauli measurement (first column (left-to-right) with a $1$ in the last row).

Note that both the columns and rows of G may be permuted; permuting columns reorders the new sequence of Pauli measurements; permuting rows merely swaps codeword generators. It may be possible to find an algorithm that iteratively applies a permutation rule convention to find a codeword matrix that allows magic states to occupy the fewest number of cells in a distillation tile.

% \vspace{-.2cm} %DEBUG 
\section{Procedure for determining code distances of distillation tiles}
\label{app:algoSTcosts}

Here we describe the procedure used to determine the spacelike distances $d_x$ and $d_z$ and timelike distances $d_m$ for lattice surgery in the distillation tiles of \cref{sec:MSDMSDdesign}. First, note that when performing distillation in the Clifford frame using TELS-assisted lattice surgery, we will need to execute two PP sets. The first PP set performs the non-Clifford gates using $\ket{T_X}$ resource states (using a classical $[n_1,k_1,d_1]$ code), and the second PP set performs a set of Pauli measurements associated with conjugating Clifford corrections through single-qubit logical measurements (using a classical $[n_2,k_2,d_2]$ code).

We first set $\delta^{(M)}$ as the error budget per magic state that is output from a distillation protocol. Logical errors may accumulate on the distilled magic states by different mechanisms. Even with noiseless lattice surgery, errors on the input magic states may may cause the final magic state to be logically wrong. This depends entirely on the choice of quantum code, here $\llbracket n_{\text{dist}}, k_{\text{dist}}, d_{\text{dist}} \rrbracket$. The logical error rate (per output magic state) of a distillation protocol with noiseless gates is
\begin{equation}
    p_L^{(M)} = \frac{l_{\text{dist}}}{k_{\text{dist}}} \Big(\frac{p(1+\eta)}{3 \eta} \Big)^{d_{\text{dist}}},
\end{equation}
according to the analysis in Sec. $1$ of Ref.\cite{Litinski19magic}, using the biased circuit-level noise model of \cref{subsec:PauliBasedReview}. Here, $l_{\text{dist}}$ is the number of weight-$d_{\text{dist}}$ fault sets that can cause a logical error.

Given the error budget $\delta^{(M)}$ and the logical error rate with noiseless gates $p_L^{(M)}$, we may now upper bound the logical error rate due to noisy lattice surgery measurements,
\begin{equation}
\delta = (\delta^{(M)} - p_L^{(M)}) k_{\text{dist}}.
\end{equation}
Note that we multiply by $k_{\text{dist}}$ since $\delta^{(M)}$ is the error budget per output magic state, but $\delta$ is the error budget of lattice surgery for the entire distillation protocol. Logical errors due to lattice surgery may occur due to spacelike or timelike errors. 

For each PP set, the logical error rate due to lattice surgery is
\begin{equation}
    p_{\text{PP1}}(d_m, n) = p_{L,X}(d_m , n) + p_{L,Z}(d_m , n) + p_{L}(p_m(d_m , n)), 
\end{equation}
where  $p_{L,X}$ and $p_{L,Z}$ are the spacelike contributions and $p_{L}$ is the timelike failure rate of TELS described in \cref{sec:NewTELS}. Here $d_m$ and $n$ refer to the lattice surgery measurement distance and the number of Pauli measurements respectively. Although we express $p_{\text{PP1}}$ as a function of two variables $d_m$ and $n$, the spacelike distance $d_x$ and $d_z$ are also input variables. The important point is that $d_m$ and $n$ are the only two variables that are different for each PP set. 

In Ref.~\cite{Chamberland22}, Eqs. $3-6$ denote the logical error rates of an $X \otimes X$ lattice surgery measurement. We obtain equations for $p_{L,X}$, $p_{L,Z}$, and $p_m$ by modifying the above equations as shown below
\begin{align}
    p_m(d_m , n) =& 0.01634 n A (21.93p)^{\tfrac{d_m+1}{2}}, \\ 
    p_{L,Z}(d_m , n) =& 0.03148 T N d_x (28.91p)^{\tfrac{d_z+1}{2}}, \\ 
    p_{L,X}(d_m , n) =& 0.0148 T \frac{F}{d_x} (0.762p)^{\tfrac{d_x+1}{2}}, 
\end{align}
where $A$ is the area of the routing space (in units of $d_x$ and $d_z$), $T$ is the average time taken to execute the parallelizable Pauli set (from \cref{sec:NewTELS}), $N$ is the maximum number of logical qubits that are concurrently used during any lattice surgery measurement in a TELS protocol, and $F$ is a pessimistic estimate of the maximum area used during a lattice surgery measurement (routing space $+$ logical qubits associated in the measurement). In \cref{app:constants}, we show equations for $F,A,N$ and $\text{Space}$ in terms of $d_x$ and $d_z$ for each of the different distillation layouts suggested in \cref{sec:MSDMSDdesign}. The time to complete the entire distillation protocol is given in \cref{subsubsec:MSDMSDClifftime}.

If we use measurement distance $d_m'$ for the first PP set and measurement distance $d_m''$ for the second PP set, then the objective is to find a set of parameters $\{d_x, d_z, d_m', d_m''\}$ that minimizes the space-time cost of a distillation factory, while ensuring the following equation is satisfied,
\begin{equation}
\label{eq:conditiondelta}
p_{\text{PP1}}(d_m', n_1) + p_{\text{PP2}}(d_m'', n_2) < \delta.
\end{equation}

Note that \cref{eq:conditiondelta}  ensures that the probability of a single logical failure event is less than $\delta$. Two independent logical failure events occur with probability $\sim \delta^2$, so we omit these higher order events from \cref{eq:conditiondelta}.

\section{Constants used to determine spacetime costs of distillation tiles}
\label{app:constants}

In this section and \tabref{tab:constants} we list the constants that are used to determine the minimum space-like and time-like distances for the distillation layouts described in \cref{app:algoSTcosts}. Distillation protocols that do not use TELS perform auto-corrected non-Clifford gate gadgets for the entire protocol, and hence contain one value each for the number of logical qubits, $N$, routing space area, $A$, and full area of lattice surgery, $F$. For protocols that perform TELS using the Clifford frame distillation circuit of \cref{subsec:MSDMSDCliff}, there are two PP sets. We display two sets of values for $N,A$, and $F$ to account for the changes between the execution of the first and second PP sets.

% \vspace{13cm} %DEBUG
% \hbox{} %DEBUG

\begin{table}
    \caption{Constants associated with layouts of distillation tiles described in \cref{sec:MSDMSDdesign}. These constants are used to determine minimum spacelike and timelike distances using the procedure in \cref{app:algoSTcosts}.}
    \centering
    \begin{tabular}{l}
    $15$-to-$1$ distillation - No temporal encoding \\
    \hline
    Space $= (2 d_z + 2 d_x + 2) (5(d_x+1))$\\
    $N = 7$ \\
    $A =(d_x+1)(d_z + 4(d_x+2))$\\
    $F = \text{Space} - 2 d_z (d_x + 1)$\\
    \hline
    $15$-to-$1$ distillation - No temp. encoding, parallelized \\
    \hline
    Space $= (2 d_z + 4 d_x + 4) (5(d_x+1))$\\
    $N = 9$ \\
    $A =(d_x+1)(d_z + 4(d_x+2))$\\
    $F = \text{Space} - 3 d_z (d_x + 1)$\\
    \hline
    $15$-to-$1$ distillation - SED2 \\
    \hline
    Space $= (2 d_z + d_x + 3) (5(d_x+1))$\\
    $N = 7$ \\
    $A = 5 (d_x+1)(d_x+3)$\\
    $F = \text{Space} - 3 d_z (d_x + 1)$\\
    $N2 = 5$ \\
    $A2 = A + 4 d_z (d_x+1)$\\
    $F2 = \text{Space} - d_z (d_x + 1)$\\
    \hline
    $15$-to-$1$ distillation - SED2, parallellized \\
    \hline
    Space $= (2 d_z + 3 d_x + 3) (7(d_x+1))$\\
    $N = 8$ \\
    $A =  (d_x+1)(2 d_z + 14 d_x+3)$\\
    $F = \text{Space} - 4 d_z (d_x + 1)$\\
    $N2 = 5$ \\
    $A2 = A + 4 d_z (d_x+1)$\\
    $F2 = \text{Space} - d_z (d_x + 1)$\\
    \hline
    $15$-to-$1$ distillation - BCH3 \\
    \hline
    Space $= (2 d_z + d_x + 3) (6(d_x+1))$\\
    $N = 10$ \\
    $A = 6 (d_x+1)(d_x+3)$\\
    $F = \text{Space} - 2 d_z (d_x + 1)$\\
    $N2 = 5$ \\
    $A2 = A + 4 d_z (d_x+1)$\\
    $F2 = \text{Space} - d_z (d_x + 1)$\\
    \hline
    \end{tabular}
    \label{tab:constants}
\end{table}

\begin{table}
    \centering
    \begin{tabular}{l}
    $15$-to-$1$ distillation - BCH3, parallellized \\
    \hline
    Space $= (2 d_z + 3 d_x + 3) (8(d_x+1))$\\
    $N = 11$ \\
    $A =  (d_x+1)(2 d_z + 16 d_x+3)$\\
    $F = \text{Space} - 3 d_z (d_x + 1)$\\
    $N2 = 5$ \\
    $A2 = A + 4 d_z (d_x+1)$\\
    $F2 = \text{Space} - d_z (d_x + 1)$\\
    \hline
    $15$-to-$1$ distillation - Golay \\
    \hline
    Space $= (2 d_z + d_x + 3) (8(d_x+1))$\\
    $N = 15$ \\
    $A = 8 (d_x+1)(d_x+3)$\\
    $F = \text{Space} -  d_z (d_x + 1)$\\
    $N2 = 5$ \\
    $A2 = A + 4 d_z (d_x+1)$\\
    $F2 = \text{Space} - d_z (d_x + 1)$\\
    \hline
$15$-to-$1$ distillation - Golay, parallellized \\
    \hline
    Space $= (2 d_z + 3 d_x + 3) (9(d_x+1))$\\
    $N = 15$ \\
    $A =  (d_x+1)(2 d_z + 18 d_x+3)$\\
    $F = \text{Space} - d_z (d_x + 1)$\\
    $N2 = 5$ \\
    $A2 = A + 4 d_z (d_x+1)$\\
    $F2 = \text{Space} - d_z (d_x + 1)$\\
    \hline
$116$-to-$12$ distillation - No temporal encoding \\
    \hline
    Space $= \max(2(d_z + d_x), 4(d_x+1)) (22d_x + d_z +23)$\\
    $N = 31$ \\
    $A =(d_x+1)(25d_x+4)$\\
    $F = \text{Space} - 13 d_z (d_x + 1)$\\
    \hline
    $116$-to-$12$ distillation - No temp. encoding, parallelized \\
    \hline
    Space $= \max(2(d_z + 3d_x), 4(d_x+1)) (23d_x + 2d_z +25)$\\
    $N = 33$ \\
    $A =(d_x+1) (30 d_x + \max(2(d_z + 3d_x), 4(d_x+1)))$\\
    $F = \text{Space} - 14 d_z (d_x + 1)$\\
    \hline
    \end{tabular}
\end{table}

\begin{table}
    \centering
    \begin{tabular}{l}
    $116$-to-$12$ distillation - Zett5, parallelized \\
    \hline
    Space $= (2 d_z + 3d_x + 3) (31(d_x+1))$\\
    $N = 46$ \\
    $A = (d_x+1)(2 d_z + 51 d_x+3)$\\
    $F = \text{Space} - 14 d_z (d_x + 1)$\\
    $N2 = 29$ \\
    $A2 = A + 19 d_z (d_x+1)$\\
    $F2 = \text{Space} - 12 d_z (d_x + 1)$\\
    \hline
    $116$-to-$12$ distillation - BCH9, parallellized \\
    \hline
    Space $= (2 d_z + 3 d_x + 3) (38(d_x+1))$\\
    $N = 58$ \\
    $A =  (d_x+1)(2 d_z + 65 d_x+3)$\\
    $F = \text{Space} - 15 d_z (d_x + 1)$\\
    $N2 = 29$ \\
    $A2 = A + 32 d_z (d_x+1)$\\
    $F2 = \text{Space} - 12 d_z (d_x + 1)$\\
    \hline
     $114$-to-$14$ distillation - No temporal encoding \\
    \hline
    Space $= \max(2(d_z + d_x), 4(d_x+1)) (23d_x + d_z + 24)$\\
    $N = 31$ \\
    $A =(d_x+1)(26d_x+4)$\\
    $F = \text{Space} - 15 d_z (d_x + 1)$\\
    \hline
    $114$-to-$14$ distillation - No temp. encoding, parallelized \\
    \hline
    Space $= \max(2(d_z + 3d_x), 4(d_x+1)) (24d_x + 2d_z + 26)$\\
    $N = 33$ \\
    $A =(d_x+1) (30 d_x + \max(2(d_z + 3d_x), 4(d_x+1)))$\\
    $F = \text{Space} - 16 d_z (d_x + 1)$\\
    \hline
    $114$-to-$14$ distillation - Zett5, parallelized \\
    \hline
    Space $= (2 d_z + 3d_x + 3) (32(d_x+1))$\\
    $N = 46$ \\
    $A = (d_x+1)(2 d_z + 51 d_x+3)$\\
    $F = \text{Space} - 16 d_z (d_x + 1)$\\
    $N2 = 29$ \\
    $A2 = A + 19 d_z (d_x+1)$\\
    $F2 = \text{Space} - 14 d_z (d_x + 1)$\\
    \hline
    \end{tabular}
\end{table}

\begin{table}
    \centering
    \begin{tabular}{l}
    $114$-to-$14$ distillation - BCH7, parallellized \\
    \hline
    Space $= (2 d_z + 3 d_x + 3) (35(d_x+1))$\\
    $N = 52$ \\
    $A =  (d_x+1)(2 d_z + 57 d_x+3)$\\
    $F = \text{Space} - 16 d_z (d_x + 1)$\\
    $N2 = 29$ \\
    $A2 = A + 25 d_z (d_x+1)$\\
    $F2 = \text{Space} - 14 d_z (d_x + 1)$\\
    \hline
    $125$-to-$3$ distillation - No temporal encoding \\
    \hline
    Space $= \max(2(d_z + d_x), 4(d_x+1)) (18d_x + d_z +19)$\\
    $N = 31$ \\
    $A =(d_x+1)(20d_x+4)$\\
    $F = \text{Space} - 4 d_z (d_x + 1)$\\
    \hline
    $125$-to-$3$ distillation - No temp. encoding, parallelized \\
    \hline
    Space $= \max(2(d_z + 3d_x), 4(d_x+1)) (20d_x + 2d_z  +22)$\\
    $N = 33$ \\
    $A =(d_x+1) (29 d_x + \max(2(d_z + 3d_x), 4(d_x+1)))$\\
    $F = \text{Space} - 5 d_z (d_x + 1)$\\
    \hline
    $125$-to-$3$ distillation - BCH7, parallelized \\
    \hline
    Space $= (2 d_z + 3d_x + 3) (30(d_x+1))$\\
    $N = 52$ \\
    $A = (d_x+1)(2 d_z + 58 d_x+3)$\\
    $F = \text{Space} - 6 d_z (d_x + 1)$\\
    $N2 = 29$ \\
    $A2 = A + 26 d_z (d_x+1)$\\
    $F2 = \text{Space} - 3 d_z (d_x + 1)$\\
    \hline
    $125$-to-$3$ distillation - BCH9, parallellized \\
    \hline
    Space $= (2 d_z + 3 d_x + 3) (33(d_x+1))$\\
    $N = 56$ \\
    $A =  (d_x+1)(2 d_z + 64 d_x+3)$\\
    $F = \text{Space} - 5 d_z (d_x + 1)$\\
    $N2 = 29$ \\
    $A2 = A + 32 d_z (d_x+1)$\\
    $F2 = \text{Space} - 3 d_z (d_x + 1)$\\
    \hline
    \end{tabular}
\end{table}

\section{Additional distillation layouts}
\label{sec:additionaldistillation}

\subsection{125-to-3 distillation}
\label{sec:125to3}

The $125$-to-$3$ magic state distillation protocol is obtained from a triorthogonal CSS quantum $\llbracket 125,3,5\rrbracket$ code. This code is constructed by puncturing the $[128,29,32]$ Reed-Muller code at any three locations~\cite{Haah18}. As a result, the quantum code will contain $3$ logical qubits, $96$ $X$-type stabilizers and $26$ $Z$-type stabilizers. After applying the circuit transformation from a gate-based model to the PBC model~\cite{Litinski19}, we are left with a sequence of $99$ commuting Pauli measurements on $29$ logical qubits, three of which will finally become the distilled magic states. These $99$ measurements form a size-$99$ PP set.

In this paper, we will consider using this protocol in a regime where the physical error rate is $p=10^{-3}$ and our target is to distill magic states with logical error probability at most $\delta^{(M)}=10^{-15}$. Using the noise model described in \cref{subsec:MSDMSDinjection} for the injection of magic states, we apply the analysis in Ref.~\cite{Litinski19magic} to determine the logical failure probability per output magic state for one round of a $125$-to-$3$ distillation scheme,
\begin{align}
    p_L^{(M)}  = & \frac{1}{3} 31 \Big ( (\epsilon_{\text{L,Z}})^5 + \frac{1}{2} 10(\epsilon_{\text{L,Z}})^4 \epsilon_{\text{L,X}} \nonumber \\
    & \quad + \frac{1}{4} 40 (\epsilon_{\text{L,Z}})^3 (\epsilon_{\text{L,X}})^2
    % \nonumber \\ & 
    + \frac{1}{8} 80(\epsilon_{\text{L,Z}})^2 (\epsilon_{\text{L,X}})^3 \nonumber \\
    &\quad + \frac{1}{16} 80 \epsilon_{\text{L,Z}} (\epsilon_{\text{L,X}})^4
    % \nonumber \\ &
    + \frac{1}{32} 32(\epsilon_{\text{L,X}})^5 \Big ) \nonumber \\
    = &\frac{31(1+\eta)^5}{729 \eta^5} p^5 \: .
\end{align}
For $p=10^{-3}$ and $\eta=100$, the probability that the distillation succeeds is $1 - p_D^{(M)} = (1-\epsilon_L)^{125} = 0.9584$ and $p_L^{(M)} = 4.47 \times 10^{-17}$. As this is sufficiently below $\delta^{(M)}$, the lattice surgery measurements used to execute the distillation protocol must be modeled with measurement distance large enough to allow for distilled magic states of logical error rate at most $\delta{(M)}$. Using the procedure in \cref{app:algoSTcosts}, we determined that the minimum spacelike distances are $d_x=13$ and $d_z=25$. 

We show two layouts for distillation tiles that do not use TELS in \cref{fig:125ue} and \cref{fig:125uepar}. These layouts perform distillation in the Pauli frame as described in Ref.~\cite{Litinski19} using auto-corrected non-Clifford gadgets (\cref{fig:PauliAutocorr}). On the layout of \cref{fig:125ue}, $99$ Pauli measurements are performed, each with measurement distance $d_m = 23$ (also derived using \cref{app:algoSTcosts}). On the layout of \cref{fig:125uepar}, Pauli measurements can be performed two at a time. Hence the time required is only the time for $50$ sequential lattice surgery measurements with measurement distance $d_m = 23$. 

In \cref{fig:125bch7par}, we show a layout for a distillation tile that performs TELS with a classical $[127,106,7]$ BCH code. Note that since there are two disjoint routing spaces (in blue and grey), the set of $n=127$ Pauli measurements can be performed in the time required for $64$ sequential measurements. This way, the time cost is nearly halved, with only a minor increase to the height of the distillation tile (two rows of routing space of height $d_x$). To obtain the time cost shown in \cref{tab:STcosts}, we used $d_m = 7$ and $c=2$, where $c$ is the maximum weight of classical errors that are corrected in a TELS protocol. Only classical errors of weight greater than or equal to three triggered detection events. Similarly, in \cref{fig:125bch9par}, we show a layout for a $125$-to-$3$ distillation tile that performs TELS with a classical $[127,99,9]$ BCH code. Here, we used the parameters $d_m = 5$ and $c=1$ to obtain the time cost shown in \cref{tab:STcosts}. For both of the layouts that use TELS, we have only discussed the TELS code used to execute the non-Clifford gates of \cref{subsec:MSDMSDCliff}. For the $125$-to-$3$ distillation protocol, there may be at most $26$ additional Pauli measurements to perform due to the conditional Clifford corrections. The results of these measurements are used to detect if there are errors in the final distilled magic states. These Pauli measurements form a PP set of size at most $26$. For the worst case where the PP set is of size 26, the TELS protocol  uses the $[31,26,3]$ BCH code with lattice surgery measurement distance $d_m = 9$.

The space requirements of all the above layouts are described as functions of $d_x$ and $d_z$ in \cref{app:constants}. Additional constants in \cref{app:constants} can be used with the procedure of \cref{app:algoSTcosts} to determine all the minimum distances (spacelike and timelike) for the distillation protocols.

\begin{figure}
    \centering
    \subfloat[\label{fig:125ue} 
    ]{\includegraphics[height=.09\textwidth]{imagesChap7/125to3ue.pdf}}
    \hspace{0.2cm}
    \subfloat[\label{fig:125uepar}
    ]{\includegraphics[height=.12\textwidth]{imagesChap7/125to3uepar.pdf}}
    \hspace{0.6cm}
    \subfloat[\label{fig:125bch7par}
    ]{\includegraphics[height=.12\textwidth]{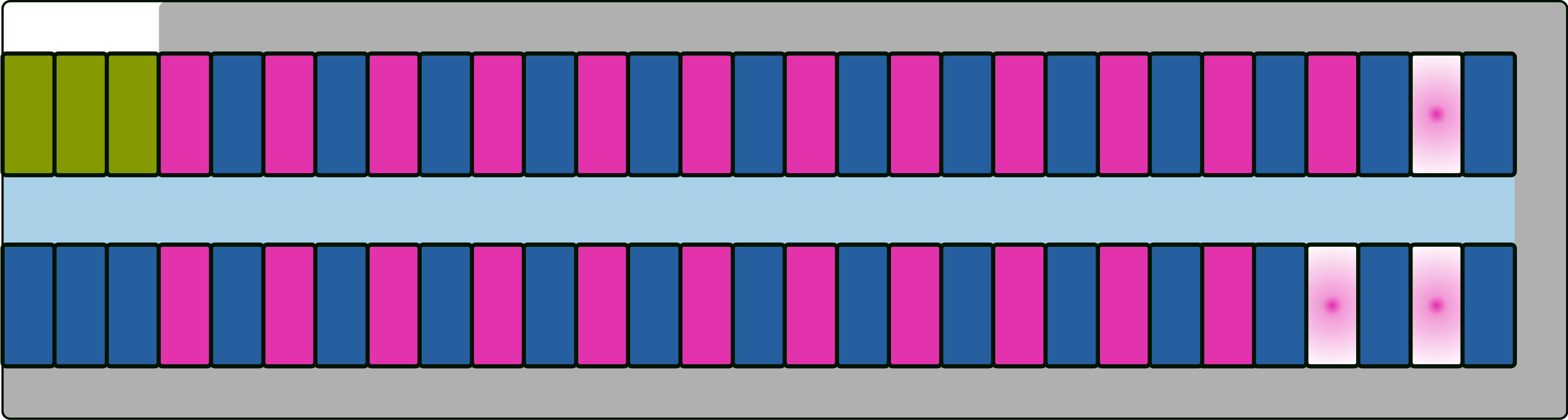}}
    \hspace{0.2cm}
    \subfloat[\label{fig:125bch9par} ]{\includegraphics[height=.12\textwidth]{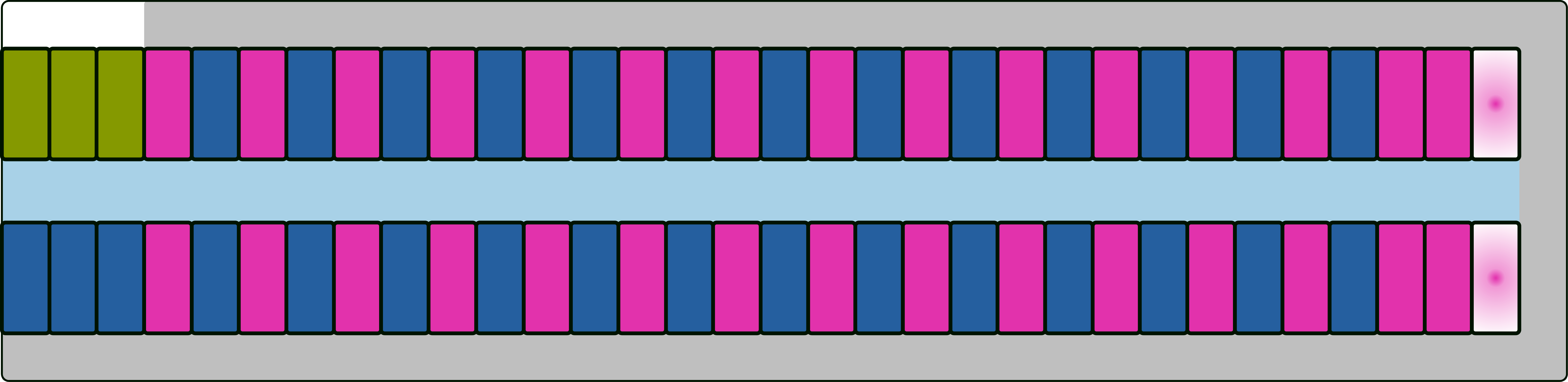}}
    \caption{Layouts of logical qubits for lattice-surgery-based $125$-to-$3$ magic state distillation. Cell color legend in caption of \cref{fig:15to1layouts}.~(a) Layout for a distillation tile without temporally encoded lattice surgery measurements, using an auto-corrected non-Clifford gate gadget. The blue routing space allows Pauli $X$-type measurements as there is access to the $X$ boundaries of all the cells. The gray routing region allows access to a $\ket 0$ ancilla with $Y$ boundary access. This region contains hardware to execute a non-Clifford gate gadget.~(b) Layout for a distillation tile that performs Pauli measurements two at a time, without temporal encoding. The long blue routing space performs one set of measurement with the auto-corrected non-Clifford gadget hardware at the left, and the large grey routing space uses the gadget hardware on the right.~(c) Layout for a distillation tile performing temporally encoded lattice surgery with the $[127,106,7]$ BCH code. Non-Clifford gates are performed two at a time, using separate routing spaces shown in gray and blue.~(d) Layout for a TELS-assisted distillation tile using the $[127,99,9]$ BCH code, performing non-Clifford gates two at a time.}
    \label{fig:125to1layouts}
\end{figure}

\subsection{116-to-12 distillation}
\label{sec:116to12}

The $116$-to-$12$ magic state distillation protocol is obtained from a triorthogonal CSS quantum $\llbracket 116,12,4\rrbracket$ code. This code is constructed by puncturing the $[128,29,32]$ Reed-Muller code at a specific set of $12$ locations as shown in Ref.~\cite{Haah18}. As a result, the quantum code will contain $12$ logical qubits, $87$ $X$-type stabilizers and $17$ $Z$-type stabilizers. After applying the circuit transformation from a gate-based model to the PBC model~\cite{Litinski19}, we are left with a sequence of $99$ commuting Pauli measurements on $29$ logical qubits, twelve of which will finally become the distilled magic states. These $99$ measurements form a size-$99$ PP set.

In this paper, we will consider using this protocol in a regime where the physical error rate is $p=10^{-4}$ and our target is to distill magic states with logical error probability at most $\delta^{(M)}=10^{-15}$. Using the noise model described in \cref{subsec:MSDMSDinjection} for the injection of magic states, we apply the analysis in Ref.~\cite{Litinski19magic} to determine the logical failure probability per output magic state for one round of a $116$-to-$12$ distillation scheme,
\begin{align}
    p_L^{(M)}  = & \frac{1}{12} 495 \Big ( (\epsilon_{\text{L,Z}})^4 + \frac{1}{2} 8(\epsilon_{\text{L,Z}})^3 \epsilon_{\text{L,X}} \nonumber \\
    & \quad + \frac{1}{4} 24 (\epsilon_{\text{L,Z}})^2 (\epsilon_{\text{L,X}})^2
    + \frac{1}{8} 32(\epsilon_{\text{L,Z}}) (\epsilon_{\text{L,X}})^3 \nonumber \\
    &\quad + \frac{1}{16} 16 (\epsilon_{\text{L,X}})^4
     \Big ) \nonumber \\
    = &\frac{495(1+\eta)^4}{972 \eta^4} p^4 \: .
\end{align}
For $p=10^{-4}$ and $\eta=100$, the probability that the distillation succeeds is $1 - p_D^{(M)} = (1-\epsilon_L)^{116} = 0.9961$ and $p_L^{(M)} = 5.3 \times 10^{-17}$. As this is sufficiently below $\delta^{(M)}$, the lattice surgery measurements used to execute the distillation protocol must be modeled with measurement distance large enough to allow for distilled magic states of logical error rate at most $\delta{(M)}$. Using the procedure in \cref{app:algoSTcosts}, we determined that the minimum spacelike distances are $d_x=9$ and $d_z=15$. 

We show two layouts for distillation tiles that do not use TELS in \cref{fig:116ue} and \cref{fig:116uepar}. These layouts perform distillation in the Pauli frame as described in Ref.~\cite{Litinski19} using auto-corrected non-Clifford gadgets (\cref{fig:PauliAutocorr}). On the layout of \cref{fig:116ue}, $99$ Pauli measurements are performed, each with measurement distance $d_m = 13$. On the layout of \cref{fig:125uepar}, Pauli measurements can be performed two at a time. Hence the time required is only the time for $50$ sequential lattice surgery measurements with measurement distance $d_m = 13$. 

In \cref{fig:116zett5par}, we show a layout for a distillation tile that performs TELS with a classical $[129,114,6]$ Zetterberg code. Note that since there are two disjoint routing spaces (in blue and grey), the set of $n=129$ Pauli measurements can be performed in the time required for $65$ sequential measurements. To obtain the time cost shown in \cref{tab:STcosts}, we used $d_m = 3$ and $c=0$. Similarly, in \cref{fig:116bch9par}, we show a layout for a $116$-to-$12$ distillation tile that performs TELS with a classical $[127,99,9]$ BCH code. Here, we used the parameters $d_m = 3$ and $c=2$ to obtain the time cost shown in \cref{tab:STcosts}. As discussed in \cref{subsec:MSDMSDCliff}, we must also execute a second PP set, now of maximum size $17$, due to the conditional Clifford corrections. The results of these measurements are used to detect if there are errors in the final distilled magic states. These Pauli measurements form a PP set of maximum size $17$, and in the worst case, the TELS protocol used is a $[43,17,7]$ BCH code with lattice surgery measurement distance $d_m = 3$ and $c=2$.

The space requirements of all the above layouts are described as functions of $d_x$ and $d_z$ in \cref{app:constants}. Additional constants in \cref{app:constants} can be used with the procedure of \cref{app:algoSTcosts} to determine all the minimum distances (spacelike and timelike) for the distillation protocols.

\begin{figure}
    \centering
    \subfloat[\label{fig:116ue}
    ]{\includegraphics[width=.332\textwidth]{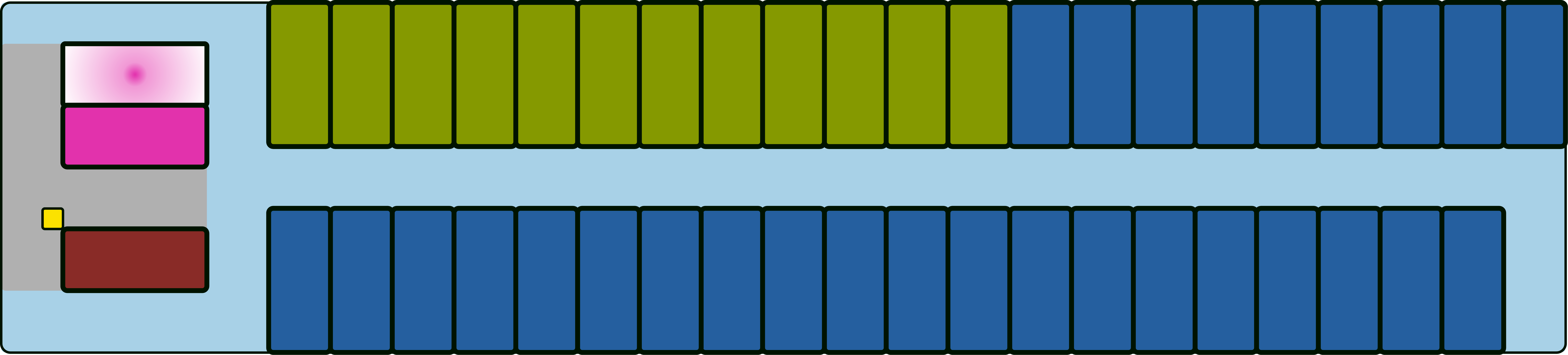}}
    \hspace{0.2cm}
    \subfloat[\label{fig:116uepar}
     ]{\includegraphics[width=.392\textwidth]{imagesChap7/116to12uepar.pdf}}
    \hspace{0.6cm}
    \subfloat[\label{fig:116zett5par}
    ]{\includegraphics[width=.41\textwidth]{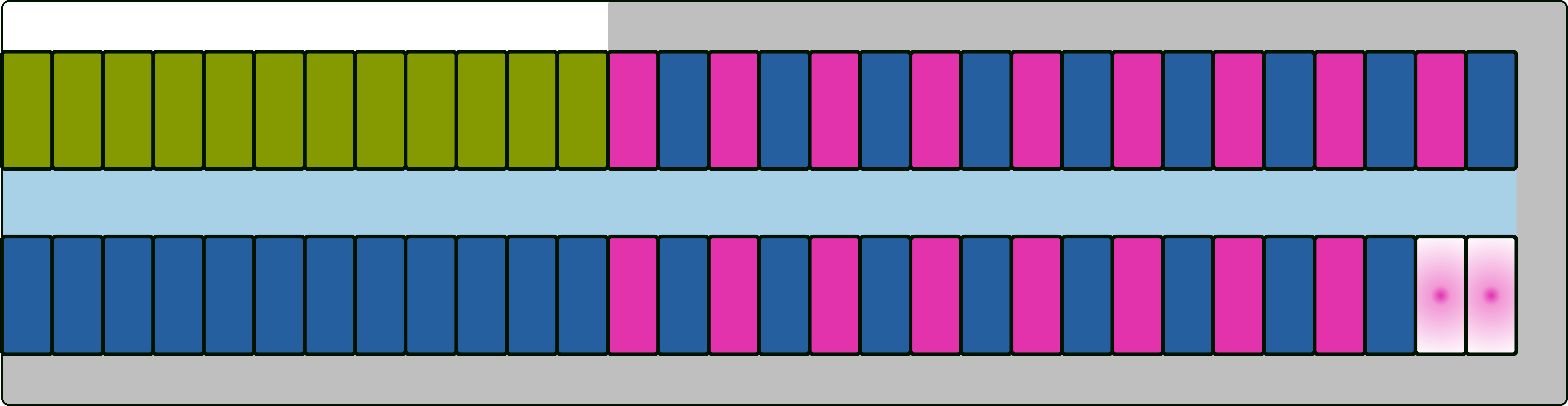}}
    \hspace{0.2cm}
    \subfloat[\label{fig:116bch9par} ]{\includegraphics[width=.5\textwidth]{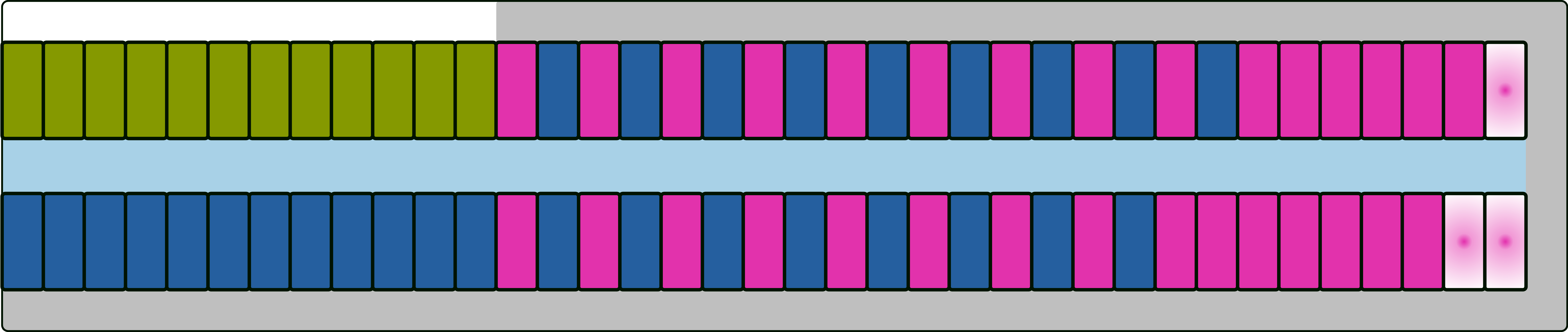}}
    \caption{Layouts of logical qubits for lattice-surgery-based $116$-to-$12$ magic state distillation. Cell color legend in caption of \cref{fig:15to1layouts}.~(a) Layout for a distillation tile without temporally encoded lattice surgery measurements, using an auto-corrected non-Clifford gate gadget.~(b) Layout for a distillation tile that performs Pauli measurements two at a time, without temporal encoding.~(c) Layout for a distillation tile performing temporally encoded lattice surgery with the $[129,114,6]$ Zetterberg code. Non-Clifford gates are performed two at a time, using separate routing spaces shown in gray and blue.~(d) Layout for a TELS-assisted distillation tile using the $[127,99,9]$ BCH code, performing non-Clifford gates two at a time.}
    \label{fig:116to12layouts}
\end{figure}

\subsection{114-to-14 distillation}
\label{sec:114to4}

The $114$-to-$14$ magic state distillation protocol is obtained from a triorthogonal CSS quantum $\llbracket 114,14,3\rrbracket$ code. This code is constructed by puncturing the $[128,29,32]$ Reed-Muller code at a specific set of $14$ locations as shown in Ref.~\cite{Haah18}. As a result, the quantum code will contain $14$ logical qubits, $85$ $X$-type stabilizers and $15$ $Z$-type stabilizers. After applying the circuit transformation from a gate-based model to the PBC model~\cite{Litinski19}, we are left with a sequence of $99$ commuting Pauli measurements on $29$ logical qubits, fourteen of which will finally become the distilled magic states. These $99$ measurements form a size-$99$ PP set.

In this paper, we will consider using this protocol in a regime where the physical error rate is $p=10^{-3}$ and our target is to distill magic states with logical error probability at most $\delta^{(M)}=10^{-10}$. Using the noise model described in \cref{subsec:MSDMSDinjection} for the injection of magic states, we apply the analysis in Ref.~\cite{Litinski19magic} to determine the logical failure probability per output magic state for one round of a $114$-to-$14$ distillation scheme,
\begin{align}
    p_L^{(M)}  = & \frac{1}{14} 30 \Big ( (\epsilon_{\text{L,Z}})^3 + \frac{1}{2} 6(\epsilon_{\text{L,Z}})^2 \epsilon_{\text{L,X}} \nonumber \\
    & \quad + \frac{1}{4} 12 (\epsilon_{\text{L,Z}}) (\epsilon_{\text{L,X}})^2
    + \frac{1}{8} 8 (\epsilon_{\text{L,X}})^3 \Big ) \nonumber \\
    = &\frac{30(1+\eta)^3}{378 \eta^3} p^3 \: .
\end{align}
For $p=10^{-3}$ and $\eta=100$, the probability that the distillation succeeds is $1 - p_D^{(M)} = (1-\epsilon_L)^{114} = 0.962$ and $p_L^{(M)} = 8.18 \times 10^{-11}$. Now, the lattice surgery measurements used to execute the distillation protocol must be modeled with measurement distance large enough to allow for distilled magic states of logical error rate at most $\delta{(M)}$. 

We show two layouts for distillation tiles that do not use TELS in \cref{fig:114ue} and \cref{fig:114uepar}. These layouts perform distillation in the Pauli frame as described in Ref.~\cite{Litinski19} using auto-corrected non-Clifford gadgets (\cref{fig:PauliAutocorr}). On the layout of \cref{fig:114ue}, $99$ Pauli measurements are performed, each with measurement distance $d_m = 15$. Using the procedure in \cref{app:algoSTcosts}, we determined that the minimum spacelike distances for this layout are $d_x=9$ and $d_z=19$. On the layout of \cref{fig:114uepar}, Pauli measurements can be performed two at a time. Hence the time required is only the time for $50$ sequential lattice surgery measurements with measurement distance $d_m = 15$. However in this case, when we calculated the minimum spacelike distances, the $Z$-distance could be dropped by two. This can be attributed to the fact that the distillation protocol finished in nearly half the time, and the probability of a logical $Z$-type error (see \cref{app:algoSTcosts}) scales linearly with time. Hence $d_x=9$ and $d_z=17$.

In \cref{fig:116zett5par}, we show a layout for a distillation tile that performs TELS with a classical $[129,114,6]$ Zetterberg code. Note that since there are two disjoint routing spaces (in blue and grey), the set of $n=129$ Pauli measurements can be performed in the time required for $65$ sequential measurements. To obtain the time cost shown in \cref{tab:STcosts}, we used $d_m = 5$ and $c=0$. Similarly, in \cref{fig:116bch9par}, we show a layout for a $114$-to-$14$ distillation tile that performs TELS with a classical $[127,106,7]$ BCH code. Here, we used the parameters $d_m = 5$ and $c=1$ to obtain the time cost shown in \cref{tab:STcosts}. As discussed in \cref{subsec:MSDMSDCliff}, we also execute a second PP set, now of maximum size $15$, due to the conditional Clifford corrections. The results of these measurements are used to detect if there are errors in the final distilled magic states. These Pauli measurements form a PP set of maximum size $15$, and in the worst case, the TELS protocol used is a $[31,16,7]$ BCH code with lattice surgery measurement distance $d_m = 5$ and $c=2$.

The space requirements of all the above layouts are described as functions of $d_x$ and $d_z$ in \cref{app:constants}. Additional constants in \cref{app:constants} can be used with the procedure of \cref{app:algoSTcosts} to determine all the minimum distances (spacelike and timelike) for the distillation protocols.

\begin{figure}
    \centering
    \subfloat[\label{fig:114ue} 
    ]{\includegraphics[width=.36\textwidth]{imagesChap7/114to14ue.pdf}}
    \hspace{0.2cm}
    \subfloat[\label{fig:114uepar}
    ]{\includegraphics[width=.41\textwidth]{imagesChap7/114to14uepar.pdf}}
    \hspace{0.8cm}
    \subfloat[\label{fig:114zett5par}
    ]{\includegraphics[width=.435\textwidth]{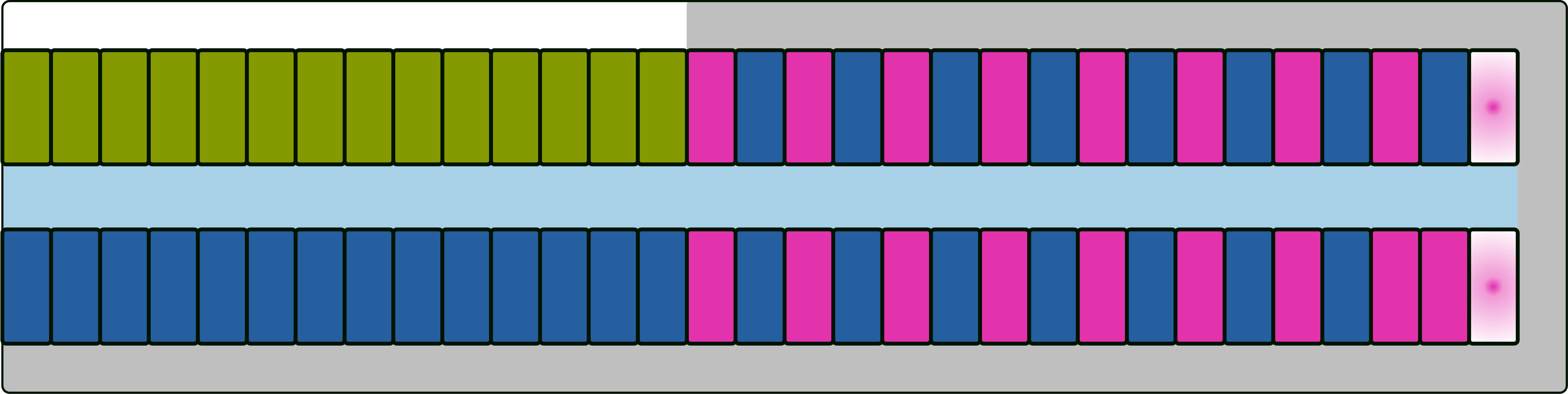}}
    \hspace{0.2cm}
    \subfloat[\label{fig:114bch7par} ]{\includegraphics[width=.48\textwidth]{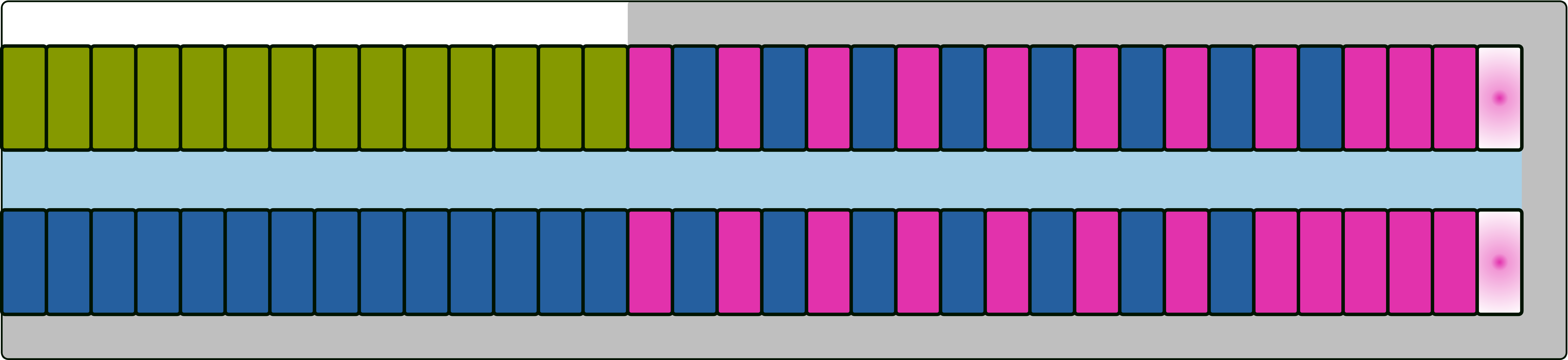}}
    \caption{Layouts of logical qubits for lattice-surgery-based $114$-to-$14$ magic state distillation with a $\llbracket 114,14,3\rrbracket$ quantum code. Cell color legend in caption of \cref{fig:15to1layouts}.~(a) Layout for a distillation tile without temporally encoded lattice surgery measurements, using an auto-corrected non-Clifford gate gadget.~(b) Layout for a distillation tile that performs Pauli measurements two at a time, without temporal encoding.~(c) Layout for a distillation tile performing temporally encoded lattice surgery with the $[129,114,6]$ Zetterberg code. Non-Clifford gates are performed two at a time, using separate routing spaces shown in gray and blue.~(d) Layout for a TELS-assisted distillation tile using the $[127,106,7]$ BCH code, performing non-Clifford gates two at a time.}
    \label{fig:114to14layouts}
\end{figure}

\end{document}